%% file: main.tex
\documentclass
[11pt,letterpaper]
{article}

\usepackage[utf8]{inputenc}
\usepackage{etex}
\usepackage{xspace,enumerate}
\usepackage[dvipsnames]{xcolor}
\usepackage[T1]{fontenc}
\usepackage[full]{textcomp}
\usepackage[american]{babel}
\usepackage{mathtools}
\usepackage{titling}
\usepackage{graphicx}
\usepackage{amsthm}
\newtheorem{theorem}{Theorem}[section]
\newtheorem*{theorem*}{Theorem}

\newtheorem{proposition}[theorem]{Proposition}
\newtheorem*{proposition*}{Proposition}
\newtheorem{lemma}[theorem]{Lemma}
\newtheorem*{lemma*}{Lemma}
\newtheorem{corollary}[theorem]{Corollary}
\newtheorem*{conjecture*}{Conjecture}
\newtheorem{fact}[theorem]{Fact}
\newtheorem*{fact*}{Fact}

\newtheorem*{hypothesis*}{Hypothesis}

\theoremstyle{definition}
\newtheorem{definition}[theorem]{Definition}
\newtheorem*{definition*}{Definition}
\theoremstyle{condition}
\newtheorem{condition}[theorem]{Condition}
\newtheorem*{condition*}{Condition}

\newtheorem{question}[theorem]{Question}

\newtheorem{algorithm}[theorem]{Algorithm}

\theoremstyle{remark}

\newtheorem*{claim*}{Claim}
\newtheorem{remark}[theorem]{Remark}
\newtheorem*{remark*}{Remark}

\newtheorem*{observation*}{Observation}

\usepackage[
letterpaper,
top=1.2in,
bottom=1.2in,
left=1in,
right=1in]{geometry}
\usepackage{newpxtext} 
\usepackage{textcomp} 
\usepackage[varg,bigdelims]{newpxmath}
\usepackage[scr=rsfso]{mathalfa}
\usepackage{bm} 
\let\mathbb\varmathbb
\usepackage{microtype}
\usepackage[
pagebackref,
colorlinks=true,
urlcolor=blue,
linkcolor=blue,
citecolor=OliveGreen,
]{hyperref}
\usepackage[capitalise,nameinlink]{cleveref}
\crefname{lemma}{Lemma}{Lemmas}
\crefname{fact}{Fact}{Facts}
\crefname{theorem}{Theorem}{Theorems}
\crefname{corollary}{Corollary}{Corollaries}
\crefname{claim}{Claim}{Claims}
\crefname{example}{Example}{Examples}
\crefname{algorithm}{Algorithm}{Algorithms}
\crefname{problem}{Problem}{Problems}
\crefname{definition}{Definition}{Definitions}
\usepackage{paralist}
\usepackage{turnstile}
\usepackage{mdframed}
\usepackage{tikz}
\usepackage{caption}
\DeclareCaptionType{Algorithm}
\usepackage{newfloat}

\usepackage{enumitem}

\newcommand{\Authornoteb}[2]{}
 \newcommand{\Authornoter}[2]{}

\newcommand{\Authornotecolored}[3]{}
\newcommand{\Authorcomment}[2]{}
\newcommand{\Authorfnote}[2]{}

\newcommand{\Pnote}{\Authornoter{P}}
\newcommand{\Anote}{\Authornoteb{A}}

\definecolor{forestgreen(traditional)}{rgb}{0.0, 0.27, 0.13}

\usepackage{boxedminipage}
\newcommand{\paren}[1]{(#1)}
\newcommand{\Paren}[1]{\left(#1\right)}

\newcommand{\abs}[1]{\lvert#1\rvert}

\newcommand{\set}[1]{\{#1\}}
\newcommand{\Set}[1]{\left\{#1\right\}}

\newcommand{\norm}[1]{\lVert#1\rVert}
\newcommand{\Norm}[1]{\left\lVert#1\right\rVert}

\newcommand{\iprod}[1]{\langle#1\rangle}
\newcommand{\Iprod}[1]{\left\langle#1\right\rangle}

\newcommand{\Esymb}{\mathbb{E}}
\newcommand{\Psymb}{\mathbb{P}}

\DeclareMathOperator*{\E}{\Esymb}

\DeclareMathOperator*{\ProbOp}{\Psymb}
\renewcommand{\Pr}{\ProbOp}

\newcommand{\from}{\colon}

\newcommand{\mper}{\,.}
\newcommand{\mcom}{\,,}
\newcommand\bdot\bullet

\DeclareMathOperator{\poly}{poly}

\DeclareMathOperator{\argmax}{argmax}

\DeclareMathOperator{\rank}{rank}

\newcommand{\Z}{\mathbb Z}
\newcommand{\N}{\mathbb N}
\newcommand{\R}{\mathbb R}

\newcommand{\cA}{\mathcal A}
\newcommand{\cB}{\mathcal B}

\newcommand{\cD}{\mathcal D}

\newcommand{\cL}{\mathcal L}

\newcommand{\cN}{\mathcal N}

\def\cov{\texttt{Cov}}

\newcommand{\wh}{\widehat}

\renewcommand{\leq}{\leqslant}
\renewcommand{\le}{\leqslant}
\renewcommand{\geq}{\geqslant}
\renewcommand{\ge}{\geqslant}
\let\epsilon=\varepsilon
\numberwithin{equation}{section}
\newcommand\MYcurrentlabel{xxx}
\newcommand{\MYstore}[2]{%
  \global\expandafter \def \csname MYMEMORY #1 \endcsname{#2}%
}
\newcommand{\MYload}[1]{%
  \csname MYMEMORY #1 \endcsname%
}
\newcommand{\MYnewlabel}[1]{%
  \renewcommand\MYcurrentlabel{#1}%
  \MYoldlabel{#1}%
}
\newcommand{\MYdummylabel}[1]{}
\newcommand{\torestate}[1]{%
  \let\MYoldlabel\label%
  \let\label\MYnewlabel%
  #1%
  \MYstore{\MYcurrentlabel}{#1}%
  \let\label\MYoldlabel%
}
\newcommand{\restatetheorem}[1]{%
  \let\MYoldlabel\label
  \let\label\MYdummylabel
  \begin{theorem*}[Restatement of \cref{#1}]
    \MYload{#1}
  \end{theorem*}
  \let\label\MYoldlabel
}
\newcommand{\restatelemma}[1]{%
  \let\MYoldlabel\label
  \let\label\MYdummylabel
  \begin{lemma*}[Restatement of \cref{#1}]
    \MYload{#1}
  \end{lemma*}
  \let\label\MYoldlabel
}
\newcommand{\restateprop}[1]{%
  \let\MYoldlabel\label
  \let\label\MYdummylabel
  \begin{proposition*}[Restatement of \cref{#1}]
    \MYload{#1}
  \end{proposition*}
  \let\label\MYoldlabel
}
\newcommand{\restatefact}[1]{%
  \let\MYoldlabel\label
  \let\label\MYdummylabel
  \begin{fact*}[Restatement of \prettyref{#1}]
    \MYload{#1}
  \end{fact*}
  \let\label\MYoldlabel
}
\newcommand{\restate}[1]{%
  \let\MYoldlabel\label
  \let\label\MYdummylabel
  \MYload{#1}
  \let\label\MYoldlabel
}

\newcommand{\eps}{\epsilon}

\allowdisplaybreaks
\newcommand*{\tr}{\mathrm{tr}}

\newcommand{\1}{\bm{1}}

\newcommand{\pE}{\tilde{\mathbb{E}}}

\def\expecf#1#2{{\bf \mathbb{E}}_{#1}\left[ #2 \right]}
\def\var#1{\mbox{\bf Var}[ #1 ]}

\def\dim#1{\mathrm{dim} (#1)}

\def\abs#1{\left|#1  \right|}

\def\norm#1{\left\| #1 \right\|}

\def \dtv{d_{\mathsf{TV}}}

\def\expecf#1#2{{\bf \mathbb{E}}_{#1}\left[ #2 \right]}
\def\var#1{\mbox{\bf Var}[ #1 ]}

\def\tzeta{\tilde{\zeta}}

\newcommand{\eqdef}{\stackrel{{\mathrm {\footnotesize def}}}{=}}

\newcommand{\sym}{\mathrm{Sym}}

\newcommand{\Sigmaa}{S}

  \newcommand{\inote}[1]{}

\usepackage{dsfont}

\def\colorful{0}

\usepackage{nicefrac}
\def\nnewcolor{1}
\ifnum\nnewcolor=1

\fi
\ifnum\nnewcolor=0

\fi

\ifnum\colorful=1
\newcommand{\new}[1]{{\color{red} #1}}

\newcommand{\newblue}[1]{{\color{blue} #1}}
\else
\newcommand{\new}[1]{{#1}}

\newcommand{\newblue}[1]{{#1}}

\fi

\newcommand{\nnnew}[1]{{#1}}

\newcommand{\calA}{\mathcal A}

\newcommand{\calC}{\mathcal C}
\newcommand{\calD}{\mathcal D}

\newcommand{\calF}{\mathcal F}

\newcommand{\calL}{\mathcal L}
\newcommand{\calM}{\mathcal M}
\newcommand{\calN}{\mathcal N}
\newcommand{\calO}{\mathcal O}

\newcommand{\calU}{\mathcal U}

\newcommand{\calX}{\mathcal X}

\def\expecf#1#2{ \mathop{\mathbb{E}}_{#1}\left[ #2 \right]}

\newcommand{\bigO}[1]{\mathcal{O}\hspace{-0.1cm}\left(#1\right)}
\newcommand{\bigOk}[1]{\mathcal{O}_k\hspace{-0.1cm}\left(#1\right)}

\title{ Robustly Learning Mixtures of $k$ Arbitrary Gaussians}

\author{
   Ainesh Bakshi \\ CMU \\ abakshi@cs.cmu.edu\\\\
   Daniel M. Kane \\ UCSD \\ dakane@ucsd.edu 
   \and
   Ilias Diakonikolas \\ UW Madison \\ ilias@cs.wisc.edu\\\\
   Pravesh K. Kothari \\ CMU \\praveshk@cs.cmu.edu 
   \and
   He Jia \\ Georgia Tech \\hjia36@gatech.edu \\\\
   Santosh S. Vempala  \\ Georgia Tech \\ vempala@gatech.edu
 }

\begin{document}

\pagestyle{empty}


\maketitle
\thispagestyle{empty} 

\input{abstract}

\clearpage


{\small
  \microtypesetup{protrusion=false}
  \setcounter{tocdepth}{2}
  \tableofcontents{}

  \microtypesetup{protrusion=true}
}

\clearpage

\pagestyle{plain}
\setcounter{page}{1}


\input{intro}

\input{related-work}

\input{techniques}

\input{prelims}


\input{tensor_decomposition}



\input{partial-clustering}
\input{recursion-lemma}

\input{full-algo-analysis}

\input{partial-clustering-upgrade}

\input{poly-list-size}

\input{parameter_recovery}

\section*{Acknowledgments}
We thank an anonymous reviewer for pointing out an issue with a technical statement 
(Fact 2.35 in the previous arXiv version, replaced by Lemmas~\ref{lem:var-zero-mean-gaussians} 
and~\ref{fact:mean-variance-subgaussian} in the current version) 
that claimed a bound on the variance of a more general class of distributions 
than what is needed in our approach.  

A.B. was supported by the Office of Naval Research (ONR)
grant N00014-18-1-2562, and the National Science Foundation (NSF) Grant No. CCF-1815840.
I.D. was supported by NSF Award CCF-1652862 (CAREER), a Sloan Research Fellowship, and a DARPA 
Learning with Less Labels (LwLL) grant. H.J. and S.S.V. were supported in part by NSF awards
AF-1909756 and AF-2007443. D.M.K. was supported by NSF Award CCF-1553288 (CAREER) and 
a Sloan Research Fellowship. Part of this work was done
while A.B. was visiting the Simons Institute for the Theory of Computing.

\clearpage



\phantomsection
  \addcontentsline{toc}{section}{References}
  \bibliographystyle{amsalpha}
  \bibliography{bib/allrefs,bib/ref,bib/custom,bib/dblp,bib/custom2,bib/mathreview}  


\appendix
 

\input{appendix}

\input{bitcomplexity}

\end{document}

%% file: abstract.tex

\begin{abstract}
We give a polynomial-time algorithm for the problem of robustly estimating a mixture of $k$ arbitrary Gaussians in $\R^d$, for any fixed $k$, in the presence of a constant fraction of arbitrary corruptions. 
This resolves the main open problem in several previous works on algorithmic robust statistics, 
which addressed the special cases of robustly estimating (a) a single Gaussian, (b) a mixture of TV-distance separated Gaussians, and (c) a uniform mixture of two Gaussians. Our main tools are an efficient \emph{partial clustering} algorithm that relies on the sum-of-squares method, and a novel \emph{tensor decomposition} algorithm 
that allows errors in both Frobenius norm and low-rank terms. 
\end{abstract}

%% file: intro.tex

\section{Introduction} \label{sec:intro}

\subsection{Background and Motivation} \label{ssec:background}

Given a collection of observations and a class of models, 
the objective of a typical learning algorithm is to find the model in the class that best fits the data.
The classical assumption is that the input data are i.i.d.\ samples generated by a statistical model 
in the given class. This is a simplifying assumption that is, at best, only approximately valid, as 
real datasets are typically exposed to some source of systematic noise. 
Robust statistics~\cite{HampelEtalBook86, Huber09} challenges this assumption by focusing
on the design of {\em outlier-robust}  estimators  --- algorithms 
that can tolerate a {\em constant fraction} of corrupted datapoints, independent of the dimension. 
Despite significant effort over several decades starting with important early works of Tukey and Huber in the 60s, 
even for the most basic high-dimensional estimation tasks,
all known computationally efficient estimators were until fairly recently highly sensitive to outliers.

This state of affairs changed with two independent works from the TCS community~\cite{DKKLMS16, LaiRV16},
which gave the first computationally efficient and outlier-robust learning algorithms for a range of ``simple''
high-dimensional probabilistic models. In particular, these works developed efficient robust estimators
for a single high-dimensional Gaussian distribution with unknown mean and covariance.
Since these initial algorithmic works~\cite{DKKLMS16, LaiRV16}, 
we have witnessed substantial research progress on algorithmic aspects of 
robust high-dimensional estimation by several communities of researchers, including TCS, machine learning,
and mathematical statistics. See Section~\ref{ssec:related} for an overview of the prior work most relevant
to the results of this paper. The reader is referred to~\cite{DK20-survey} for a  recent survey on the topic. 

One of the main original motivations for the development of algorithmic robust statistics within the TCS community 
was the problem of learning high-dimensional Gaussian mixture models.
A {\em Gaussian mixture model (GMM)} is a convex combination of Gaussian distributions, i.e.,
a distribution on $\R^d$ of the form $\calM  = \sum_{i=1}^k w_i \mathcal{N}(\mu_i, \Sigma_i)$, where
the weights $w_i$, mean vectors $\mu_i$, and covariance matrices $\Sigma_i$ are unknown.
GMMs are {\em the} most extensively studied latent variable model in the statistics and machine learning literatures, 
starting with the pioneering work of Karl Pearson in 1894~\cite{Pearson:94}, 
which introduced the method of moments in this context.

In the absence of outliers, a long line of work initiated by Dasgupta~\cite{Dasgupta:99, AroraKannan:01, VempalaWang:02, AchlioptasMcSherry:05, KSV08, BV:08} gave efficient clustering algorithms for GMMs under various separation assumptions. 
Subsequently, efficient learning algorithms were obtained~\cite{KMV:10, MoitraValiant:10, BelkinSinha:10, HardtP15} 
under minimal information-theoretic conditions. Specifically, Moitra and Valiant~\cite{MoitraValiant:10} and Belkin and Sinha~\cite{BelkinSinha:10} designed 
the first polynomial-time learning algorithms for arbitrary Gaussian mixtures with any fixed number of components.
These works qualitatively characterized the complexity of this fundamental learning problem in the noiseless setting.
Alas, all aforementioned algorithms are very fragile in the presence of corrupted data.
Specifically, a {\em single} outlier can completely compromise
their performance. 

Developing efficient learning algorithms for high-dimensional GMMs in the more realistic {\em outlier-robust} setting
--- the focus of the current paper --- has turned out to be significantly more challenging. 
This was both one of the original motivations and the main open problem 
in the initial robust statistics works~\cite{DKKLMS16, LaiRV16}. 
We note that~\cite{DKKLMS16} developed a robust density estimation algorithm 
for mixtures of {\em spherical} Gaussians --- a very special case of our problem
where the covariance of each component is a multiple of the identity --- 
and highlighted a number of key technical obstacles that need to be overcome in order to handle the general case.
Since then, a number of works have made algorithmic progress on important special cases of the general problem.
These include faster robust clustering for the spherical case under minimal separation 
conditions~\cite{HopkinsL18, KothariSS18, DiakonikolasKS18-mixtures}, 
robust clustering for separated (and potentially non-spherical) Gaussian mixtures~\cite{BK20, DHKK20}, 
and robustly learning {\em uniform} mixtures of two arbitrary Gaussian components~\cite{Kane20}.

This progress notwithstanding, the algorithmic task of robustly learning a mixture of a constant number
(or even two) arbitrary Gaussians (with arbitrary weights) 
has remained a central open problem in this field, as highlighted recently~\cite{DVW19-vignette}.

This discussion motivates the following question, whose resolution is the main
result of this work:

\begin{question} \label{q:robust-gmm}
Is there a $\poly(d, 1/\eps)$-time robust GMM learning algorithm, in the presence of an $\eps$-fraction of outliers, 
that has a dimension-independent error guarantee, for an arbitrary mixture 
of any constant number of arbitrary Gaussians on $\R^d$?
\end{question}

\subsection{Our Results} \label{ssec:results}

To formally state our main result, we define
the model of robustness we study.
We focus on the following standard data corruption model 
that generalizes Huber's contamination model~\cite{Huber64}. 

\begin{definition}[Total Variation Contamination Model] \label{def:adv}
Given a parameter $0< \epsilon < 1/2$ and a class of distributions $\mathcal{F}$ on $\mathbb{R}^d$, 
the \emph{adversary} operates as follows: The algorithm specifies the number of samples $n$.
The adversary knows the true target distribution $X \in \mathcal{F}$ and selects a distribution $F$ such that
$\dtv(F, X) \leq \eps$. Then $n$ i.i.d.\ samples are drawn from $F$ and are given as input to the algorithm.
\end{definition}

Intuitively, the parameter $\epsilon$ in Definition~\ref{def:adv} 
quantifies the power of the adversary. The total variation contamination 
model is strictly stronger than Huber's contamination model. 
Recall that in Huber's model~\cite{Huber64}, the adversary generates 
samples from a mixture distribution $F$ of the form $F = (1-\epsilon) X + \epsilon N$, 
where $X$ is the unknown target distribution and $N$ is an adversarially chosen noise distribution. 
That is, in Huber's model the adversary is only allowed to add outliers.

\begin{remark}
The {\em strong contamination model}~\cite{DKKLMS16} is a strengthening of 
the total variation contamination, where an adversary can see 
the clean samples and then arbitrarily replace an $\eps$-fraction of these points 
to obtain an $\eps$-corrupted set of samples.
Our robust learning algorithm succeeds in this strong contamination model, with the additional 
requirement that we can obtain two sets of independent $\eps$-corrupted samples from the unknown mixture. 
\end{remark}

In the context of robustly learning GMMs, we want to design an efficient algorithm 
with the following performance: Given a sufficiently large set of 
samples from a distribution that is $\eps$-close in total variation distance to an unknown GMM $\calM$ on $\R^d$, 
the algorithm outputs a hypothesis GMM $\widehat{\calM}$ such that with high probability 
the total variation distance $\dtv(\widehat{\calM}, \calM)$ is small. Specifically, we want 
$\dtv(\widehat{\calM}, \calM)$ to be only a function of $\eps$ and independent 
of the underlying dimension $d$.


The main result of this paper is the following:
\inote{Cite the formal version of the theorem that appears later in the main body.} \Pnote{done.}

\begin{theorem}[Main Result, See Corollary~\ref{thm:robust-GMM-arbitrary-poly-time}]\label{thm:main-informal}
There is an algorithm with the following behavior:
Given $\eps>0$ and a multiset of $n = d^{O(k)} \poly(\log(1/\eps))$ samples from a distribution $F$ on $\R^d$ such that
$\dtv(F, \calM) \leq \eps$, for an unknown target $k$-GMM $\calM = \sum_{i=1}^k w_i \mathcal{N}(\mu_i, \Sigma_i)$, 
the algorithm runs in time \new{$\poly(n) \poly_k(1/\eps)$} and outputs a $k$-GMM hypothesis 
$\widehat{\calM}  = \sum_{i=1}^k \widehat{w}_i \mathcal{N}(\widehat{\mu}_i, \widehat{\Sigma}_i)$ 
such that with high probability we have that 
$\dtv(\widehat{\calM}, \calM ) \leq g(\eps, k)$. Here $g: \R_+ \times \Z_+ \to \R_+$ is a function such that
$\lim_{\eps \rightarrow 0} g(\eps, k) = 0$. 
\end{theorem}

Theorem~\ref{thm:main-informal} gives the first polynomial-time {\em robust proper learning} algorithm, 
with dimension-independent error guarantee, for {\em arbitrary} $k$-GMMs, for any fixed $k$. 
This is the first polynomial-time algorithm for this problem, even for $k=2$. 

\new{Since the dissemination of~\cite{BDK+20:arxiv}, we have improved
the guarantees of Theorem~\ref{thm:main-informal} in two ways.
First, by refining the guarantees of two of our components in the algorithm, we are able to obtain a robust proper learner
with error guarantee of $\poly_k(\eps)$ (Theorem~\ref{thm:poly-proper-inf}). Second, we show that
the same algorithm also achieves the stronger parameter estimation guarantee (Theorem~\ref{thm:param-inf}).
These refinements build heavily build on the techniques in~\cite{BDK+20:arxiv} and represent
work subsequent to ~\cite{BDK+20:arxiv, LM20}. We describe them in detail in the following.}

\new{\paragraph{Independent and Concurrent Work}
Theorem~\ref{thm:main-informal} is the main result obtained in \cite{BDK+20:arxiv}.
In independent and concurrent work to~\cite{BDK+20:arxiv}, 
Liu and Moitra~\cite{LM20} obtained a closely related result under stronger assumptions on the input mixture --- specifically, 
$1/f(k)$ lower bound on the component weights, $\poly(\eps/d)$ lower bound and $\poly(d/\eps)$ upper bound on the 
eigenvalues of each component covariance, and a $\geq g(k)$ TV-distance separation between every pair of components.
Under these assumptions, they obtain a formally stronger robust \emph{parameter estimation} guarantee in time that grows exponentially in $f(k), g(k)$ with error guarantees that decay exponentially in $f(k), g(k)$.
We discuss this work and its connection to 
Theorem~\ref{thm:main-informal} in more detail in Section~\ref{ssec:conc}.}

\paragraph{Discussion}
Before proceeding, we make a few important remarks about Theorem~\ref{thm:main-informal}.
\begin{enumerate}[leftmargin = *]
\item \textit{Sample Complexity and Runtime}: 
Our algorithm succeeds whenever the sample size $n$ satisfies $n \geq n_0 = d^{O(k)}/\poly(\epsilon)$. 
The running time of our algorithm is \new{$\poly(n) \poly_k(1/\eps)$}. 
Statistical query lower bounds~\cite{DKS17-sq} suggest that $d^{\Omega(k)}$ samples are necessary for \new{efficiently} 
learning GMMs, even for approximation to constant accuracy in the simpler setting without outliers  
and under the more restrictive \emph{clustering} setting (where components are pairwise well-separated in total variation distance). \new{This provides some evidence that the sample-time tradeoff achieved
by Theorem~\ref{thm:main-informal}} is qualitatively optimal (within absolute constant factors in the exponent). 
\new{We note that the algorithm establishing Theorem~\ref{thm:main-informal} works in the standard bit-complexity model 
of computation 
and its running time is polynomial in the bit-complexity of the input parameters. }
We discuss the numerical accuracy required to implement our algorithm in Appendix~\ref{app:bit-comp}.

In the noiseless case, the first polynomial-time learning algorithm for $k$-GMMs on $\R^d$ 
was given in~\cite{MoitraValiant:10, BelkinSinha:10}. 
In particular, the sample complexity and running time of the~\cite{MoitraValiant:10} algorithm is $(d/\eps)^{q(k)}$, 
for some function $q(k) = k^{\Omega(k)}$. We observe that our running time and sample complexity are exponentially better than the guarantees for the noiseless case in~\cite{MoitraValiant:10, BelkinSinha:10}. \new{Moreover}, 
as we explain in Section~\ref{ssec:techniques}, the \cite{MoitraValiant:10, BelkinSinha:10} algorithms 
are very sensitive to outliers and an entirely new approach is required to obtain an efficient robust learning algorithm. 

\item \textit{Handling Arbitrary Weights}: The algorithm of Theorem~\ref{thm:main-informal}
succeeds {\em without any assumptions} on the weights of the mixture components. 
\new{We emphasize that this is an important feature and not a technicality. Prior work~\cite{BK20, DHKK20, Kane20}, 
as well as the concurrent work~\cite{LM20}, cannot handle the case of general weights --- even for the case of $k=2$
components. In fact, for the special case of uniform weights, we give a simpler algorithm 
for robustly learning GMMs (presented in Theorem~\ref{thm:robust-GMM-equiweigthed}). 
This algorithm naturally generalizes to give a sample complexity and running time that 
grows {\em exponentially} in $1/w_{\min}$, where $w_{\min}$ is the minimum weight of any component in the mixture. 
Handling the general case (i.e., obtaining a fully polynomial-time algorithm, not incurring an exponential cost in $1/w_{\min}$)
requires genuinely new algorithmic ideas and is one of the key technical innovations 
in the proof of Theorem~\ref{thm:main-informal}.}

\item \textit{Handling Arbitrary Covariances}: 
\new{The algorithm of Theorem~\ref{thm:main-informal} does not require
assumptions on the variances of the component covariances, 
modulo basic limitations posed by numerical computation issues}. 
Specifically, our algorithm works even if some of the component covariances are rank-deficient 
(i.e., have directions of $0$ variance) with running time scaling polynomially in the bit-complexity 
of the unknown component means and covariances. Such a dependence on the bit complexity of the input parameters is unavoidable -- there exist\footnote{For e.g., for unit vector $v=(1/\sqrt{3},1/\sqrt{3},1/\sqrt{3},0,0,\ldots,0)$ and for every choice of rational covariance $\Sigma$, the total variation distance between $\cN(0,I-vv^{\top})$ and $\cN(0,\Sigma)$ is the maximum possible $1$.} 
examples of rank-deficient covariances with irrational entries 
such that the total variation distance between the corresponding Gaussian 
and every Gaussian with covariance matrix of rational entries is the maximum possible value of one. 

\item \textit{Error Guarantee}: The function $g$ quantifying the final error guarantee of our basic algorithm 
is $g(\eps, k) = 1/(\log(1/\eps))^{C_k}$, for some function $C_k$ that goes to $0$ when $k$ increases. 
Importantly, for any fixed $k$, the final error guarantee of our algorithm depends only on $\eps$, tends to $0$ as $\epsilon \rightarrow 0$ and is independent of the dimension $d$. \new{In Theorem~\ref{thm:poly-proper-inf}, 
we show that, by modifying our algorithm, we can obtain improved error -- 
scaling as a fixed polynomial in $\epsilon$. This turns out to be quantitatively close to 
best possible for any robust proper learning algorithm.}

\end{enumerate}

Our work is most closely related to the recent paper by Kane~\cite{Kane20}, 
which gave a polynomial-time robust learning algorithm for the {\em uniform} $k=2$ case, 
i.e., the case of two {\em equal weight} components, 
and the polynomial \-time algorithms~\cite{BK20, DHKK20} 
for the problem under the (strong) assumption that the component Gaussians 
are pairwise well-separated in total variation distance. 


Our algorithm builds on the ideas in the works~\cite{BK20, DHKK20} that gave efficient 
clustering algorithms for any fixed number $k$ of components,  under the crucial assumption 
that the components have pairwise total variation distance close to $1$. 
In this case, the above works actually succeed in efficiently \emph{clustering} the input sample into $k$ groups, 
such that each group contains the samples generated from one of the Gaussians, 
up to some small misclassification error. In contrast, the main challenge in this work is the information-theoretic impossibility of clustering in our setting where there are no separation assumptions. As we will explain in the proceeding
discussion, while we draw ideas from~\cite{BK20, DHKK20, Kane20}, 
a number of significant conceptual and technical challenges need to be overcome 
in the non-clusterable setting.

\paragraph{Improvements to Theorem~\ref{thm:main-informal}.} We now describe refinements of our main theorem.

\paragraph{Improving Error to a Fixed Polynomial in $\epsilon$} 
It turns out that the inverse poly-logarithmic accuracy (in $1/\epsilon$) 
\new{in the final error guarantee of} Theorem~\ref{thm:main-informal} can be traced 
to \new{an exhaustive search subroutine in our novel tensor decomposition subroutine and probability of success of our rounding algorithm in our partial clustering routine. Via natural (and conceptually simple) quantitative improvements to 
these two ingredients, we obtain an algorithm achieving the qualitatively 
nearly best possible error of $\poly_k(\epsilon)$. 
Specifically, we show:}


\begin{theorem}[Robustly Learning $k$-Mixtures with $\poly(\epsilon)$-error, 
Informal, see Corollary \ref{thm:robust-GMM-arbitrary-poly-eps}] \label{thm:poly-proper-inf}
There is an algorithm with the following behavior:
Given $\eps>0$ and a multiset of $n = d^{O(k)} \poly_k(1/\epsilon)$ samples from a distribution $F$ 
on $\R^d$ such that $\dtv(F, \calM) \leq \eps$, for an unknown target $k$-GMM $\calM = \sum_{i=1}^k w_i \mathcal{N}(\mu_i, \Sigma_i)$, 
the algorithm runs in time  $\poly(n)\poly_k(1/\eps)$ and outputs a $k$-GMM hypothesis 
$\widehat{\calM}  = \sum_{i=1}^k \widehat{w}_i \mathcal{N}(\widehat{\mu}_i, \widehat{\Sigma}_i)$ 
such that with high probability we have that 
$\dtv(\widehat{\calM}, \calM ) \leq \bigO{\eps^{c_k}}$, 
where $c_k$ depends only on $k$.   
 \end{theorem} 

\paragraph{Robust Parameter Recovery} 
Finally, we show that the \emph{same} algorithm as in Theorems~\ref{thm:main-informal} 
and \ref{thm:robust-GMM-arbitrary-poly-eps} actually implies that 
the recovered mixture of Gaussians is close in \emph{parameter} 
distance to the unknown target mixture. 
Such parameter estimation results are usually stated under the assumption 
that every pair of components of the unknown mixture are separated in total variation distance. 
In this work, we provide a stronger version of this parameter estimation guarantee. 

More specifically, in the theorem below, we prove that whenever the components of the input mixture 
can be clustered together into some groups such that all mixtures in a group are close (and thus, indistinguishable), 
there exists a similar clustering of the output mixture such that all parameters (weight, mean, and covariances) 
of each cluster are close within $\poly_k(\epsilon)$ in total variation distance. In particular this means that for each significant component of the input mixture, there is a component of the output mixture with very close parameters.

We note that~\cite{LM20} gave a parameter estimation guarantee (under additional assumptions on the mixture weights and component variances) whenever every pair of components in the unknown mixture are \new{$f(k)$-far} in total variation distance, 
\new{where $f$ can be any function of $k$, but the choice of $f$ affects the exponent in the running time and error guarantee of the~\cite{LM20} algorithm.)} 

\new{By strengthening one of the structural results in their argument, we establish the following:}



\begin{theorem}[Parameter Recovery, See Theorem~\ref{cor:parameter-estimation-main-technical}] \label{thm:param-inf}
Given $\eps>0$ and a multiset of $n = d^{O(k)} \poly_k(1/\eps)$ samples 
from a distribution $F$ on $\R^d$ such that
$\dtv(F, \calM) \leq \eps$, for an unknown target $k$-GMM $\calM = \sum_{i=1}^k w_i \mathcal{N}(\mu_i, \Sigma_i)$, 
the algorithm runs in time \new{$\poly(n) \poly_k(1/\eps)$} and outputs a $k'$-GMM hypothesis 
$\widehat{\calM}  = \sum_{i=1}^{k'} \widehat{w}_i \mathcal{N}(\widehat{\mu}_i, \widehat{\Sigma}_i)$ with $k'\le k$
such that with high probability we have that 
there exists a partition of $[k]$ into $k'+1$ sets $R_0,R_1,\dots,R_{k'}$ such that
\begin{enumerate}
\item 
Let $W_i=\sum_{j\in R_i}w_j$, \new{$i \in \{0, 1, \ldots, k'\}$.}
Then, for all $i\in[k']$, we have that
\begin{align*}
|W_i-\widehat w_{i}|&\le \poly_k(\epsilon), \textrm{ and} \\
\dtv(\mathcal{N}(\mu_j, \Sigma_j),\mathcal{N}(\widehat{\mu}_i, \widehat{\Sigma}_i))&\le\poly_k(\epsilon)\quad \forall j\in R_i \;.
\end{align*}
\item
The \new{total weight} of exceptional components in $R_0$ is $\new{W_0 \leq} \poly_k(\epsilon)$.
\end{enumerate}
\end{theorem}

If we assume additionally that any pair of components in the unknown mixture has total variation distance at least $\poly_k(\eps)$, then the following result follows directly from Theorem~\ref{thm:param-inf}.

\begin{corollary}\label{cor:robust-param-est}
Let $\calM = \sum_{i=1}^k w_i \mathcal{N}(\mu_i, \Sigma_i)$ be an unknown target $k$-GMM satisfying
the following conditions: (i) $\dtv(\mathcal{N}(\mu_i, \Sigma_i), \mathcal{N}(\mu_j, \Sigma_j))\ge\eps^{f_1(k)}$ for all  $i \neq j$,
and (ii) $S=\{i\in[k]:w_i\ge\eps^{f_2(k)}\}$ is a subset of $[k]$, where $f_1(k),f_2(k)$ are sufficiently small functions of $k$. 
Given $\eps>0$ and a multiset of $n = d^{O(k)} \poly_k(1/\eps)$ samples from a distribution $F$ on $\R^d$ 
such that $\dtv(F, \calM) \leq \eps$, there exists an algorithm that runs in time \new{$\poly(n) \poly_k(1/\eps)$} 
and outputs a $k'$-GMM hypothesis 
$\widehat{\calM}  = \sum_{i=1}^{k'} \widehat{w}_i \mathcal{N}(\widehat{\mu}_i, \widehat{\Sigma}_i)$ with $k'\le k$
such that with high probability there exists a bijection $\pi:S\to[k']$ satisfying the following: 
For all $i\in S$, it holds that
\begin{align*}
|w_i-\widehat w_{\pi(i)}|& \le\poly_k(\epsilon)\\
\dtv(\mathcal{N}(\mu_i, \Sigma_i),\mathcal{N}(\widehat{\mu}_{\pi(i)}, \widehat{\Sigma}_{\pi(i)}))& \le\poly_k(\epsilon).
\end{align*}
\end{corollary}

\noindent \new{We note that both the pairwise separation between the components and the lower bounds on the weights in Corollary~\ref{cor:robust-param-est} scale as a fixed polynomial in $\eps$ (for fixed $k$), 
which is qualitatively information-theoretically necessary.}

%% file: related-work.tex
\subsection{Concurrent and Independent Work} \label{ssec:conc}

\new{In independent and concurrent work with \cite{BDK+20:arxiv}, which established Theorem~\ref{thm:main-informal},}
\cite{LM20} gave an efficient outlier-robust parameter learning algorithm for mixtures of Gaussians under additional assumptions. 
Unlike our main theorem (Theorem~\ref{thm:main-informal}), the algorithm in \cite{LM20} \new{succeeds under} 
three crucial assumptions: 1) all mixing weights are bounded below \new{by $1/f(k)$, where $f$ is a function of $k$
that appears in the exponent of the running time and output error}, 2) all components have covariances with
eigenvalues at least $\poly(\eps/d)$ and at most $\poly(d/\eps)$ 
\new{(in particular, they cannot handle rank-deficient components)}, 
and 3) every pair of component Gaussians are separated by $g(k)$ in total variation distance 
(regardless of how \new{small the fraction of corruptions $\epsilon$} is).
\new{(We again note that each of these obtaining an algorithm without these assumptions requires significant new ideas in both the design and the analysis of the algorithm. Our techniques
lead to efficient algorithms without any of these assumptions.)}
The analysis of \cite{LM20} obtained a robust \emph{parameter recovery} guarantee 
on the output of their algorithm. In contrast, our Theorem~\ref{thm:main-informal} 
does not make any assumption on the mixture weights, the covariances or component separations, 
yielding the weaker guarantee of \emph{proper learning} 
(i.e., returning a mixture of Gaussians that is close in total variation distance to the unknown uncorrupted mixture). 
The error guarantee offered in our Theorem~\ref{thm:main-informal}, while being a constant (for any fixed $k,\epsilon$) 
is inverse poly-logarithmic in $1/\epsilon$. \new{We subsequently were able to improve these guarantees 
to achieve $\poly_k(\eps)$ error and parameter recovery guarantees.}


\Pnote{commenting out the first line since we seem to have said this already in the paragraph above.}
We also note that the techniques adopted in the two concurrent works (\cite{BDK+20:arxiv} and \cite{LM20})  
are significantly different. On the one hand, both works make essential use of
recent advances in clustering non-spherical mixtures~\cite{BK20,DHKK20}. 
While~\cite{LM20} relies on the clustering algorithm of \cite{DHKK20}, 
our work leverages a {\em key modification} to the clustering algorithm of \cite{BK20}
that allows us to obtain a fixed polynomial-time algorithm 
(as opposed to incurring an exponential cost in $1/w_{\min}$, where $w_{\min}$ is the minimum mixing weight),
at the cost of handling only a {\em subclass} of clusterable Gaussian mixtures. 
Indeed, our novel variant of clustering combined with our new spectral clustering subroutine 
is the key to not incurring an exponential cost in $1/w_{\min}$. In the next step, where the input mixture can be assumed to be non-clusterable,  
\cite{LM20} uses a new sum-of-squares based algorithm for parameter estimation. 
In contrast, our work does not use sum-of-squares method in this component and instead relies on a new list-decoding algorithm for tensor decomposition that makes no assumption on the underlying components.

\subsection{Related Prior Work} \label{ssec:related}

The algorithmic question of designing efficient robust estimators in high dimensions has been extensively 
studied in recent years. After the initial papers~\cite{DKKLMS16, LaiRV16}, 
a number of works developed robust estimators for a range of statistical problems.
These include efficient outlier-robust algorithms for sparse estimation~\cite{BDLS17, DKKPS19-sparse}, 
learning graphical models~\cite{ChengDKS18}, linear regression~\cite{KlivansKM18, DKS19-lr, bakshi2020robust, ZhuJS20, CherapanamjeriATJFB2020}, stochastic optimization~\cite{PrasadSBR2018,DiakonikolasKKLSS2018sever},
and connections to non-convex optimization~\cite{CDGS20, ZhuJS20}. 
Notably, the robust estimators developed in some of these works~\cite{DKKLMS16, LaiRV16, DKK+17} 
are scalable in practice and yield a number of applications in exploratory data analysis~\cite{DKK+17} 
and adversarial machine learning~\cite{TranLM18, DiakonikolasKKLSS2018sever}.
The reader is referred to~\cite{DK20-survey} for a recent survey.

Our partial clustering algorithm makes essential use 
of the Sum-of-squares based {\em proofs to algorithms} framework (see~\cite{TCS-086} for an exposition). This framework, 
beginning with~\cite{MR3388192-Barak15}, uses the {\em Sum-of-squares method} to design algorithms for statistical estimation problems, and has 
led to some of the most general outlier-robust learning algorithms. 
This includes computationally efficient outlier-robust estimators of the mean, covariance, 
and low-degree moments of structured distributions, with applications to ICA~\cite{KS17}, linear regression~\cite{KlivansKM18,bakshi2020robust,ZhuJS20}, clustering spherical mixtures~\cite{HopkinsL18,KStein17}, and clustering non-spherical mixtures~\cite{BK20,DHKK20}. 
The sum-of-squares method also gives a generally applicable scheme to handle the 
{\em list-decodable learning} setting~\cite{BalcanBV08, CSV17}, where a majority of the input points 
are corrupted, yielding efficient list-decodable learners for mean estimation~\cite{KStein17}, regression~\cite{karmalkar2019list,raghavendra2020list}, and subspace clustering/recovery~\cite{raghavendra2020subspace,bakshi2020list}.


%

Our work also has connections to the usage of tensor decomposition algorithms for learning statistical models. 
\cite{HK} used fourth-order tensor decomposition to obtain a polynomial-time algorithm for mixtures of spherical Gaussians with linearly independent means (with condition number guarantees). 
This result was extended via higher-order tensor decomposition for non-spherical Gaussian mixtures 
in a smoothed analysis model~\cite{GeHK15}. Fourth order tensor decomposition has earlier been 
used in~\cite{FJK:96} for the ICA problem~\cite{FJK:96}, and extended to general ICA with 
higher-order tensor decomposition by \cite{GVX14}. Such results rely on additional and 
non-trivial assumptions on the parameters of the mixture components in order to succeed, 
and are incomparable to our tensor-decomposition result that does not make 
any assumptions on the parameters of the mixture components. Indeed, this is 
the key innovation in our tensor decomposition algorithm that relaxes the guarantees on the output 
(we output a small list of candidate parameters) under a priori bounds on distance between components 
that is ensured by our partial clustering subroutine. This relaxation of tensor decomposition, 
and the new procedure that accomplishes it, is one of the main contributions of our paper. 

Finally, we point out that~\cite{DKS17-sq} gave an SQ lower bound for learning 
(fully) clusterable Gaussian mixtures without outliers, which provides evidence that a $d^{\Omega(k)}$ 
dependence is necessary in both the sample complexity and runtime
of any algorithm that learns GMMs. 




\subsection{Organization}
The structure of this paper is as follows: In Section \ref{sec:prelims}, we provide relevant background and technical facts. 
In Section \ref{sec:list-recovery}, we describe and analyze our new tensor decomposition algorithm.
In Section \ref{sec:robust-partial-clustering}, we use a sum-of-squares based approach to partially cluster a mixture. 
In Section \ref{sec:spectral-separation}, we give a spectral separation algorithm to identify thin components. 
In Section \ref{sec:full-algo-analysis}, we put all these pieces together to prove Theorem~\ref{thm:main-informal}.
In Section \ref{sec:robust-partial-clustering-upg}, we present a refinement of our partial clustering procedure that improves the probability of success to a constant independent of the minimum weight of any component in the input mixture. 
In Section \ref{sec:poly-proper}, we present an efficient algorithm that replaces an exhaustive search subroutine in the tensor decomposition algorithm and combines it with the improved partial clustering subroutine to get a $\poly_k(\epsilon)$-error guarantee for robust proper learning of Gaussian mixtures and prove Theorem~\ref{thm:poly-proper-inf}.
Finally, in Section~\ref{sec:param-recovery}, we show that our algorithm in fact achieves the stronger parameter estimation guarantees and prove Theorem~\ref{thm:param-inf}.
\Anote{Update.} \Pnote{done.}

%% file: techniques.tex
\subsection{Overview of Techniques } \label{ssec:techniques}

\new{\subsubsection{Proof of Theorem~\ref{thm:main-informal}}} \label{ssec:old-techniques}

In this section, we give a bird's eye view of our algorithm and the main ideas that go into it. 
Recall that our goal is to design an efficient algorithm that takes an $\epsilon$-corrupted sample $Y$ 
from a mixture of $k$-Gaussians $\calM = \sum_i w_i \cN(\mu_i,\Sigma_i)$ 
and outputs a mixture $\widehat{\calM} = \sum_i \hat w_i \cN(\widehat \mu_i, \widehat \Sigma_i)$ 
such that the total variation distance between $\calM$ and $\widehat{\calM}$ is bounded above 
by a dimension-independent function of $\epsilon$ (bounded above by $1/(\log(1/\epsilon))^{k^{-O(k^2)}}$). 
Specifically, we want the running time of our algorithm to be bounded above by a polynomial 
in the dimension $d$ and $1/\epsilon$, for any fixed $k$. 

In the non-robust setting (i.e., for $\epsilon =0$), the algorithm of Moitra and Valiant~\cite{MoitraValiant:10}, 
extending their work with Kalai~\cite{KMV:10}, solves this problem. 
However, natural attempts to adapt their method to tolerate outliers run into immediate difficulties. 
The starting point of~\cite{MoitraValiant:10} is to observe that if a mixture of $k$ Gaussians 
has every pair of components separated in total variation distance by at least $\delta$, 
then a random univariate projection of the mixture has a pair of components that are $\delta/\sqrt{d}$-separated 
in total variation distance. Their algorithm uses this observation to piece together estimates of the mixture 
when projected to several carefully chosen directions to get an estimate of the high-dimensional mixture. 
Notice, however, that such a strategy meets with instant roadblock in the presence of outliers: 
the fraction of outliers, being a dimension-independent constant, completely overwhelms the total variation 
distance between components in any one direction making them indistinguishable\footnote{Informally speaking, one could hope to show that outliers projected into a random direction cannot be too adversarial, but it is unclear how to use this observation not in the least because the algorithm of~\cite{MoitraValiant:10} requires a somewhat carefully tailored choice of projections.}. 

For a reader familiar with the work in algorithmic robust statistics, 
this may not come as a surprise --- to handle outliers, we almost always need to develop
a completely new algorithm, even in the outlier-free setting.

In particular, as we next describe, our approach diverges from the method of~\cite{MoitraValiant:10} 
at the very beginning, and instead relies on a careful interleaving of two new algorithmic primitives: 
(1) a new \emph{partial clustering algorithm} based on the sum-of-squares (SoS) method, 
and (2) a new \emph{tensor decomposition} method for decomposing a symmetrized 
sum of tensor powers of $d \times d$ matrices. 

\paragraph{The Key Determinant: Clusterability.} 
Our first and key conceptual contribution is to deal with the case of \emph{partially clusterable} mixtures 
differently from those that are not partially clusterable. We call a mixture partially clusterable 
if there is a pair of components that have total variation distance larger than $1-\Omega_k(1)$ 
(we will call such components {\em well-separated} in what follows). 
We note that even the setting when the mixture is {\em fully} clusterable 
(i.e., every pair of components is well-separated), 
the learning problem captures many hard special cases (such as subspace clustering) and is highly non-trivial. 
Two recent works~\cite{DHKK20,BK20} gave a polynomial-time algorithm for the fully clusterable
case, using the sum-of-squares method. Interestingly, it turns out that the clustering algorithm of~\cite{BK20} 
(specifically, their Lemma~6.4) can be generalized (see Theorem~\ref{thm:partial-clustering-non-poly}) 
to the partial clustering setting, i.e., the setting where we are guaranteed to have a pair of components 
that are well-separated (with no guarantees on the remaining components). 
This gives an algorithm with running time of $d^{(k/\alpha)^{O(k)}}$ to partition the input sample 
into components so that each piece of the partition is (effectively) a $(\poly(\alpha/k)+\epsilon)$-corrupted 
sample from disjoint sub-mixtures. Here, $\alpha$ is the smallest mixing weight. 
As we will soon see, such a partial clustering algorithm will be too slow to yield our final guarantees 
of a fully polynomial time algorithm (when the smallest weight is too small), 
and we will soon discuss a more efficient variant that will suffice for our purposes. 

\paragraph{Approximate Isotropic Transformation.} 
By applying our partial-clustering algorithm, we can effectively assume that the input is an $\epsilon$-corrupted 
sample from a mixture with every pair of components \emph{at most} $(1-\Omega_k(1))$-far
in total variation distance. At this point, we would like to make the mixture isotropic --- that is, we would like 
to assume that the mean of the mixture is $\approx 0$ and the covariance of the mixture is $\approx I$. 
In the setting with no outliers, this is simply a matter of computing the empirical mean and covariance 
and applying an appropriate affine transformation to the input points. However, in the setting with outliers, 
even this task is somewhat non-trivial. A natural idea is to use the algorithm for robust covariance estimation 
with bounded error in \emph{spectral norm} from~\cite{KS17} that works for all \emph{certifiably subgaussian} 
distributions (the same work also establishes that arbitrary mixtures of Gaussians are certifiably subgaussian). 
However, it turns out that our algorithm needs dimension-independent error guarantee on the estimated covariance 
in \emph{Frobenius norm} (instead of the weaker spectral error guarantee). Fortunately, the recent work~\cite{BK20} 
(Theorem 7.1 in their paper and Fact~\ref{fact:param-estimation-main}) gives precisely such an algorithm 
for robust covariance estimation that relies on the stronger property of \emph{certifiable hyper-contractivity} 
(we verify that this property holds for mixtures of Gaussians in 
Lemma~\ref{lem:mixtures-of-certifiably-hypercontractive-distributions} of Section~\ref{sec:prelims}). 

\paragraph{Mixtures with Pairwise Close Components.} 
After the first two steps, we can effectively assume that we are working with 
an $\epsilon$-corrupted sample from a mixture that is approximately isotropic 
and every pair of components is not too far in total variation distance. 
Why is this latter guarantee useful? 
As established in the recent works~\cite{DHKK20,BK20}, such a bound translates 
into a guarantees (with respect to natural norms) on the parameters of the component Gaussians. 
In particular, using this translation in our setting implies that after partial clustering plus 
an approximate isotropic transformation, we have that: 
1) $\Norm{\Sigma_i -I}_F \leq \poly(\alpha, k)$, 
2) $\Norm{\mu_i}_2 \leq 2/\sqrt{\alpha}$, 
and 3) $\Sigma_i \succeq \frac{1}{\poly(\alpha, k)} I$ for every $i$ 
(recall that $\alpha$ is the minimum weight in the mixture). 
In this case, it turns out that in order to learn the unknown mixture with error guarantees in total variation distance, 
it suffices to obtain $\poly_k(\epsilon)$-error estimates of the $\mu_i, \Sigma_i$'s in Frobenius norm. 

\paragraph{Symmetrized Tensors and the Work of Kane~\cite{Kane20}.} 
The key first step in addressing such a mixture was taken in the very recent work of Kane~\cite{Kane20}, 
who gave a polynomial-time algorithm to robustly learn an \emph{equiweighted} mixture of two Gaussians. 
For this special case, after isotropic transformation, one can effectively assume that the two means 
are $\pm \mu$ and the two covariances are $I \pm \Sigma$.  
Kane's idea is to look at a certain tensor (``Hermite tensor'') that can be built using the $4$-th and $6$-th raw moments of the mixture. 
Since we must use outlier-robust algorithms to estimate these tensors, we can obtain estimates that are accurate 
only up to constant error in Frobenius norm of the tensor. 
Kane's key observation is that \emph{for the special case of $k=2$ components}, 
one can build two different Hermite tensors, one of which is rank-one with component $\approx \mu$ 
(and thus one can immediately ``read off'' $\mu$); the other only has a tensor power of $\Sigma$. 
This second tensor is of the form $\hat T_4 = \mathrm{Sym}((\Sigma-I) \otimes (\Sigma-I))+E$, 
where $\Norm{E}_F = O_k(\sqrt{\epsilon})$ 
and $\mathrm{Sym}$ refers to symmetrizing over all possible permutations 
of the ``4 modes of the tensor''. Unlike the case of the mean, 
one cannot simply ``read-off''\footnote{It is helpful to visualize a single entry of this tensor for, say, the case when $i,j,k,\ell$ are all distinct: 
$\hat T_4(i,j,k,\ell) = \frac{1}{3} (\Sigma(i,j) \Sigma(k,\ell) + \Sigma(i,k) \Sigma(j,\ell) + \Sigma(i,\ell)\Sigma(j,k)) + error$. 
Notice that obtaining entries of $\Sigma$ from $T_4$ is formally a task of solving noisy quadratic equations.} 
$\Sigma$ from $T_4$, but Kane gives a simple method to accomplish this. 
As noted in~\cite{Kane20}, it is not clear how to extend this to non-equiweighted mixtures of $k=2$ Gaussians, 
and going to even $k =3$ components requires substantially new ideas. 

\paragraph{List-decodable Tensor Decomposition.} 
Our method for $k > 2$ works by abstracting and generalizing 
key aspects of the~\cite{Kane20} somewhat ad-hoc approach for the case of $k=2$ 
and combining it with new ideas. 
This is necessary because of several issues, as we soon discuss: 
1) as stated, Kane's approach does not work as is even for $k=2$ when the mixture is not equiweighted, 
2) it is not known whether one can build tensors that ``separate'' out the means 
and the covariances, as~\cite{Kane20} managed to do for $k=2$, 
3) the relevant tensors will not be (symmetrizations of) rank-$1$ tensors, up to noise. 
It is worth noting, in fact, that such a gap is information-theoretically inherent, 
as shown in~\cite{MoitraValiant:10}; learning arbitrary mixtures of $k$ 
Gaussians \emph{provably} requires at least $2k$ moments of the mixture. 
Another way of seeing this is by considering
the parallel pancakes construction of~\cite{DKS17-sq}, where the
authors produce an example of a mixture of $k$ Gaussians whose first
$k$ moments match the standard Gaussian exactly despite not being
close in total variational distance. It should be noted that,
in this example, the component Gaussians are all equal to the standard
Gaussian except in one hidden direction. Thus, we cannot hope to
identify the components exactly with just $O(1)$ many moments.
(However, for this example at least, we might hope to identify these
components up to whatever is going on in this one (unknown) hidden
direction. In fact, more complicated constructions can be made to have
several such hidden directions.) 
The somewhat surprising fact (that we establish 
in more detail below) is that by looking at only the first four moments of our mixture, 
we can learn all of the components up to some errors taking place along a
bounded number of hidden directions. In particular, we can learn the
covariance matrices of the component Gaussians up to some {\em low rank}
error terms. We elaborate on this new idea below.

For the sake of the intuition, it is helpful to focus on the simpler case where all the means are zero. 
In this case, the estimated 4th Hermite tensor (built from estimated raw moments of degree at most $4$ of the mixture) has the following form: 
\[
\hat T_4 = \sum_{i=1}^k w_i \mathrm{Sym}( (\Sigma_i-I) \otimes (\Sigma_i-I) + E) \;,
\]
where $E$ is a $4$-tensor with $\Norm{E}_F =O_k(\sqrt{\epsilon})$.
Given the form of this tensor, it is natural to think of applying tensor decomposition algorithms, 
by thinking of $\Sigma_i-I$ as a $d^2$-dimensional vector. 
However, we run into the issue of uniqueness of tensor decomposition, 
since we are dealing with 2nd order tensors (once we view $\Sigma_i-I$ as a $d^2$-dimensional vector). 
One might imagine computing higher-order tensors of similar forms to overcome the uniqueness issues, 
but this runs into two major complications: first, the symmetrization operation introduces spurious terms 
that do not have the sum of tensor-power structure required for such an algorithm to succeed. 
Indeed, this is far from just being an annoying technicality --- the recent work of 
Garg, Kayal, and Saha~\cite{DBLP:journals/corr/abs-2004-06898} 
addresses precisely such a tensor decomposition problem via algebraic techniques. 

Unfortunately, as is explicitly pointed out as one of the main open question in their work (see page 14), 
because of their reliance on algebraic techniques, 
their algorithm is highly brittle and in particular, may not even be able to handle the benign noise 
that comes from estimating tensors from independent (uncorrupted) samples. 
(Of course, our setting has to deal with the malicious noise introduced due to the adversarial outliers.) 
Second, even if one were to get hold of the tensor without the symmetrization operation, 
the only applicable tensor decomposition algorithm (recall that we do not make \emph{any} genericity assumptions on the components that are typically required by tensor decomposition algorithms) 
is the result of Barak, Kelner, and Steurer~\cite{MR3388192-Barak15}. 
However, the~\cite{MR3388192-Barak15} result, while being efficient in its dependence 
on the number of components, has exponential dependence on the target error, 
which is prohibitively expensive for our application. 

Our idea is to give up on the goal of recovering the unique decomposition of the tensor $\hat T_4$ above, 
and start by applying an operation that is a common trick in most tensor decomposition algorithms. 
In our context, this trick amounts to taking a random matrix (say, with independent standard Gaussian entries) 
$P$ and ``collapsing'' the last two modes of $\hat T_4$ with $P$ 
(i.e., computing $\hat S(i,j) = \sum_{k,\ell}\hat T_4 (i,j,k,\ell) P(k,\ell)$) 
to obtain a matrix $Q$. 
In the usual tensor decomposition procedures, we are interested in proving that one can recover all the information 
about the components of the tensor from $Q$. We will not be able to prove such a statement here. 
Instead, our key observation is that one can choose a matrix $P$ \emph{of rank $\poly(k)$} 
and argue that the resulting $\hat S$ is $O_k(\epsilon)$-close to one of the $\Sigma_i-I$ 
\emph{up to an error term of $O(k^2)$ rank.} To see this, note that in the symmetrization
\[
\hat T_4(i,j,k,\ell) = \frac{1}{3} (\Sigma(i,j) \otimes \Sigma(k,\ell) + 
\Sigma(i,k) \otimes \Sigma(j,\ell) + \Sigma(i,\ell) \otimes \Sigma(j,k)) + \mbox{ error}
\]
applying a rank-one matrix $P$ to the modes $k,\ell$ will reduce the first term to the matrix $\Sigma$, 
while the latter two terms become rank-one terms! When the tensor is a sum of $k$ such symmetrized tensors, 
we will use $k$ such rank-one matrices, and take a linear combination of them to get a weighted sum of the $\Sigma$'s 
plus a term of rank $O(k^2)$. As we show, the linear combination can be chosen such that only one 
of the component $\Sigma$'s survives (up to small Frobenius norm).  
Moreover, such a low-rank $P$ can be obtained efficiently by simply choosing $\poly(k)$ 
random rank-$1$ matrices and exhaustively searching 
over an appropriate cover of the $O(k^2)$-dimensional subspace spanned by them. 

\paragraph{Subspace Enumeration to Recover the Low-rank Terms.} 
We then show that we can use the generated estimates $\hat S$ and the tensors $\hat T_m$ for $m \leq 4k$ 
to find a $\poly_k(\epsilon)$-dimensional subspace $V_*$ such that the low-rank error matrix 
in the estimated $\hat S$ have their range space essentially contained inside $V_*$. 
Our next step involves a subspace enumeration over $V_*$ 
to output a list of a bounded number of parameters 
such that a mixture defined by some $k$ of them 
must be close in total variation distance to the input mixture. 

Our final step involves a relatively standard hypothesis testing procedure, 
using a robust tournament that goes over each of the candidate mixtures in our generated list 
and finds one that is approximately the closest in total variation distance to a (fresh) set of corrupted samples. 

\paragraph{An Algorithm with Exponential Dependence in $1/\epsilon$.} 
The above steps suffice to immediately obtain an algorithm 
whose running time grows exponentially in the reciprocal of the minimum weight of the mixture. 
This gives a polynomial-time algorithm (see Theorem~\ref{thm:robust-GMM-equiweigthed} for details) 
for robustly learning arbitrary equiweighted mixtures of $k$ Gaussians. 
When the weights are not all equal, notice that we can treat any component with weight $\leq \epsilon$ as outliers, 
which effectively means that $\alpha \geq \epsilon$. 
Thus, our discussion already yields a $d^{f(k/\epsilon)}$-time algorithm in the general setting. 

In order to improve this running time to have a fixed polynomial dependence on $d$ (independent of $1/\epsilon$), 
we rely on a new partial clustering result that weakens the separation guarantee of total variation distance. 
Our final algorithm then involves a recursive interleaving of the partial clustering 
and tensor decomposition steps with a new \emph{Recursive Spectral Clustering} subroutine. 
We discuss these steps next.

\paragraph{Partial Clustering.} 
The key bottleneck in the running time guarantee of the algorithm described above is the partial clustering step, 
so it is important to examine the cause for the exponential dependence on the minimum weight in the running time. 
The algorithm relies on a recently established characterization of well-separated pair of Gaussians 
$\cN(\mu_1,\Sigma_1)$ and $\cN(\mu_2,\Sigma_2)$ 
in terms of three geometric distances between their parameters: 
1) Mean Separation: $\exists v \text{ such that } \iprod{\mu_1-\mu_2,v}^2 \geq \Delta (v^\top (\Sigma_1 + \Sigma_2)v)$, 
2) Relative Frobenius Separation: $\Norm{\Sigma_1^{-1/2} (\Sigma_2 -\Sigma_1) \Sigma^{-1/2}}_F \geq \Delta$, 
and 3) Spectral Separation: $\exists v$ such that $v^{\top}\Sigma_1 v \geq \Delta v^{\top} \Sigma_2 v$ 
or $v^{\top} \Sigma_2 v \geq \Delta v^{\top} \Sigma_1 v$. 
The main idea of the algorithm is to give efficient (low-degree) \emph{sum-of-squares} certificates 
of \emph{simultaneous intersection bounds} that show that any cluster a natural SoS relaxation finds 
cannot significantly overlap with two well-separated clusters simultaneously. 
This step requires sum-of-squares certificates for two natural analytic properties: 
\emph{certifiable hypercontractivity} and \emph{certifiable anti-concentration} 
(introduced in the recent works~\cite{karmalkar2019list, raghavendra2020list}, 
and also used in~\cite{raghavendra2020subspace,bakshi2020list}). 
The bottleneck that results in the bad running time for us is the degree 
of the sum-of-squares certificate needed for certifiable anti-concentration 
(which grows polynomially in $1/\alpha$). 

It is not known whether there is a sum-of-squares certificate of much smaller degree for certifiable anti-concentration. 
To make progress here, we observe that the only usage of this certificate occurs in dealing 
with spectrally separated pairs of Gaussians in the mixture. Indeed, we give a new partial clustering algorithm 
that works in fixed polynomial time, whenever there is a pair of Gaussian components separated 
either via their means or an appropriate variant of the relative Frobenius distance. 

\paragraph{Tensor Decomposition Needs to be Augmented.} 
While we gain on the running time through our new partial clustering algorithm, 
the guarantees of the tensor decomposition subroutine we discussed above 
are no longer enough to guarantee a recovery of parameters 
that result in a mixture close in total variation distance. 
Because of the three conditions that we assumed in the working of the tensor decomposition algorithm, 
we can no longer guarantee the third one that gives a lower bound on the smallest eigenvalue of every covariance 
(relative to the covariance of the mixture). In particular, we can end up in a situation where, 
even though we have a list of parameters that contain Frobenius-norm-close estimates of the covariances, 
the estimates are not enough to provide a total variation distance guarantee. 
(Consider for example a ``skinny'' direction where the variance of some component is very small, or even $0$. 
Then we have to learn the parameters more precisely!)

\paragraph{Spectral Separation of Thin Components.} 
It turns out that the above is the only way the algorithm can fail at this point --- one or more covariance matrices 
have a very small eigenvalue (if not, the Frobenius norm error would imply TV-distance error). 
But since we have estimates of the covariances, we can find such a small eigenvector. 
Now we observe that since the mixture is nearly isotropic (i.e., the overall variance in each direction is $\sim 1$), 
if some component has very small variance along a direction, then 
the components must be separable along this direction. 
We show that it is possible to efficiently cluster the mixture after projecting it to this direction, 
so that each cluster has strictly fewer components. 
We then recursively apply the entire algorithm on the clusters obtained, 
which will each have strictly fewer components. 

\paragraph{Polynomial Complexity.} 
To avoid $\eps$-dependence in the exponent of the dimension or exponential dependence 
in the minimum mixing weight, we use our new partial clustering algorithm that does not rely 
on certifiable anti-concentration, by avoiding spectral proximity guarantees, and moving the work 
of separating along small eigenvalue directions to later in the algorithm. 
Our tensor decomposition also has a dependence on the minimum mixing weight. 
To circumvent this exponential dependence on the minimum weight, 
a natural approach would be to ignore components lighter than some threshold (that depends 
on the target error) and treat them as corruptions. However, this intuitive approach runs into a difficulty. 
In order to get nontrivial error guarantees on the tensor decomposition, 
we need that the minimum mixing weight is significantly {\em larger} than the fraction of outliers 
(since the decomposition involves generating a list of candidate hypotheses, one of which must be accurate). 
To solve this problem, we show that we can set a minimum weight threshold that depends 
on the number of remaining components of the mixture, and ``remove'' some (but not all) components, 
so that the remaining mixture has minimum weights above this threshold; 
and the threshold is also sufficiently larger than the total weight of small components treated as corruptions. 
This makes the overall computational complexity dependence on $d$ truly polynomial for any fixed $k$, 
and avoids any dependence on the true minimum mixing weight. 

\new{

\subsubsection{Proofs of Theorems~\ref{thm:poly-proper-inf} and~\ref{thm:param-inf}} \label{ssec:new-techniques}

\paragraph{Robust Proper Learning with $\poly_k(\eps)$ Error.}
The algorithm of Theorem~\ref{thm:main-informal} achieves error polynomial in $\epsilon$, 
alas in time exponential in $\poly_k(1/\eps)$. This exponential dependence on $1/\epsilon$ comes from two sources: 
(1) the exhaustive subspace enumeration within the space $V_*$ of significant eigenvalues, 
and (2) the error probability in the partial clustering algorithm, 
which then necessitates super-polynomial enumeration. 
We use additional techniques to reduce both dependencies (and hence the overall running time) 
to $\poly_k(1/\eps)$. 

The key bottleneck in the partial clustering step is the rounding algorithm that obtains a list of candidate clusters each of size roughly $\alpha n$ where $\alpha$ is the weight of the smallest component cluster (note that this can be arbitrarily small compared to $1/k$).  The rounding guarantees that any such candidate cluster cannot simultaneously contain a significant fraction of points from two components whose covariances are separated in Frobenius distance. However, the candidates could be indiscriminate in collecting an arbitrary subset of points from ``nearby'' clusters. Our rounding algorithm in Section 4 simply guesses a partition of them such that every pair  a $\exp(-1/\alpha)$ success probability in guessing a 2 partition of points such that any cluster (up to a tiny fraction of errors) appears on one side only. 

Our new clustering algorithm improves on this random guessing step in the rounding by observing that points that come from nearby clusters must have covariances that are close in Frobenius norm. Thus our rounding applies a robust covariance estimation algorithm (with error guarantees in Frobenius norm) to each candidate cluster of size roughly $\alpha n$ and then collects candidate clusters whenever their estimated covariances are close. We show that this procedure coalesces the $O(1/\alpha)$ clusters into a collection of at most $k$. The random guessing step now succeeds with $\exp(-k)$ probability resolving the bottleneck in the discussion above. 

The other bottleneck is in the subspace enumeration over $V_*$ outputs a list of size with exponential dependence on the dimension of $V_*$. 
To reduce the list size of tensor decomposition, first we observe that we can use an elementary filtering 
technique to denoise the data, since the number of samples needed is small. Specifically, 
if the sample complexity of the algorithm is polynomial in the dimension and error parameter, 
we can set both the dimension of $V_*$ and the error parameter to be polynomials in $1/\eps$, 
so that the number of samples needed is $O(1/\eps)$. Then, 
with probability $(1-\eps)^{O(1/\eps)}=\Omega(1)$, a sample of size $O(1/\eps)$ 
drawn from the total variation contamination model has no noise. 
For such a clean sample, we can apply a non-robust algorithm to the subspace $V_*$. 
If the time complexity of the algorithm is polynomial in the dimension and error parameter, 
then the running time to apply it on $V_*$ will be $\poly_k(\eps)$, as desired.
The next step is to prove that such an algorithm exists. We will use the algorithm in Theorem~8 of~\cite{MoitraValiant:10}. 
A small technical issue is that the latter assumes that any pair of components in the mixture has TV distance 
at least $\delta$, where $\delta$ is the error parameter. 
We show that with an appropriately chosen parameter $\delta'=\poly(\delta)$, 
any pair of components with TV distance less than $\delta'$ are close enough 
so that the algorithm cannot distinguish the pair from a single Gaussian, 
since the algorithm only requires a polynomial number of samples. 
If we merge all such pairs, any pair in the mixture is separated by $\delta'$,
and then we can apply Theorem~8 in~\cite{MoitraValiant:10}.
For each estimate $\hat S$ of the main algorithm, we can recover the low-rank error of $\hat S$ 
by learning the mixture in the subspace $V_*$. Combining the estimates $\hat S$ 
and the estimates in the subspace $V_*$, we get a list of parameters of size $\poly_k(\eps)$. 

\paragraph{Parameter Recovery.}
We show that for any two Gaussian mixtures, if they are close in TV distance, 
their components are also close in TV distance, which implies that we can recover 
the components or the parameters of the mixture. This result generalizes Theorem 8.1 in~\cite{LM20}, 
which has three additional assumptions: 
(i) each component has variance at least $\poly(\eps/d)$ and at most $\poly(d/\eps)$, 
(ii) each pair of components has TV distance at least $\poly_k(\eps)$, and
(iii) the minimal weights of both mixtures are at least $\poly_k(\eps)$. 
\cite{LM20} also proved the conclusion under the assumption that any parameters (means and covariances) 
are identical or separated. We reduce the general case to this simplification. 
The first step is to deal with the components with small weights. 
We use a threshold $\eps' =\poly_k(\eps)$ of weights such that if we treat components with weights smaller 
than the threshold as noise, other components have weights at least $\poly_k(\eps')$. 
The second step is a partial clustering on the union of the components of the two mixtures, 
after which the components within each cluster are pairwise not too close, 
i.e., have TV distance bounded by $1-\poly(\eps')$. Then we can modify the parameters in each cluster slightly, 
so that the resulting parameters for different components are either identical or have a minimum separation. 
Thus, we reduce the general case of  two arbitrary mixtures to this special case. 
Overall, this is a purely information-theoretic statement --- for each significant weight component 
of one mixture, there will be a component in the second mixture with very close mean and covariance. 
We note that this is not necessarily a 1-1 mapping between components, 
which is impossible in general without further assumptions.
}

%% file: prelims.tex
\section{Preliminaries} \label{sec:prelims}

\paragraph{Basic Notation.}
For a vector $v$, we use $\norm{v}_2$ to denote its Euclidean norm.
For an $n \times m$ matrix $M$, we use $\norm{M}_{\textrm{op}} = \max_{\norm{x}_2=1} \norm{Mx}_2$
to denote the operator norm of $M$ and $\norm{M}_F = \sqrt{\sum_{i,j} M_{i,j}^2}$ to denote the Frobenius norm of $M$. We sometimes use the notation $M(i,j)$ to index the corresponding entries in $M$. For an $n \times n$ symmetric matrix $M$, we use $\succeq$ to denote the PSD/Loewner ordering over eigenvalues of $M$ and $\tr\Paren{M} = \sum_{i \in [n]} M_{i,i}$ to denote the trace of $M$.  We use $U\Lambda U^{\top}$ to denote the eigenvalue decomposition,
where $U$ is an $n \times n$ matrix with orthonormal columns and $\Lambda$ is
the $n \times n$ diagonal matrix of the eigenvalues. We use
$M^{\dagger} = U \Lambda^{\dagger} U^{\top} $ to denote the Moore-Penrose pseudoinverse, where $\Lambda^{\dagger}$ inverts the non-zero eigenvalues of $M$. If $M \succeq 0$, we use $M^{\dagger/2} = U \Lambda^{\dagger/2} U^{\top}$ to denote taking the square-root of the inverted non-zero eigenvalues.

For $d \times d$ matrices $A,B$, the Kronecker product of $A,B$, denoted by $A \otimes B$,
is indexed by $(i,j), (k,\ell) \in [d] \times [d]$ and has entries $(A \otimes B)((i,j),(k,\ell)) = A(i,k) B(j,\ell)$.
We will equip every tensor $T$ with the norm $\Norm{\cdot}_F$ that simply corresponds to the $\ell_2$-norm
of any flattening of $T$ to a vector.
The notation $T(\cdot, \cdot, x, y)$ is used to denote collapsing two modes of the tensor by plugging in $x, y$.
For a positive integer $\ell$ and vector $v$, we also use
$v^{\otimes \ell} = \underbrace{v \otimes v \ldots \otimes v}_{\ell \textrm{times}}$.

We use the notation
$\calM = \sum_{i \in [k]} w_i \calN\Paren{\mu_i , \Sigma_i}$ to represent a $k$-mixture of Gaussians.
The total variation distance between two probability distributions on $\R^d$ with densities
$p, q$ is defined as $\dtv(p, q) = \frac{1}{2}\int_{\R^d}|p(x)-q(x)|dx$. We also use $\expecf{}{\cdot}$, $\var{\cdot}$ and $\textsf{Cov}(\cdot)$ to denote the expectation, variance and covariance of a random variable.

\new{For a finite dataset $X$, we will use $Z \in_u X$ to denote that $Z$ is the uniform distribution on $X$. We will sometimes use the term mean (resp. covariance) of $X$ to refer to $\expecf{Z \in_u X}{Z}$ (resp. $\textsf{Cov}_{Z \in_u X}(Z)$).}

\subsection{Gaussian Background} \label{ssec:prelims-gaussian}

The first few facts in this subsection can be found in Kane~\cite{Kane20}.

\begin{fact}\label{GaussianTVFact}
The total variation distance between two Gaussians $\calN\Paren{\mu_1, \Sigma_1}$ and $\calN\Paren{\mu_2,\Sigma_2}$ can be bounded above as follows:
$$
\dtv\Paren{\calN(\mu_1,\Sigma_1),\calN(\mu_2,\Sigma_2)} = \bigO{ \Paren{\mu_1-\mu_2}^\top \Sigma_1^{\dagger}\Paren{\mu_1-\mu_2}  + \Norm{\Sigma_1^{\dagger/2} \Paren{ \Sigma_2 - \Sigma_1} \Sigma_1^{\dagger/2}}_F  } \;.
$$
\end{fact}


\begin{fact}[Theorem 2.4 in \cite{Kane20}]\label{CovarianceLearnerTheorem}
Let $\calD$ be a distribution on $\R^{d\times d}$, where $\calD$ is supported on the subset of $\R^{d\times d}$ corresponding to the set of symmetric PSD matrices. Suppose that $\E[\calD]=\Sigma$ and that for any symmetric matrix $A$ we have that $\var{\tr(AX)} = \bigO{\sigma^2\|\Sigma^{1/2} A \Sigma^{1/2}\|_F^2}.$ Then, for $\epsilon \ll \sigma^{-2}$, there exists a polynomial-time algorithm that given sample access to an $\eps$-corrupted set of samples from $\calD$
returns a matrix $\hat{\Sigma}$ such that with high probability $\|\Sigma^{-1/2}(\Sigma -\hat{\Sigma})\Sigma^{-1/2}\|_F = O(\sigma\sqrt{\eps}).$
\end{fact}

\begin{fact}[Proposition 2.5 in \cite{Kane20}] \label{momentsProp}
Let $G \sim \calN(\mu,\Sigma)$ be a Gaussian in $\R^d$. Then, we have that
$$
\E[G^{\otimes m}]\Paren{i_1,\ldots, i_m} = \sum_{\substack{\textrm{Partitions }P\textrm{ of }[m]\\ \textrm{ into sets of size 1 and 2}}} \bigotimes_{\{a,b\}\in P} \Sigma\Paren{i_a,i_b} \bigotimes_{\{c\}\in P} \mu\Paren{i_c} \;.
$$
\end{fact}

We will work with the coefficient tensors of $d$-dimensional Hermite polynomials:
\begin{definition}[Hermite Tensors] \label{def:degree_m_hermite_tensor}
Define the degree-$m$ Hermite polynomial tensor as
$$
h_m(x) :=  \sum_{\substack{\textrm{Partitions }P\textrm{ of }[m]\\ \textrm{ into sets of size 1 and 2}}} \bigotimes_{\{a,b\}\in P} -I\Paren{i_a,i_b} \bigotimes_{\{c\}\in P} x\Paren{i_c} \;.
$$
\end{definition}

We will use the following fact that relates Hermite moments to the raw moments of any distribution.
\begin{fact}[Hermite vs Raw Moments, see, e.g., \cite{wiki:inverse_hermite}] \label{fact:hermite-vs-raw-moments}
For any real-valued random variable $u$, and $m \in \N$, $\max_{i \leq m} |\E u^i - \E_{z \sim \cN(0,1)} z^i| \leq 2^{O(m)} \max_{i \leq m} |\E h_m(u)|$. Similarly, $\max_{i \leq m} |\E h_m(u)| \leq 2^{O(m)} \max_{i \leq m} |\E u^i - \E_{z \sim \cN(0,1)} z^i|$.
\end{fact}

\begin{fact}[Lemma 2.7 in \cite{Kane20}] \label{HermiteExpectationLem}
If $G\sim \calN(\mu,I+\Sigma)$, then we have that
$$\E[h_m(G)] = \sum_{\substack{\textrm{Partitions }P\textrm{ of }[m]\\ \textrm{ into sets of size 1 and 2}}}
\bigotimes_{\{a,b\}\in P} \Sigma\Paren{i_a,i_b} \bigotimes_{\{c\}\in P} \mu\Paren{i_c} \;.$$
\end{fact}

\begin{fact}[Lemma 2.8 in \cite{Kane20}] \label{HermiteVarianceLem}
If $G\sim \calN(\mu,I+\Sigma)$, then $\E[h_m(G)\otimes h_m(G)]$ is equal to
$$\sum_{\substack{\textrm{Partitions }P\textrm{ of }[2m]\\ \textrm{ into sets of size 1 and 2}}}
\bigotimes_{\substack{\{a,b\}\in P\\a,b\textrm{ in same half of }[2m]}} \Sigma\Paren{i_a,i_b}
\bigotimes_{\substack{\{a,b\}\in P\\a,b\textrm{ in different halves of }[2m]}} (I+\Sigma)\Paren{i_a,i_b}\bigotimes_{\{c\}\in P} \mu\Paren{i_c} \;.$$
\end{fact}

\begin{lemma}[Slight Strengthening of Lemma 5.2 in \cite{Kane20} ]\label{Hermite Covariance Lemma}
For $G\sim \calN(\mu,\Sigma)$, the covariance matrix of $h_m(G)$ satisfies:
$$\left\|\cov(h_m(G)) \right\|_{\textrm{op}} \leq \Norm{\E \left[h_m(G) \otimes h_m(G) \right]}_{\textrm{op}}
= \mathcal{O}\Paren{m(1+\Norm{\Sigma}_F+\Norm{\mu}_2)}^{2m} \;.$$
\end{lemma}
This follows from the proof of Lemma 5.2 in \cite{Kane20} by noting that the number of terms in the sum is at most $2^m$ times the number of partitions of $[2m]$ into sets of size $1$ and $2$, which is at most $O(m)^{2m}$.

Next, we use upper and lower bounds on low-degree polynomials of Gaussian random variables. We defer the proof of the subsequent Lemma to Appendix \ref{sec:omitted_proofs}.

\begin{lemma}[Concentration of low-degree polynomials] \label{lem:concentration_of_low_degree_gaussians}
Let $T$ be a $d$-dimensional, degree-$4$ tensor such that $\Norm{T}_F \leq \Delta$
for some $\Delta>0$ and let $x, y \sim \calN(0, I)$. Then, with probability at least $1-1/\poly(d)$, the following holds:
\begin{equation*}
\Norm{ T\Paren{\cdot, \cdot, x, y}}^2_F \leq \bigO{\log(d) \Delta^2} \;.
\end{equation*}
\end{lemma}

Note that for any matrix $M$, $\Iprod{M, x \otimes y}$, where $x,y \sim \calN(0,I)$,
is a degree-$2$ polynomial in Gaussian random variables. As a result, we have the following
anti-concentration inequality.

\begin{lemma}[Anti-concentration of bi-linear forms, \cite{CW01}] \label{lem:anti-concentration_of_bilinear_forms}
Let $M$ be a $d \times d$ matrix and let $x , y \sim \calN(0, I)$.
Then, for any $\zeta \in (0, 1)$, the following holds:
\begin{equation*}
\Pr \left[ \Iprod{M , x \otimes y}^2 \leq \zeta \expecf{}{\Iprod{M , x \otimes y}^2}   \right] \leq \bigO{\sqrt{\zeta}} \;.
\end{equation*}
\end{lemma}


\subsection{Sum-of-Squares Proofs and Pseudo-distributions}
\label{sec:sos_proofs}

We refer the reader to the monograph~\cite{TCS-086} and the lecture notes~\cite{BarakS16} for a detailed exposition 
of the sum-of-squares method and its usage in average-case algorithm design.

Let $x = (x_1, x_2, \ldots, x_n)$ be a tuple of $n$ indeterminates and let $\R[x]$ be the set of polynomials 
with real coefficients and indeterminates $x_1,\ldots,x_n$. We say that a polynomial $p\in \R[x]$ is a 
\emph{sum-of-squares (sos)} if there exist polynomials $q_1,\ldots,q_r$ such that $p=q_1^2 + \cdots + q_r^2$.

\subsubsection{Pseudo-distributions}

Pseudo-distributions are generalizations of probability distributions.
We can represent a discrete (i.e., finitely supported) probability distribution over $\R^n$ 
by its probability mass function $D\from \R^n \to \R$ such that $D \geq 0$ 
and $\sum_{x \in \mathrm{supp}(D)} D(x) = 1$.
Similarly, we can describe a pseudo-distribution by its mass function 
by relaxing the constraint $D\ge 0$ to passing certain low-degree non-negativity tests.

Concretely, a \emph{level-$\ell$ pseudo-distribution} is a finitely-supported function $D:\R^n \rightarrow \R$ 
such that $\sum_{x} D(x) = 1$ and $\sum_{x} D(x) f(x)^2 \geq 0$ for every polynomial $f$ of degree at most $\ell/2$.
(Here, the summations are over the support of $D$.)
A straightforward polynomial-interpolation argument shows that every level-$\infty$-pseudo distribution 
satisfies $D\ge 0$ and is thus an actual probability distribution.
We define the \emph{pseudo-expectation} of a function $f$ on $\R^n$ with respect to a pseudo-distribution $D$, 
denoted $\pE_{D(x)} f(x)$, as
\begin{equation}
  \pE_{D(x)} f(x) = \sum_{x} D(x) f(x) \,\mper
\end{equation}
The degree-$\ell$ moment tensor of a pseudo-distribution $D$ is the tensor $\pE_{D(x)} (1,x_1, x_2,\ldots, x_n)^{\otimes \ell}$.
In particular, the moment tensor has an entry corresponding to the pseudo-expectation of all monomials of degree at most $\ell$ in $x$.
The set of all degree-$\ell$ moment tensors of probability distribution is a convex set.
Similarly, the set of all degree-$\ell$ moment tensors of degree-$d$ pseudo-distributions is also convex.
Unlike moments of distributions, there is an efficient separation oracle for moment tensors of pseudo-distributions.
\begin{fact}[\cite{MR939596-Shor87,MR1748764-Nesterov00,MR1846160-Lasserre01,MR3050242-Parrilo13}] 
\label[fact]{fact:sos-separation-efficient}
For any $n,\ell \in \N$, the following set has an $n^{O(\ell)}$-time weak separation oracle 
(in the sense of \cite{MR625550-Grotschel81}):
\begin{equation}
\Set{ \pE_{D(x)} (1,x_1, x_2, \ldots, x_n)^{\otimes d} \mid \text{ degree-d pseudo-distribution $D$ over $\R^n$}}\,\mper
\end{equation}
\end{fact}
This fact, together with the equivalence of weak separation and optimization \cite{MR625550-Grotschel81}, 
allows us to efficiently optimize over pseudo-distributions (approximately) --- this algorithm is referred to as 
the sum-of-squares algorithm. The \emph{level-$\ell$ sum-of-squares algorithm} optimizes over the space 
of all level-$\ell$ pseudo-distributions that satisfy a given set of polynomial constraints (defined below).

\begin{definition}[Constrained pseudo-distributions]
Let $D$ be a level-$\ell$ pseudo-distribution over $\R^n$.
Let $\cA = \{f_1\ge 0, f_2\ge 0, \ldots, f_m\ge 0\}$ be a system of $m$ polynomial inequality constraints.
We say that \emph{$D$ satisfies the system of constraints $\cA$ at degree $r$}, denoted $D \sdtstile{r}{} \cA$, 
if for every $S\subseteq[m]$ and every sum-of-squares polynomial $h$ with 
$\deg h + \sum_{i\in S} \max\set{\deg f_i,r}$, we have that $\pE_{D} h \cdot \prod _{i\in S}f_i  \ge 0$.

We write $D \sdtstile{}{} \cA$ (without specifying the degree) if $D \sdtstile{0}{} \cA$ holds.
Furthermore, we say that $D\sdtstile{r}{}\cA$ holds \emph{approximately} if the above inequalities 
are satisfied up to an error of $2^{-n^\ell}\cdot \norm{h}\cdot\prod_{i\in S}\norm{f_i}$, 
where $\norm{\cdot}$ denotes the Euclidean norm\footnote{The choice of norm is not important here because 
the factor $2^{-n^\ell}$ swamps the effects of choosing another norm.} 
of the coefficients of a polynomial in the monomial basis.
\end{definition}

We remark that if $D$ is an actual (discrete) probability distribution, then we have that 
$D\sdtstile{}{}\cA$ if and only if $D$ is supported on solutions to the constraints $\cA$. 
We say that a system $\cA$ of polynomial constraints is \emph{explicitly bounded} 
if it contains a constraint of the form $\{ \|x\|^2 \leq M\}$.
The following fact is a consequence of \cref{fact:sos-separation-efficient} and \cite{MR625550-Grotschel81}:
\begin{fact}[Efficient Optimization over Pseudo-distributions]
There exists an $(n+ m)^{O(\ell)} $-time algorithm that, given any explicitly bounded and satisfiable 
system\footnote{Here, we assume that the bit complexity of the constraints in $\cA$ is $(n+m)^{O(1)}$.} 
$\cA$ of $m$ polynomial constraints in $n$ variables, outputs a level-$\ell$ pseudo-distribution 
that satisfies $\cA$ approximately. \label{fact:eff-pseudo-distribution}
\end{fact}

\paragraph{Basic Facts about Pseudo-Distributions.}

We will use the following Cauchy-Schwarz inequality for pseudo-distributions:
\begin{fact}[Cauchy-Schwarz for Pseudo-distributions] \label{fact:pseudo-expectation-cauchy-schwarz}
Let $f, g$ be polynomials of degree at most $d$ in indeterminate $x \in \R^d$. 
Then, for any degree-$d$ pseudo-distribution $\tzeta$, we have that
$\pE_{\tzeta}[fg] \leq \sqrt{\pE_{\tzeta}[f^2]} \sqrt{\pE_{\tzeta}[g^2]}$.
\end{fact}

\begin{fact}[H{\"o}lder's Inequality for Pseudo-Distributions] \label{fact:pseudo-expectation-holder}
Let $f, g$ be polynomials of degree at most $d$ in indeterminate $x \in \R^d$. 
Fix $t \in \N$. Then, for any degree-$d t$ pseudo-distribution $\tzeta$, we have that
$\pE_{\tzeta}[f^{t-1}g] \leq \paren{\pE_{\tzeta}[f^t]}^{\frac{t-1}{t}} \paren{\pE_{\tzeta}[g^t]}^{1/t}$.
\end{fact}

\begin{corollary}[Comparison of Norms] \label{fact:comparison-of-pseudo-expectation-norms}
Let $\tzeta$ be a degree-$t^2$ pseudo-distribution over a scalar indeterminate $x$. 
Then, we have that $\pE[x^{t}]^{1/t} \geq \pE[x^{t'}]^{1/t'}$ for every $t' \leq t$. 
\end{corollary}



\subsubsection{Sum-of-squares proofs}

Let $f_1, f_2, \ldots, f_r$ and $g$ be multivariate polynomials in $x$.
A \emph{sum-of-squares proof} that the constraints $\{f_1 \geq 0, \ldots, f_m \geq 0\}$ imply the constraint $\{g \geq 0\}$ 
consists of  polynomials $(p_S)_{S \subseteq [m]}$ such that
\begin{equation}
g = \sum_{S \subseteq [m]} p_S \cdot \Pi_{i \in S} f_i \mper
\end{equation}
We say that this proof has \emph{degree $\ell$} if for every set $S \subseteq [m]$ 
the polynomial $p_S \Pi_{i \in S} f_i$ has degree at most $\ell$.
If there is a degree $\ell$ SoS proof that $\{f_i \geq 0 \mid i \leq r\}$ implies $\{g \geq 0\}$, we write:
\begin{equation}
\{f_i \geq 0 \mid i \leq r\} \sststile{\ell}{}\{g \geq 0\} \mper
\end{equation}
For all polynomials $f, g\colon \R^n \to \R$ and for all functions $F\colon \R^n \to \R^m$, 
$G\colon \R^n \to \R^k$, $H\colon \R^{p} \to \R^n$ such that each of the coordinates of the outputs are polynomials of the inputs, 
we have the following inference rules.

The first one derives new inequalities by addition or multiplication:
\begin{equation} \label{eq:sos-addition-multiplication-rule}
\frac{\cA \sststile{\ell}{} \{f \geq 0, g \geq 0 \} } {\cA \sststile{\ell}{} \{f + g \geq 0\}}, \frac{\cA \sststile{\ell}{} \{f \geq 0\}, \cA \sststile{\ell'}{} \{g \geq 0\}} {\cA \sststile{\ell+\ell'}{} \{f \cdot g \geq 0\}}\mper
\end{equation}
The next one derives new inequalities by transitivity: 
\begin{equation} \label{eq:sos-transitivity}
\frac{\cA \sststile{\ell}{} \cB, \cB \sststile{\ell'}{} C}{\cA \sststile{\ell \cdot \ell'}{} C}\mper 
\end{equation}
Finally, the last rule derives new inequalities via substitution:
\begin{equation} \label{eq:sos-substitution}
\frac{\{F \geq 0\} \sststile{\ell}{} \{G \geq 0\}}{\{F(H) \geq 0\} \sststile{\ell \cdot \deg(H)} {} \{G(H) \geq 0\}}\tag{substitution}\mper
\end{equation}
Low-degree sum-of-squares proofs are sound and complete if we take low-level pseudo-distributions as models.
Concretely, sum-of-squares proofs allow us to deduce properties of pseudo-distributions that satisfy some constraints.
\begin{fact}[Soundness] \label{fact:sos-soundness}
If $D \sdtstile{r}{} \cA$ for a level-$\ell$ pseudo-distribution $D$ and there exists a sum-of-squares proof 
$\cA \sststile{r'}{} \cB$, then $D \sdtstile{r\cdot r'+r'}{} \cB$.
\end{fact}
If the pseudo-distribution $D$ satisfies $\cA$ only approximately, soundness continues to hold 
if we require an upper bound on the bit-complexity of the sum-of-squares $\cA \sststile{r'}{} B$ 
(i.e., the number of bits required to write down the proof). In our applications, the bit complexity of all 
sum-of-squares proofs will be $n^{O(\ell)}$ (assuming that all numbers in the input have bit complexity $n^{O(1)}$). 
This bound suffices in order to argue about pseudo-distributions that satisfy polynomial constraints approximately.

The following fact shows that every property of low-level pseudo-distributions can be derived by low-degree sum-of-squares proofs.
\begin{fact}[Completeness] \label{fact:sos-completeness}
Suppose that $d \geq r' \geq r$ and $\cA$ is a collection of polynomial constraints with degree at most $r$, 
and $\cA \sststile{}{} \{ \sum_{i = 1}^n x_i^2 \leq B\}$ for some finite $B$.
Let $\{g \geq 0 \}$ be a polynomial constraint. If every degree-$d$ pseudo-distribution that satisfies $D \sdtstile{r}{} \cA$ 
also satisfies $D \sdtstile{r'}{} \{g \geq 0 \}$, then for every $\epsilon > 0$, 
there is a sum-of-squares proof $\cA \sststile{d}{} \{g \geq - \epsilon \}$.
\end{fact}

\paragraph{Basic Sum-of-Squares Proofs.}
We will require the following basic SoS proofs.

\begin{fact}[Operator norm Bound] \label{fact:operator_norm}
Let $A$ be a symmetric $d\times d$ matrix and $v$ be a vector in $\mathbb{R}^d$. Then, we have that
\[\sststile{2}{v} \Set{ v^{\top} A v \leq \|A\|_2\|v\|^2_2 } \;.\]
\end{fact}

\begin{fact}[SoS H{\"o}lder's Inequality] \label{fact:sos-holder}
Let $f_i, g_i$, for $1 \leq i \leq s$, be scalar-valued indeterminates. 
Let $p$ be an even positive integer.  
Then, 
\[
\sststile{p^2}{f,g} \Set{  \Paren{\frac{1}{s} \sum_{i = 1}^s f_i g_i^{p-1}}^{p} \leq \Paren{\frac{1}{s} \sum_{i = 1}^s f_i^p} \Paren{\frac{1}{s} \sum_{i = 1}^s g_i^p}^{p-1}}\mper
\]
\end{fact}
\noindent Observe that using $p =2$ above yields the SoS Cauchy-Schwarz inequality. 

\begin{fact}[SoS Almost Triangle Inequality] \label{fact:almost-triangle-sos}
Let $f_1, f_2, \ldots, f_r$ be indeterminates. Then, we have that
\[
\sststile{2t}{f_1, f_2,\ldots,f_r} \Set{ \Paren{\sum_{i\leq r} f_i}^{2t} \leq r^{2t-1} \Paren{\sum_{i =1}^r f_i^{2t}}}\mper
\]
\end{fact}

\begin{fact}[SoS AM-GM Inequality, see Appendix A of~\cite{MR3388192-Barak15}] \label{fact:sos-am-gm}
Let $f_1, f_2,\ldots, f_m$ be indeterminates. Then, we have that
\[
\sststile{m}{f_1, f_2,\ldots, f_m} \Set{ \Paren{\frac{1}{m} \sum_{i =1}^n f_i }^m \geq \Pi_{i \leq m} f_i} \mper
\]
\end{fact}

We defer the proofs of the two subsequent lemmas to Appendix \ref{sec:omitted_proofs}.

\begin{lemma}[Spectral SoS Proofs] \label{lem:spectral-sos-proofs}
Let $A$ be a $d \times d$ matrix. Then for $d$-dimensional vector-valued indeterminate $v$, we have:
\[
\sststile{2}{v} \Set{ v^{\top}Av \leq \Norm{A}_2 \Norm{v}_2^2}\mper
\]
\end{lemma}

\begin{fact}[Cancellation within SoS, Lemma 9.2~\cite{BK20}]
Let $a,C$ be scalar-valued indeterminates. Then, 
\[
\Set{a \geq 0} \cup \Set{a^t \leq C a^{t-1}} \sststile{2t}{a,C} \Set{a^{2t} \leq C^{2t}}\mper
\]
\label{fact:sos-cancellation}
\end{fact}

\begin{lemma}[Frobenius Norms of Products of Matrices] \label{lem:frob-of-product}
Let $B$ be a $d \times d$ matrix valued indeterminate for some $d \in \N$. Then, for any $0 \preceq A \preceq I$, 
\[
\sststile{2}{B} \Set{\Norm{AB}_F^2 \leq \Norm{B}_F^2}\mcom
\]
and, 
\[
\sststile{2}{B} \Set{\Norm{BA}_F^2 \leq \Norm{B}_F^2}\mcom
\]
\end{lemma}

\subsection{Analytic Properties of Gaussian Distributions} \label{ssec:prelim-analytic}

The following definitions and results describe the analytic properties of Gaussian distributoins that we will use. We also state the guarantees of known robust estimation algorithms for estimating the mean, covariance and moment tensors of Gaussian mixtures here.

\paragraph{Certifiable Subgaussianity.}

We will make essential use of the following definition.

\begin{definition}[Certifiable Subgaussianity (Definition 5.1 in~\cite{KS17})]
For $t \in \N$ and an absolute constant $C> 0$, a distribution $\cD$ on $\R^d$ is said to be $t$-certifiably $C$-subgaussian
if for every even $t' \leq t$, we have that
\[
\sststile{t'}{v} \Set{ \E_{\cD} \iprod{x,v}^{t'} \leq (Ct')^{t'/2} \Paren{\E_{\cD} \iprod{x,v}^{2}}^{t'/2} }\mper
\]
\end{definition}

\begin{fact}[Mixtures of Certifiably Subgaussian Distributions, Analogous to Lemma 5.4 in~\cite{KS17}] \label{fact:subgaussianity-of-mixtures}
Let $\cD_1, \cD_2, \ldots, \cD_q$ be $t$-certifiably $C$-subgaussian distributions on $\R^d$.
Let $p_1, p_2, \ldots, p_q$ be non-negative weights such that $\sum_i p_i = 1$ and $p = \min_{i \leq q}p_i$.
Then, the mixture $\sum_{i} p_i \cD_i$ is $t$-certifiably $C/p$-subgaussian.
\end{fact}


\paragraph{Certifiable Anti-Concentration.}

The first is \emph{certifiable anti-concentration} --- an SoS formulation of classical anti-concentration inequalities --- that was introduced in~\cite{karmalkar2019list, raghavendra2020list}.

In order to formulate certifiable anti-concentration, we start with a univariate even polynomial $p$ that serves as a uniform approximation to the delta function at $0$ in an interval around $0$. Such polynomials are constructed
in~\cite{karmalkar2019list,raghavendra2020list} (see also~\cite{DGJ+:10}).
Let $q_{\delta,\Sigma}(x,v)$ be a multivariate (in $v$) polynomial defined by
$q_{\delta, \Sigma}(x,v) = \Paren{v^{\top}\Sigma v}^{2s} p_{\delta,\Sigma}\Paren{\frac{\iprod{x,v}}{\sqrt{v^{\top}\Sigma v}}}$.
Since $p_{\delta,\Sigma}$ is an even polynomial, $q_{\delta,\Sigma}$ is a polynomial in $v$.

\begin{definition}[Certifiable Anti-Concentration] \label{def:certifiable-anti-concentration}
A mean-$0$ distribution $D$ with covariance $\Sigma$ is $2s$-certifiably $(\delta,C\delta)$-anti-concentrated
if for $q_{\delta,\Sigma}(x,v)$ defined above,  there exists a degree-$2s$ sum-of-squares proof of the following
two unconstrained polynomial inequalities in indeterminate $v$:
\[
\Set{\iprod{x,v}^{2s} + \delta^{2s} q_{\delta,\Sigma}(x,v)^2 \geq \delta^{2s} \Paren{v^{\top} \Sigma v}^{2s}} \text{ , } \Set{\E_{x \sim D} q_{\delta,\Sigma}(x,v)^2 \leq C\delta \Paren{v^{\top}\Sigma v}^{2s}}\mper
\]
An isotropic subset $X \subseteq  \R^d$ is $2s$-certifiably $(\delta,C\delta)$-anti-concentrated
if the uniform distribution on $X$ is $2s$-certifiably $(\delta,C\delta)$-anti-concentrated.
\end{definition}
\begin{remark}
The function $s(\delta)$ can be taken to be $O(\frac{1}{\delta^2})$ for standard Gaussian distribution and the uniform distribution on the unit sphere (see~\cite{karmalkar2019list} and~\cite{bakshi2020list}). 
\end{remark}

\paragraph{Certifiable Hypercontractivity.}
Next, we define \emph{certifiable hypercontractivity} of degree-$2$ polynomials that formulates (within SoS)
the fact that higher moments of degree-$2$ polynomials of distributions (such as Gaussians)
can be bounded in terms of appropriate powers of their 2nd moment.

\begin{definition}[Certifiable Hypercontractivity] \label{def:certifiable-hypercontractivity}
An isotropic distribution $\cD$ on $\R^d$ is said to be $h$-certifiably $C$-hypercontractive
if there is a degree-$h$ sum-of-squares proof of the following unconstrained polynomial inequality
in $d \times d$ matrix-valued indeterminate $Q$:
\[ \E_{x \sim \cD} \paren{x^{\top}Qx}^{h} \leq \paren{Ch}^{h} \Paren{\E_{x \sim \cD} \paren{x^{\top}Qx}^2}^{h/2}\mper
\]
A set of points $X \subseteq \R^d$ is said to be  $C$-certifiably hypercontractive
if the uniform distribution on $X$ is  $h$-certifiably $C$-hypercontractive.
\end{definition}

Hypercontractivity is an important notion in high-dimensional probability and analysis on product spaces~\cite{MR3443800-ODonnell14}. Kauers, O'Donnell, Tan and Zhou~\cite{DBLP:conf/soda/KauersOTZ14} showed certifiable hypercontractivity of Gaussians and more generally product distributions with subgaussian marginals. Certifiable hypercontractivity strictly generalizes the better known \emph{certifiable subgaussianity} property (studied first in~\cite{KS17}) that controls higher moments of linear polynomials.


Observe that the definition above is affine invariant. In particular, we immediately obtain:

\begin{fact}\label{fact:affine-invariance-certifiable-hypercontractivity}
Given $t\in\mathbb{N}$,  if a random variable $x$ on $\R^d$ has $t$-certifiable $C$-hypercontractive degree-$2$ polynomials, then so does $Ax$ for any $A \in \R^{d \times d}$.
\end{fact}

As observed in \cite{KS17}, the Gaussian distribution is $t$-certifiably $1$-subgaussian and $t$-certifiable
$1$-hypercontractive for every $t$. Next, we establish certifiable hypercontractivity for mixtures of Gaussians. We defer the proofs to Appendix \ref{sec:omitted_proofs}.

\begin{lemma}[Shifts Cannot Decrease Variance]\label{lem:shifts-only-increase-variance}
Let $\cD$ be a distribution on $\R^d$, $Q$ be a $d \times d$ matrix-valued indeterminate, and $C$ be a scalar-valued indeterminate.
Then, we have that
\[
\sststile{2}{Q,C} \Set{\E_{x\sim \cD} \left[ \Paren{Q(x)-\E_{x\sim\cD}[Q(x)] }^2 \right] \leq \E_{x\sim\cD} \left[\Paren{Q(x)-C}^2\right] }\mper
\]
\end{lemma}

\begin{lemma}[Shifts of Certifiably Hypercontractive Distributions] \label{lem:shifts-of-certifiably-hypercontractive-distributions}
Let $x$ be a mean-$0$ random variable with distribution $\cD$ on $\R^d$ with $t$-certifiably $C$-hypercontractive degree-$2$ polynomials. Then, for any fixed constant vector $c \in \R^d$, the random variable $x+c$ also has $t$-certifiable $4C$-hypercontractive degree-$2$ polynomials.
\end{lemma}

\begin{lemma}[Mixtures of Certifiably Hypercontractive Distributions] \label{lem:mixtures-of-certifiably-hypercontractive-distributions}
Let $\cD_1, \cD_2, \ldots, \cD_k$ have $t$-certifiable $C$-hypercontractive degree-$2$ polynomials on $\R^d$,
for some fixed constant $C$. Then, any mixture $\cD=\sum_i w_i \cD_i$ also has $t$-certifiably $(C/\alpha)$-hypercontractive
degree-$2$ polynomials for $\alpha = \min_{i \leq k, w_i > 0} w_i$.
\end{lemma}

\begin{corollary}[Certifiable Hypercontractivity of Mixtures of $k$ Gaussians] \label{cor:hypercontractivity_of_mixtures}
Let $\calM$ be a $k$-mixture of Gaussians $\sum_i w_i \cN(\mu_i,\Sigma_i)$ with weights $w_i \geq \alpha$ for every $i \in [k]$.
Then, for all $t\in\N$, $\cD$ has $t$-certifiably $4/\alpha$-hypercontractive degree-$2$ polynomials.
\end{corollary}


We will use the following robust mean estimation algorithm for bounded covariance distributions~\cite{DKK+17}:
\begin{fact}[Robust Mean Estimation for Bounded Covariance Distributions] \label{fact:robust-mean-estimation}
There is a $\poly(n)$ time algorithm that takes input an $\epsilon$-corruption $Y$ of a collection
of $n$ points $X \subseteq \R^d$, 
and outputs an estimate $\hat{\mu}$ satisfying
$\Norm{\E_{x \sim_u X} x - \hat{\mu}}_2 \leq
O(\sqrt{\epsilon}) \Norm{\E_{x \sim_u X} (x-\E_{x \sim_u X}x)(x-\E_{x \sim_u X}x)^{\top}}_{2}$.
\end{fact}

We will use the following robust covariance estimation algorithm from~\cite{KS17}:
\begin{fact}[Robust Covariance Estimation, ~\cite{KS17}]  \label{fact:ks-cov-est}
For every $C>0, \epsilon >0$ and even $k \in \N$ such that $Ck \epsilon^{1-2/k} \leq c$ for some small enough absolute constant $c$, there exists a polynomial-time algorithm that given an (corrupted) sample $S$ outputs an estimate of the covariance $\hat{\Sigma} \in \R^{d \times d}$ with the following guarantee: there
exists $n_0 \geq (C+d)^{O(k)}/\epsilon$ such that if $S$ is an $\epsilon$-corrupted sample with size $|S| \geq n_0$ of a $k$-certifiably $C$-subgaussian distribution $D$ over $\R^d$ with mean $\mu \in \R^d$
and covariance $\Sigma \in \R^{d \times d}$, then with high probability:

\[
(1-\delta)\Sigma \succeq \hat{\Sigma} \succeq (1+\delta) \Sigma 
\]
for $\delta\leq O(Ck) \epsilon^{1-2/k}$.
\end{fact}

We will also require the following robust estimation algorithm with Frobenius distance guarantees proven
for certifiably hypercontractive distributions in~\cite{BK20}. Since we obtain estimates to the true covariance in Lowner ordering, we can obtain the subspace spanned by the inliers exactly, project on to this subspace and apply Theorem 7.1 in~\cite{BK20}.

\begin{fact}[Robust Mean and Covariance Estimation for Certifiably Hypercontractive Distributions,
Theorem 7.1 in~\cite{BK20}] \label{fact:param-estimation-main}
Given $t \in \mathbb{N}$, and $\epsilon > 0$ sufficiently small so that $Ct\epsilon^{1-4/t} \ll 1$\footnote{This notation means that we needed $Ct\epsilon^{1-2/t}$ to be at most $c_0$ for some absolute constant $c_0 > 0$.}, for some absolute constant $C>0$.
Then, there is an algorithm that takes input $Y$, an $\epsilon$-corruption of a sample $X$ of size $n$
with mean $\mu_*$, covariance $\Sigma_*$, and $2t$-certifiably $C$-hypercontractive degree-$2$
polynomials, runs in time $n^{O(t)}$, and outputs an estimate $\hat{\mu}$ and $\hat{\Sigma}$ satisfying:
\begin{enumerate}
\item $\Norm{\Sigma^{\dagger/2}_*(\mu_*-\hat{\mu})}_2 \leq O(Ct)^{1/2} \epsilon^{1-1/t}$,
\item $(1-\eta) \Sigma_* \preceq \hat{\Sigma} \preceq (1+\eta) \Sigma_*$ for $\eta \leq O(Ck) \epsilon^{1-2/t}$, and,
\item $\Norm{\Sigma_*^{\dagger/2} \Paren{ \hat{\Sigma} - \Sigma_*} \Sigma_*^{\dagger/2} }_F \leq (C't) O(\epsilon^{1-1/t})$,
\end{enumerate}
\nnnew{where $C'=\max\Set{C,B}$ for the smallest possible $B>0$ such that for $d \times d$-matrix-valued indeterminate $Q$, $\sststile{2}{Q}\Set{\E_{\cD} \Paren{x^{\top}Qx-\E_{\cD}x^{\top}Qx}^2 \leq B \Norm{\Sigma_*^{1/2} Q\Sigma_*^{1/2}}_F^2}$.\footnote{The first two guarantees here hold for the larger class of certifiably subgaussian distributions and were proven in~\cite{KS17} (see Theorem  1.2). Gaussian distribution (with arbitrary mean and covariance) are $t$-certifiably $1$-subgaussian for all $t$ and their mixtures (similar to Lemma~\ref{lem:mixtures-of-certifiably-hypercontractive-distributions} and explicitly proven in Lemma 5.4 of~\cite{KS17}) are $t$-certifiably $O(1/\alpha)$-subgaussian where $\alpha$ is the minimum mixing weight.}}
\end{fact}

The last line in the above fact asserts a bound (along with a degree $2$ SoS proof) on the variance of degree $2$ polynomials in terms of the Frobenius norm of its coefficient matrix. In the next few claims, we verify this property via  elementary arguments for the two classes of distributions relevant to this paper. We note that whenever a distribution  satisfies the bounded variance property (without an SoS proof), it also satisfies the property via a degree $2$ SoS proof using Lemma~\ref{lem:spectral-sos-proofs}. Thus, asking for an SoS proof of degree $2$ in this context poses no additional restrictions on the distribution. Nevertheless, we provide explicit and direct SoS proofs in the following.

We first note that this property of having \emph{certifiable bounded variance} is closed under linear transformations. 

\begin{lemma}[Linear Transformations of Certifiably Bounded-Variance Distributions] \label{lem:bounded-variance-linear-transform}
For $d \in \N$, let $x$ be a random variable with distribution $\cD$ on $\R^d$ such that for $d \times d$ matrix-valued indeterminate $Q$, $\sststile{2}{Q} \Set{\E_{x \sim \cD}(x^{\top}Qx-\E_{\cD}x^{\top}Qx)^2 \leq \Norm{\Sigma^{1/2}Q\Sigma^{1/2}}^2_F }$. Let $A$ be an arbitrary $d \times d$ matrix and let $x' = Ax$ be the random variable with covariance $\Sigma' = AA^{\top}$. Then, we have that
\[
\sststile{2}{Q} \Set{\E_{x' \sim \cD'}({x'}^{\top}Qx'-\E_{\cD'}{x'}^{\top}Qx')^2 \leq \Norm{{\Sigma'}^{1/2}Q{\Sigma'}^{1/2}}^2_F}\mper
\]
\end{lemma}

\begin{lemma}[Variance of Degree-$2$ Polynomials of Standard Gaussians] \label{lem:var-zero-mean-gaussians}
We have that
\[
\sststile{2}{Q} \Set{\E_{\cN(0,I)} \Paren{x^{\top}Qx - \E_{\cN(0,I)} x^{\top}Qx}^2 \leq 3 \Norm{Q}_F^2}\mper
\]
\end{lemma}
\begin{remark}
As is easy to verify, the same proof more generally holds for any distribution that has the same first four 
moments as the zero-mean Gaussian distribution.
\end{remark}

As an immediate corollary of the previous two lemmas, we have:
\begin{corollary}[Variance of Degree-$2$ Polynomials of Zero-Mean, Arbitrary Covariance Gaussians] 
\label{lem:var-zero-mean-arbitrary-cov-gaussians}
For any $0 \preceq \Sigma$, we have that
\[
\sststile{2}{Q} \Set{\E_{\cN(0,\Sigma)} \Paren{x^{\top}Qx - \E_{\cN(0,\Sigma)} x^{\top}Qx}^2 \leq 3 \Norm{\Sigma^{1/2}Q\Sigma^{1/2}}_F^2}\mper
\]
\end{corollary}

We next prove that the same property holds for mixtures of Gaussians satisfying certain conditions.


\begin{lemma}[Variance of Degree-$2$ Polynomials of Mixtures]
\label{fact:mean-variance-subgaussian} 
Let $\calM = \sum_i w_i \cD_i$ be a $k$-mixture of distributions $\cD_1, \cD_2, \ldots, \cD_k$ with means $\mu_i$ and covariances $\Sigma_i$. Let $\mu = \sum_i w_i \mu_i$ be the mean of $\calM$. Suppose that each of $\cD_1,\cD_2, \ldots, \cD_k$ have certifiably $C$-bounded-variance i.e. for $Q$: a symmetric $d \times d$ matrix-valued indeterminate. 
\[
\sststile{2}{Q} \Set{\E_{x' \sim \cD_i}({x'}^{\top}Qx'-\E_{\cD_i}{x'}^{\top}Qx')^2 \leq C\Norm{{\Sigma'}^{1/2}Q{\Sigma'}^{1/2}}^2_F}\mper
\]
Further, suppose that for some $H > 1$, $\norm{\mu_i-\mu}_2^2, \Norm{\Sigma_i-I}_F \leq H$ for every $1 \leq i \leq k$. 
Then, we have that
\[
\sststile{2}{Q} \Set{ \expecf{x\sim \calM}{ \Paren{ x^\top Qx-\expecf{x \sim \calM}{x^\top Qx } }^2 } \leq 100C H^2\Norm{Q}_F^2 }\mper
\]
\end{lemma}


As an immediate corollary of Lemma~\ref{lem:bounded-variance-linear-transform} and Lemma~\ref{fact:mean-variance-subgaussian}, 
we obtain:

\begin{lemma}[Variance of Degree-$2$ Polynomials of Mixtures of Gaussians]
\label{fact:mean-variance-subgaussian-arbitrary-covariance}
Let $\calM = \sum_i w_i \cN(\mu_i,\Sigma_i)$ be a $k$-mixture of Gaussians with $w_i \geq \alpha$, 
mean $\mu = \sum_i w_i \mu_i$ and covariance $\Sigma = \sum_i w_i ((\mu_i-\mu)(\mu_i-\mu)^{\top} + \Sigma_i)$. 
Suppose that for some $H > 1$, $\Norm{\Sigma^{\dagger/2}(\Sigma_i-\Sigma)\Sigma^{\dagger/2}}_F \leq H$ for every $1 \leq i \leq k$. 
Let $Q$ be a symmetric $d \times d$ matrix-valued indeterminate. Then for $H' = \max \{H, 1/\alpha\}$, 
\[
\sststile{2}{Q} \Set{ \expecf{x\sim \calM}{ \Paren{ x^\top Qx-\expecf{x \sim \calM}{x^\top Qx } }^2 } \leq 100 {H'}^2\Norm{\Sigma^{1/2}Q\Sigma^{1/2}}_F^2 }\mper
\]
\end{lemma}

\paragraph{Analytic Properties are Inherited by Samples.}

The following lemma can be proven via similar, standard techniques as in several prior works~\cite{KS17,karmalkar2019list,
bakshi2020list,BK20}.

\begin{fact} \label{fact:moments-to-analytic-properties}
Let $D$ be a distribution on $\R^d$ with mean $\mu$ and covariance $\Sigma$. Let $t \in \N$. Let $X$ be a sample from $D$ such that, $\Norm{\frac{1}{|X|} \sum_{x \in X} (1,\bar{x})^{\otimes t} - \E_{x \sim D} (1,\bar{x})^{\otimes t}}_F \leq d^{-O(t)}$. Here, $\bar{x} = \Sigma^{\dagger/2}(x-\mu_i)$. Then,
\begin{enumerate}
\item If $D$ is $2t$-certifiably $C$-subgaussian, then the uniform distribution on $X$ is $t$-certifiably $2C$-subgaussian.
\item If $D$ has $2t$-certifiably $C$-hypercontractive degree $2$ polynomials, then the uniform distribution on $X$ has $t$-certifiably $2C$-hypercontractive degree $2$ polynomials.
\item If $D$ is $2t$-certifiably $C\delta$-anti-concentrated, then the uniform distribution on $X$ is $t$-certifiably $2C\delta$-anti-concentrated.
\nnnew{\item If $\sststile{2}{Q}\Set{ \E_{x \sim \cD}(x^{\top}Qx-\E_{x \sim \cD}x^{\top}Qx)^2 \leq C \Norm{Q}_F^2}$, then, for the uniform distribution $\cD_X$ on $X$, $\sststile{2}{Q}\Set{ \E_{x \sim \cD_X}(x^{\top}Qx-\E_{x \sim \cD_X}x^{\top}Qx)^2 \leq 2C \Norm{Q}_F^2}.$}
\end{enumerate}
\end{fact}

\subsection{Deterministic Conditions on the Uncorrupted Samples} \label{ssec:det-conds}

\new{
In this section, we describe the set of deterministic conditions on the set of uncorrupted samples,
under which our algorithms succeed. We will require the following definition.

\begin{definition}\label{def:eps-corrupted}
Fix $0< \eps<1/2$.
We say that a multiset $Y$ of points in $\R^d$ is an $\eps$-corrupted version (or an $\eps$-corruption)
of a multiset $X$ of points in $\R^d$ if $|X \cap Y| \geq \max\{ (1-\eps)|X|,(1-\eps)|Y| \} $.
\end{definition}

Throughout this paper and unless otherwise specified, we will use $X$ to denote a multiset of i.i.d.
samples from the target $k$-mixture ${\cal M} = \sum_{i=1}^k w_i G_i$, where $G_i = \cN(\mu_i, \Sigma_i)$.
We will use $X_i$ for the subset of points in $X$ drawn from $G_i$, i.e., $X = \cup_{i=1}^k X_i$.

We will use $Y$ to denote an $\eps$-corrupted version of $X$, as per Definition~\ref{def:eps-corrupted}.
In this {\em strong contamination model}, the adversary can see the clean samples from $X$ before
they decide on the $\eps$-corruption $Y$.
The strong contamination model is known to subsume the total variation contamination of
Definition~\ref{def:adv} (see, e.g., Section 2 of~\cite{DKKLMS16}). We note that our robust learning
algorithm succeeds in this stronger contamination model, with the additional requirement that we can obtain two sets
of independent $\eps$-corrupted samples from ${\cal M}$. (The second set is needed to run a hypothesis testing
routine after we obtain a small list of candidate hypotheses.)
}

Our algorithm works for any finite set of points in $\R^d$
that satisfies a natural set of deterministic conditions.
As we will show later in this section, these deterministic conditions are
satisfied with high probability by a sufficiently large set of i.i.d. samples from any
$k$-mixture of Gaussians.


\begin{condition}[Good Samples] \label{cond:convergence-of-moment-tensors}
Let $\calM=\sum_{i=1}^k w_i \cN(\mu_i,\Sigma_i)$ be a $k$-mixture of Gaussians in $\R^d$.
\new{Let $X$ be a set of $n$ points in $\R^d$.
We say that $X$ satisfies Condition~\ref{cond:convergence-of-moment-tensors}
with respect to $\calM$ with parameters $(\gamma,t)$ if there is a partition of $X$
as $X_1\cup X_2 \cup \ldots \cup X_k$ such that:
\begin{enumerate}
\item For all $i \in [k]$ with $w_i \geq \gamma$, any positive integer $m\leq t$, and any $v\in \R^d$,
$$\left|\frac{1}{n}\sum_{x\in X_i} \langle v, x-\mu_i \rangle^m - w_i \E_{x\sim \cN(\mu_i,\Sigma_i)}[\langle v, x-\mu_i \rangle^m]\right| \leq w_i \, \gamma \, m! \, (v^T \Sigma_i v)^{m/2} \;. $$
\item For all $i \in [k]$ and any halfspace $H\subset \R^d$, we have that
$\left| |X_i \cap H|/n - w_i \Pr_{x\sim \cN(\mu_i,\Sigma_i)}[x\in H] \right| \leq \gamma.$
\end{enumerate}}
\end{condition}

\new{We will also need the following consequences
of Condition~\ref{cond:convergence-of-moment-tensors}.
The first one is immediate.

\begin{lemma}
\label{lem:covergence_affine_invariant}
Condition \ref{cond:convergence-of-moment-tensors} is invariant under affine transformations.
In particular, if $A(x):\R^d \rightarrow \R^{d'}$ is an affine transformation, and if $X$ satisfies Condition~\ref{cond:convergence-of-moment-tensors}
with respect to $\calM$ with parameters $(\gamma, t)$, then $A(X)$ satisfies
Condition~\ref{cond:convergence-of-moment-tensors} with respect to $A(\calM)$ with parameters $(\gamma,t)$.
\end{lemma}

We note that the first part of Condition~\ref{cond:convergence-of-moment-tensors}
implies that higher moment tensors are close in Frobenius distance.

\begin{lemma}\label{moment closeness lemma}
If $X$ satisfies Condition \ref{cond:convergence-of-moment-tensors} with respect to $\calM=\sum_i w_i \cN(\mu_i,\Sigma_i)$
with parameters $(\gamma,t)$, then if $w_i \geq \gamma$ for all $i \in [k]$,
and if for some $B\geq 0$ we have that $\Norm{\mu_i}_2^2,\Norm{\Sigma_i}_{\mathrm{op}} \leq B$ for all $i \in [k]$,
then for all $m\leq t$, we have that:
$$
\Norm{\E_{x\in_u X}[x^{\otimes m}] - \E_{x\sim M}[x^{\otimes m}]}_F^2 \leq \gamma^2 m^{O(m)}B^m d^m \;.
$$
\end{lemma}

We note that Condition~\ref{cond:convergence-of-moment-tensors} also behaves
well with respect to taking submixtures.

\begin{lemma}\label{submixture condition lemma}
Let $\calM=\sum_i w_i \cN(\mu_i,\Sigma_i)$. Let $S\subset [k]$ with $\sum_{i\in S} w_i = w$, and let
$\calM'=\sum_{i\in S} (w_i/w) \cN(\mu_i,\Sigma_i)$. Then if $X$ satisfies Condition \ref{cond:convergence-of-moment-tensors} with respect to $\calM$ with parameters $(\gamma,t)$ for some $\gamma<1/(2k)$
 with the corresponding partition being $X=X_1\cup X_2 \cup\ldots \cup X_k$,
 then $X' = \bigcup_{i\in S} X_i$ satisfies Condition \ref{cond:convergence-of-moment-tensors}
 with respect to $\calM'$ with parameters $(O(k\gamma/w),t)$.
\end{lemma}

Finally, we show that given sufficiently many i.i.d. samples from a $k$-mixture of Gaussians,
Condition \ref{cond:convergence-of-moment-tensors} holds with high probability.

\begin{lemma} \label{lem:det-suffices}
Let $\calM=\sum_{i=1}^k w_i \cN(\mu_i,\Sigma_i)$ and let $n$ be an integer at least $k t^{Ct} d^{t}/\gamma^3$,
for a sufficiently large universal constant $C>0$, some $\gamma>0$, and some $t \in \N$.
If $X$ consists of $n$ i.i.d. samples from $\calM$,
then $X$ satisfies Condition \ref{cond:convergence-of-moment-tensors}
with respect to $\calM$ with parameters $(\gamma, t)$ with high probability.
\end{lemma}

\noindent The proofs of the preceding lemmas can be found in Appendix~\ref{sec:omitted_proofs}.
}

\subsection{Hypothesis Selection} \label{ssec:tournament}
Our algorithm will require a procedure to select a hypothesis from a list of candidates
that contains an accurate hypothesis. A number of such procedures are known
in the literature (see, e.g.,~\cite{Yatracos85, DL:01, DDS12stoclong, DK14, DDS15, DKKLMS16}).
Here we will use the following variant from \cite{Kane20}, showing that we can efficiently
perform a hypothesis selection (tournament) step with access to $\epsilon$-corrupted samples.

\begin{fact}[Robust Tournament, \cite{Kane20}]\label{lem:tournament}
Let $X$ be an unknown distribution, $\eta \in (0, 1)$,
and let $H_1,\dots,H_n$ be distributions with explicitly computable probability density functions
that can be efficiently sampled from. Assume furthermore than $\min_{1\le i\le n}(\dtv(X,H_i))\le \eta$.
Then there exists an efficient algorithm that given access to $O(\log(n)/\eta^2)$ $\epsilon$-corrupted samples from $X$,
where $\eps \le \eta$, along with $H_1,\dots,H_n$, computes an $m \in [n]$ such that with high probability we have that
$\dtv(X,H_m)=O(\eta) \;.$
\end{fact}

%% file: tensor_decomposition.tex


\section{List-Recovery of Parameters via Tensor Decomposition}
\label{sec:list-recovery}




In this section, we give an algorithm that takes samples from a $k$-mixture of Gaussians,
whose component means and covariances are not too far from each other
in natural norms, and outputs a dimension-independent size list of candidate
$k$-tuples of parameters (i.e., means and covariances) one of which is
guaranteed to be close to the true target $k$-tuple of parameters. Our approach involves a new
tensor decomposition procedure that works in the absence of any non-degeneracy
conditions on the components.

The goal of this section is to prove the following theorem:

\begin{theorem}[Recovering Candidate Parameters when Component Covariances
are close in Frobenius Distance]\label{thm:list-recovery-by-tensor-decomposition}
Fix any $\alpha> \epsilon > 0, \Delta > 0$.
There is an algorithm  that takes input $X$, a sample from a $k$-mixture of Gaussians
$\calM = \sum_i w_i \cN(\mu_i,\Sigma_i)$ satisfying Condition~\ref{cond:convergence-of-moment-tensors}
with parameters \new{$\gamma = \eps d^{-8k}k^{-Ck}$, for $C$ a sufficiently large universal constant,
and $t=8k$, and let $Y$ be an $\eps$-corruption of $X$}.
If $w_i\ge\alpha$, $\Norm{\mu_i}_2 \leq \frac{2}{\sqrt{\alpha}}$ and  $\Norm{\Sigma_i - I}_F \leq \Delta$ for every $i \in [k]$,
then\new{, given $k,Y$ and $\eps$,} the algorithm outputs a list $L$ of at most
$\new{\exp\Paren{ \log(1/\epsilon) \Paren{k + 1/\alpha + \Delta}^{O(k)}/\eta^2 } }$
candidate hypotheses (component means and covariances),
such that \new{with probability at least $0.99$} there exist $\{\hat \mu_i, \hat \Sigma_i\}_{i \in [k]} \subseteq L$ satisfying
$\Norm{\mu_i - \hat \mu_i}_2 \leq \bigO{\frac{\Delta^{1/2}}{\alpha}} \eta^{G(k)}$ and
$\Norm{\Sigma_i - \hat\Sigma_i}_F \leq \bigO{k^4} \frac{\Delta^{1/2}}{\alpha} \eta^{G(k)}$ for all $i \in [k]$.
Here, $\eta = (2 k)^{4k} \bigO{1/\alpha+ \Delta}^{4k} \sqrt{\epsilon}$, $G(k) = \frac{1}{C^{k+1} (k+1)!}$.
The running time of the algorithm is $\poly(|L|, \new{|Y|, d^k})$.
\end{theorem}



In the body of this section, we establish Theorem~\ref{thm:list-recovery-by-tensor-decomposition}.
The structure of this section is as follows:
In Section~\ref{ssec:tensor-alg}, we describe our algorithm, which is then analyzed
in Sections~\ref{ssec:tensor-analysis}-\ref{ssec:mv}.

\subsection{List-Decodable Tensor Decomposition Algorithm} \label{ssec:tensor-alg}

In this section, we describe our tensor decomposition algorithm,
which is given in pseudocode below (Algorithm~\ref{algo:list-recovery-tensor-decomposition}).

\input{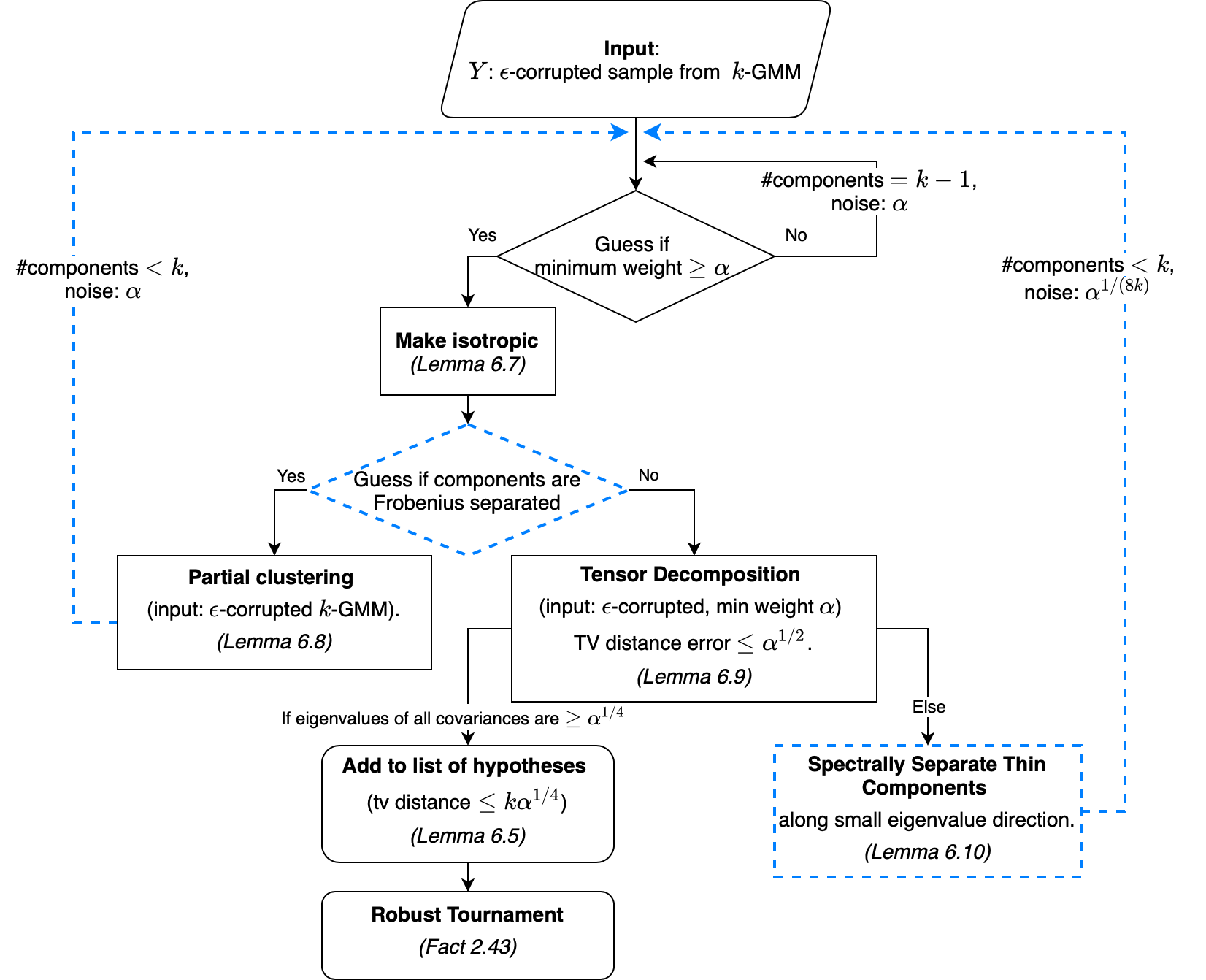}

\subsection{Analysis of Algorithm} \label{ssec:tensor-analysis}


We analyze the three main steps of Algorithm~\ref{algo:list-recovery-tensor-decomposition} in the following lemmas.
We will prove the following three propositions in the subsequent subsections that analyze Steps 1, 2 and 3
of Algorithm \ref{algo:list-recovery-tensor-decomposition}. For Step~1, we show that when $X$ satisfies Condition~\ref{cond:convergence-of-moment-tensors}, the empirical estimates of the moment tensors obtained
by applying the robust mean estimation algorithm to $X$ are sufficiently close to the moment tensors
of the input mixture $\calM$.

\begin{proposition}[Robustly Estimating Hermite Polynomial Tensors]\label{prop:robust-estimation-hermite-tensors}
For any integer $m\leq {4k}$, and $\Delta \in \R_+$, there exists an algorithm with running time $\poly_m(d/\eps)$
that takes an $\epsilon$-corruption $Y$ of $X$, \new{a set satisfying Condition \ref{cond:convergence-of-moment-tensors}
with respect to $\calM = \sum_{i = 1}^k w_i \cN(\mu_i,\Sigma_i)$ with parameters $\gamma = \eps d^{-m}m^{-Cm}$,
for $C$ a sufficiently large constant, and $t=\newblue{2}m$}. If $w_i \geq \alpha$, $\Norm{\mu_i}_2 \leq 2/\sqrt{\alpha}$, and $\Norm{\Sigma_i -I}_F \leq \Delta$ for each $i \in [k]$, then the algorithm outputs a tensor $\hat{T}_m$ such that
$\Norm{ \hat{T}_m - \expecf{}{h_m(\calM)} }_F \leq \eta$,
for $\eta =\mathcal{O}\left(m(1+1/\alpha+\Delta)\right)^m\sqrt{\eps}$.
\end{proposition}

\noindent The proof of Proposition~\ref{prop:robust-estimation-hermite-tensors} is deferred to
Section~\ref{ssec:robust-estimation-hermite-tensors}.

Next, we analyze Step 2 of the algorithm and prove that, with non-negligible probability,
randomly collapsing two modes of $\hat T_4$ yields a matrix $\hat S$ such that
$\hat S - (\Sigma_i-I) = P_i + Q_i$, where $P_i$ has small Frobenius norm
and $Q_i$ is a rank-$\mathcal{O}(k^2)$ matrix.

\begin{proposition}[Tensor Decomposition up to Low-Rank Error]
\label{prop:tensor_decomposition_upto_lowrank}
Let $\calM = \sum_{i = 1}^k w_i \cN(\mu_i,\Sigma_i)$ be a $k$-mixture of Gaussians
satisfying $w_i \geq \alpha$, $\Norm{\mu_i}_2 \leq 2/\sqrt{\alpha}$, and $\Norm{\Sigma_i -I}_F \leq \Delta$ for each $i \in [k]$.
For $0<\eta<1$, let $\hat{T}_4$ be a tensor such that $\Norm{\E[h_4(\calM)]- \hat{T}_4 }_F \leq \eta$,
and let $D$ \newblue{be a sufficiently large constant multiple of} ${ k^4/(\alpha \sqrt{\eta})}$.
For all $j \in [4k]$, let $x^{(j)}, y^{(j)} \sim \calN(0, I)$ be independent
and $a_j \sim \calU[-D, D]$, where $ \calU[-D, D]$ is the uniform distribution over the interval $[-D, D]$,
and let $\hat{S} = \sum_{j\in[4k]} a_j \hat{T}_{4}\Paren{\cdot, \cdot ,x^{(j)},y^{(j)}}$.
Then, for each $i \in [k]$, with probability at least $ \Paren{ \eta /\Paren{ k^5\Paren{\Delta^4 + 1/\alpha^4 }} }^{4k} $,
over the  choice of $x^{(j)},y^{(j)}$ and $a_j$, we have that $\hat{S}-(\Sigma_i-I)= P_i + Q_i$,
where $\Norm{P_i}_F = \bigO{\sqrt{\eta/\alpha}}$, $\Norm{Q_i}_F =\bigO{ \frac{1+\Delta^2}{\sqrt{\eta}\alpha^3}}$
and $\mathrm{rank}(Q_i) = \bigO{k^2}$.
\end{proposition}

\noindent The proof of Proposition~\ref{prop:tensor_decomposition_upto_lowrank}
is given in Section~\ref{ssec:td_up_to_low_rank}.

Finally, in Step~3, for any $\hat S$ such that $\hat S - (\Sigma_i-I) = P_i + Q_i$,
where $P_i$ has small Frobenius norm and $Q_i$ is a rank $\mathcal{O}(k^2)$ matrix,
we find a low-dimensional subspace $V'$ such that the range space of $Q_i$ is approximately contained in $V'$.
We will use $V'$ to exhaustively search for $\mathcal{O}(k^2)$ rank matrices to find candidates for $Q_i$.

\begin{proposition}[Low-Dimensional Subspace $V'$ for Exhaustive Search]
\label{prop:construction-of-low-dim-subspace-for-enumeration}
Let $\calM = \sum_{i = 1}^k w_i \cN(\mu_i,\Sigma_i)$ be a $k$-mixture of Gaussians
satisfying $w_i \geq \alpha$, $\Norm{\mu_i}_2 \leq 2/\sqrt{\alpha}$, and $\Norm{\Sigma_i -I}_F \leq \Delta$ for each $i \in [k]$.
Let $\Norm{\hat T_m - \expecf{}{h_m(\calM)} }_F \leq \eta$, for each $1 \leq m \leq 4k$, and some $\eta >0$.
Let $V$ be the span of all the left singular vectors of the $d \times d^{m-1}$ matrix
obtained by the natural flattening of $\hat T_m$ with singular values at least $ 2\eta$.
For each $1 \leq i \leq k$, let $S_i=\Sigma_i-I$ and $\hat{S_i}$ be a $d \times d$ matrix such that $\hat S_i - S_i = P_i + Q_i$,
where $\Norm{P_i}_F \leq \bigO{ \sqrt{\eta/\alpha}}$, $Q_i$ has rank $\mathcal{O}(k^2)$,
and $\Norm{Q_i}_F \leq \bigO{ \frac{1+\Delta^2}{\sqrt{\eta}\alpha^3}}$.
Let $V'$ be the span of $V$ plus all singular vectors of $\hat{S_i}$ of singular values at least $\delta$ for all $i$.
Then, for $\delta = 2 \eta^{1/(C^{k+1} (k+1)!)}$ with a sufficiently large constant $C>0$, we have that:
\begin{enumerate}
	\item $\dim{V'} \leq \Paren{\calO\Paren{k(1+ 1/\alpha + \Delta)}^{4k+5}}/\eta^2$.
	\item There is a vector $\mu_i' \in V'$ such that $\Norm{\mu_i - \mu_i'}^2_2 \leq \frac{20}{\alpha^2}\sqrt{\delta} \Delta$.
	\item There are $q = \mathcal{O}(k^2)$ unit vectors $v_1, v_2, \ldots, v_q \in V'$ and scalars $\tau_1, \tau_2, \ldots, \tau_q \in \left[ -10(1+\Delta^2)/(\sqrt{\eta}\alpha^5), 10(1+\Delta^2)/(\sqrt{\eta}\alpha^5) \right]$
	such that $\Norm{Q_i-\sum_{i =1}^q \tau_i v_iv_i^\top}_F \leq \bigO{\frac{k^2}{\alpha} \delta^{1/4} \Delta^{1/2}}$.
\end{enumerate}
\end{proposition}

\noindent The proof of Proposition~\ref{prop:construction-of-low-dim-subspace-for-enumeration} is given in
Section~\ref{ssec:td_recover_low_rank}.

We can now use these propositions to complete the proof of Theorem~\ref{thm:list-recovery-by-tensor-decomposition}.

\begin{proof}[Proof of Theorem~\ref{thm:list-recovery-by-tensor-decomposition}]
Using Proposition~\ref{prop:robust-estimation-hermite-tensors}, Step 1 of the algorithm outputs estimates $\hat T_i$ for $i \in [4k]$ such that $\max_{m \in [4k]} \Norm{\hat T_m - \E h_m(\calM)}_F \leq \eta$. Next, %
by the standard coupon collector analysis, using Proposition~\ref{prop:tensor_decomposition_upto_lowrank}
and repeating Step 2 of the algorithm $\ell' = 100\log k\Paren{ \eta /\Paren{ k^5\Paren{\Delta^4 + 1/\alpha^4 }} }^{-4k}$ times,
guarantees that with probability at least $1-1/(100 k)^{100}$, for every $1 \leq i \leq k$, there are $\hat{S}_i \in L$ such that $\hat S_i - (\Sigma_i-I) = P_i + Q_i$ for $P_i,Q_i$ satisfying $\Norm{P_i}_F \leq \sqrt{\eta/\alpha}$, $\Norm{Q_i}_F \leq \frac{1+\Delta^2}{\sqrt{\eta} \alpha^5} $ and $Q_i$ has rank $\mathcal{O}(k^2)$.

Next, Proposition~\ref{prop:construction-of-low-dim-subspace-for-enumeration} implies that for every such $\hat{S}_i \in L'$, we can construct a subspace $V' =V'_{\hat S_i}$ of dimension $\bigO{(k(1+ 1/\alpha + \Delta))^{4k+5} /\new{\eta^2}}$ such that $V'$ contains $\mu_i'$ that satisfies $\Norm{\mu_i - \mu_i'}^2_2 \leq \frac{\Delta}{\alpha^2} \cdot \sqrt{\delta}$, and there is a rank $O(k^2)$ matrix $\hat Q_i$ with range space contained in $V'$ such that $\Norm{Q_i - \hat Q_i}_F \leq \calO(\frac{k^2}{\alpha} \delta^{1/4} \Delta^{1/2})$.

Now, let $V_{\tau} \subseteq V'$ be a $\tau = \delta^{1/4}$-cover, in $\ell_2$-norm,
of vectors with $\ell_2$ norm at most $2/\sqrt{\alpha}$ in $V'$.
Then, since $\Norm{\mu_i}_2 \leq \frac{2}{\sqrt{\alpha}}$, there is a vector $\hat \mu_i \in V_{\tau}$
such that $\Norm{\mu_i - \hat \mu_i}^2_2 \leq \tau + \frac{20}{\alpha^2} \sqrt{\delta} \Delta \leq \frac{40}{\alpha^2} \sqrt{\delta} \Delta$.

Further, there exist $\tau_1, \tau_2, \ldots \tau_{\mathcal{O}(k^2)}$ in a $\tau$-cover of
$[-10(1+\Delta^2)/(\sqrt{\eta}\alpha^5), 10(1+\Delta^2)/(\sqrt{\eta}\alpha^5)]$
and vectors $v_1, v_2, \ldots, v_{\mathcal{O}(k^2)} \in V_\tau$ such that
$\Norm{\sum_{i = 1}^{\mathcal{O}(k^2)} \tau_i v_iv_i^{\top} - Q_i}_F \leq \mathcal{O}(k^4\delta^{1/4} \Delta^{1/2}/\alpha )$.
In particular, $\hat{\Sigma}_i = I+\hat S - \sum_{i = 1}^{\mathcal{O}(k^2)} \tau_i v_iv_i^{\top}$
satisfies
\begin{equation}
\norm{\hat{\Sigma}_i-\Sigma_i}_F = \mathcal{O}(\sqrt{\eta}) +  \bigO{ \frac{ k^4 \delta^{1/4} \Delta^{1/2} } {\alpha} } = \bigO{ \frac{ k^4 \delta^{1/4} \Delta^{1/2} } {\alpha}} \;.
\end{equation}
The size of this search space for every fixed $\hat{S} \in L'$ can be bounded above
by $\Paren{\frac{1+\Delta^2}{\delta \alpha^5}}^{\mathcal{O}(k^5\dim{V'})}$.
Thus, the size of $L$ can be bounded from above by
\begin{equation*}
\begin{split}
k^5 \Paren{ \frac{\Delta^4}{\eta} + \frac{1}{\alpha^4 \eta } }^{4k}  \cdot \Paren{\frac{1+\Delta^2}{\delta \sqrt{\eta} \alpha^5} }^{\mathcal{O}(k^5\dim{V'})}
& \leq  \exp\Paren{ \log(1/\epsilon) \Paren{k + 1/\alpha + \Delta}^{O(k)}/\eta^2 } \;.
\end{split}
\end{equation*}
This completes the proof.
\end{proof}

\subsection{Robust Estimation of Hermite Tensors} \label{ssec:robust-estimation-hermite-tensors}
In this section, we will prove Proposition~\ref{prop:robust-estimation-hermite-tensors}.









\begin{proof}[Proof of Proposition~\ref{prop:robust-estimation-hermite-tensors}]
Consider the uniform distribution on the uncorrupted sample $X$.
We want to analyze the effect of applying the robust mean estimation algorithm (Fact~\ref{fact:robust-mean-estimation})
to the points $h_m(x)$, for $x \in X$. In order for us to apply Fact~\ref{fact:robust-mean-estimation},
we need to ensure that \new{the uniform distribution on} $\{h_m(x)\}_{x \in X}$ has bounded covariance.
This step gives us a good approximation
to $\E_{x \sim_u X} h_m(x)$.  In order for us to obtain an approximation to $\E h_m(\calM)$, we need to bound
the difference between $\E h_m(\calM)$ and $\E_{x \sim_u X} h_m(x)$. We will do both these steps below.

The second part is immediate.
\new{By the definition of $h_m(X)$, we have that
$$
\Norm{\frac{1}{|X|} \sum_{x \in X} h_m(x)-\E h_m(\calM)}_F \leq \sum_{j\leq m/2} m^{2j} d^j \Norm{\frac{1}{|X|} \sum_{x \in X} x^{\otimes (m-2j)}-\E \calM^{\otimes (m-2j)}}_F.
$$
By Lemma \ref{moment closeness lemma}, this is at most $O(1+\Delta+1/\alpha)^m m^{O(m)} d^{m/2}\gamma \leq \eta/2$.
We note that a similar argument bounds
$$\Norm{\frac{1}{|X|} \sum_{x \in X} h_m(x) {\otimes h_m(x)} - \E h_m(\calM) \new{\otimes h_m(\calM)}}_F \leq \eta^2.$$
} 
Let us now verify the first part. We proceed via bounding the operator norm of the covariance of $h_m(\calM)$.
We can then use the bound on the Frobenius norm
$\Norm{\frac{1}{|X|} \sum_{x \in X} h_m(x)\new{\otimes h_m(X)}-\E h_m(\calM)\new{\otimes h_m(\calM)}}_F$
to get a bound on $\Norm{\frac{1}{|X|} \sum_{ x\in X} h_{m}(x)h_m(x)^{\top}}_{\new{\mathrm{op}}}$
(the operator norm of the canonical square flattening of the of the $2m$-th empirical Hermite moment tensor of $X$).
This will complete the proof.

Let $G_i=\cN(\mu_i,\Sigma_i)$ be the components of $\calM$.
We have that
\begin{equation}
\begin{split}
\cov(h_m(\calM)) =& \sum_{i\in[k]} w_i \cov(h_m(G_i)) \\
&+ \frac{1}{2}\sum_{i,j\in[k]} w_iw_j \Paren{\E[h_m(G_i)]-\E[h_m(G_j)]} \Paren{\E[h_m(G_i)]-\E[h_m(G_j)]}^\top \;.
\end{split}
\end{equation}
By Lemma \ref{Hermite Covariance Lemma}, we have that for all $i\in [k]$, it holds
\begin{equation*}
\Norm{\cov(h_m(\new{G}_i))}_{\textrm{op} }
=\mathcal{O}\left(m(1+\norm{\mu_i}_2+\norm{\Sigma_i-I}_F)\right)^{2m}=\mathcal{O}\left(m(1+2/\sqrt{\alpha}+\Delta)\right)^{2m} \;,
\end{equation*}
where for any matrix $M$,  $\Norm{M}_{\textrm{op}} = \max_{ \Norm{u}_2=1 } \Norm{Mu}_2$ is the operator norm of the matrix.
Further, for any $i, j\in[k]$,
\begin{equation}
\begin{split}
\Norm{\Paren{\E[h_m(G_i)]-\E[h_m(G_j)]} \Paren{\E[h_m(G_i)]-\E[h_m(G_j)]}^\top}_{\textrm{op}}
& =  \Norm{ \E[h_m(G_i)]-\E[h_m(G_j)]}_{2} \\
& = \mathcal{O}\Paren{m(1+1/\alpha+\Delta)}^{2m} \;.
\end{split}
\end{equation}
This claim follows from the triangle inequality of the operator norm.
\end{proof}

\subsection{List-Recovery of Covariances up to Low-Rank Error} \label{ssec:td_up_to_low_rank}
\input{td_upto_low_rank}

\input{td_recover_low_rank}

%% file: algorithm.tex

\begin{mdframed}
  \begin{algorithm}[List-Recovery of Candidate Parameters via Tensor Decomposition]
    \label{algo:list-recovery-tensor-decomposition}\mbox{}
    \begin{description}
    \item[Input:] An $\epsilon$-corruption $Y$ of a sample $X$ from a $k$-mixture of
    Gaussians $\calM = \sum_i w_i \cN(\mu_i,\Sigma_i)$.
    \item[Requirements: ] The guarantees of the algorithm hold if the mixture parameters
    and the sample $X$ satisfy:
    \begin{enumerate}
        \item $w_i \geq \alpha$ for all $i \in [k]$,
        \item $\Norm{\mu_i}_2 \leq 2/\sqrt{\alpha}$ for all $i\in[k]$,
        \item $\Norm{\Sigma_i - I}_F \leq \Delta$ for all $i\in[k]$.
        \item $X$ satisfies Condition~\ref{cond:convergence-of-moment-tensors} with parameters \new{$(\gamma, t)$,
        where $\gamma = \eps d^{-8k}k^{-Ck}$, for $C$ a sufficiently large universal constant, and $t=8k$}.
    \end{enumerate}

    \item[Parameters: ] $\eta = (2 k)^{4k} (C k(1/\alpha+ \Delta))^{4k} \sqrt{\epsilon}$, $D = C(k^4/(\alpha\sqrt{\eta}))$,
    $\delta = 2 \eta^{1/(C^{k+1} (k+1)!)} $, $\ell'= 100\log k\Paren{ \eta /\Paren{ k^5\Paren{\Delta^4 + 1/\alpha^4 }} }^{-4k}$,
for some sufficiently large absolute constant $C>0$, $\lambda  = 4\eta$, $\phi = 10(1+\Delta^2)/(\sqrt{\eta}\alpha^5)$.
    \item[Output:] A list $L$ of hypotheses such that there exists at least one, $\{\hat{\mu}_i, \hat{\Sigma}_i\}_{i \leq k}\in L$, satisfying: $\Norm{\mu_i - \hat \mu_i}_2 \leq \bigO{\frac{\Delta^{1/2}}{\alpha}} \eta^{G(k)}$ and
$\Norm{\Sigma_i - \hat\Sigma_i}_F \leq \bigO{k^4} \frac{\Delta^{1/2}}{\alpha} \eta^{G(k)}$, where $G(k) = \frac{1}{C^{k+1} (k+1)!}$.

    \item[Operation:]\mbox{}
    \begin{enumerate}
    \item \textbf{Robust Estimation of Hermite Tensors}: For $m \in [4k]$, compute $\hat T_m$ such that $\max_{m \in [4k]} \Norm{\hat T_m - \expecf{}{h_m(\calM)} }_F \leq \eta$ using the robust mean estimation algorithm in Fact~\ref{fact:robust-mean-estimation}.

    \item \textbf{Random Collapsing of Two Modes of $\hat{T}_4$}: Let $L'$ be an empty list. Repeat $\ell'$ times:
    For $ j \in [4k]$, choose independent standard Gaussians in $\mathbb{R}^d$, denoted by $x^{(j)}, y^{(j)} \sim \calN(0, I)$,
    and uniform draws $a_1, a_2, \ldots, a_t$ from $[-D,D]$. Let $\hat S$ be a $d \times d$ matrix such that
    for all $r,s\in[d]$, $\hat S(r,s) = \sum_{j \in [4k]} a_j \hat T_4(r,s,x^{(j)}, y^{(j)}) =
    \sum_{j \in [4k] } a_j \sum_{g,h \in [d]} \hat T_4 (r,s,g,h) x^{(j)}(g) y^{(j)}(h)$.
    Add $\hat S$  to the list $L'$.

    \item \textbf{Construct Low-Dimensional Subspace for Exhaustive Search}: Let $V$ be the span of all singular vectors of the natural $d\times d^{m-1}$ flattening of $\hat T_m$ with singular values $\geq \lambda$ for $m \leq 4k$. For each $\hat S \in L'$, let $V'_{\hat S}$ be the span of $V$ plus all the singular vectors of $\hat S$ with singular value larger than $\delta^{1/4}$.

    \item \textbf{Enumerating Candidates in $V'_{\hat S}$}: Initialize $L$ to be the empty list.
    For each $\hat S \in L'$, let $V_{\delta^{1/4}}$ be a $\delta^{1/4}$-cover of vectors in $V'_{\hat{S}}$
    with $\ell_2$-norm at most $2/\sqrt{\alpha}$. Enumerate over vectors $\hat \mu$ in $V_{\delta^{1/4} }$.
    Let $k' =C k^2$
    and let $\calC_{\delta^{1/4}}$ be a $\delta^{1/4}$-cover of the interval $[-\phi, \phi]^{k'}$.
    For $\{ \tau_j \}_{j\in [k']} \in \calC_{\delta^{1/4}}$ and for all $\{ v_j \}_{j\in [k']} \in V_{\delta^{1/4}}$,
    let $\hat{Q} = \sum_{j \in [k']} \tau_j v_j v_j^\top$. Add  $\{ \hat \mu,I+\hat{S}+\hat{Q} \}$ to $L$.

    \end{enumerate}

    \end{description}
  \end{algorithm}
\end{mdframed}

%% file: td_upto_low_rank.tex

In this section, we prove Proposition~\ref{prop:tensor_decomposition_upto_lowrank}.
We first set some useful notation. We will write $\Sigmaa_i \eqdef \Sigma_i-I$ throughout this section.
We will also use $\Sigmaa_i'$ to denote $\Sigmaa_i + \mu_i \otimes \mu_i$.

We first show that for every $i$, there exists a matrix $P$ such that $\Paren{\sum_{i\in [k]} w_i \Sigmaa'_i \otimes \Sigmaa'_i } \Paren{\cdot, \cdot, P}$ is close to $\Sigmaa'_i$.

\begin{lemma}[Existence of a $2$-Tensor] \label{lem:existence_2_tensor}
Under the hypothesis of Proposition~\ref{prop:tensor_decomposition_upto_lowrank}, for each $i\in[k]$,
there exists a matrix $P$ such that $\Norm{P}_F =  \bigO{1/(\sqrt{\eta} \alpha)}$ and
$\Norm{ T'_4\Paren{\cdot, \cdot, P} - \Sigmaa_i'}_F = \bigO{\sqrt{\eta/\alpha}}$,
where $T'_4 = \Paren{\sum_{i\in [k]} w_i \Sigmaa'_i \otimes \Sigmaa'_i }$.
\end{lemma}

Note that throughout this section it will be useful to think of $T_4'$ as a $d^2 \times d^2$ matrix rather than as a tensor. In this case, we can think of $T_4'$ as $\sum_{i=1}^k w_i (\Sigmaa'_i) (\Sigmaa'_i)^T.$ From standard facts about positive semidefinite matrices it follows that $\Sigmaa_i'$ is in the image of $T_4'$, and Lemma \ref{lem:existence_2_tensor} is just a slightly robustified version of this (saying that we can find an approximate preimage that it not itself too large).

The proof of this Lemma \ref{lem:existence_2_tensor} will involve linear programming duality with an infinite system of constraints. As the application of duality with infinitely many constraints has some technical issues, we state below an appropriate version of duality.

\begin{fact}[Linear Programming Duality for Compact, Convex Constraint Sets]
\label{fact:duality_compact}
Let $K \subset \R^{n+1}$ be a compact convex set. There exists an $x\in \R^n$ so that $(x,1)\cdot z > 0$ for all $z\in K$ if and only if there is no element $(0,0,\ldots,0,a)\in K$ for any $a\leq 0$.
\end{fact}
This fact can be proved by noting that if no such $a$ exists, there must be a hyperplane separating $K$ from the set of such points $(0,a)$. This separating hyperplane will be of the form $(z,y)\in H$ if and only if $y=x\cdot z$ for some $x$ and this $x$ will provide the solution to the linear system.

\begin{proof}[Proof of Lemma \ref{lem:existence_2_tensor}]

To show that such a $P$ exists for each $i$, we apply linear programming duality.
In particular, the conditions imposed on $P$ define a linear program, which has a feasible
solution unless there is a solution to the dual linear program. For sufficiently large constants $c_1$ and $c_2$, consider the following primal in the variable $P$:
\begin{align} \label{eqn:primal}
\Iprod{v, P}  & \leq \frac{c_1}{ \sqrt{\eta} \alpha} \Norm{v}_F     &  \forall\hspace{0.1in} v \in \mathbb{R}^{d \times d} \\
\label{eqn:primal_a}
\Iprod{u, T'_4\Paren{\cdot, \cdot, P} - \Sigmaa'_i } & \leq c_2\sqrt{\eta}  \Norm{u}_F
& \forall\hspace{0.1in} u \in \mathbb{R}^{d \times d}.
\end{align}
It is not hard to see that $\Norm{P}_F \leq \frac{c_1}{ \sqrt{\eta} \alpha}$ if and only if \eqref{eqn:primal} holds for all $v$ and $\Norm{ T'_4\Paren{\cdot, \cdot, P} - \Sigmaa_i'}_F \leq c_2{\sqrt{\eta/\alpha}}$ if and only if \eqref{eqn:primal_a} holds for all $u$. Throughout the proof, we suggest that the reader think of $u$ and $v$ as vectors in $d^2$-dimensional vector space.

Our goal is to show that there exists a feasible solution $P$ such that \eqref{eqn:primal} and \eqref{eqn:primal_a} hold simultaneously for all $u, v \in \R^{d\times d}$. We first note that this is equivalent to saying that
\begin{align}
\label{eqn:dual}
\Iprod{v, P} + \Iprod{u, T'_4\Paren{\cdot, \cdot, P} }  -\Iprod{u, \Sigmaa'_i }  & \leq \frac{c_1}{\sqrt{\eta} \alpha} \Norm{v}_F   +  c_2 \sqrt{\eta}  \Norm{u}_F  \;,
\end{align}
for all $u,v\in \R^{d\times d}$. This is not quite in the form necessary to apply Fact~\ref{fact:duality_compact}, so we note that this is in turn equivalent to saying that
\begin{align}
\label{eqn:dual reformulation}
\Iprod{v, P} + \Iprod{u, T'_4\Paren{\cdot, \cdot, P} }  -\Iprod{u, \Sigmaa'_i }  & \leq 1  \;,
\end{align}
for all $u,v\in \R^{d\times d}$ so that $\frac{c_1}{\sqrt{\eta} \alpha} \Norm{v}_F   +  c_2 \sqrt{\eta}  \Norm{u}_F \leq 1$, and $u \in \mathrm{span}\{\Sigmaa_i'\}$. As this is a convex set of linear equations, we have by Fact \ref{fact:duality_compact} that there exists such a $P$ unless there exists such a pair of $u$ and $v$ so that the coefficient of $P$ in Equation \eqref{eqn:dual reformulation} is $0$ and so that the resulting inequality of constants is either false or holds with equality. In particular, the coefficient of $P$ vanishes if and only if $v = - T'_4\Paren{u, \cdot, \cdot} $. We then get a contradiction only if for some $u \in \mathrm{span}\{\Sigmaa_i'\}$
\begin{align}\label{duality equation in u}
-\Iprod{u, \Sigmaa'_i } & \geq 1 \geq \frac{c_1}{\sqrt{\eta} \alpha} \Norm{T_4'\Paren{u, \cdot, \cdot}}_F   +  c_2 \sqrt{\eta}  \Norm{u}_F.
\end{align}
We claim that this is impossible.

In particular, squaring Equation \eqref{duality equation in u} would give
\begin{equation}
\label{eqn:contradict_this}
\begin{split}
\Iprod{u, \Sigmaa'_i}^2 & \geq  \Paren{ \frac{c_1}{ \sqrt{\eta} \alpha} \Norm{ T'_4\Paren{u, \cdot, \cdot} }_F  +  c_2 \sqrt{\eta} \Norm{u}_F  }^2\\
& \geq \frac{c}{ \alpha}  \Norm{ T'_4\Paren{u, \cdot, \cdot} }_F \cdot \Norm{u}_F \;,
\end{split}
\end{equation}
for some large enough constant $c>1$, where the last inequality follows from the AM-GM inequality.
However, using the dual characterization of the Frobenius norm, we have
\begin{equation}
\label{eqn:dual_frob}
\begin{split}
\Norm{ T'_4\Paren{u, \cdot, \cdot} }_F & \geq  \frac{\Iprod{ u , T'_4\Paren{u, \cdot, \cdot} } }{\Norm{u}_F}
\geq \frac{w_i}{\Norm{u}_F} \Iprod{ u , \Sigmaa'_i}^2 \;,
\end{split}
\end{equation}
where the last inequality follows from $T'_4$ containing a $w_i \Sigmaa_i\otimes \Sigmaa_i$ term, and the other terms contributing non-negatively.
Rearranging Equation~\eqref{eqn:dual_frob}, we have
\begin{equation*}
 \Iprod{ u , \Sigmaa'_i}^2  \leq  \frac{ 1}{ w_i } \Norm{ T'_4\Paren{u, \cdot, \cdot} }_F \Norm{u}_F  \leq \frac{1}{\alpha}\Norm{ T'_4\Paren{u, \cdot, \cdot} }_F \Norm{u}_F\;.
\end{equation*}
This contradicts Equation~\eqref{eqn:contradict_this} unless $T_4'\Paren{u, \cdot, \cdot} = 0$. This therefore suffices to prove the feasibility of the primal.
\end{proof}

We have thus shown that there is some matrix $P$ so that $T_4'\Paren{P, \cdot, \cdot} $ suffices for our purposes. We need to show that our appropriate random linear combination of $x^{(j)}\otimes y^{(j)}$ suffices. In fact, we will show that with reasonably high probability over our choice of $x^{(j)},y^{(j)}$ that there is some linear combination of the $x^{(j)}\otimes y^{(j)}$ (with coefficients that are not too large) so that their projection onto the space spanned by the $\Sigmaa_i'$ (which is all that matters when applying $T_4'$) equal to $P$.

For the sake of intuition, we note that if we removed the bound on the coefficients, we would need that the projections of the $x^{(j)}\otimes y^{(j)}$ spanned $\mathrm{span}\{\Sigmaa_i'\}$. Since there are at least $k$ of them, this will hold unless there is some $v\in \mathrm{span}\{\Sigmaa_i'\}$ so that $v$ is orthogonal to all of the $x^{(j)}\otimes y^{(j)}$. This shouldn't happen because each $x^{(j)}\otimes y^{(j)}$ is very unlikely to be orthogonal to $v$.

To deal with the constraint that the coefficients are not too large, we use linear programming duality to show that there will be a solution unless there is some $v$ that is \emph{nearly} orthogonal to all of the $x^{(j)}\otimes y^{(j)}$. Again, this is unlikely to happen for any individual term, and thus, by independence, highly unlikely to happen for all $j$ simultaneously. Combining this with a cover argument will give our proof.

\begin{lemma}[Existence of a Bi-Linear Form] \label{lem:existence_of_bilinear}
Given the preconditions in Proposition \ref{prop:tensor_decomposition_upto_lowrank},
with probability at least $99/100$
over the choice of $x^{(j)},y^{(j)}$,
there exist $b_j \in [-D, D]$ for $j \in [4k]$, where $D = \bigO{k^4/(\sqrt{\eta}\alpha)}$,
such that the projection of $\sum_{j=1}^{t} b_j x^{(j)}\otimes y^{(j)}$ onto the space spanned
by the $\Sigmaa_i'$ is $P$, where $P$ satisfies the conclusion of Proposition \ref{lem:existence_2_tensor}.
\end{lemma}
\begin{proof}
To prove this lemma, we again use a linear programming based argument.
Consider the following (primal) linear program in the variables $b_j$, for $j \in [4k]$:
\begin{align}\label{eqn:primal_2}
\sum_{j \in [4k]} b_j \langle \Sigmaa'_i, x^{(j)} \otimes y^{(j)} \rangle
& = \langle \Sigmaa'_i,P \rangle   & \forall\hspace{0.1in} i\in [k]  \\
\label{eqn:primal_2_a}
-D  \leq b_j &\leq D   & \forall\hspace{0.1in} j \in [4k]
\end{align}
We note that a set of $b_j$ satisfying Equation \eqref{eqn:primal_2} will have the projection of $\sum_{j\in[4k]} b_j x^{(j)} \otimes y^{(j)}$ onto the span of the $\Sigmaa_i'$ be the same as the projection of $P$, and that if the $b_j$'s satisfy Equation \eqref{eqn:primal_2_a} then we will have $|b_j|\leq D$ for all $j$. Thus, it suffices to show that with high probability over our choice of $x^{(j)}$ and $y^{(j)}$ that the above system is feasible.

We will show this by linear programming duality (since this is now a finite system of equations, we can use standard results rather than Fact \ref{fact:duality_compact}). In particular, we have that Equations \eqref{eqn:primal_2} and \eqref{eqn:primal_2_a} are simultaneously satisfiable unless there are real numbers $c_i$ and non-negative real numbers $z_j,z_j'$ so that
$$
\sum_{i=1}^k c_i \sum_{j \in [4k]} b_j \langle \Sigmaa'_i, x^{(j)} \otimes y^{(j)} \rangle + \sum_{j\in [4k]} (z_j-z_j')b_j \leq \sum_{i=1}^k c_i  \langle \Sigmaa'_i,P \rangle+\sum_{j\in [4k]}(z_j+z_j')D
$$
yields a contradiction. Setting $v= \sum_{i=1}^k c_i \Sigmaa'_i$, the above simplifies to
\begin{align}\label{equation to contradict 2}
\sum_{j \in [4k]} b_j\left(\langle v, x^{(j)} \otimes y^{(j)} \rangle+z_j-z_j' \right) & \leq \langle v,P \rangle+\sum_{j\in [4k]}(z_j+z_j')D
\end{align}
We note that in order for Equation \eqref{equation to contradict 2} to be a contradiction, it must be the case that the coefficients of $b_j$ are all $0$. In particular, we must have
$$
z_j'-z_j = \langle v, x^{(j)} \otimes y^{(j)} \rangle
$$
for all $j$. In particular, this means that 
$$z_j+z_j' \geq \left| \langle v, x^{(j)} \otimes y^{(j)} \rangle \right|.$$
In such a case, the right hand side of Equation \eqref{equation to contradict 2} will be at least
$$
\langle v,P \rangle+\sum_{j\in [4k]}\left| \langle v, x^{(j)} \otimes y^{(j)} \rangle \right|D
$$
Therefore, Equation \eqref{equation to contradict 2} can only yield a contradiction if there exists a $v\in \mathrm{span}\{\Sigmaa'_i\}$ so that
\begin{equation}\label{v to contradict equation}
\langle v,P \rangle < - \sum_{j\in [4k]}\left| \langle v, x^{(j)} \otimes y^{(j)} \rangle \right|D.
\end{equation}

We want to show that with high probability over our choice of $x^{(j)}, y^{(j)}$ that there is no $v\in \mathrm{span}\{\Sigmaa'_i\}$ satisfying Equation \eqref{v to contradict equation}. In fact, we will show that for every such $v$ that 
$$
\sum_{j\in [4k]}\left| \langle v, x^{(j)} \otimes y^{(j)} \rangle \right| \geq \frac{c_1}{\sqrt{\eta}\alpha}\Norm{v}_F.
$$
We can scale $v$ so that $\Norm{v}_F = 1$, and it suffices to show that
\begin{equation}
\label{eqn:to_obtain_contradiction}
\sum_{j \in [4k]}\left|\Iprod{ \tilde{v} ,  x^{(j)} \otimes y^{(j)}   }\right|  \geq \Paren{\frac{c_1 }{\sqrt{\eta} \alpha D }}
\end{equation}
holds for all unit vectors $v$ in $\mathrm{span}\{\Sigmaa_i'\}$ with high probability.

Since we need to show that infinitely many equations all hold with high probability, we will use a cover argument. In particular, we can construct $\calC$, a $\tau$-cover for all unit vectors $v$ in the span of the $\Sigmaa_i'$, where we take $\tau = \Paren{\frac{c'_1 }{ k^2 \sqrt{\eta} \alpha D}}$. Since this is a cover of a unit sphere in a $k$-dimensional subspace, we can construct such a cover so that
$\abs{\calC} = \bigO{ 1/ \tau }^{k}$. Replacing $v$ with the closest point in $\calC$,
denoted by $v'$, it suffices to show that with high probability for all $v$ that
\begin{equation}\label{cover bound equation}
\sum_{j \in [4k]}\left|\Iprod{ {v} ,  x^{(j)} \otimes y^{(j)}   } \right| \geq \sum_{j \in [4k]}\left|\Iprod{ {v'} ,  x^{(j)} \otimes y^{(j)}   }\right| - \sum_{j \in [4k]}\left|\Iprod{ {v}  - {v}',  x^{(j)} \otimes y^{(j)}   } \right|  \geq \Paren{\frac{2c_1 }{\sqrt{\eta} \alpha D}}.
\end{equation}
We begin by bounding the terms
$$
\sum_{j \in [4k]}\left|\Iprod{ {v}  - {v}',  x^{(j)} \otimes y^{(j)}   } \right|.
$$
For this we notice by Cauchy-Schwartz that each term is at most $\Norm{v-v'}_F$ times the Frobenius norm of the projection of $x^{(j)} \otimes y^{(j)}$ onto the span of the $\Sigmaa_i'$. We note that for any $k$-dimensional subspace $W$ with orthonormal basis $w_1,\ldots,w_k$ we have that
\begin{align*}
\E\left[\Norm{\mathrm{Proj}_W(x^{(j)} \otimes y^{(j)} )}_F^2 \right] & = \sum_{i=1}^k \left|\Iprod{ w_i,  x^{(j)} \otimes y^{(j)}   } \right|^2\\
& = k.
\end{align*}
Therefore, with high probability over the choice of $x^{(j)}, y^{(j)}$ each of the projections of $x^{(j)} \otimes y^{(j)}$ onto the span of the $\Sigmaa_i'$ has Frobenius norm $\tilde O(\sqrt{k})$. Therefore, if this condition holds over our choice of $x^{(j)}$ and $y^{(j)}$, we can show Equation \eqref{cover bound equation} if we can show that
\begin{align}\label{simplified cover equation}
\sum_{j \in [4k]}\left|\Iprod{ {v'} ,  x^{(j)} \otimes y^{(j)}   } \right| & \geq \Paren{\frac{c_1 }{\sqrt{\eta} \alpha D}} \geq \Paren{\frac{2c_1 }{\sqrt{\eta} \alpha D}} - \tau\tilde O(k^{3/2})
\end{align}
for all $v' \in\mathcal{C}$.

Each term in $\sum_{j \in [4k]}\Iprod{ {v}' ,  x^{(j)} \otimes y^{(j)}}$ is a random bi-linear form
given by $z_j = \sum_{\ell, p\in [d]} {v}'_{\ell,p} x^{(j)}_{\ell}y^{(j)}_{p}$.
Then, we have that $\expecf{}{z_j} = 0$ and
\begin{equation*}
\begin{split}
\expecf{}{z^2_j} = \expecf{}{ \Paren{\sum_{\ell, p\in [d]} {v}'_{\ell,p} x^{(j)}_{\ell}y^{(j)}_{p}}^2} & = \sum_{\ell, \ell', p, p'} \expecf{}{{v}'_{\ell,p} {v}'_{\ell',p'} x^{(j)}_{\ell} x^{(j)}_{\ell'} y^{(j)}_{p} y^{(j)}_{p'} } \\
& = \sum_{\ell, p\in [d] } \Paren{{v}'_{\ell,p}}^2 \cdot \expecf{}{ \Paren{x^{(j)}_{\ell}}^2} \cdot  \expecf{}{ \Paren{y^{(j)}_{p}}^2}\\
& = 1 \;,
\end{split}
\end{equation*}
where the last equality follows from ${\tilde{v}'}_F = 1$.

Using Lemma \ref{lem:anti-concentration_of_bilinear_forms} with $\zeta = \frac{2c_1 }{\sqrt{\eta} \alpha D}$,
\begin{equation}
\Pr\left[ \abs{z_j} \leq \frac{c_1 }{\sqrt{\eta} \alpha D}\right] \leq c_5 \Paren{\frac{2c_1 }{\sqrt{\eta} \alpha D}}^{1/2} \;.
\end{equation}
However, we note that Equation \eqref{simplified cover equation} will hold unless $\abs{z_j} \leq \frac{c_1 }{\sqrt{\eta} \alpha D}$ for all $j\in [4k]$.
Since the $z_j$'s are independent, we  conclude that
\begin{equation}
\Pr\left[ \sum_{j \in [4k]}\left|\Iprod{ {v}' ,  x^{(j)} \otimes y^{(j)}   }\right|
\leq \frac{c_1 }{\sqrt{\eta} \alpha D}\right] \leq O\Paren{\frac{c_1 }{\sqrt{\eta} \alpha D}}^{2k} \;.
\end{equation}
Since the above argument holds for any $v' \in \calC$, we can union bound over all elements in the cover $\calC$,
and the probability that there exists a $\tilde{v}'$ in the cover that does not satisfy Equation \eqref{simplified cover equation}
is at most $O\Paren{  k^2 \sqrt{\eta} \alpha D }^{k} \cdot O\Paren{\frac{c_1 }{\sqrt{\eta} \alpha D}}^{2k}$.
Setting $D$ to be a sufficiently large multiple of $(k^4/ (\sqrt{\eta} \alpha))$ suffices to
conclude that with probability at least $1-1/{\poly(k)}$, the primal is feasible.
\end{proof}


\begin{proof}[Proof of Proposition~\ref{prop:tensor_decomposition_upto_lowrank}]
We begin by bounding the Frobenius norm of $\hat{T}_4$.
Let $T_4 = \E[h_4(\calX)]$. It  then follows from Lemma \ref{HermiteExpectationLem} that
\begin{equation}
\label{eqn:expected_4th_hermite_tensor}
T_4 = \sym\left(\sum_{i=1}^k w_i \left(3\Sigmaa_i\otimes \Sigmaa_i +6\Sigmaa_i\otimes \mu_i^{\otimes 2}+\mu_i^{\otimes 4} \right) \right) \;.
\end{equation}
Further, $\Norm{\Sigmaa_i\otimes \Sigmaa_i}_F \leq \Norm{\Sigmaa_i}^2_F\le\Delta^2$,
$\Norm{\Sigmaa_i\otimes \mu_i^{\otimes 2}}_F \leq \Norm{\Sigmaa_i}_F\Norm{\mu_i}^2_2\le4\Delta/\alpha$,
and $\Norm{\mu_i^{\otimes 4}}_F \leq \Norm{\mu_i}^4_2\le16/\alpha^2$.
Since $T_4$ is an average of terms of the form $\Sigmaa_i^{\otimes 2}, \Sigmaa_i\otimes \mu_i^{\otimes 2}$
and $\mu_i^{\otimes 4}$, and each such term is upper bounded, we can conclude that
$\Norm{T_4}_F = \bigO{\Delta^2 + 1/\alpha^2 }$, and by the triangle inequality that
$\Norm{ \hat{T}_4 }_F \leq \bigO{\Delta^2 + 1/\alpha^2 + \eta}$.
Let $\Sigmaa_i'=\Sigmaa_i+\mu_i^{\otimes 2}$ and let $T_4':= \sum_{i=1}^k w_i \left(\Sigmaa_i'\otimes \Sigmaa_i' \right)$.
We can then rewrite Equation~\eqref{eqn:expected_4th_hermite_tensor} as follows:
\begin{equation}
\label{eqn:rewrite_4th_hermite_tensor}
T_4 = \sym\left(\sum_{i=1}^k w_i \left(3\Sigmaa_i'\otimes \Sigmaa_i' -2\mu_i^{\otimes 4} \right) \right) \;.
\end{equation}
For $j \in [4k]$, let $x^{(j)}, y^{(j)} \sim \calN(0, I)$. Collapsing two modes of $\hat{T}_4$, it follows from Equation \eqref{eqn:rewrite_4th_hermite_tensor} that for any fixed $j$,
\begin{equation}
\label{eqn:symmeterizing_hermite_tensor}
\begin{split}
\hat{T}_4\Paren{ \cdot, \cdot , x^{(j)},  y^{(j)}} &= \Paren{\hat{T}_4 - T_4}\hspace{-0.06in}\Paren{ \cdot, \cdot , x^{(j)},  y^{(j)}} + T_4\Paren{ \cdot, \cdot , x^{(j)},  y^{(j)}}\\
&  = \Paren{\hat{T}_4 - T_4}\hspace{-0.06in}\Paren{ \cdot, \cdot , x^{(j)},  y^{(j)}}  + \sym\left(\sum_{i=1}^k w_i \left(3\Sigmaa_i'\otimes \Sigmaa_i' -2\mu_i^{\otimes 4} \right) \right)\Paren{\cdot, \cdot, x^{(j)}, y^{(j)}} \\
& =  \Paren{\hat{T}_4 - T_4 + T'_4}\hspace{-0.06in}\Paren{ \cdot, \cdot , x^{(j)},  y^{(j)}} + \sum_{i\in [k]} w_i \Paren{\Sigmaa_i' x^{(j)}} \otimes \Paren{\Sigmaa_i' y^{(j)}} \\
& \hspace{0.2in} + \sum_{i\in [k]} w_i\Paren{\Sigmaa_i' y^{(j)}} \otimes \Paren{\Sigmaa_i' x^{(j)}}  + \sum_{i \in [k]} w_i \Paren{ -2 \mu_i^{\otimes 2} \Iprod{\mu_i, x^{(j)}} \Iprod{\mu_i, y^{(j)}} } \;,
\end{split}
\end{equation}
where we use that $\sym(\cdot)$ is a linear operator satisfying $\sym\Paren{\mu_i^{\otimes 4}} = \mu_i^{\otimes 4}$, and
\begin{equation*}
\sym\Paren{ \Sigmaa'_i \otimes \Sigmaa'_i } = \frac{1}{3} \Sigmaa'_i  \otimes \Sigmaa'_i + \frac{1}{3} \Sigmaa'_i  \oplus \Sigmaa'_i + \frac{1}{3} \Sigmaa'_i  \ominus \Sigmaa'_i
\end{equation*}
where for indices $(i_1, i_2, i_3, i_4)$, $\Paren{\Sigmaa'_i  \oplus \Sigmaa'_i}\Paren{i_1, i_2, i_3, i_4} = \Paren{\Sigmaa'_i  \otimes \Sigmaa'_i}\Paren{i_1, i_3, i_2, i_4}$ and $\Paren{\Sigmaa'_i  \ominus \Sigmaa'_i}\Paren{i_1, i_2, i_3, i_4} = \Paren{\Sigmaa'_i  \otimes \Sigmaa'_i}\Paren{i_1, i_4, i_2, i_3}$.

Next, it follows from Lemma \ref{lem:existence_2_tensor} that there exists a matrix $\tilde{P}_i$ such that
$\Norm{\tilde{P}_i}_F = \bigO{1/(\sqrt{\eta} \alpha)}$ and
$\Norm{T'_4\Paren{\cdot, \cdot, \tilde{P}_i} - \Sigmaa'_i }_F = \bigO{\sqrt{\eta/\alpha}}$.
Furthermore, with probability at least $0.99$, there exists a sequence of $b_j \in [-D,D]$,
for $j \in [4k]$, such that $T'_4\Paren{\cdot,\cdot, \sum_{j \in[4k] } b_j x^{(j)} \otimes y^{(j)} } = T'_4\Paren{\cdot, \cdot, \tilde{P}_i}$.

Consider a cover, $\calC$, of the interval $[-D, D]$ with points spaced at intervals of length
$\tau = \bigO{ \frac{\sqrt{\eta}}{\alpha k\Paren{\Delta^4 + 1/\alpha^4}} }$. Since we uniformly sample $a_j$'s, with probability at least $(\tau/D)^{\calO(k)}$, for all $j\in [4k]$, $\abs{b_j-a_j} \leq \tau$, and we condition on this event. Thus,
\begin{equation}
\label{eqn:closeness_T'_and_S'}
\begin{split}
\Norm{ T'_4\Paren{\cdot,\cdot, \sum_{j \in[4k] } a_j x^{(j)} y^{(j)} }  - \Sigmaa'_i}_F & \leq  \Norm{ T'_4\Paren{\cdot,\cdot, \sum_{j \in[4k] } b_j x^{(j)}\otimes y^{(j)} }  - \Sigmaa'_i}_F + \Norm{ T'_4\Paren{\cdot,\cdot, \sum_{j \in[4k] } (b_j-a_j)x^{(j)}\otimes y^{(j)} } }_F\\
& \leq \bigO{\sqrt{\eta/\alpha} } +\bigO{\tau\Delta^2}
\le \bigO{\sqrt{\eta/\alpha} } \;.
\end{split}
\end{equation}
Taking the linear combinations with coefficients $a_j$ in Equation \eqref{eqn:symmeterizing_hermite_tensor}, we have
\begin{equation} \label{eqn:final_guarantee_split_frob_low_rank}
\begin{split}
\hat{T}_4\Paren{ \cdot, \cdot , \sum_{j\in [4k]} a_j x^{(j)}\otimes  y^{(j)}} - \Sigmaa_i
& =  \Paren{\hat{T}_4 - T_4 + T'_4 }\hspace{-0.04in}\Paren{ \cdot, \cdot , \sum_{j \in [4k]}a_j x^{(j)} \otimes y^{(j)}} - \Sigmaa'_i - \mu_i \otimes \mu_i  \\
& \hspace{0.2in} + \sum_{j \in [4k]} a_j \sum_{i\in [k]} w_i \Paren{\Sigmaa_i' x^{(j)}} \otimes \Paren{\Sigmaa_i' y^{(j)}} + \sum_{j \in [4k]} a_j \sum_{i\in [k]} w_i\Paren{\Sigmaa_i' y^{(j)}} \otimes \Paren{\Sigmaa_i' x^{(j)}}  \\
& \hspace{0.2in}  +\sum_{j \in [4k]} a_j \sum_{i \in [k]} w_i \Paren{ -2 \mu_i^{\otimes 2} \Iprod{\mu_i, x^{(j)}} \Iprod{\mu_i, y^{(j)}} } \;.
\end{split}
\end{equation}
Setting $P_i = \Paren{\hat{T}_4-T_4 + T'_4}\Paren{\cdot,\cdot, \sum_{j \in[4k] } a_j x^{(j)} y^{(j)}} - \Sigmaa'_i$, it follows from Lemma \ref{lem:concentration_of_low_degree_gaussians} that with probability at least $0.99$, $\Paren{\hat{T}_4 -T_4}\Paren{\cdot,\cdot, \sum_{j \in[4k] } a_j x^{(j)} y^{(j)}}$ has Frobenius norm $\bigO{k D \eta}$ and it follows from Equation \eqref{eqn:closeness_T'_and_S'} that with probability at least $0.99$, $T'_4\Paren{\cdot,\cdot, \sum_{j \in[4k] } a_j x^{(j)} y^{(j)}} - \Sigmaa'_i$ has Frobenius norm $\bigO{\sqrt{\eta/\alpha}}$. Setting the remaining terms to $Q_i$, with probability at least $0.99$ we can bound their Frobenius norm as follows:
\begin{equation}
\begin{split}
\Norm{Q_i}_F & \leq \Norm{\mu_i \otimes \mu_i}_F + \Norm{ \Paren{\sum_{i \in [k]} w_i \Sigmaa'_i \oplus \Sigmaa'_i + w_i\Sigmaa'_i\ominus \Sigmaa'_i -2 w_i\mu_i^{\otimes 4} }\Paren{\cdot, \cdot, \sum_{j \in [4k]} a_j x^{(j)} y^{(j)} }  }_F \\
& \leq \frac{4}{\alpha} + \Paren{2 \max_{i\in k} \Norm{\Sigmaa'_i}^2_F + \frac{32}{\alpha^2} + k\tau }\cdot \Norm{\tilde{P}}_F \\
& \leq \frac{4}{\alpha}+ \bigO{ \frac{1}{\sqrt{\eta}\alpha} \Paren{\Delta +\frac{1}{\alpha}}^2}\\
& \le \bigO{\frac{1+\Delta^2}{\sqrt{\eta}\alpha^3}} \;,
\end{split}
\end{equation}
where the first inequality follows from the triangle inequality,
the second follows from our assumptions that $\Norm{\mu_i}_2 \leq 2/\sqrt{\alpha}$,
$\sum_{j \in [4k]} b_j x^{(j)} y^{(j)} = \tilde{P}_i$ in the span of the $\Sigmaa'_i$, and $\abs{a_j - b_j} \leq \tau$ for all $j\in[4k]$,
and the third inequality follows from the definition of $\Sigmaa'_i$,
the bound on $\Norm{\tilde{P}}_F$ and the bound on $\Norm{\Sigmaa_i -I}_F$.
\end{proof}

%% file: td_recover_low_rank.tex

\subsection{Finding a Low-dimensional Subspace for Exhaustive Search} \label{ssec:td_recover_low_rank}

In this subsection, we will prove Proposition~\ref{prop:construction-of-low-dim-subspace-for-enumeration}. 


We start by extending Theorem 4 of~\cite{MoitraValiant:10}, 
which shows that large parameter distance between pairs of univariate Gaussian mixtures 
implies large distance between their low-degree moments. 
In the following, we use $M_j(F) = \expecf{}{F^j}$ to denote the 
$j$-th moment of a distribution $F$. We show:

\begin{lemma}\label{lem:mv}
There exists a constant $C>0$ such that the following holds: Fix any $D>0$ and $0\le \beta\le1/(2(2k-1)!D^{2k-3})$. Suppose that $F=\sum_{i=1}^k w_i\cN(\mu_i,\sigma_i^2)$ is a univariate $k$-mixture of Gaussians with $w_i \geq \beta$, and $|\mu_i|,\sigma_i\le D$, 
for all $i \in [k]$. If $|\mu_i|+|\sigma_i^2-1|\ge\beta$ for some $i\le k$, then 
$$
\max_{j \in [2k] } \abs{ M_j(F)-M_j\Paren{\calN(0,1)} } \ge \beta^{C^{k+1}(k+1)!-1} \;.
$$
\end{lemma}

\noindent We give the proof of Lemma~\ref{lem:mv} in Section~\ref{ssec:mv}.

\begin{lemma}[Bounding $\mu_i$'s and $S_i$'s in non-influential directions for $\expecf{}{h_m(\calM)}$] 
\label{lem:bound-on-projection-onto-non-influential-directions}
Let $\calM = \sum_{i = 1}^k w_i \cN(\mu_i,\Sigma_i)$ be a $k$-mixture of Gaussians on $\R^d$ 
satisfying $w_i \geq \alpha$, $\Norm{\mu_i}_2 \leq 2/\sqrt{\alpha}$, and $\Norm{\Sigma_i-I}_F \leq \Delta$ for every $i \in [k]$. 
For some  $B \in \R$, let  $u \in \R^d$ be a unit vector such that $\abs{\expecf{}{ h_m\Paren{ \Iprod{\calM, u} } }} \leq B$ 
for all $ m \in [2k]$. Then, for $\delta = 2^{O(k)} B^{1/(C^{k+1}(k+1)!)}$ and $S_i=\Sigma_i-I$, we have that:
\begin{enumerate}
\item for all $i \leq k$, $\abs{\iprod{u,\mu_i}}, \abs{u^{\top}(I-\Sigma_i)u} \leq \delta$,
\item $\Norm{S_i u}_2^2 \leq 20\delta \Delta/\alpha^2 + B/\alpha$,
\end{enumerate}
where $C>0$ is a fixed universal constant. 
\end{lemma}

\begin{proof}
The 1-D random variable $\iprod{u,\calM}$ is a mixture of Gaussians described by 
$\sum_{i = 1}^k w_i \cN(\iprod{\mu_i,u}, u^{\top} \Sigma_i u)$. Towards a contradiction, 
assume that there is an $i \in [k]$ such that $\abs{\iprod{u,\mu_i}} + \abs{u^{\top}(I-\Sigma_i)u} \geq \delta$. 
Then, applying Lemma~\ref{lem:mv}, yields that there is a $j\in [2k]$ such that 
$|M_j(\iprod{u,\calM}) - M_j(\cN(0,1))| \geq \delta^{C^{k+1} (k+1)! -1}$. 
Applying Fact~\ref{fact:hermite-vs-raw-moments} implies that there exists an $m \in [2k]$ such that 
\[
\abs{\expecf{}{ h_m(\iprod{u,\calM}) } } > 2^{-O(k)} \delta^{C^{k+1} (k+1)!-1} \gg B \;,
\] 
yielding a contradiction. 

We can now prove the second part. 
Recall that for $S_i = \Sigma_i -I$ for every $i$, we have that
\[
\expecf{}{h_4(\calM)} = \sum_{i = 1}^k w_i \textrm{Sym} \Paren{ 3 \Paren{ S_i \otimes S_i}  + 6 \Paren{S_i \otimes \mu_i^{\otimes 2}} + \mu_i^{\otimes 4} }\mper
\]
We consider the $d \times d$ matrix obtained by the natural flattening of the $d \times d$ tensor 
$u^{\otimes 2} \cdot \expecf{}{h_4(\calM)}$. 
Then, we can write:
\begin{multline}
u^{\otimes 2} \cdot \expecf{}{h_4(\calM)} = \sum_{i = 1}^k w_i \Bigl( (u^{\top}S_iu) S_i + 2(S_iu)(S_iu)^{\top} +  \iprod{u,\mu_i}^2 S_i \\ + 2\iprod{u,\mu_i} \mu_i (S_i u)^{\top} +  2\iprod{u,\mu_i}(S_i u)  \mu_i^{\top} + (u^{\top} S_i u) \mu_i \mu_i^{\top} + \iprod{u,\mu_i}^2\mu_i \mu_i^{\top} \Bigr) \;.
\end{multline}
Now, from the first part, we know that for all $i \in [k]$, $|u^{\top}S_i u| \leq \delta$ and the hypothesis of the lemma gives us that $\Norm{S_i}_F= \Norm{\Sigma_i -I}_F \leq \Delta$. Thus, for each $i$, the first term in the summation above has Frobenius norm at most $\Delta \delta$. Using that $\iprod{u,\mu_i}_2^2 \leq \delta^2$ from the first part of the lemma, 
yields that, for each $i$, the Frobenius norm of the third term is at most $\Delta \delta^2$. 

Next, using in addition that $\Norm{\mu_i}_2 \leq 2/\sqrt{\alpha}$ yields that, for each $i$, 
the Frobenius norm of the 4th and 5th terms are at most $2\delta\Delta/\sqrt{\alpha}$ 
and the Frobenius norm of the 6th and 7th terms are at most $\delta/\alpha$. 
Thus, for each $i$ and all but the 2nd term in the summation above, 
we have an upper bound on the Frobenius norm of $4\delta\Delta/\alpha$. 

Now, since $\abs{\expecf{}{h_4\Paren{ \Iprod{\calM, u} }} } \leq B$, and $u$ is a unit vector, 
we have that $\Norm{u^{\otimes 2} \expecf{}{ h_4(\calM)} }_F \leq B$. 
Thus, combining the aforementioned argument with the triangle inequality, we have 
for each $i$, 
\begin{equation*}
\begin{split}
\Norm{S_i u}_2^2 = \Norm{ S_i u \Paren{S_i u}^\top }_F & \leq \frac{1}{\alpha} \Norm{ u^{\otimes 2} \cdot \expecf{}{h_4(\calM)} }_F 
+ \sum_{i \in [k]} w_i \Bigl( \Paren{u^{\top}S_iu + \Iprod{u, \mu_i}^2 }   \Norm{ S_i  }_F \Bigr)  \\
& \hspace{0.2in} + \sum_{i \in [k]} 4 w_i \Bigl( \Paren{ \Iprod{u, \mu_i}  }   \Norm{ \mu_i }_2 \Norm{S_i u}_2 \Bigr) + \sum_{i \in [k] } 4 w_i \Bigl( \Paren{ \Iprod{u, \mu_i}^2 + u^\top S_i u  }   \Norm{ \mu_i }^2_2 \Bigr) \\
& \leq B/\alpha  +  15 \delta  \Delta /\alpha \;,
\end{split}
\end{equation*}
and the claim follows.  
\end{proof}



\begin{lemma}[Subspace covering all the means and large singular vectors of $S_i = \Sigma_i-I$] 
\label{lem:subspace-from-large-estimated-moment-tensor}
Let $\calM = \sum_{i = 1}^k w_i \cN(\mu_i,\Sigma_i)$ be a $k$-mixture of Gaussians on $\R^d$
satisfying $w_i \geq \alpha$, $\Norm{\mu_i}_2 \leq 2/\sqrt{\alpha}$, 
and $\Norm{\Sigma_i-I}_F \leq \Delta$ for all $i \in [k]$. Given $0<\eta<1$, 
let $\hat T_m$ satisfy $\Norm{\hat T_m-\expecf{}{h_m(\calM)} }_F \leq \eta$ 
for every $ m \in [4k]$ and let $\lambda \geq 2 \eta$. 
Let $V$ be the span of all the left singular vectors of the $d \times d^{m-1}$ matrix 
obtained by the natural flattening of $\hat T_m$ with singular values at least $\lambda$. 
Then, for $\delta = 2 \lambda^{1/(2C^{k+1} (k+1)!)}$, we have that:
\begin{enumerate}
\item $\dim{V} \leq \Paren{ 4k\eta^2 + k^{O(k)} } \bigO{1+ 1/\alpha + \Delta}^{4k} /\lambda^2$, 
\item Let 
\[
V_{\inf} = \{\mu_i\}_{i \in [k] } \cup \Set{ v \mid \exists i \in [k],  \text{ s.t. $\Norm{v}_2 =1$ and $v$ is an eigenvector of $S_i$ and } \Norm{S_i v}_2 \geq \sqrt{\delta} }_{i \leq k}\mper
\] 
Then, for every unit vector $v \in V_{\inf}$, $\Norm{v-\Pi_V v}^2_2 \leq 20\delta^{1/4} \Delta/\alpha^2$, 
where $\Pi_V v$ is the projection of $v$ onto $V$. 
\end{enumerate}
\end{lemma}

\begin{proof}
From Fact \ref{HermiteExpectationLem}, we have that
$\expecf{}{ h_m(\calM)}  = \sum_{i \in [k]} w_i \expecf{}{ h_m(\new{G_i})}$,
\new{where $G_i = \calN(\mu_i, \Sigma_i)$},
and since $\Norm{\mu_i}_2 \leq 2/\sqrt{\alpha}$ and $\Norm{\Sigma_i - I}_F \leq \Delta $, 
it follows that $\Norm{\expecf{}{ h_m(\calM)} }^2_F \leq  \calO\Paren{m( 1 + 1/\alpha + \Delta)}^{4m}$.
From Proposition~\ref{prop:robust-estimation-hermite-tensors}, we know that 
\[
 \Norm{ \hat T_m}^2_F \leq 2\Norm{ \hat T_m - \expecf{}{ h_m(\calM)} }^2_F 
 + 2\Norm{ \expecf{}{ h_m(\calM)} }^2_F \leq  \eta^2 + \calO\Paren{m(1 + 1/\alpha + \Delta)}^{4m} \;.
 \] 
Thus, the number of singular vectors of the $d \times d^{m-1}$ flattening of $\hat T_m$ 
with a singular value $ \geq \lambda$ is at most 
$(\eta^2 +   \bigO{m(1 + 1/\alpha + \Delta)}^{4m})/\lambda^2$. 
Summing up this bound for all $m \in [4k]$, yields the claimed upper bound on $\dim{V}$. 

For the second part, we will first bound $\iprod{u,v}$ for any unit vector $u$ orthogonal to the subspace $V$. 
Towards this, observe that since $u$ is orthogonal to $V$ and $\norm{u}_2 = 1$, we have
$$
\Norm{u \cdot \expecf{}{ h_m(\calM)}}_F 
\leq \norm{u \hat T_m}_F+\norm{\hat T_m-\E h_m(\calM)}_F\\
\leq \lambda + \eta \leq 2 \lambda \;,
$$
where $u \cdot \expecf{}{ h_m(\calM)}$ is a matrix-vector product of $u$ with 
a $d \times d^{m-1}$ flattening of $\expecf{}{ h_m(\calM)}$. 
For $\delta = 2 \lambda^{1/(C^{k+1} (k+1)!)}$, 
applying Lemma~\ref{lem:bound-on-projection-onto-non-influential-directions} yields that 
\begin{equation}
\label{eqn:upperbound_along_u}
\iprod{\mu_i,u}^2 +  \Norm{S_i u}_2^2 \leq \delta^2+20 \delta \Delta/\alpha\le20 \delta \Delta/\alpha^2 \;.
\end{equation} 
Now, if $v$ is one of the $\mu_i$'s, then we immediately get from Equation~\ref{eqn:upperbound_along_u} 
that $\iprod{v,u}^2 \leq 20 \delta \Delta/\alpha^2$. Similarly, note that if $v$ is a unit length eigenvector 
of $S_i$ satisfying $\Norm{S_iv}_2^2 \geq \sqrt\delta$, then, 
\begin{equation*}
\Iprod{u, v}^2 = \frac{1}{\Norm{S_iv}^2_2} \Iprod{u, S_iv}^2 = 
\frac{1}{\Norm{S_iv}^2_2} \Iprod{S_iu, v}^2 \leq \frac{\Norm{S_iu}^2_2}{\Norm{S_iv}^2_2} \;.
\end{equation*}
In both cases, setting $u = (v-\Pi_V v)/\Norm{v-\Pi_V v}_2$ completes the proof. 
\end{proof}

We can now complete the proof of Proposition~\ref{prop:construction-of-low-dim-subspace-for-enumeration}:

\begin{proof}[Proof of Proposition~\ref{prop:construction-of-low-dim-subspace-for-enumeration}]
We know that $\hat S_i - P_i -S_i$ is a symmetric, rank-$k'$ matrix such that $k'=\bigO{k^2}$, 
described by the eigenvalue decomposition $\sum_{i = 1}^{k'} \tau_i v_i v_i^{\top}$, 
where $v_i$'s are the eigenvectors and $\tau_i$'s are the corresponding eigenvalues. 
Since $\Norm{S_i}_F \leq \Delta$ and 
$$\Norm{\hat S_i}_F 
\leq \Norm{P_i}_F+ \Norm{Q_i}_F + \Norm{S_i}_F 
\leq  \bigO{ \sqrt{\eta/\alpha} } +  \bigO{ \frac{1+\Delta^2}{\sqrt{\eta} \alpha^3}  } + \Delta
= \bigO{ \frac{1+\Delta^2}{\sqrt{\eta} \alpha^3}} \;,
$$ 
we have that the number of singular values of $\hat S_i$ that exceed $\delta^{1/4}$ 
is at most $\calO \Paren{\frac{1+\Delta^2}{\sqrt{\eta}\alpha^3\sqrt{\delta}}}$. 
Recall that from Lemma~\ref{lem:subspace-from-large-estimated-moment-tensor} 
it follows that the dimension of the subspace $V$ is at most 
$k^{\mathcal{O}(k)} \mathcal{O}(1+ 1/\alpha + \Delta)^{4k}/\lambda^2$. 
Thus, the dimension of $V'$ is at most
\begin{equation*}
k^{\mathcal{O}(k)} \bigO{ \frac{\Paren{1+ 1/\alpha + \Delta}^{4k}}{\lambda^2} }  +  \bigO{\frac{1+\Delta^2}{\sqrt{\eta}\alpha^3\sqrt{\delta}}} =  \bigO{ \frac{k^{\mathcal{O}(k)} \Paren{1+ \frac{1}{\alpha} + \Delta}^{4k+5} }{\eta^2}} \;.
\end{equation*}
Since $V'$ contains $V$ constructed in Lemma~\ref{lem:subspace-from-large-estimated-moment-tensor}, 
we immediately obtain that for every $\mu_i$, 
$\Norm{\mu_i-\Pi_{V'} \mu_i}^2_2 \leq \frac{20}{\alpha^2} \sqrt\delta \Delta$.

Next, let $u$ be a unit vector orthogonal to $V'$. 
Then, since $V'$ contains the $V$ described in Lemma~\ref{lem:subspace-from-large-estimated-moment-tensor}, 
we know that $\Norm{S_i u}^2_2 \leq \frac{20}{\alpha^2} \sqrt{\delta} \Delta$. 
Similarly, since $V'$ contains all eigenvectors of $\hat S_i$ with singular values exceeding $\delta^{1/4}$, 
we know that $\Norm{\hat S_i u}_2^2 \leq \delta^{1/2}$. 
Thus, we can conclude that $\Norm{ (\hat S_i - S_i)u}_2^2 \leq \frac{100}{\alpha^2} \sqrt{\delta} \Delta$. 
Let $Q_i=\sum_{j=1}^{k'}\tau_j v_jv_j^\top$ with orthonormal $v_j\in\R^d$. 
We know such $\tau_j$'s and $v_j$'s exist because of the upper bound on $\mathrm{rank}(Q_i)$. 
Therefore, for any $j$, $|v_j^{\top} (\hat S_i -S_i)u| \leq \frac{10}{\alpha} \delta^{1/4} \Delta^{1/2}$. 
On the other hand, for any $j$, we have that
\[
v_j^{\top} (\hat S_i- S_i) u \geq \iprod{v_j,u} \tau_j - \Norm{P_i}_F = \iprod{v_j,u} \tau_j - \mathcal{O}(\sqrt{\eta}) \;.
\]
Combining the two bounds above, yields that whenever $\tau_j \geq \delta^{1/4}$,
\begin{equation*}
|\iprod{v_j,u}| \leq \mathcal{O}(\sqrt{\eta}/ \tau_j) + \frac{10}{\alpha \tau_j} \delta^{1/4} \Delta^{1/2} \leq \frac{10}{\alpha} \delta^{1/2} \Delta^{1/2} \;.
\end{equation*}
Thus, the matrix $\hat Q_i = \sum_{j = 1}^{k'} \tau_j \Pi_{V'} v_j (\Pi_{V'}v_j)^{\top}$ 
has its range space in $V'$ and satisfies 
\begin{equation*}
\Norm{\hat Q_i - Q_i}_F \leq \bigO{ k^2 \delta^{1/4} }  + \bigO{  \frac{k^2}{\alpha} \delta^{1/2} \Delta^{1/2} }  = \bigO{  \frac{k^2}{\alpha} \delta^{1/2} \Delta^{1/2} } \;.
\end{equation*}
\end{proof}

\subsection{Parameter vs Moment Distance for Gaussian Mixtures} \label{ssec:mv}

In this subsection, we prove Lemma~\ref{lem:mv}. 
To that end, we will use the following two results; the second one is from~\cite{MoitraValiant:10}.

\begin{lemma}\label{lem:moment-distance-from-parameter}
Suppose $\calN(\mu_1,\sigma_1^2)$ and $\calN(\mu_2,\sigma_2^2)$ are univariate Gaussians 
with $|\mu_i|, |\sigma_i|\le D$, for some $D \in \R_{+}$. If $|\mu_1-\mu_2|+|\sigma_1^2-\sigma_2^2|\le\beta$, 
then the distance between raw moments of two Gaussians is
$$
\left|M_j(\calN(\mu_1,\sigma_1^2))-M_j(\calN(\mu_2,\sigma_2^2))\right|\le (j+1)!D^{j-1}\beta.
$$
\end{lemma}
\begin{proof}

By Proposition~\ref{momentsProp}, the $j$-th raw moment of a Gaussian $\cN(\mu,\sigma^2)$
is a sum of monomials in $\mu$ and $\sigma^2$ of degree $j$. 
There are at most $(j+1)!$ terms in the polynomial. Thus, changing the mean 
or the variance by at most $\beta$ will change the $j$-th moment by at most $(j+1)!D^{j-1}\beta$.
\end{proof}

\begin{theorem}(\cite{MoitraValiant:10})\label{thm:mv}
Let $F,F'$ be two univariate mixtures of Gaussians: 
$F=\sum_{i=1}^{k} w_i \calN(\mu_i,\sigma_i^2)$ and 
$F'=\sum_{i=1}^{k'} w_i'\calN(\mu_i',{\sigma_i'}^2).$
There is a constant $c>0$ such that, for any $\beta<c$, if $F,F'$ satisfy:
\begin{enumerate}
\item $w_i,w_i'\in[\beta,1]$
\item $|\mu_i|,|\mu_i'|\le 1/\beta$
\item $|\mu_i-\mu_{i'} |+|\sigma_i^2-\sigma_{i'}^2|\ge\beta$ and $|\mu_i'-\mu_{i'}'|+|{\sigma_i'}^2-{\sigma_{i'}'}^2|\ge\beta$ for all $i\neq i'$
\item $\beta \le\min_{\pi}\sum_i\left(|w_i-w_{\pi(i)}'|+|\mu_i-\mu_{\pi(i)}'|+|\sigma^2-{\sigma_{\pi(i)}'}^2|\right)$, 
where the minimization is taken over all mappings $\pi:\{1,\dots,k\}\to\{1,\dots,k'\}$,
\end{enumerate}
then
$$
\max_{j \in [2(k+k'-1)] } |M_j(F)-M_j(F')|\ge\beta^{O(k)} \;.
$$
\end{theorem}

We are now ready to complete the proof of Lemma~\ref{lem:mv}.

\begin{proof}[Proof of Lemma~\ref{lem:mv}]

We proceed via induction on $k$. Consider the base case, i.e., $k=1$. 
Then, either $\abs{\mu_1} \geq \beta/2 $ or $\abs{\sigma_1-1}\geq \beta/2$, 
and thus the first or second moment differ by at least $\beta^2/4$. 
Let the inductive hypothesis be that Lemma~\ref{lem:mv} holds for at most $k$ components. 

Consider the case where $|\mu_i-\mu_{i'}|+|\sigma_i^2-\sigma_{i'}^2|\ge\beta^{C^k k!}$ for all pairs of components $i, i' \in [k]$. Then,  by Theorem~\ref{thm:mv}, we have that
$$\max_{j \in [2k] }|M_j(F)-M_j(\cN(0,1))|\ge\beta^{C^{k+1}k!}\ge\beta^{C^{k+1}(k+1)!-1} \;,$$
and the lemma follows. 

Otherwise, we know that there exists a pair of components with parameter distance less than $\beta^{C^k k!}$. 
In this case, we merge these two components and get a $(k-1)$-mixture $F'$. 
By Lemma~\ref{lem:moment-distance-from-parameter}, the distance between 
the $j$-th moments of $F'$ and $F$ is at most $(j+1)!D^{j-1}\beta^{C^kk!}$. 
Since we still have $|\mu_i'|+|{\sigma_i'}^2-1| \geq \beta- 3\beta^{C^kk!}$ for all components $i$ in $F'$, 
the inductive hypothesis implies that 
$$
\max_{j\in [2k-2 ] }|M_j(F')-M_j(\cN(0,1))|\ge\left(\beta- 3\beta^{C^kk!}\right)^{C^{k}(k)!-1} \;.
$$
By the triangle inequality, we can write
\begin{align*}
\max_{j\in [2k] }|M_j(F)-M_j(\cN(0,1))|
&\ge\max_{j\in [2k-2]}|M_j(F')-M_j(\cN(0,1))|-\max_{j\in [2k-2]  }|M_j(F)-M_j(F')|\\
&\ge\left(\beta- 3\beta^{C^kk!}\right)^{C^{k}k!-1}-(2k-1)!D^{2k-3}\beta^{C^kk!}\\
&\ge\beta^{C^{k+1}(k+1)!-1} \;.
\end{align*}
The last inequality follows from the assumption that $\beta\le1/(2(2k-1)!D^{2k-3})$.
This completes the proof of Lemma~\ref{lem:mv}.
\end{proof}

%% file: partial-clustering.tex

\section{Robust Partial Cluster Recovery}
\label{sec:robust-partial-clustering} 
In this section, we give two robust \emph{partial clustering} algorithms. 
A partial clustering algorithm takes a set of points $X = \cup_{i \leq k}X_i$ 
with true clusters $X_1, X_2, \ldots, X_k$ and outputs a partition of the sample $X = X'_1 \cup X'_2$ such that $X'_1= \cup_{i \in S} X_i$ and 
$X'_2 = \cup_{i \not \in S} X_i$, 
for some subset $S \subseteq [k]$ of size $1 \leq |S| < k$. 
That is, a partial clustering algorithm partitions the sample into two non-empty parts 
so that each part is a sample from a ``sub-mixture''. This is a weaker guarantee than clustering the entire mixture, which must find each of the original $X_i$'s. We show that the relaxed guarantee is feasible even when the mixture as a whole 
is not clusterable. In our setting, we will get an approximate (that is, a small fraction of points are misclassified) partial clustering that works for $\epsilon$-corruptions $Y$ of any i.i.d. sample $X$ 
from a mixture of $k$ Gaussians, as long as there is a pair of components in the original mixture 
that have large total variation distance between them. 

A partial clustering algorithm such as above was one of the innovations in~\cite{BK20} 
that allowed for a polynomial-time algorithm for clustering all fully clusterable Gaussian mixtures. 

In this section, we  build on the ideas in~\cite{BK20} to derive two new partial clustering algorithms 
that work even when the original mixture is not fully clusterable. Both upgrade the results of ~\cite{BK20} by handling mixtures with arbitrary weights $w_i$s instead of uniform weights and handling mixtures where not all pairs of components are well-separated in TV distance. The first algorithm succeeds under the information-theoretically minimal separation assumption (i.e. separation in total variation distance) but runs in time exponential in the inverse mixing weight. The second algorithm is a key innovation of this paper -- it gives an algorithm that runs in polynomial time in the inverse mixing weight at the cost of handling separation only in relative Frobenius distance. This improved running time guarantee (at the cost of strong separation requirement that we mitigate through a novel standalone spectral separation step in Section~\ref{sec:spectral-separation}) is crucial to obtaining the fully polynomial running time in our algorithm.

In order to state the guarantees of our algorithms, 
we first formulate a notion of parameter separation 
(same as the one employed in~\cite{BK20,DHKK20}) as the next definition.

\begin{definition}[$\Delta$-Parameter Separation] \label{def:separated-mixture-model}
We say that two Gaussian distributions $\cN(\mu_1,\Sigma_1)$ and $\cN(\mu_2,\Sigma_2)$ 
are $\Delta$-parameter separated if at least one of the following three conditions hold:
\begin{enumerate}
  \item \textbf{Mean-Separation:} $\exists v \in \R^d$ such that 
  $\langle \mu_1 - \mu_2, v \rangle^2 > \Delta^2 v^{\top} (\Sigma_1 + \Sigma_2) v$,
  \item \textbf{Spectral-Separation:} $\exists v \in \R^d$ such that $v^{\top} \Sigma_1 v > \Delta \, v^{\top}\Sigma_2v$,
  \item \textbf{Relative-Frobenius Separation:} $\Sigma_i$ and $\Sigma_j$ have the same range space and $\Norm{ \Sigma_1^{\dagger/2} (\Sigma_2 - \Sigma_1) \Sigma_1^{\dagger/2}}^2_F > \Delta^2 \Norm{\Sigma_1^{\dagger} \Sigma_2}_{\mathrm{op}}^2$.
\end{enumerate}
\end{definition}
As shown in ~\cite{DHKK20,BK20}, if a pair of Gaussians is $(1-\exp(-O(\Delta \log \Delta))$-separated in total variation distance, then, they are $\Delta$-parameter separated. 


Our first algorithm succeeds in robust partial clustering whenever there is a pair of component Gaussians that are 
$\Delta$-parameter separated. The running time of this algorithm grows exponentially 
in the reciprocal of the minimum weight in the mixture.  

\begin{theorem}[Robust Partial Clustering in TV Distance] \label{thm:partial-clustering-non-poly}
Let $0 \leq \epsilon < \alpha \leq 1$, and $\eta > 0$.  There is an algorithm with the following guarantees: 
Let $\{\mu_i,\Sigma_i\}_{i \leq k}$ be means and covariances of $k$ unknown Gaussians. 
Let $Y$ be an $\epsilon$-corruption of a sample $X$ of size $n\geq \paren{dk}^{Ct}/\epsilon$ for a large enough constant $C>0$, from $\calM = \sum_i w_i \cN(\mu_i,\Sigma_i)$ satisfying Condition~\ref{cond:convergence-of-moment-tensors} with parameters $t = (k/\eta)^{O(k)}$ and 
$\gamma \leq \epsilon d^{-8t}k^{-Ct}$, for a sufficiently large constant $C>0$.  
Suppose further that $w_i \geq \alpha > 2\epsilon$ for every $i$ and that 
there are $i, j$ such that $\cN(\mu_i,\Sigma_i)$ and $\cN(\mu_j,\Sigma_j)$ 
are $\Delta$-parameter separated for $\Delta = (k/\eta)^{O(k)}$. 

Then, the algorithm on input $Y$,
runs in time $n^{(k/\eta)^{O(k)}}$, and with probability at least $2^{-O\Paren{\frac{1}{\alpha} \log\Paren{\frac{k}{\eta \alpha}}}}$ 
over the draw of $X$ and the algorithm's random choices, 
the algorithm outputs a partition of $Y$ into $Y_1,Y_2$ satisfying: 
\begin{enumerate}
 \item \textbf{Partition respects clustering:} for each $i$, 
 $\max \{ \frac{k}{n}|Y_1 \cap X_i|, \frac{k}{n}|Y_2 \cap X_i|\} \geq 1-\eta- O(\eps/\alpha^4)$, and, 
 \item \textbf{Partition is non-trivial:} $\max_{i}\frac{k}{n} |X_i \cap Y_1|, \max_{i}\frac{k}{n} |X_i \cap Y_2| \geq 1-\eta-O(\eps/\alpha^4)$. 
\end{enumerate}
\end{theorem}

Our proof of the above theorem is based on a relatively straightforward extension of the ideas of~\cite{BK20}, albeit
with two key upgrades 1)  allowing the input mixtures to have arbitrary mixing weights 
(at an exponential cost in the inverse of the minimum weight) and 2) handling mixtures where some pair of components may not be well-separated in TV distance. 

In order to get our main result that gives a fully polynomial algorithm (including in the inverse mixing weights), 
we will use a incomparable variant of the above partial clustering method that only handles 
a weaker notion of parameter separation, but runs in fixed polynomial time.

\begin{theorem}[Robust Partial Clustering in Relative Frobenius Distance] \label{thm:partial-clustering-poly}
Let $0 \leq \epsilon < \alpha/k \leq 1$ and $t \in \N$. There is an algorithm with the following guarantees: 
Let $\{\mu_i,\Sigma_i\}_{i \leq k}$ be means and covariances of $k$ unknown Gaussians. 
Let $Y$ be an $\epsilon$-corruption of a sample $X$ of size $n\geq \paren{dk}^{Ct}/\epsilon$ for a large enough constant $C>0$, from 
$\calM = \sum_i w_i \cN(\mu_i,\Sigma_i)$ that satisfies Condition~\ref{cond:convergence-of-moment-tensors} 
with parameters $2t$ and $\gamma \leq \epsilon d^{-8t}k^{-Ck}$, for a large enough constant $C >0$. 
Suppose further that $w_i \geq \alpha > 2 \epsilon$ for each $i \in [k]$, 
and that for some $t \in \N$, $\beta > 0$ there exist $i,j \leq k$ such that 
$\Norm{\Sigma^{\dagger/2}(\Sigma_i -\Sigma_j)\Sigma^{\dagger/2}}_F^2 = \Omega\left((k^2 t^4)/(\beta^{2/t} \alpha^4)\right)$, 
where $\Sigma$ is the covariance of the mixture $\calM$. 
Then, the algorithm runs in time $n^{O(t)}$, and with probability at least $2^{-O\Paren{ \frac{1}{\alpha} \log\Paren{\frac{k}{\eta}} } }$ over the random choices of the algorithm, 
outputs a partition $Y = Y_1 \cup Y_2$ satisfying:   
\begin{enumerate}
 \item \textbf{Partition respects clustering:} for each $i$, $\max \left\{ \frac{1}{w_i n}|Y_1 \cap X_i|, \frac{1}{w_i n}|Y_2 \cap X_i|\right\} \geq 1- \beta - O(\epsilon/\alpha^4)$, and, 
 \item \textbf{Partition is non-trivial:} $\max_{i}\frac{1}{w_i n} |X_i \cap Y_1|, \max_{i}\frac{1}{w_i n} |X_i \cap Y_2| \geq 1- \beta - O(\epsilon/\alpha^4)$. 
\end{enumerate}
\end{theorem}

The starting point for the proof of the above theorem is the observation that the running time of our first algorithm above is exponential in the inverse mixing weight almost entirely because of dealing with spectral separation 
(which requires the use of ``certifiable anti-concentration'' that we define in the next subsection). 
We formulate a variant of relative Frobenius separation (that is directly useful to us) and prove that 
whenever the original mixture has a pair of components separated in this notion, 
we can in fact obtain a fully polynomial partial clustering algorithm building on the ideas in~\cite{BK20}.






\subsection{Algorithm}

Our algorithm will solve SoS relaxations of a polynomial inequality system. 
The constraints here use the input $Y$ to encode finding a sample $X'$ 
(the intended setting being $X'=X$, the original uncorrupted sample) 
and a cluster $\hat{C}$ in $X'$ of size $= \alpha n$, indicated by $z_i$s 
(the intended setting is simply the indicator for any of the $k$ true clusters) 
satisfying properties of Gaussian distribution (certifiable hypercontractivity and anti-concentration).

Covariance constraints introduce a matrix valued indeterminate $\Pi$ intended to be the square root of $\hat{\Sigma}$, the sos variable for the covariance of a single component.
\begin{equation}
\text{Covariance Constraints: $\cA_1$} = 
  \left \{
    \begin{aligned}
      &
      &\Pi
      &=UU^{\top}\\
      &
      &\Pi^2
      &=\hat{\Sigma}\\
    \end{aligned}
  \right \}
\end{equation}
The intersection constraints force that $X'$ be $\epsilon$-close to $Y$ (and thus, $2\epsilon$-close to unknown sample $X$). 
\begin{equation}
\text{Intersection Constraints: $\cA_2$} = 
  \left \{
    \begin{aligned}
      &\forall i\in [n],
      & m_i^2
      & = m_i\\
      &&
      \textstyle\sum_{i\in[n]} m_i 
      &= (1-\epsilon) n\\
      &\forall i \in [n],
      &m_i (y_i-x'_i)
      &= 0
    \end{aligned}
  \right \}
\end{equation}
The subset constraints introduce $z$, which indicates the subset $\hat{C}$ intended to be the true clusters of $X'$.
\begin{equation}
\text{Subset Constraints: $\cA_3$} = 
  \left \{
    \begin{aligned}
      &\forall i\in [n].
      & z_i^2
      & = z_i\\
      &&
      \textstyle\sum_{i\in[n]} z_i 
      &= \alpha n\\
    \end{aligned}
  \right \}
\end{equation}

Parameter constraints create indeterminates to stand for the covariance $\hat{\Sigma}$ and mean $\hat{\mu}$ of $\hat{C}$ (indicated by $z$).
\begin{equation}
\text{Parameter Constraints: $\cA_4$} = 
  \left \{
    \begin{aligned}
      &
      &\frac{1}{\alpha n}\sum_{i = 1}^n z_i \Paren{x'_i-\hat\mu}\Paren{x'_i-\hat\mu}^{\top}
      &= \hat{\Sigma}\\
      &
      &\frac{1}{\alpha n}\sum_{i = 1}^n z_i x'_i
      &= \hat\mu\\
    \end{aligned}
  \right \}
\end{equation}


\text{Certifiable Hypercontractivity : $\cA_4$}=\\
\begin{equation}
  \left \{
    \begin{aligned}
     &\forall t \leq 2s
     &\frac{1}{\alpha^2 n^2} \sum_{i,j \leq n} z_i z_j (Q(x'_i-x'_j)-\E_zQ)^{2t}
     &\leq \Paren{\frac{1}{\alpha^2 n^2} \sum_{i,j \leq n} z_i z_j (Q(x'_i-x'_j)-\E_zQ)^{2}}^t\\
     &
     &\nnnew{\frac{1}{\alpha^2 n^2} \sum_{i,j \leq n} z_i z_j \Paren{Q(x'_i-x'_j)-\E_z Q}^{2}}
     &\nnnew{\leq \frac{6}{\alpha^2} \Norm{Q}_F^{2}}
    \end{aligned}
  \right \}
\end{equation}
Here, we used the shorthand $\E_z Q = \frac{1}{\alpha^2 n^2} \sum_{i,j \leq n} z_i z_j Q(x'_i-x'_j)$.

In the constraint system for our first algorithm, we will use the following certifiable anti-concentration constraints on $\hat{C}$ for $\delta = \alpha^{-\poly(k)}$ and $\tau = \alpha/\poly(k)$ and $s(u) = 1/u^2$ for every $u$.
\begin{equation}
\text{Certifiable Anti-Concentration : $\cA_5$} =
  \left \{
    \begin{aligned}
      &
      &\frac{1}{\alpha^2 n^2}\sum_{i,j=  1}^n z_i z_j q_{\delta,\Sigma}^2\left(\Paren{x'_i-x'_j},v\right)
      &\leq 2^{s(\delta)} C\delta \Paren{v^{\top}\Sigma v}^{s(\delta)}\\
      &
      &\frac{1}{\alpha^2 n^2}\sum_{i,j=  1}^n z_i z_j q_{\tau,\Sigma}^2\left(\Paren{x'_i-x'_j},v\right)
      &\leq 2^{s(\tau)} C\tau \Paren{v^{\top}\Sigma v}^{s(\eta)}\\
     \end{aligned}
    \right\}
 \end{equation}

We note that the constraint system for our second algorithm (running in fixed polynomial time), we will not use $\cA_5$. 
Towards proving Theorems~\ref{thm:partial-clustering-poly} and~\ref{thm:partial-clustering-non-poly} we use the following algorithm that differs only in the degree of the pseudo-distribution computed and the constraint system that the pseudo-distribution satisfies. 

\begin{mdframed}
  \begin{algorithm}[Partial Clustering]
    \label{algo:robust-partial-clustering}\mbox{}
    \begin{description}
    \item[Given:]
        A sample $Y$ of size $n$. An outlier parameter $\epsilon > 0$ and an accuracy parameter $\eta > 0$. 
    \item[Output:]
      A partition of $Y$ into partial clustering $Y_1 \cup Y_2$.
    \item[Operation:]\mbox{}
    \begin{enumerate}
   		\item \textbf{SDP Solving:} Find a pseudo-distribution $\tzeta$ satisfying $\cup_{i =1}^5 \cA_i$ ($\cup_{i = 1}^4 \cA_i$ for Theorem~\ref{thm:partial-clustering-poly}) such that $\pE_{\tzeta}z_i \leq \alpha+o_d(1)$ for every $i$. If no such pseudo-distribution exists, output fail.
      	\item \textbf{Rounding:} Let $M = \pE_{z \sim \tzeta} [zz^{\top}]$. 
       		\begin{enumerate} 
        		 \item Choose $\ell = \bigO{\frac{1}{\alpha} \log (k/\eta)}$ rows of $M$ uniformly at random and independently. 
       			\item For each $i \leq \ell$, let $\hat{C}_i$ be the indices of the columns $j$ such that $M(i,j) \geq \eta^2 \alpha^5/k$.
       		
    	\item Choose a uniformly random $S \subseteq [\ell]$ and output $Y_1 = \cup_{i \in S} \hat{C}_i$ and $Y_2 = Y \setminus Y_1$.
    	\end{enumerate}
      \end{enumerate}
    \end{description}
  \end{algorithm}
\end{mdframed}

\subsection{Analysis}

\paragraph{Simultaneous Intersection Bounds.}

The key observation for proving the first theorem is the following lemma 
that gives a sum-of-squares proof that no $z$ that satisfies the constraints $\cup_{i= 1}^5 \cA_i$ 
can have simultaneously large intersections with the $\Delta$-parameter 
separated component Gaussians.

\begin{lemma}[Simultaneous Intersection Bounds for TV-separated case] \label{lem:simultaneous-intersection-bounds}
Let $Y$ be an $\epsilon$-corruption of a sample $X$ of size $n\geq \paren{dk}^{Ct}/\epsilon$ for a large enough constant $C>0$, from $\calM = \sum_i w_i \cN(\mu_i,\Sigma_i)$ satisfying Condition~\ref{cond:convergence-of-moment-tensors} with parameters $t = (k/\eta)^{O(k)}$ and 
$\gamma \leq \epsilon d^{-8t}k^{-Ct}$, for a sufficiently large constant $C>0$.  
Suppose further that $w_i \geq \alpha > 2\epsilon$ for every $i$ and that 
there are $i, j$ such that $\cN(\mu_i,\Sigma_i)$ and $\cN(\mu_j,\Sigma_j)$ 
are $\Delta$-parameter separated for $\Delta = (k/\eta)^{O(k)}$. 
Then, there exists a partition of $[k]$ into $S \cup L$ such that, $|S|, |L| < k$ 
and for $z(X_r) = \frac{1}{w_r n}\sum_{i \in X_r} z_i $,
\[
\Set{\cup_{i=1}^5\cA_i} \sststile{{(k/\eta \alpha)}^{\poly(k)}}{z} \Set{\sum_{i \in S, j \in L} z(X_i)z(X_{j}) \leq  O(k^2 \epsilon/\alpha) + \eta/\alpha}\mper
\]
\end{lemma}

\noindent The proof of Lemma~\ref{lem:simultaneous-intersection-bounds} is given in 
Section~\ref{ssec:simultaneous-intersection-bounds}.

For the second theorem, we use the following version that strengthens the separation assumption 
and lowers the degree of the sum-of-squares proof (and consequently the running time of the algorithm) as a result.

\begin{lemma}[Simultaneous Intersection Bounds for Frobenius Separated Case] 
\label{lem:simultaneous-intersection-bounds-poly-clustering}
Let $X$ be a sample of size $n\geq \paren{dk}^{Ct}/\epsilon$ for a large enough constant $C>0$, from 
$\calM = \sum_i w_i \cN(\mu_i,\Sigma_i)$ that satisfies Condition~\ref{cond:convergence-of-moment-tensors} 
with parameters $2t$ and $\gamma \leq \epsilon d^{-8t}k^{-Ck}$, for a large enough constant $C >0$. 
Suppose further that $w_i \geq \alpha > 2 \epsilon$ for each $i \in [k]$, 
and that for some $t \in \N$, $\beta > 0$ there exist $i,j \leq k$ such that 
$\Norm{\Sigma^{\dagger/2}(\Sigma_i -\Sigma_j)\Sigma^{\dagger/2}}_F^2 = \Omega\left((k^2 t^4)/(\beta^{2/t} \alpha^2)\right)$, 
where $\Sigma$ is the covariance of the mixture $\calM$. 
Then, for any $\epsilon$-corruption $Y$ of $X$, there exists a partition of $[k] = S \cup T$ such that  
\[
\Set{\cup_{i = 1}^4\cA_i } \sststile{2t}{z} \Set{\sum_{i \in S} \sum_{j\in T} z(X_i) z(X_j) \leq O(k^2) \beta +O(k^2)\epsilon/\alpha}\mper
\]
Here,  $z(X_r) = \frac{1}{w_r n}\sum_{i \in X_r} z_i $ for every $r$.
\end{lemma}

\noindent The proof of Lemma~\ref{lem:simultaneous-intersection-bounds-poly-clustering} is given in 
Section~\ref{ssec:simultaneous-intersection-bounds-poly-clustering}.

Notice that the main difference between the above two lemmas 
is the constraint systems they use. Specifically, the second lemma does \emph{not} enforce 
certifiable anti-concentration constraints. As a result, there is a difference in 
the degree of the sum-of-squares proofs they claim; the degree of the SoS proof in the second lemma
does not depend on the inverse minimum mixture weight. 

First, we complete the proof of the Theorem~\ref{thm:partial-clustering-non-poly}. 
The proof of Theorem~\ref{thm:partial-clustering-poly} is exactly the same except 
for the use of Lemma~\ref{lem:simultaneous-intersection-bounds-poly-clustering} 
(and thus has the exponent in the running time independent of $1/\alpha$) 
instead of Lemma~\ref{lem:simultaneous-intersection-bounds}.

\begin{proof}[Proof of Theorem~\ref{thm:partial-clustering-non-poly}]

Let $\eta' = O(\eta^2 \alpha^3/k)$. We will prove that whenever $\Delta \geq \poly(k/\eta')^k = \poly\Paren{\frac{k}{\eta \alpha}}^k$, Algorithm~\ref{algo:robust-partial-clustering}, when  run with input $Y$, with probability at least $0.99$, 
recovers a collection $\hat{C}_1, \hat{C}_2,\ldots, \hat{C}_{\ell}$ of $\ell = \bigO{\frac{1}{\alpha} \log k/\eta}$ 
subsets of indices satisfying $|\cup_{i \leq \ell} \hat{C}_i| \geq (1-\eta'/k^{40})n$ 
such that there is a partition $S \cup L = [\ell]$, $0 < |S| <\ell$ satisfying:
\begin{equation}\label{eq:partial-recovery-clusterwise}
\min\Biggl\{ \frac{1}{\alpha n} |\hat{C}_i \cap \cup_{j \in S} X_j|, \frac{1}{\alpha n} |\hat{C}_i \cap \cup_{j \in L} X_j| \Biggr\} \leq 100 \eta' /\alpha^3 + O(\epsilon/\alpha^4)\mper 
\end{equation}
We first argue that this suffices to complete the proof.
Split $[\ell]$ into two groups $G_S, G_L$ as follows.
For each $i$, let $j = \argmax_{r \in [\ell]} \frac{1}{\alpha n} |\hat{C}_i \cap X_r|$.
If $j \in S$, add it to $G_S$, else add it to $G_L$.
Observe that this process is well-defined - i.e, there cannot be $j\in S$ and $j' \in L$ that both maximize $\frac{1}{\alpha n} |\hat{C}_i \cap X_r|$ as $r$ varies over $[k]$. To see this, WLOG, assume $j \in S$.  Note that $\frac{1}{\alpha n} |\hat{C}_i \cap X_j| \geq 1/k$. Then, we immediately obtain: $\frac{1}{\alpha n} |\cup_{j \in S} X_j \cap \hat{C}_i| \geq 1/k$. Now, if we ensure that $\eta' \leq \alpha^3/k^2$ and $\epsilon \leq O(\alpha^4/k)$, then, $\frac{1}{\alpha n} |\hat{C}_i \cap \cup_{j' \in L} X_{j'}|$ is at most the RHS of \eqref{eq:partial-recovery-clusterwise} which is $\ll 1/k$. This completes the proof of well-definedness. 
Next, adding up \eqref{eq:partial-recovery-clusterwise} for each $i \in S$ yields that 
\[
\frac{1}{|\hat{C}_i|} | \Paren{\cup_{i \in G_S} X_i} \cap \cup_{j \in L} X_j| \leq O(\log(k/\eta')/\alpha)  \Paren{\eta' + \bigO{\epsilon/\alpha}},
\] 
where we used that $|G_S| \leq \ell$. Combined with $|\cup_{i \leq \ell} \hat{C}_i| \geq (1-\eta'/k^{40})n$, we obtain that 
\[
|\cup_{i \in G_S} X_i| \geq 1- \eta'/k^{40} - O(\log(k/\eta')/\alpha)  \Paren{\eta' + O(\epsilon/\alpha)} = \eta + \bigO{\log(k/\eta \alpha)\epsilon/\alpha^2}
\]
for $\eta' \leq \bigO{\eta^2 \alpha^3/k}$.

We now go ahead and establish \eqref{eq:partial-recovery-clusterwise}. 
Let $\tzeta$ be a pseudo-distribution satisfying $\cA$ of degree $(k/\eta)^{\poly(k)}$ satisfying $\pE_{\tzeta} z_i = \alpha$ for every $i$. Such a pseudo-distribution exists. To see why, let $\tzeta$ be the actual distribution that always sets $X' = X$, chooses an $i$ with probability $w_i$ and outputs a uniformly subset $\hat{C}$ of size $\alpha n$ of $X_i$ conditioned on $\hat{C}$ satisfying $\cA$. Then, notice that since $X$ satisfies Condition~\ref{cond:convergence-of-moment-tensors}, by Fact~\ref{fact:moments-to-analytic-properties}, the uniform distribution on each $X_i$ has $t$-certifiably $C$-hypercontractive degree $2$ polynomials and is $t$-certifiably $C\delta$-anti-concentrated. By an concentration argument using high-order Chebyshev inequality, similar to the proof of Lemma~\ref{lem:det-suffices} (applied to uniform distribution on $X_i$ of size $n \geq (dk)^{O(t)}$, $\hat{C}$ chosen above satisfies the constraints $\cA$ with probability at least $1-o_d(1)$. Observe that the probabililty that $z_i$ is set to $1$ under this distribution is then at most $\alpha+o_d(1)$. Thus, such a distribution satisfies all the constraints in $\cA$. 

Next, let $M = \pE_{\tzeta}[zz^{\top}]$. Then, we claim that:
\begin{enumerate}
  \item $o_d(1) + \alpha \geq M(i,j) \geq 0$ for all $i,j$,
  \item $M(i,i) \in  \alpha \pm o_d(1)$ for all $i$,
  \item $\E_{j \sim [n]} M(i,j) \geq \alpha^2-o_d(1)$ for every $i$.
\end{enumerate}

The proofs of these basic observations are similar to those presented in Chapter 4.3 of~\cite{TCS-086} (see also the proof of Theorem 5.1 in ~\cite{BK20}):
Observe that $\cA \sststile{4}{} \Set{z_i z_j = z_i^2 z_j^2 \geq 0}$ for every $i$. Thus, by Fact~\ref{fact:sos-completeness}, $\pE[z_i z_j] \geq 0$ for every $i,j$. Next, observe that $\cA \sststile{2}{} \Set{(1-z_i) = (1-z_i)^2 \geq 0}$ for every $i$ and thus, $\cA \sststile{2}{} \Set{z_i (1-z_j) \geq 0}$. Thus, by Fact~\ref{fact:sos-completeness} again, we must have $\pE[z_i z_j] \leq \pE[z_i] \leq \alpha + o_d(1)$. Finally, $\cA \sststile{2}{} \Set{\sum_{j} z_i z_j = z_i \sum_j z_j = \alpha n z_i}$. Thus, by Fact~\ref{fact:sos-completeness} again, we must have $\sum_{j} M(i,j) = \sum_{j} \pE[z_i z_j] = \alpha n \sum_j \pE[z_i] \in (\alpha^2 \pm o_d(1)) n$. Let $B_i$ be the entries in the $i$-th row $M_i$ that are larger than $\alpha^2/2$. Then, by (1) and (2), we immediately derive that $B_i$ must have at least $\alpha n/2$ elements. 
Call an entry of $M$ large if it exceeds $\alpha^2 \eta'$. For each $i$, let $B_i$ be the set of large entries in row $i$ of $M$. Then, using (3) and (1) above gives that $|B_i| \geq \alpha(1-\alpha\eta')n$ for each $1 \leq  i \leq n$. Next, call a row $i$ ``good'' if $\frac{1}{\alpha n}\min \{ \abs{\cup_{r \in L} X_r \cap B_i}, \abs{\cup_{r' \in S} X_{r'} \cap B_i}\} \leq 100 \eta' /\alpha^3 + O(\epsilon/\alpha^4)$. Let us estimate the fraction of rows of $M$ that are good. 

Towards that goal, let us apply Lemma~\ref{lem:simultaneous-intersection-bounds} with $\eta$ set to $\eta'$ and use Fact~\ref{fact:sos-completeness} (SoS Completeness),  to obtain $\sum_{r \in S, r' \in L}\E_{i \in X_r} \E_{j \in X_{r'}} M(i,j)  \leq \eta' +O(\epsilon/\alpha)$. Using Markov's inequality, with probability $1-\alpha^3/100$ over the uniformly random choice of $i$, $\E_{ j \in X_{r'}} M(i,j) \leq 100 \frac{1}{\alpha^3}\eta' + O(\epsilon/\alpha^4)$. Thus, $1-\alpha^3/100$ fraction of the rows of $M$ are good.

Next, let $R$ be the set of $\frac{100}{\alpha} \log\Paren{\frac{k^{50}}{\eta'} }$ rows sampled in the run of the algorithm and set $\hat{C}_i = B_i$ for every $i \in R$. The probability that all of them are good is then at least $(1-\alpha^3/100)^{\frac{100}{\alpha} \log \Paren{\frac{k^{50}}{\alpha \eta'}}} \geq 1-\alpha$. Let us estimate the probability that $|\cup_{i \in R} \hat{C}_i| \geq (1-\eta'/k^{40})n$. For a uniformly random $i$, the chance that a given point $t \in B_i$ is at least $\alpha(1-\alpha \eta')$. Thus, the chance that $t \not \in \cup_{i \in R}B_i$ is at most $(1-\alpha/2)^{100/\alpha \log\Paren{ k^{50}/(\alpha \eta')} } \leq \eta'/k^{50}$. Thus, the expected number of $t$ that are not covered by $\cup_{i \in R} \hat{C}_i$ is at most $n \eta' /k^{50}$. Thus, by Markov's inequality, with probability at least $1-1/k^{10}$, $1-\eta'/k^{40}$ fraction of $t$ are covered in $\cup_{i \in R} \hat{C}_i$. By the above computations and a union bound, with probability at least $1-\eta'/k^{10}$ both the conditions below hold simultaneously: 1) each of the $\frac{100}{\alpha} \log\Paren{ k^{50}/\eta'}$ rows $R$ sampled are good and 2) $|\cup_{i \in R} \hat{C}_i| \geq (1-\eta'/k^{40})n$.  This completes the proof.
\end{proof}




\subsection{Proof of Lemma~\ref{lem:simultaneous-intersection-bounds}} \label{ssec:simultaneous-intersection-bounds}

Our proof is based on the following simultaneous intersection bounds from~\cite{BK20}. 
We will use the following lemma that forms the crux of the analysis of the clustering algorithm in~\cite{BK20}:

\begin{lemma}[Simultaneous Intersection Bounds, Lemma~5.4 in~\cite{BK20}] \label{lem:intersection-bounds-non-poly}
Fix $\delta >0, k \in \N$. Let $X = X_1 \cup X_2 \cup \ldots X_k$ be a good sample of size $n$ from a $k$-mixture 
$\sum_i w_i \cN(\mu_i,\Sigma_i)$ of Gaussians. Let $Y$ be any $\epsilon$-corruption of $X$. 
Suppose there are $r,r' \leq k$ such that one of the following three conditions hold for some $\Delta \geq (k/\delta)^{O(k)}$:
\begin{enumerate}
\item there exists a $v$ such that $v^{\top}\Sigma(r') v > \Delta v^{\top}\Sigma(r') v$ and 
$B = \max_{i \leq k} \frac{v^{\top} \Sigma(i)v}{v^{\top}\Sigma(r')v}$, or 

\item there exists a $v \in \R^d$ such that $\Iprod{\mu(r) -\mu(r'),v}_2^2 \geq \Delta^2 v^{\top} \Paren{\Sigma(r) +\Sigma(r')} v$, or, 

\item $\Norm{\Sigma(r')^{-1/2}\Sigma(r)\Sigma(r')^{-1/2} - I}_F^2 \geq 
\Delta^2 \Paren{ \Norm{\Sigma(r')^{-1/2}\Sigma(r)^{1/2}}_{\mathrm{op}}^4}$.
\end{enumerate}
Then, for the linear polynomial $z(X_r) = \frac{1}{\alpha n} \sum_{i \in X_r} z_i$ in indeterminates $z_i$s satisfies:
\[
\Set{ \cup_{i \leq 5}  \cA_i } \sststile{(k/\delta)^{O(k)}\log (2B)}{z} \Set{z(X_r)z(X_{r'}) \leq  O(\sqrt{\delta}) + O(\epsilon/\alpha) }\mper
\]
\end{lemma}

\begin{proof}[Proof of Lemma~\ref{lem:simultaneous-intersection-bounds}]
Without loss of generality, assume that the pair of separated components are $\cN(\mu_1,\Sigma_1)$ and $\cN(\mu_2,\Sigma_2)$. Let us start with the case when the pair is spectrally separated. 
Then, there is a $v \in \R^d$ such that $\Delta v^{\top} \Sigma_1v \leq v^{\top} \Sigma_2 v$. 

Consider an ordering of the true clusters along the direction $v$, renaming cluster indices if needed, such that $v^{\top} \Sigma_1 v \leq v^{\top}\Sigma_2 v \leq \ldots v^{\top}\Sigma_k v$. Let $j \leq k'$ be the largest integer such that $\poly(k/\eta)v^{\top}\Sigma_j v \leq v^{\top}\Sigma_{j+1} v$. Further, observe that since $j$ is defined to be the largest index which incurs separation $\poly(k/\eta)$, all indices in $[j,k]$ have spectral bound at most $\poly(k/\eta)$ and thus $\frac{v^{\top} \Sigma_k v}{v^{\top}\Sigma_j v} \leq \poly(k/\eta)^{k}$.



Applying Lemma~\ref{lem:intersection-bounds-non-poly} with the above direction $v$ to every $r <j$ and $r' \geq j$ and observing that the parameter $B$ in each case is at most $\frac{v^{\top} \Sigma_k v}{v^{\top}\Sigma_j v} \leq \Delta^{k}$ yields:

\[
\cA \sststile{ O({k}^2 s^2 \poly \log (\Delta) )}{z} \Set{z(X_r)z(X_{r'}) \leq O(\epsilon/\alpha) +\sqrt{\delta}}\mper
\]

Adding up the above inequalities over all $r \leq j-1$ and $r' \geq j+1$ and taking $S = [j-1]$, $T = [k]\setminus [j-1]$ completes the proof in this case. 

Next, let us take the case when $\cN(\mu_1,\Sigma_1)$ and $\cN(\mu_2,\Sigma_2)$ are mean-separated. WLOG, suppose $\langle \mu_1,v \rangle \leq \langle \mu_2, v \rangle \ldots \leq \langle \mu_k, v\rangle$. Then, we know that $\langle \mu_k-\mu_1, v \rangle \geq \Delta v^{\top} \Sigma_i v$. Thus, there must exist an $i$ such that $\langle \mu_i - \mu_{i+1},v \rangle \geq \Delta v^{\top} \Sigma_i v/k$. Let $S = [i]$ and $L= [k] \setminus S$. Applying Lemma~\ref{lem:intersection-bounds-non-poly} and arguing as in the previous case (and noting that $\kappa = \poly(k)$) completes the proof.

Finally, let us work with the case of relative Frobenius separation. Since $\|\Sigma_1^{-1/2}\Sigma_k^{1/2}\| \leq \poly(k)$, the hypothesis implies that $\Norm{\Sigma_1- \Sigma_2}_F \geq \Delta/\poly(k)$. Let $B = \Sigma_1 -\Sigma_2$ and let $A = B/\Norm{B}_F$.
WLOG, suppose $\iprod{\Sigma_1,A}\leq \ldots \iprod{\Sigma_k,A}$.  
Then, since $\iprod{\Sigma_k,A} - \iprod{\Sigma_1,A} \geq \Delta/\poly(k)$, there must exist an $i$ such that $\iprod{\Sigma_{i+1},A} - \iprod{\Sigma_{i},A} \geq \Delta/\poly(k)$. Let us now set $S = [i]$ and $L = [k] \setminus S$. 

Then, for every $i \in S$ and $j \in L$, we must have: $\iprod{\Sigma_j,A} - \iprod{\Sigma_i,A}\geq \Delta/\poly(k)$. 
Thus, $\Norm{\Sigma_j-\Sigma_i}_F \geq \Delta/\poly(k)$. 
And thus, $\Delta/\poly(k)  \leq \Norm{\Sigma_j-\Sigma_i}_F \leq \Norm{\Sigma_i^{-1/2}\Sigma_j^{1/2}}_2^2 \Norm{\Sigma_i^{-1/2}\Sigma_j\Sigma_i^{-1/2}-I}_F$. Rearranging and using the bound on $\Norm{\Sigma_i^{-1/2}\Sigma_j^{1/2}}_2^2$ yields that $\Norm{\Sigma_i^{-1/2}\Sigma_j\Sigma_i^{-1/2}-I}_F \geq \Delta/\poly(k)$. 

A similar argument as in the two cases above now completes the proof.

\end{proof}

\subsection{Proof of Lemma~\ref{lem:simultaneous-intersection-bounds-poly-clustering}} 
\label{ssec:simultaneous-intersection-bounds-poly-clustering}

We use $\E_z$ as a shorthand for $\frac{1}{\alpha n}\sum_{i = 1}^n z_i$. We will write $\frac{1}{w_r n} \sum_{j \in X_r} z_j = z(X_r)$. 
Note that $z(X_r) \in [0,1]$. And finally, we will write $z'(X_r) = \frac{1}{w_r n} \sum_{j \in X_r} z_j \1(x_j = y_j)$ -- the version of $z(X_r)$ that only sums over non-outliers. 

We will use the following technical facts in the proof:

\begin{fact}[Lower Bounding Sums, Fact 4.19~\cite{BK20}] \label{fact:lower-bounding-sums-SoS}
Let $A,B,C,D$ be scalar-valued indeterminates. Then, for any $\tau >0$,
\[
\Set{0 \leq A, B \leq A+B \leq 1} \cup \Set{0\leq C,D} \cup \Set{C+D\geq \tau} \sststile{2}{A,B,C,D} \Set{ AC+ BD \geq \tau AB}\mper
\]

\end{fact}

\begin{fact}[Cancellation within SoS, Lemma 9.2 in~\cite{BK20}] \label{fact:cancellation-within-sos-constant-rhs}
For indeterminate $a$ and any $t \in \N$, 

\[
\Set{a^{2t} \leq 1} \sststile{2t}{a} \Set{ a\leq 1}\mper
\]
\end{fact}





\begin{lemma}[Lower-Bound on Variance of Degree 2 Polynomials] \label{lem:lower-bound-variance}
Let $Q \in \R^{d \times d}$. Then, for any $i,j \leq k$,  and $z'(X_r) = \frac{1}{w_r n} \sum_{i \in X_r} z_i \1(x_i = y_i) $, we have:
\begin{align*}
\cA \sststile{4}{z} \Biggl\{z'(X_r)^2 z'(X_r')^2 \leq  \frac{(32Ct)^{2t}}{(\E_{X_r} Q - \E_{X_{r'}}Q)^{2t}}  & \Biggl(\frac{\alpha^4}{w_r^2 w_{r'}^2}  \Paren{\E_z (Q-\E_z Q)^{2}}^{t} \\
&+ \frac{\alpha^2}{w_r^2} \Paren{\E_{X_{r'}} (Q-\E_{X_{r'}}Q)^{2}}^t + \frac{\alpha^2}{w_{r'}^2} \Paren{\E_{X_r}(Q-\E_{X_r}Q)^{2}}^t\Biggr)\Biggr\}\mper
\end{align*}
\end{lemma}
\begin{proof}
Let $z_i' = z_i \1(x_i = y_i)$ for every $i$. 
Using the substitution rule and non-negativity constraints of the $z_i$'s, we have
\begin{equation}
\label{eqn:lowerbound_expec_Q}
\begin{split}
\cA \sststile{4}{z} \Biggl\{ \E_z (Q-\E_z Q)^{2t}  & = \frac{1}{\alpha^2 n^2} \sum_{i, j \leq n } z_i' z_j' \Paren{Q(x_i - x_j)- \E_z Q}^{2t} \\
& \geq \frac{1}{\alpha^2 n^2} \sum_{i, j \in X_r \text{ or } i,j \in X_{r'}} z_i' z_j' \Paren{Q(x_i - x_j)- \E_z Q}^{2t}\Biggl\}
\end{split}
\end{equation}
Using the SoS almost triangle inequality, we have 

\begin{equation}
\label{eqn:lowerbound_expec_Q_almost_tri}
\begin{split}
\cA \sststile{4}{z} \Biggl\{  & \frac{1}{\alpha^2 n^2} \sum_{i, j \in X_r \text{ or } i,j \in X_{r'}} z_i' z_j' \Paren{Q(x_i - x_j)- \E_z Q}^{2t} \\
& \geq \Paren{ \frac{1}{2^{2t}}} \Paren{ \frac{1}{\alpha^2 n^2} \sum_{i, j \in X_r } z_i' z_j' \Paren{\E_{X_r}Q -\E_z Q}^{2t} - \frac{1}{\alpha^2n^2} \sum_{i, j \in X_r } z_i' z_j' \Paren{Q(x_i - x_j)- \E_{X_r}Q}^{2t} } \\
& \hspace{0.2in} + \Paren{\frac{1}{2^{2t}}}\Paren{ \frac{1}{\alpha^2 n^2} \sum_{i, j \in  X_r } z_i' z_j' \Paren{\E_{X_{r'}}Q- \E_z Q}^{2t}-  \frac{1}{\alpha^2 n^2} \sum_{i, j \in  X_{r'} } z_i' z_j' \Paren{Q(x_i - x_j)- \E_{X_{r'}}Q}^{2t}}\\
&= 2^{-2t} \Paren{ (w_r/\alpha)^2 z(X_r)^2 \Paren{\E_{X_r}Q -\E_z Q}^{2t} - \frac{1}{\alpha^2 n^2} \sum_{i, j \in  X_r } \Paren{Q(x_i - x_j)- \E_{X_r}Q}^{2t}}\\
&\hspace{0.2in}+ 2^{-2t}\Paren{ (w_{r'}/\alpha)^2 z(X_{r'})^2 \Paren{\E_{X_{r'}}Q- \E_z Q}^{2t}- \frac{1}{\alpha^2 n^2} \sum_{i, j \in  X_{r'} } \Paren{Q(x_i - x_j)- \E_{X_{r'}}Q}^{2t} }\Biggl\}
\end{split}
\end{equation} 
Using  Fact~\ref{fact:lower-bounding-sums-SoS}, we can further simplify the above as follows: 

\begin{equation}
\label{eqn:lowerbound_expec_Q}
\begin{split}
\cA \sststile{4}{z} \Biggl\{ \E_z (Q-\E_z Q)^{2t}  
& \geq  2^{-6t} \frac{w_r^2 w_{r'}^2}{\alpha^4} z'(X_{r})^2z'(X_{r'})^2 (\E_{X_r}Q-\E_{X_{r'}}Q)^{2t}  \\& -  2^{-6t}(w_r/\alpha)^2  \E_{X_r} (Q-\E_{X_r}Q)^{2t} - 2^{-6t}
 (w_{r'}/\alpha)^2  \E_{X_{r'}} (Q-\E_{X_{r'}}Q)^{2t} \\
 &\geq 2^{-6t} \frac{w_r^2 w_{r'}^2}{\alpha^4} z'(X_{r})^2z'(X_{r'})^2 (\E_{X_r}Q-\E_{X_{r'}}Q)^{2t} - (w_r/\alpha)^2 (Ct)^{2t}\Paren{\E_{X_r} (Q-\E_{X_r}Q)^{2}}^t\\
&- (w_{r'}/\alpha)^2  (Ct)^{2t}\Paren{\E_{X_{r'}} (Q-\E_{X_{r'}}Q)^{2}}^t \Biggl\}
\end{split}
\end{equation}
where the last inequality follows from the Certifiable Hypercontractivity constraint ($\calA_4$). Rearranging completes the proof.

\end{proof}

We can use the lemma above to obtain a simultaneous intersection bound guarantee when there are relative Frobenius separated components in the mixture. 
\begin{lemma} \label{lem:pair-rel-frob}
Suppose $\Norm{\Sigma^{-1/2}(\Sigma_r-\Sigma_{r'})\Sigma^{-1/2}}_F^2 \geq 10^8 \frac{C^6 t^4}{{\beta}^{2/t}\alpha^2}$. Then, for $z'(X_r) = \frac{1}{\alpha n} \sum_{i \in X_r} z_i \cdot \1(y_i = x_i)$ for every $r$,  
\[
\cA \sststile{2t}{w} \Set{z'(X_r) z'(X_r') \leq \beta}\mper
\]
\end{lemma}
\begin{proof}
We work with the transformation $x_i \rightarrow \Sigma^{-1/2}x_i$. Let $\Sigma_z' = \Sigma^{-1/2}\Sigma_z \Sigma^{-1/2}$, $\Sigma_r' = \Sigma^{-1/2} \Sigma_r\Sigma^{-1/2}$ and $\Sigma_{r'}' =\Sigma^{-1/2}\Sigma_{r'}\Sigma^{-1/2}$ be the transformed covariances. Note that transformation is only for the purpose of the argument - our constraint system does not depend on $\Sigma$.

Notice that $\Norm{\Sigma_{r'}'}_2 \leq \frac{1}{w_r}$ and $\Norm{\Sigma_{r'}'}_2 \leq \frac{1}{w_{r'}}$.

We now apply Lemma~\ref{lem:lower-bound-variance} with $Q = \Sigma_{r}'-\Sigma_{r'}'$. Then, notice that $\E_{X_r}Q - \E_{X_{r'}} Q = \Norm{\Sigma_r' - \Sigma_{r'}'}_F^2 = \Norm{Q}_F^2$. Then, we obtain:

\begin{multline}\label{eq:basic-bound}
\cA \sststile{2t}{z} \Biggl\{z'(X_r)^2 z'(X_{r'})^2 \\\leq \Paren{\frac{32Ct}{\E_{X_r} Q - \E_{X_{r'}}Q}}^{2t}  \Paren{\frac{\alpha^4}{w_r^2 w_{r'}^2} \Paren{\E_z (Q-\E_z Q)^2}^t + \frac{\alpha^2}{w_r^2} \Paren{\E_{X_{r'}} (Q-\E_{X_{r'}}Q)^2}^t + \frac{\alpha^2}{w_{r'}^2} \Paren{\E_{X_r}(Q-\E_{X_r}Q)^2}}^t\Biggr\}\mper
\end{multline}

\nnnew{Since $X_{r}$ and $X_{r'}$ have certifiably $C$-bounded variance polynomials for $C = 4$ (as a consequence of Condition~\ref{cond:convergence-of-moment-tensors} and Fact~\ref{fact:moments-to-analytic-properties} followed by an application of Lemma~\ref{lem:frob-of-product}), we have:}
\[
\cA \sststile{2}{Q} \Set{\E_{X_{r'}} (Q-\E_{X_{r'}}Q)^2 \leq \nnnew{6\Norm{{\Sigma_{r'}'}^{1/2}Q{\Sigma_{r'}'}^{1/2}}_F^2} \leq \frac{6}{w_{r'}^2} \Norm{Q}_F^2} \;,
\]
and
\[
\cA \sststile{2}{Q} \Set{\E_{X_{r}} (Q-\E_{X_{r}}Q)^2 \leq \nnnew{6\Norm{{\Sigma_{r}'}^{1/2}Q{\Sigma_{r}'}^{1/2}}_F^2} \leq \frac{6}{w_{r}^2} \Norm{Q}_F^2} \mper
\]
Finally, using the bounded-variance constraints in $\cA$, we have: 
\[
\cA \sststile{4}{Q,z} \E (Q-\E_z Q)^2 \leq \frac{6}{\alpha^2}\Norm{Q}_F^2\mper
\]

Plugging these estimates back in \eqref{eq:basic-bound} yields:
\begin{equation}\label{eq:final-bound}
\cA \sststile{4}{z} \Set{z'(X_r)^2 z'(X_{r'})^2 \leq \frac{(1000Ct)^{2t}}{\alpha^{2t} \Norm{Q}_F^{2t}}} \mper
\end{equation}

Plugging in the lower bound on $\Norm{Q}_F^{2t}$ and applying Fact~\ref{fact:cancellation-within-sos-constant-rhs} completes the proof. 
\end{proof}



We can use the above lemma to complete the proof of Lemma~\ref{lem:simultaneous-intersection-bounds-poly-clustering}:

\begin{proof}[Proof of Lemma~\ref{lem:simultaneous-intersection-bounds-poly-clustering}]
WLOG, assume that $\Sigma = I$. 
Let $Q = \Sigma_r - \Sigma_{r'}$ and let $\bar{Q}= Q/\Norm{Q}_F$.
Consider the numbers $v_i = \tr(\Sigma_r \cdot Q)$. 
Then, we know that $\max_{i,j} |v_i - v_j| \geq \Norm{Q}_F$. 
Thus, there must exist a partition of $[k] = S \cup T$ such that $|v_i -v_j| \geq \Norm{Q}_F/k$ whenever $i \in S$ and $j \in T$. 
 
Thus, for every $i \in S$  and $j \in T$, $\Norm{\Sigma_i - \Sigma_j}^2_F \geq \Norm{Q}^2_F/k^2 = 10^8\frac{C^6 t^4}{{(\beta^{2/t} \alpha^2)}}$. 
We can now apply Lemma above to every $i \in S, j \in T$, observe that $\cA \sststile{4}{} \Set{z(X_r) z(X_{r'}) \leq z'(X_r) z'(X_{r'}) + 2 \epsilon/\alpha}$, and add up the resulting inequalities to finish the proof.

\end{proof}

\subsection{Special Case: Algorithm for Uniform and Bounded Mixing Weights}
In this subsection, we obtain a polynomial time algorithm when the input mixture has weights that are bounded from below. This includes the case of uniform weights and when the minimum mixing weight is at least some function of $k$. At a high level, our algorithm  partitions the sample into clusters as long as there is a pair of components separated in TV distance and given samples that are not clusterable, runs the tensor decomposition algorithm to list decode. 
We then use standard robust tournament results to pick a hypothesis from the list. 

\begin{theorem}[Robustly Learning Mixtures of Gaussians with Bounded Weights]
\label{thm:robust-GMM-equiweigthed}
Given $0< \epsilon < \bigOk{1}$,
let $Y = \{ y_1, y_2, \ldots , y_n\}$ be a multiset of $n \geq n_0=\poly_k\Paren{d, 1/\eps}$ $\epsilon$-corrupted samples 
from a $k$-mixture of Gaussians $\calM =  \sum_{i \leq k} w_i \calN\Paren{\mu_i, \Sigma_i}$, such that $w_i \geq \alpha$. 
Then, there exists an algorithm with running time $\poly_k(n^{1/\alpha} )\cdot\exp\Paren{ \poly_k(1/\alpha,1/\epsilon)}$ 
such that with probability at least $9/10$ it outputs a hypothesis $k$-mixture of Gaussians 
$\widehat{\calM} = \sum_{i \leq k} \hat{w}_i \calN\Paren{\hat{\mu_i}, \hat{\Sigma}_i}$ 
such that $d_{\textrm{TV}}\Paren{\calM, \widehat{\calM}} = \bigOk{\eps}$.  
\end{theorem}

\input{uniform_weights}

%% file: uniform_weights.tex
Briefly, our algorithm simply does the following: 
\begin{enumerate}
  \item \textbf{Clustering via SoS}: Guess a partition of the mixture such that each component in the partition is not clusterable. Let the resulting partition have $t \leq k$ components. In parallel, try all possible ways to run Algorithm \ref{algo:robust-partial-clustering} repeatedly to obtain a partition of the samples, $\{\tilde{Y}_j\}_{j \in [t]}$ into exactly $t$ components. For each such partition repeat the following.

  \item \textbf{Robust Isotropic Transformation}: Run the algorithm corresponding to Lemma \ref{lem:sampling_from_nearl_isotropic} on each set $\tilde{Y}_j$ to make the sample approximately isotropic. 
  Grid search for weights over $[\alpha, 1/k]^k$ with precision $\alpha$.
  \item \textbf{List-Decoding via Tensor Decomposition}: Run Algorithm \ref{algo:list-recovery-tensor-decomposition} on each $\tilde{Y}_j$. Concatenate the lists to obtain $\calL$.
   \item \textbf{Robust Tournament}: Run the tournament from Fact~\ref{lem:tournament} over all the hypotheses in $\calL$, and output the winning hypothesis. 
\end{enumerate}

\begin{proof}[Proof Sketch]
 Setting $\Delta = (k^{k^{O(k)}})$, it follows from Theorem \ref{thm:partial-clustering-non-poly} that we obtain a partition of $Y$ into $\{\tilde{Y}_j\}_{j \in [t]}$, for some $t\in[k]$ such that $\tilde{Y}_j$ has at most $\bigO{k\epsilon/\alpha}$ outliers, $\Paren{1-\bigO{k\epsilon/\alpha}}$-fraction of samples from at least one component of the input mixture and the resulting samples are not $\Delta$-separated (see Definition \ref{def:separated-mixture-model}). 
 It then follows from Lemma \ref{lem:sampling_from_nearl_isotropic} that the mean $\mu_j$ and covariance $\Sigma_j$ of $\tilde{Y}_j$ satisfy : a) $\Norm{\mu_j}_2 \leq \bigO{ \sqrt{\epsilon}k^{1.5}/\alpha^{1.5} } $, 
 b) $\Paren{1- \sqrt{\epsilon}k^{1.5}/\alpha^{1.5} } I \preceq \Sigma_j\preceq \Paren{1- \sqrt{\epsilon}k^{1.5}/\alpha^{1.5}  } I$, 
 and c) $\Norm{\Sigma_j - I }_F \leq \bigO{\sqrt{\epsilon}k^{1.5}/\alpha^{1.5}}$. 

 Each component, $\tilde{Y}_j$, of the partition  can have at most $k$ components. Assuming these correspond to $\{w^{(j)}_i,\mu^{(j)}_i \Sigma^{(j)}_i \}_{i \in [k]}$, observe, $\sum_{i \in[k]} w^{(j)}_i \Sigma^{(j)}_i + w^{(j)}_i \mu^{(j)}_i \Paren{\mu^{(j)}_i}^\top \preceq \Paren{1+ \sqrt{\epsilon}k^{1.5}/\alpha^{1.5} } I$. Thus, we have that 
 $\Norm{\mu^{(j)}_i}^2_2 \leq \Paren{1 +\sqrt{\epsilon}k^{1.5} }/\alpha^{2.5}$ 
 and combined with not being $\Delta$-separated, it follows that for all $i' \in[k]$,
\begin{equation*}
\begin{split}
\Norm{\Sigma^{(j)}_{i'} - I}_F  = \Norm{\Sigma^{(j)}_{i'} - \Sigma_j + \Paren{\Sigma_j - I}  }_F 
& \leq  \Norm{ \Sigma^{(j)}_{i'}  - \sum_{i \in[k]} w^{(j)}_i \Sigma^{(j)}_i  +  \sum_{i\in[k]} w^{(j)}_i \mu^{(j)}_i \Paren{\mu^{(j)}_i}^\top }_F + \Norm{\Sigma_j - I }_F \\
& \leq \Norm{   \sum_{i \in[k]} w^{(j)}_i \Paren{ \Sigma^{(j)}_{i'}  - \Sigma^{(j)}_i  }  }_F + \bigO{k^{1.5}/\alpha^{2.5}} \\
& \leq  \bigO{\Delta/\alpha} \;.
\end{split}
\end{equation*}
There are at most $\bigO{k^k}$ ways in which we can partition the set of input points such that each resulting component is not partially clusterable. We run the algorithm in parallel for each one. 
Then, for the correct iteration, we apply Theorem \ref{thm:list-recovery-by-tensor-decomposition} to get a list $\calL$ of size $\exp( \poly_k(1/\alpha, 1/\epsilon) )$ such that it contains a hypothesis $\{\hat{w}^{(j)}_i, \hat{\mu}^{(j)}_i, \hat{\Sigma}^{(j)}_i \}_{i\in [k]}$ such that $\abs{\hat{w}^{(j)}_i- w^{(j)}_i } \leq \alpha$, $\Norm{\hat{\mu}^{(j)}_i - \mu^{(j)}_i}_2 \leq \bigOk{\epsilon}$ and  $\Norm{\hat{\Sigma}^{(j)}_i - \Sigma^{(j)}_i }_F \leq \bigOk{\epsilon}$. Since $\Paren{1- 1/\Delta } I \preceq \Sigma^{(j)}_i $, it then follows from Lemma \ref{lem:frobenius_to_tv} that the hypothesis is $\bigOk{\epsilon}$-close to the input in total variation distance.

 Algorithm~\ref{algo:robust-partial-clustering} is called at most $\bigO{k^k}$ times, and along with the robust isotropic transformation, this requires $\poly_k\Paren{n^{1/\alpha},1/\epsilon}$. The grid search contributes a multiplicative factor of $(1/\alpha)^k$. The tensor decomposition algorithm and robust hypothesis section  $\poly_k(n^{1/\alpha} )\cdot\exp\Paren{ \poly_k(1/\alpha,1/\epsilon)}$.
\end{proof}

%

%



%

%% file: recursion-lemma.tex

\section{Spectral Separation of Thin Components}
\label{sec:spectral-separation}

\new{In this section, we show how to efficiently separate a thin
component, if such a component exists, given sufficiently accurate
approximations to the component means and covariances.
This is an important step in our overall algorithm and is
required to obtain total variation distance guarantees.

Specifically, the main algorithmic result of this section is described
in the following lemma:
}

\begin{lemma}\label{cor:recursion}
There is a polynomial-time algorithm with the following properties:
Let $\calM=\sum_{i=1}^k w_i G_i$ with $G_i = \cN(\mu_i,\Sigma_i)$ be a $k$-mixture of Gaussians on $\R^d$,
and let $X$ be a set of points in $\R^d$ satisfying Condition \ref{cond:convergence-of-moment-tensors}
with respect to $\calM$ for some parameters $(\gamma,t)$. The algorithm takes input parameters
$\eta, \delta$, satisfying $0< \delta < \eta < 1/(100k)$, and $Y$, an $\eps$-corrupted version of $X$,
as well as candidate parameters $\{\hat{\mu}_i, \hat{\Sigma}_i\}_{i \leq k}$. Then as long as
\begin{enumerate}
\item $\cov(\calM) \succeq I/2$,
\item $\|\mu_i-\hat{\mu}_i\|_2< \delta$ and $\|\Sigma_i - \hat{\Sigma}_i \|_F < \delta$, for all $i \in [k]$, and
\item there exists an $s \in [k]$ such that $\Sigma_s$ has an eigenvalue $< \eta$,
\end{enumerate}
the algorithm outputs a partition of $Y$ into $Y_1 \cup Y_2$
\new{such that there is a non-trivial partition of $[k]$ into $Q_1\cup Q_2$,
so that letting $\calM_j$, $j \in \{1, 2\}$, be proportional to
$\sum_{i\in Q_j}w_i G_i$ and $W_j =\sum_{i\in Q_j} w_i$,
then $Y_j$ is an $((O(k^\newblue{2}\gamma)+\tilde O(\eta^{1/2k}))/\newblue{W}_j)$-corruption of a set satisfying Condition \ref{cond:convergence-of-moment-tensors} with respect to $\calM_j$
with parameters $(O(k\gamma/\newblue{W}_j),t)$. }
\end{lemma}

The key component in the proof of Lemma~\ref{cor:recursion} is the following lemma:

\begin{lemma} \label{lem:recursion}
Let $\calM = \sum_{i=1}^k w_i G_i$ with $G_i = \cN(\mu_i,\Sigma_i)$
be a $k$-mixture of Gaussians in $\R^d$ with $\cov(\calM) \succeq I/2$. Suppose that, for some
$0< \delta < 1/(100k)$, we are given $\hat{\mu}_i$ and $\hat{\Sigma}_i$ satisfying
$\|\mu_i-\hat{\mu}_i\|_2< \delta$ and $\|\Sigma_i - \hat{\Sigma}_i \|_F < \delta$, for all $i \in [k]$.
Suppose furthermore that for some $\eta>\delta$, there is a $\Sigma_{s}$, $s \in [k]$,
with an eigenvalue less than $\eta$. \new{There exists a computationally efficient algorithm that takes
inputs $\eta$, $\delta$, $\hat{\mu}_i$, $\hat{\Sigma}_i$, and computes a function $F: \R^d \to \{0,1\}$ such that:}
\begin{enumerate}
\item For each $i \in [k]$, $F(G_i)$ returns the same value in $\{0,1\}$ with probability
at least $1-\tilde{O}_k(\eta^{1/(2k)})$. We define the most likely value of $F(G_i)$ to be this value.
\item There exist $i, j \in [k]$ such that the most likely values of $F(G_i)$ and
$F(G_j)$ are different.
\end{enumerate}
\new{Furthermore, $F(x)$ can be chosen to be of the form $f(v\cdot x)$,  for some $v\in \R^d$,
and $f:\R \rightarrow \{0,1\}$ is an $O(k)$-piecewise constant function.}
\end{lemma}

\noindent Given Lemma~\ref{lem:recursion}, it is easy to finish the proof of Lemma~\ref{cor:recursion}.

\begin{proof}[Proof of Lemma~\ref{cor:recursion}]

We simply take the candidate parameters, obtain $F$ from Lemma~\ref{lem:recursion},
and partition $Y = Y_1 \cup Y_2$, so that $F$ is constant on both $Y_1$ and $Y_2$. \new{We let $Q_j$ be the set of $i$ so that $F(G_i)$ returns the value $j-1$ with large probability. Letting the partition of $X$ for Condition \ref{cond:convergence-of-moment-tensors} be $X=X_1\cup \ldots \cup X_k$, we let $X^j = \bigcup_{i\in Q_j} X_i$. Lemma \ref{submixture condition lemma} shows that the $X^j$ satisfy the appropriate conditions for $\calM_j$.
It remains to prove that $Y_j$ equals $X^j$ with a sufficiently small rate of corruptions.}
The fraction of points misclassified by $F$ equals $\epsilon$ (the fraction of outliers in the sample $Y$)
plus the misclassification error of $F$. \new{We note that given the form of $F$ and the fact that
the uncorrupted samples in $Y$ satisfy Condition \ref{cond:convergence-of-moment-tensors}, the fraction of misclassified samples from each component $i$ is at most the probability that a random sample from $G_i$ gets misclassified (at most $\tilde{O}_k(\eta^{1/(2k)})$ by Lemma \ref{lem:recursion}) plus $O(k\gamma)$. Summing this over components, gives Lemma \ref{cor:recursion}.}
\end{proof}

Let us now describe the algorithm to prove Lemma~\ref{lem:recursion} (and evaluate $F$),
which is given in pseudocode below (Algorithm~\ref{algo:recursive-clustering}).



\begin{mdframed}
  \begin{algorithm}[Algorithm for Spectrally Separating Thin Components]
    \label{algo:recursive-clustering}\mbox{}
    \begin{description}
    \item[Input:] 
    Estimated parameters $\Set{\hat\mu_i,\hat\Sigma_i}_{i \leq k}$, \new{parameters $\eta, \delta$}.

    \item[Output:] A function $F: \R^d \rightarrow \{0,1\}$. 

    \item[Operation:]\mbox{}
       \begin{enumerate}
       \item Find a unit-norm direction $v$ such that there exists $s\in[k]$, $v^T \hat{\Sigma}_{s} v < 2\eta$.
       \item Compute $(v^T \hat{\Sigma}_i v)$ for all $i \in [k]$.
       \begin{enumerate}
       \item If there exists $j \in [k]$ such that $(v^T \hat{\Sigma}_j v) > \sqrt{\eta}$,
       find a $t$ such that $\sqrt{\eta} > t > 2\eta$ and there is no $j \in [k]$ with
       $t < v^T \hat{\Sigma}_j v < t \, \Omega(\eta^{-1/(2k)})$. Set $F(x) = 1$ if there is an $i$ such that $|v \cdot (x-\hat{\mu}_i) | < \sqrt{t} \log(1/\eta)$ and $0$ otherwise.
       \item Otherwise, compute $v \cdot \hat{\mu}_i$ for all $i \in [k]$. Find a $t$
       between the minimum and the maximum of $v \cdot \hat{\mu}_i$ such that
       there is no $v \cdot \hat{\mu}_{i}$ within $1/(20k)$ of $t$.
       Set $F(x) = 1$ if $v\cdot x >t$ and $0$ otherwise.
       \end{enumerate}
       \end{enumerate}
    \end{description}
  \end{algorithm}
\end{mdframed}

\medskip

\begin{proof}[Proof of Lemma~\ref{lem:recursion}]

Let $v$ be a unit vector and $s \in [k]$ such that $v^T \hat{\Sigma}_{s} v < 2\eta$.
By assumption, we have that $\var{v \cdot \calM} \geq 1/2$.
Furthermore,
$$\var{v \cdot \calM} = \sum_i w_i (v^T \Sigma_i v) + \sum_i w_i (v \cdot (\mu_i-\mu))^2
\leq \sum_i w_i (v^T \Sigma_i v) + \sum w_i (v \cdot (\mu_i-\mu_{s}))^2 \;,$$
where $\mu$ is the mean of $\calM$. This means that either
there exists  $j \in [k]$ such that $(v^T \Sigma_j v) > 1/4$, or there exists
$j \in [k]$ such that $|v \cdot (\mu_j-\mu_s)| > 1/4$.
Since we have approximations of these quantities to order $\delta$,
we have that there is $j \in [k]$ such that $(v^T \hat{\Sigma}_j v) > 1/10$
or that there is $j \in [k]$ with $|v \cdot (\hat{\mu}_j-\hat{\mu}_s)| > 1/10$.

We first consider the case that there is a $j \in [k]$ such that $(v^T \hat{\Sigma}_j v) > \sqrt{\eta}$.
Since there is a $j \in [k]$ with $(v^T \hat{\Sigma}_j v) > \sqrt{\eta}$ and
another $s \in [k]$ with $(v^T \hat{\Sigma}_s v) < 2\eta$,
there must be some $\sqrt{\eta} > t > 2\eta$ such that there is no $j \in [k]$
with $t < v^T \hat{\Sigma}_j v < t \, \Omega(\eta^{-1/(2k)})$.
Otherwise, there must be at least one $\hat\Sigma_i$ in each $2\eta\le\Omega(\eta^{-1/(2k)})^i\le\sqrt{\eta}$,
where we need more than $k$ components.

For a given $x$, we define $F(x)$ to be $1$ if there exists $i$ such that
$|v \cdot (x-\hat{\mu}_i) | < \sqrt{t} \log(1/\eta)$, and $F(x) = 0$ otherwise.

To show that this works, we note that for all $i \in [k]$,
if $v^T \hat{\Sigma}_i v \leq t$, then $\var{v \cdot G_i} \le t+\delta$,
and since $|v \cdot (\mu_i-\hat{\mu}_i)|<\delta$,
by the Gaussian tail bound, we have that
$$
\Pr_{x\sim G_i}\left(|x-\mu_i|\ge(\sqrt{t} \log(1/\eta)-\delta)\right) \leq
\exp\left(-\frac{(\sqrt{t} \log(1/\eta)-\delta)^2}{2(t+\delta)}\right)=O(\eta) \;.
$$
Thus, all but an $O(\eta)$-fraction of the samples of $G_i$ have $F(x)=1$.

On the other hand, for components $i$ with
$v^T \hat{\Sigma}_i v \gg t \eta^{-1/(2k)}$, we have that
$\var{v \cdot G_i} \gg t \eta^{-1/(2k)}$.
Then, the density of $G_i$ is at most $1/\sqrt{2\pi t \eta^{-1/(2k)}}$.
So, the probability that a sample from $v \cdot G_i$ lies in any interval of length $2\sqrt{t} \log(1/\eta)$ is at most
$$
\frac{1}{\sqrt{2\pi t \eta^{-1/(2k)}}}2\sqrt{t} \log(1/\eta)=\tilde{O}(\eta^{1/(4k)}) \;.
$$
Since there are $k$ such intervals, the probability that $F(x)$ is $1$ when $x$ is drawn from $G_i$
is at most $\tilde{O}_k (\eta^{1/(4k)})$. This completes our proof of point (1), and point (2)
follows from the fact that we know of component $G_j$ in one class and $G_s$ in the other class.

We next consider the case where $(v^T \hat{\Sigma}_j v) \leq  \sqrt{\eta}$ for all
$j \in [k]$, and where $|v \cdot (\hat{\mu}_j - \hat{\mu}_s)| > 1/10$ for some $j \in [k]$.
Then we can find some $t$ between $v \cdot \hat{\mu}_j$ and $v \cdot \hat{\mu}_s$ such
that no $v \cdot \hat{\mu}_{i}$ is within $1/(20k)$ of $t$.
In this case, we define $F(x)$ be $1$ if $v \cdot x > t$ and $0$ otherwise.
To show part (1), first consider $i \in [k]$ such that $v \cdot \hat{\mu}_i  < t-1/(20k)$.
Then we have that $v \cdot \mu_i < t-1/(30k)$. Furthermore,
$\var{v \cdot G_j} \le \delta+\sqrt{\eta}$. Therefore, the probability that $v \cdot G_i > t$
is at most $\exp(-\Omega_k((\delta+\sqrt{\eta})^{-2}))$, which is sufficient.

A similar argument holds in the other direction for $i \in [k]$ such that $v \cdot \hat{\mu}_i > t+1/(20k)$,
and statement (2) holds because we know that there are both kinds of components.
This completes the proof.
\end{proof}




%% file: full-algo-analysis.tex

\section{Robust Proper Learning: Proof of Theorem~\ref{thm:main-informal}}
\label{sec:full-algo-analysis}

In this section, we show how to combine the partial clustering, tensor decomposition,  and recursive clustering algorithms to establish
our main result. The main theorem we prove is as follows: 

\begin{theorem}[Robustly Learning $k$-Mixtures of Arbitrary Gaussians] \label{thm:robust-GMM-arbitrary}
Given $0< \epsilon < 1/k^{k^{O(k^2)}}$ and a multiset $Y = \{ y_1, y_2, \ldots , y_n\}$ of $n$ i.i.d. 
samples from a distribution $F$ such that $\dtv(F, \calM) \leq \eps$,
for an unknown $k$-mixture of Gaussians $\calM = \sum_{i \leq k} w_i\calN\Paren{\mu_i, \Sigma_i}$, 
where \new{$n \geq n_0=d^{O(k)}/\poly(\eps)$,} 
Algorithm \ref{algo:robust-GMM-arbitrary} runs 
\new{in time $n^{O(1)}\exp\left(O(k)/\eps^2\right)$}
and with probability at least $0.99$ outputs a hypothesis $k$-mixture of Gaussians 
$\widehat{\calM} = \sum_{i \leq k} \hat{w}_i \calN\Paren{\hat{\mu_i}, \hat{\Sigma}_i}$ 
such that $d_{\textrm{TV}}\Paren{\calM, \widehat{\calM}} = \bigO{\eps^{c_k}}$, 
with $c_k=1/(100^k C^{(k+1)!}k!\mathrm{sf}(k+1))$, where $C>0$ is a universal constant 
and $\mathrm{sf}(k)=\Pi_{i \in[k]}(k-i)!$ is the super-factorial function.  
\end{theorem} 

As an immediate corollary, we obtain the following:

\begin{corollary}[Robustly Learning $k$-Mixtures of Gaussians in Polynomial Time] \label{thm:robust-GMM-arbitrary-poly-time}
Given $0< \epsilon < 1/\exp\left(k^{k^{O(k^2)}}\right)$,  and a 
multiset $Y = \{ y_1, y_2, \ldots , y_n\}$ of $n$ i.i.d. samples from a distribution $F$ 
such that $\dtv(F, \calM) \leq \eps$, for an unknown $k$-mixture of Gaussians 
$\calM = \sum_{i \leq k} w_i\calN\Paren{\mu_i, \Sigma_i}$, 
where $n \geq n_0=d^{O(k)}\log^{O(1)}(1/\epsilon)$, 
there exists an algorithm that runs in time $\poly_k(n,1/\eps)$ \Anote{is this accurate?}
and with probability at least $0.99$ outputs a $k$-mixture of Gaussians 
$\widehat{\calM} = \sum_{i \leq k} \hat{w}_i \calN\Paren{\hat{\mu_i}, \hat{\Sigma}_i}$ 
such that $\dtv \Paren{\calM, \widehat{\calM}} = \bigO{ \left(1/\log(1/\epsilon)\right)^{1/\left(k^{O(k^2)}\right) }}$.
\end{corollary} 

\noindent The corollary follows by running Algorithm \ref{algo:robust-GMM-arbitrary} with $\epsilon \gets \sqrt{1/\log(1/\epsilon)}$ 
and applying Theorem \ref{thm:robust-GMM-arbitrary}.

\medskip

The algorithm establishing Theorem~\ref{thm:robust-GMM-arbitrary} is given in pseudocode below. Algorithm~\ref{algo:cluster_or_decode} takes as input a corrupted sample from a $k$-mixture of Gaussians 
and outputs a set of $k$ mixing weights, means, and covariances, 
such that the resulting mixture is close to the input mixture in total variation distance 
with non-negligible probability.
Algorithm~\ref{algo:robust-GMM-arbitrary} simply runs Algorithm~\ref{algo:cluster_or_decode} many 
times to create a small list of candidate hypotheses (consisting of mixing weights, means, and covariances), and finally runs a robust tournament to outputs a winner. This boosts the probability of success to at least $0.99$.

\medskip

\input{full_algorithm}


\subsection{Analysis of Algorithm~\ref{algo:robust-GMM-arbitrary}}

To prove Theorem~\ref{thm:robust-GMM-arbitrary}, we will require the following intermediate results.
We defer some proofs in this subsection to Appendix~\ref{sec:omitted_proofs}.




We use the following lemma to relate the Frobenius distance of covariances to the total variation distance between two Gaussians, 
when the eigenvalues of the covariances are bounded below. 

\begin{lemma}[Frobenius Distance to TV Distance]\label{lem:frobenius_to_tv}
Suppose $\cN(\mu_1,\Sigma_1),\cN(\mu_2,\Sigma_2)$ are Gaussians with $\norm{\mu_1-\mu_2}_2 \le\delta$ and 
$\norm{\Sigma_1-\Sigma_2}_F\le\delta$. If the eigenvalues of $\Sigma_1$ and $\Sigma_2$ are at least $\lambda>0$, then 
$\dtv(\cN(\mu_1,\Sigma_1),\cN(\mu_2,\Sigma_2))=O(\delta/\lambda)$.
\end{lemma}



\input{helper_lemmas}


\subsection{Proof of the Main Theorem}\label{sec:proof-main-thm}
We are now ready to complete the proof of Theorem~\ref{thm:robust-GMM-arbitrary}.

\begin{figure}[h!]
\includegraphics[scale=0.8]{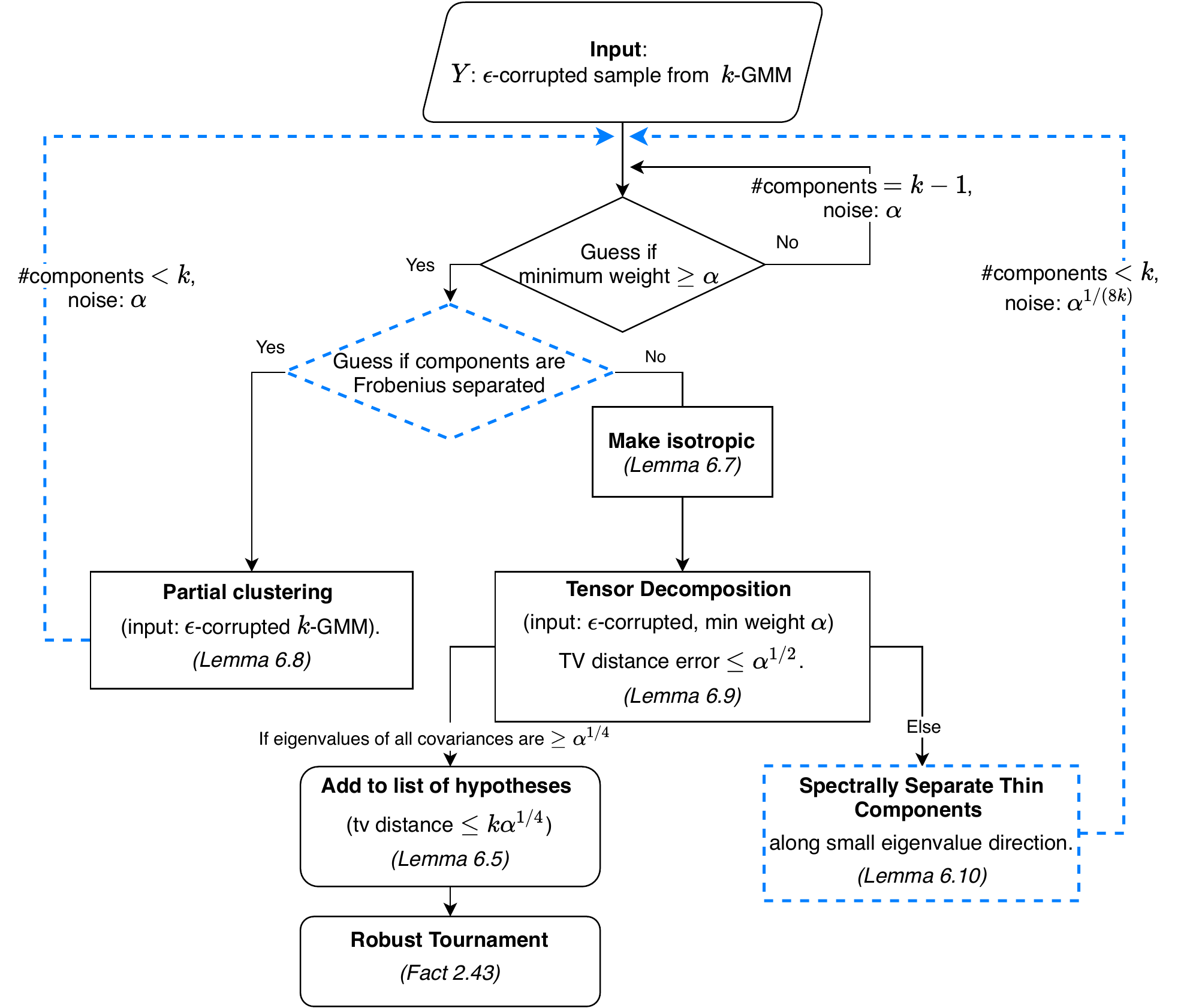}
\centering
\caption{If we assume a $1/\poly(k)$ lower bound on minimum weight, 
then we can skip all blue steps above; the partial clustering is carried out till it can no longer be done within a cluster and then followed by the tensor decomposition step.}
\end{figure}

\begin{proof}[Proof of Theorem \ref{thm:robust-GMM-arbitrary}]
\new{
We divide the proof into two parts: first we show that Algorithm~\ref{algo:cluster_or_decode} outputs a hypothesis $\widehat\calM=\sum_{i\in[k]}\hat w_i\calN(\hat\mu_i,\hat\Sigma_i)$ such that $\widehat\calM$ and $\calM$ are $\bigO{\epsilon^{c_k}}$-close in total variation distance with probability at least $\exp\left(-O(k)/\eps^2\right)$; then we show that Algorithm~\ref{algo:robust-GMM-arbitrary} outputs a $k$-mixture of Gaussians $\widehat \calM$ such that $\widehat \calM$ and $\calM$ are $\bigOk{\epsilon^{c_k}}$-close in total variation distance with probability 0.99.}

We proceed the first part by induction on $k$. 
Let $c_k = \frac{1}{(100)^k C^{(k+1)!}\mathrm{sf}(k+1)k!}$ be a scalar that only depends on $k$, 
where $C>0$ is a sufficiently large universal constant.

\textbf{Induction Hypothesis}: 
\new{
Let $X$ be a set of points satisfying Condition \ref{cond:convergence-of-moment-tensors} with respect to a $k$-mixture of Gaussians $\calM$ for some parameters $\gamma=\eps d^{-8k}k^{-C'k}$, 
where $C'$ is a sufficiently large constant and $t=8k+48$. Given a set $Y$ being an $\epsilon$-corrupted version of $X$ of size $n$, the outlier parameter $\eps$ and the component-number parameter $k$,
Algorithm~\ref{algo:robust-GMM-arbitrary} returns a $k$-mixture of Gaussians $\widehat{\calM}$ 
such that $\widehat{\calM}$ and $\calM$ are $\bigOk{\epsilon^{c_k}}$-close in total variation distance with probability $\exp\left(-(3k-2)/\eps^2\right)$.
}

\textbf{Base Case}: 
For $k=1$, the algorithm returns the single Gaussian with mean $\hat\mu$ and $\hat\Sigma$ at Step 2. Suppose the true Gaussian is $N(\mu,\Sigma)$.
It follows from the proof of Lemma~\ref{lem:sampling_from_nearl_isotropic}, 
\begin{equation*}
\Norm{  \Sigma^{\dagger/2} \Paren{\hat{\mu} - \mu}}_2=  \Norm{\Sigma^{\dagger/2} \Paren{\hat{\mu} - \mu}}_2 \leq \bigO{\sqrt{\epsilon}}
\end{equation*}
and 
\begin{equation*}
\Norm{\Sigma^{\dagger/2} \Paren{ \hat{\Sigma} - \Sigma }\Sigma^{\dagger/2}  }_F \leq \bigO{\sqrt{\epsilon}} \;,
\end{equation*}
and thus it follows from Fact \ref{GaussianTVFact} that the total variation distance 
between the hypothesis Gaussian and  the true Gaussian is at most $\bigO{\sqrt{\epsilon}}$. 
We can then conclude that the base case is true.

\textbf{Inductive Step}: We assume that our induction hypothesis holds for any $m<k$ and then prove that the induction hypothesis holds for $k$.
\paragraph{Small Clusters Can be Treated as Noise.}
\new{
Conditioning on the base case being true, we begin by guessing 
whether the minimum weight is less than $\epsilon^{1/(10C^{k+1}(k+1)!)}$ with equal probability.

Let $w_{\min}=\min_i w_i$. If $w_{\min}\le\epsilon^{1/(10C^{k+1}(k+1)!)}$, our algorithm takes step 1 with probability 0.5. In this case,
we treat the smallest component as noise and recurse on the set of samples $Y$. We set the number of components to be $k-1$ 
and the fraction of outliers being $\epsilon +\epsilon^{1/(10C^{k+1}(k+1)!)}\le 2\epsilon^{1/(10C^{k+1}(k+1)!)}$. By Lemma~\ref{submixture condition lemma}, $Y$ is an $2\epsilon^{1/(10C^{k+1}(k+1)!)}$-corrupted version of a set satisfying Condition~\ref{cond:convergence-of-moment-tensors} with respect to a $(k-1)$-mixture for parameters $\gamma=\bigO{k\eps d^{-8k}k^{-C'k}/(1- w_{\min})}\le\eps d^{-8(k-1)}(k-1)^{-C'(k-1)}$ and $t=8k+48$.
Thus applying the inductive hypothesis to $Y$, we learn the mixture up to total variation distance $\bigOk{ \Paren{2\epsilon^{1/(10C^{k+1}(k+1)!)}}^{c_{k-1}}} \le\bigOk{\epsilon^{c_k}}$ with probability 
$0.5\exp\left(-(3(k-1)-2)/\eps^2\right)\ge\exp\left(-(3k-2)/\eps^2\right)$.  Now we may assume for all $i\in[k]$, $ w_i \geq \epsilon^{1/(10C^{k+1}(k+1)!)}$.
}
\paragraph{Mixture is Covariance Separated.}
Let $\alpha=\epsilon^{1/(10C^{k+1}(k+1)!)}$ and 
$\psi_1 = \{\exists \hspace{0.05in} \calN(\mu_i, \Sigma_i), \calN(\mu_j, \Sigma_j) \mid \Norm{\Sigma_i - \Sigma_j }_F > \alpha^{-1/2} \}$ 
be the event that the samples were drawn from a mixture that is \textit{covariance separated}. 
First, consider the case where $\psi_1$ is true. We will run 3(a) with probability 0.5.
Then it follows from Lemma \ref{lem:covariance_sep} that $Y$ 
can be partitioned into $Y_1$ and $Y_2$ in time $d^{O(1)}$, such that they both have at least one component 
and the fraction of outliers in each set $Y_1,Y_2$ is at most $\epsilon^{1/(10C^{k+1}(k+1)!)}$ with probability $\alpha^{O(k\log(k/\alpha))}$. 
\new{
Then, we can guess the number of components in $Y_1$ and we will be correct with probability $1/k$. Conditioned on our guess being correct, let $Y_1$ consist of $k_1$ components and $Y_2$ consist of $k_2$ components and $k_1+k_2=k$.

Let $Q_1\cup Q_2$ be the non-trivial partition of $[k]$ in Lemma~\ref{lem:covariance_sep}, $\calM_j$ be a distribution proportional to $\sum_{i\in Q_j}w_iG_i$ and $W_j=\sum_{i\in Q_j}w_i$, then By Lemma~\ref{submixture condition lemma}, 
$Y_j$ is an $\bigO{\epsilon^{1/(10C^{k+1}(k+1)!)}}$-corrupted version of $\bigcup_{i\in Q_j} X_i$ satisfying Condition~\ref{cond:convergence-of-moment-tensors} with respect to $\calM$ with parameters $\gamma=\bigO{k\eps d^{-8k}k^{-C'k}/\alpha}\le\eps d^{-8k_j}(k_j)^{-C'k_j}$ and $t=8k+48$.}
Then, applying the inductive hypothesis on $Y_j$ for $j=1,2$, with number of components $k_j$, 
we can learn the mixtures $\calM_j$ up to total variation distance error 
$\bigOk{\epsilon^{ c_{k_j}/(10C^{k+1}(k+1)!) }}$ 
\new{
with probability $\exp\left(-(3k_j-2)/\eps^2\right)$. Finally if this is the case, we combine the two hypotheses on $Y_1,Y_2$ by multiplying each weight in the hypothesis of $Y_j$ by $|Y_j|/|Y|$ and then taking union of two hypotheses.
Then our combining method gives a final output that learns our full hypothesis to total variation distance error
$\bigOk{\epsilon^{ c_{k_1}/(10C^{k+1}(k+1)!) }}+\bigOk{\epsilon^{ c_{k_2}/(10C^{k+1}(k+1)!) }}\le\bigOk{\epsilon^{c_k}}$ with probability at least 
$0.5\cdot0.5\cdot\frac{1}{k}\cdot\alpha^{O(k\log(k/\alpha))} 
\exp\left(-(3k_1-2)/\eps^2\right)
\exp\left(-(3k_2-2)/\eps^2\right)
\ge\exp\left(-(3k-2)/\eps^2\right)$.
}

\paragraph{Mixture is not Covariance Separated.}
Next, consider the case where $\psi_1$ is false. 
\nnnew{ With probability 0.5, the algorithm guesses correctly and executes Step 2}. Since the mixture is not covariance separated, we satisfy the preconditions of  Lemma~\ref{lem:sampling_from_nearl_isotropic}, and after applying the transformation in Step 2, $\Sigma$, the covariance of the mixture $\calM$, 
is $\sqrt{\eps}k/\alpha$-close to the $r$-dimensional identity, 
where $r$ is the rank of $\Sigma$. However, since we obtain the subspace exactly, 
we can simply project all samples on the subspace and we drop the $r$ in the subsequent exposition.

Let $X'$ be the set of points obtained by applying the Affine transformation from Step 2 as defined in Lemma \ref{lem:sampling_from_nearl_isotropic}.  Then, $X'$  satisfies Condition~\ref{cond:convergence-of-moment-tensors} with respect to a nearly isotropic mixture and parameters $\gamma=\eps d^{-8k}k^{-C'k}$ and $t=8k+48$ so that we can continue the algorithm with $X'$. Whenever we return a hypothesis in the following steps, we will first apply the inverse of the transformation on our estimates $\hat\mu_i$ and $\hat\Sigma_i$. Since total variation distance is affine invariant, we have the same error guarantee in total variation distance after applying the transformation.
From now on, we reduce to the case where $\Sigma$ is $\sqrt{\eps}k/\alpha$-close to the Identity.

There is a 50\% chance our algorithm runs Step 3(b) and we will analyze the remainder of this case under that assumption.
\new{
It follows from Lemma \ref{lem:list_decoding} that we obtain a hypothesis $\{\hat{\mu}_i, \hat{\Sigma}_i \}_{i \in [k]}$ such that 
$\Norm{\mu_i - \hat{\mu}_i}_2 \leq\bigO{\epsilon^{1/(20C^{k+1}(k+1)!)}}$ and 
$\Norm{\Sigma_i - \hat{\Sigma}_i }_F \leq \bigO{\epsilon^{1/(20C^{k+1}(k+1)!)}}$ with probability $\exp\left(-1/\eps^2\right)$. 
}
Conditioned on the hypothesis being correct, we now split into two cases: either all eigenvalues of all the estimated covariances are large 
(in which case we obtain total variation distance guarantees), 
or there is a direction along which we can project and cluster further. 



\paragraph{Covariance Estimates have Large Eigenvalues.}
For the hypothesis $\{\hat{\mu}_i, \hat{\Sigma}_i \}_{i \in [k]}$ from the last step, we compute all the eigenvalues of the estimated covariance matrices, 
$\hat{\Sigma}_i$, for all $i \in [k]$.    
If, for all $i \in [k]$, $\lambda_{\min}\Paren{\hat{\Sigma}_i} \geq c \epsilon^{1/(40C^{k+1}(k+1)!)}$, 
we land in Step 3(b).i that we guess the mixing weights $\hat w_i$ uniformly in the range $[0,1]$ and then we output the corresponding hypothesis $ \{ \hat{w}_i, \hat{\mu}_i, \hat{\Sigma}_i \}_{i \in [k]}$.
With probability at least $\eps^k$, $\hat w_i$ are within $\eps$ of the true mixing weights. 
Under this condition, by Lemma~\ref{lem:frobenius_to_tv}, the mixture 
$\widehat{\calM}= \sum_{i\in[k]} \hat{w}_i\calN\Paren{\hat{\mu}_i, \hat{\Sigma}_i}$ is 
$\bigOk{\epsilon^{1/(40C^{k+1}(k+1)!)}}\le\bigOk{\epsilon^{c_k}}$-close to $\calM$ in 
total variation distance 
\new{
with probability $0.5\cdot0.5\cdot\eps^k\cdot\exp\left(-1/\eps^2\right)\ge\exp\left(-(3k-2)/\eps^2\right)$. 
}

\paragraph{One Covariance Has a Small Eigenvalue.}
Consider the case (Step 3(b).ii) where there exists a unit-norm direction $v$ 
and an estimate $\hat{\Sigma}_i$ such that $v^\top \hat{\Sigma}_i v \leq c \epsilon^{1/(40C^{k+1}(k+1)!)}$. 
It then follows from Lemma \ref{lem:hyperplane_clustering} that we can partition $Y$ into $Y_1$ and $Y_2$ 
such that each has at least one cluster and the total number of outliers in both $Y_1$ and $Y_2$
is at most $\bigO{\epsilon^{1/(80kC^{k+1}(k+1)!)}}n$.
If $Y_1$ or $Y_2$ has size less than $\epsilon^{1/(400kC^{k+1}(k+1)!)}n$, 
then we can treat it as noise and get an additive $O(\epsilon^{1/(400kC^{k+1}(k+1)!)})$-error in total variation distance. 
Otherwise, the fraction of outliers in both sets is at most 
$\bigO{(\epsilon^{1/(80kC^{k+1}(k+1)!)}n)/(\epsilon^{1/(400kC^{k+1}(k+1)!)}n)}=\bigO{\epsilon^{1/(100kC^{k+1}(k+1)!)}}$.
\new{
We then guess the number of components, $k_1$, in $Y_1$ with success probability $1/k$. Let $k_2=k-k_1$ be the number of components in $Y_2$. Then, conditioned on this event holding,
$Y_j$ is an $\bigO{\epsilon^{1/(100kC^{k+1}(k+1)!)}}$-corrupted version of a set satisfying Condition~\ref{cond:convergence-of-moment-tensors} with respect to a mixture of $k_j$ components with parameter $\gamma=k\eps d^{-8k}k^{-C'k}/\alpha\le\eps d^{-8(k_j)}(k_j)^{-C'(k_j)}$ and $t=8k+48$.
}
We can apply the inductive hypothesis to $Y_1$ 
with number of components $k_1$ and fraction of outliers $\bigO{\epsilon^{1/(100kC^{k+1}(k+1)!)}}$, 
and conclude that we learn the components of $Y_1$ 
to total variation distance $\bigOk{\epsilon^{c_{k_1}/(100kC^{k+1}(k+1)!)}}$ 
\new{
with probability $\exp\left(-(3k_1-2)/\eps^2\right)$. 
A similar argument holds for $Y_2$.
Finally if this is the case, we combine the two hypotheses on $Y_1,Y_2$ by multiplying each weight by $|Y_j|/|Y|$ and then taking union of two hypotheses.
Then our combining method gives a final output that learns our full hypothesis to total variation distance error
$\bigOk{\epsilon^{c_{k_1}/(100C^{k+1}(k+1)!)  }}+\bigOk{\epsilon^{c_{k_2}/(100C^{k+1}(k+1)!) }}+\bigO{\epsilon^{1/(400kC^{k+1}(k+1)!)}}\le\bigOk{\epsilon^{c_k}}$ with probability at least 
$0.5\cdot0.5\cdot\frac{1}{k}\cdot\exp\left(-1/\eps^2-(3k_1-2)/\eps^2-(3k_2-2)/\eps^2\right)\ge\exp\left(-(3k-2)/\eps^2\right)$.

\paragraph{Sample Size and Running Time of Algorithm~\ref{algo:cluster_or_decode}}
By Lemma~\ref{lem:det-suffices}, we need $n\ge kt^{C't}d^t/\gamma^3$ samples to generate $X$ satisfying Condition~\ref{cond:convergence-of-moment-tensors} with parameters $(\gamma,t)$. We set $\gamma=\eps d^{-8k}k^{-C'k}$ and $t=8k+48$. Then $n\ge n_0=(8k)^{O(k)}d^{O(k)}/\eps^3$.
The running time in each sub-routine we invoke is dominated by the running time 
of the tensor decomposition algorithm, and by Lemma~\ref{lem:list_decoding} in the worst case 
this is $\poly(|L|,n)=\poly\left(\exp\left(1/\eps^2\right),d^{O(k)}/\eps^3\right)=
d^{O(k)}\exp\left(1/\eps^2\right)$.

This completes the first part of the proof.
}

\paragraph{Aggregating Hypotheses.}
\new{
We run Algorithm~\ref{algo:cluster_or_decode} repeatedly on set $Y$ and add the return hypothesis into a list $\calL$ until with probability 0.99, there exists a hypothesis $\widehat\calM \in \calL$ such that $\widehat 
\calM$ and $\calM$ are $\bigOk{\eps^{c_k}}$-close in total variation distance. Since Algorithm~\ref{algo:cluster_or_decode} outputs a correct mixture with probability $\exp\left(-(3k-2)/\eps^2\right)$, we will run Algorithm~\ref{algo:cluster_or_decode} for $\exp\left(O(k)/\eps^2\right)$ times. Then the total running time is $\exp\left(O(k)/\eps^2\right)\cdot d^{O(k)}\exp\left(1/\eps^2\right)=d^{O(k)}\exp\left(O(k)/\eps^2\right)$.
}


\paragraph{Robust Tournament.}
Then we need to run a robust tournament in order to find a hypothesis that is close to the true mixture in total variation distance.
Fact~\ref{lem:tournament}  shows that we can do this efficiently only with access to an $\epsilon$-corrupted set of samples 
of size $\bigOk{\log(1/\eps)/\eps^{2c_k}}$.


This completes the proof.
\end{proof}

%% file: full_algorithm.tex
\begin{mdframed}
 \begin{algorithm}[Algorithm for Robustly Learning Arbitrary GMMs]
   \label{algo:robust-GMM-arbitrary}\mbox{}
   \begin{description}
   \item[Input:] An outlier parameter $\epsilon>0$ and a component-number parameter $k\in \N$. 
   An $\epsilon$-corrupted sample $Y = \{y_1, y_2, \ldots, y_n \}$ from a $k$-mixture of Gaussians 
   $\calM = \sum_{i \in [k]} w_i \cN(\mu_i,\Sigma_i)$. 

   \item[Parameters:]  Let  $c_k=1/(100^k C^{(k+1)!}\mathrm{sf}(k+1)k!)$ be a scalar function of $k$, where $\textrm{sf}(k)=\prod_{i\in[k]}(k-1)!$ and $C$ is a sufficiently large constant.
   
   \item[Output:] A set of parameters $\{ (\hat w_i, \hat \mu_i, \hat\Sigma_i)\}_{i \in [ k]}$,  such that with probability at least $0.99$ the mixture 
   $\wh{\calM} = \sum_{i \in [ k]} \hat w_i \cN(\hat \mu_i, \hat \Sigma_i)$ is $\bigO{\epsilon^{c_k}}$-close in total variation distance to $\calM$. 

   \item[Operation:]\mbox{}
   \begin{enumerate}
   \item Let $\calL = \{ \phi\}$ be an empty list. Repeat the following $\exp\left(O(k)/\eps^2\right)$ times : 
   \begin{enumerate}
      \item Run Algorithm \ref{algo:cluster_or_decode} with input $Y$, fraction of outliers $\eps$, and number of components $k$. 
      Let the resulting output be a set of $k$ mixing weights, means and covariances, denoted by 
      $\left\{ (\hat w_{i}, \hat \mu_i ,\hat \Sigma_i)\right\}_{i \in [k]} $. 
      Add $\left\{ (\hat w_{i}, \hat \mu_i ,\hat \Sigma_i)\right\}_{i \in [k]} $ to $\calL$.
   \end{enumerate}
   \item Run the robust tournament from Fact~\ref{lem:tournament} over all the hypotheses in 
   $\calL$. Output the winning hypothesis, denoted by 
   $\{(\hat w_i,\hat \mu_i,\hat\Sigma_i)\}_{i\in[ k]}$. 
   \end{enumerate}
   
   \end{description}
 \end{algorithm}
\end{mdframed}

\medskip

\begin{mdframed}
 \begin{algorithm}[Cluster or List-Decode]
   \label{algo:cluster_or_decode}\mbox{}
   \begin{description}
   \item[Input:] An outlier parameter $0<\epsilon<1$ and a component-number parameter $k\in \N$. 
   An $\epsilon$-corrupted version $Y = \{y_1, y_2, \ldots, y_n \}$ of $X$, 
   where $X$ is a set of $n$ samples from a $k$-mixture of Gaussians 
   $\calM = \sum_{i \in [k]} w_i \cN(\mu_i,\Sigma_i)$ such that 
   $X$ satisfies Condition~\ref{cond:convergence-of-moment-tensors} with respect to $\calM$ 
   with parameters $(\eps d^{-8k}k^{-C'k},8k+48)$,  where $C'>0$ is a sufficiently large constant.

   \item[Parameters:]  Let  $c_k=1/(100^k C^{(k+1)!}\mathrm{sf}(k+1)k!)$ be a scalar function of $k$,
   where $\textrm{sf}(k)=\prod_{i\in[k]}(k-1)!$ and $C$ is a sufficiently large constant.  

   \item[Output:] A set of parameters $\{ (\hat{w}_i, \hat \mu_i, \hat\Sigma_i)\}_{i \in [k]}$ such that with probability at least 
   $\exp\left(-O(k)/\eps^2\right)$, 
   $d_{\textrm{TV}}\left(\sum_{i \in [k]}\hat{w}_i \calN(\hat{\mu}_i, \hat{\Sigma}_i) , \calM \right) \leq \bigO{\epsilon^{c_k}}$.


   \item[Operation:]\mbox{}
   \begin{enumerate}
\item \textbf{Treat Light Component as Noise}: {If $k = 0$, ABORT.} With probability $1/2$, run Algorithm \ref{algo:cluster_or_decode} on samples $Y$, with fraction of outliers $\epsilon + \epsilon^{1/(10C^{k+1}(k+1)!)}$ and number of components $k-1$. { Return the resulting set of estimated parameters, $\{ (\hat{w}_i, \hat \mu_i, \hat \Sigma_i) \}_{i \in [k-1]}$, appended with $(0, 0,I)$.} Else, do the following:
\\ \textit{// We guess whether the event that the minimum mixing weight $\alpha$ is at least $\epsilon^{1/(10C^{k+1}(k+1)!)}$ \\ // holds. If it does not, we proceed with the algorithm. Else, we treat the 
smallest weight \\
// component as noise and recurse with $k-1$ components.}
\item \textbf{Robust Isotropic Transformation}: \nnnew{With probability $0.5$,} run the algorithm corresponding to 
Lemma~\ref{lem:sampling_from_nearl_isotropic} on the samples $Y$, and let $\hat{\mu},\hat{\Sigma} $ be the robust estimates 
of the mean and covariance. If $k=1$, return $\left(\hat{w}=1, \hat{\mu}, \hat{\Sigma} \right)$. Else, compute $\hat{U} \hat{\Lambda} \hat{U}^\top $, the eigendecomposition of $\Sigma$, and for all $i \in [n]$, apply the affine transformation $y_i \to \hat{U}^\top \hat{\Sigma}^{\dagger/2}\Paren{y_i- \hat{\mu}}$.
 \\ \textit{// The resulting estimates $\hat{\mu}, \hat{\Sigma}$ satisfy Lemma \ref{lem:sampling_from_nearl_isotropic}, and the uncorrupted samples are\\// effectively drawn from a nearly isotropic $k$-mixture.}

\item With probability $1/2$, run either $(a)$ or $(b)$ in the following: 
\begin{enumerate}
   \item \textbf{Partial Clustering via SoS}: Run Algorithm \ref{algo:robust-partial-clustering} with outlier parameter $\epsilon$ and accuracy parameter $\epsilon^{1/(5C^{k+1}(k+1)!)}$. Let $Y_1$, $Y_2$ be the partition returned. Guess the number of components in $Y_1$ to be some $k_1 \in [k-1]$ uniformly at random. Run Algorithm \ref{algo:cluster_or_decode} with input $Y_1$, fraction of outliers $ \epsilon^{1/(10c^{k+1}(k+1)!)}$, and number of components $k_1$, and let $\Big\{ ( \hat w^{(1)}_i, \hat \mu^{(1)}_i, \hat \Sigma^{(1)}_i )\Big\}_{i \in [k_1]}$ be the resulting output. Similarly, run Algorithm \ref{algo:cluster_or_decode} with input $Y_2$, fraction of outliers $\epsilon^{1/(10c^{k+1}(k+1)!)}$, and number of components $k - k_1$, and let $\Big\{ ( \hat w^{(2)}_i, \hat \mu^{(2)}_i, \hat \Sigma^{(2)}_i )\Big\}_{i \in [k-k_1]}$ be the resulting output. 
   \nnnew{Output the set $\Big\{ ( \hat w^{(1)}_i |Y_1|/|Y| ,  \hat \mu^{(1)}_i, \hat{\Sigma}^{(1)}_i   )\Big\}_{i \in [k_1]} \cup \Big\{ ( \hat w^{(2)}_i |Y_2|/|Y|,   \mu^{(2)}_i,  \hat \Sigma^{(2)}_i )\Big\}_{i \in [k-k_1]}$.} 
      \\ \textit{// When the mixture is covariance separated, the preconditions of Theorem \ref{thm:partial-clustering-poly} 
      are  \\
      // satisfied (see Lemma \ref{lem:covariance_sep}). The partition is non-trivial, and the fraction of outliers \\
      // increases from $\epsilon \to \epsilon^{1/(10c^{k+1} (k+1)!)}$ .}   

   \item \textbf{ List-Decoding via Tensor Decomposition}: Run Algorithm \ref{algo:list-recovery-tensor-decomposition} 
   and let $L$ be the resulting list of hypotheses such that each hypothesis is a set of parameters $\{(\hat{\mu}_i, \hat{\Sigma}_i) \}_{i \in[k]}$. 
   Let $\tau = \Theta\Paren{  \epsilon^{1/(40C^{k+1}(k+1)!)}}$ be an eigenvalue threshold. 
   Select a hypothesis, $\{(\hat{\mu}_i, \hat{\Sigma}_i) \}_{i \in[k]} \in L$ uniformly at random. 
 \\  \textit{// Conditioned on not being covariance separated, we satisfy the preconditions of \\
 // Theorem \ref{thm:list-recovery-by-tensor-decomposition} (see Lemma \ref{lem:list_decoding}). 
 The output is a list that contains $\{\hat{\mu}_i, \hat{\Sigma}_i\}_{i \in [k]}$  \\
 // such that for all $i \in [k]$, $\Norm{ \hat{\mu}_i - \mu_i }_2 =\bigO{\epsilon^{1/(20C^{k+1}(k+1)!)}}$ and \\ 
 // $\Norm{ \hat{\Sigma}_i - \Sigma_i }_F =\bigO{\epsilon^{1/(20C^{k+1}(k+1)!)}} $.}
  
   \begin{enumerate}
   \item \textbf{Large Eigenvalues}: If for all $i \in [k]$, $\hat{\Sigma}_i \succeq \tau I$, 
   sample $\hat{w}_i$ from $[0,1]$ uniformly at random such that $\sum_{i} \hat w_i = (1\pm k\epsilon)$. 
   Return 
   $\Big\{ \Paren{ \hat{w}_i,  \hat{U} \hat{\Lambda}^{1/2} \hat{U}^\top\hat{\mu}_i + \hat\mu,  \hat{U} \hat{\Lambda}^{1/2} \hat{U}^\top \hat{\Sigma}_i \hat{U} \hat{\Lambda}^{1/2} \hat{U}^\top  } \Big\}_{i \in [k]}$. \\
   \textit{// If all estimated covariances have all eigenvalues larger than $\tau$, the recursion \\
   // bottoms out and the hypothesis is returned.
   }
   \item \textbf{Spectral Separation of Thin Components}: Else, $\exists v, i$ s.t. $v^\top \hat{\Sigma}_i v \leq \tau$.  
   Run the  algorithm corresponding to Lemma \ref{cor:recursion} with input $Y$, parameter estimates $\{(\hat{\mu}_i, \hat{\Sigma}_i) \}_{i \in [k]}$ and threshold $\tau$.  Let $Y_1$ and $Y_2$ be the resulting partition. 
   \\ \textit{// Use small eigenvalue directions to partition the points. 
   }
\begin{enumerate}
   \item If $\min(|Y_1|, |Y_2|) < \epsilon^{1/(400kC^{k+1}(k+1)!)}n$, run Algorithm \ref{algo:cluster_or_decode} with input $Y$, fraction of outliers $ 2\epsilon^{1/(400kC^{k+1}(k+1)!)}$ and number of components being $k-1$, and let $\Big\{ ( \hat w^{(1)}_i, \hat \mu^{(1)}_i, \hat \Sigma^{(1)}_i )\Big\}_{i \in [k_1]}$ be the resulting output. Output the resulting hypothesis $\Big\{(\hat{w}_i,  \hat{U} \hat{\Lambda}^{1/2} \hat{U}^\top \hat{\mu}_i + \hat\mu ,  \hat{U} \hat{\Lambda}^{1/2} \hat{U}^\top \hat{\Sigma}_i \hat{U} \hat{\Lambda}^{1/2} \hat{U}^\top ) \Big\}_{i \in [k-1]} \cup (0, 0, I)$.

   \item Else, select $k_1 \in [k-1]$ uniformly at random. Run Algorithm \ref{algo:cluster_or_decode} with input $Y_1$, fraction of outliers $\epsilon^{1/(100kC^{k+1}(k+1)!)}$ and number of components being $k_1$. Similarly, run Algorithm \ref{algo:cluster_or_decode} with input $Y_2$, fraction of outliers $\epsilon^{1/(100kC^{k+1}(k+1)!)}$ and number of components $k-k_1$, and let $\Big\{ ( \hat w^{(2)}_i, \hat \mu^{(2)}_i, \hat \Sigma^{(2)}_i )\Big\}_{i \in [k-k_1]}$ be the resulting output. Output the set $\Big\{ ( \hat w^{(1)}_i |Y_1|/|Y| , \hat{U} \hat{\Lambda}^{1/2} \hat{U}^\top \hat \mu^{(1)}_i + \hat \mu,  \hat{U} \hat{\Lambda}^{1/2} \hat{U}^\top \hat{\Sigma}^{(1)}_i \hat{U} \hat{\Lambda}^{1/2} \hat{U}^\top )\Big\}_{i \in [k_1]} \cup \Big\{ ( \hat w^{(2)}_i |Y_2|/|Y|,  \hat{U} \hat{\Lambda}^{1/2} \hat{U}^\top \mu^{(2)}_i + \hat \mu,  \hat{U} \hat{\Lambda}^{1/2} \hat{U}^\top \hat{\Sigma}^{(2)}_i \hat{U} \hat{\Lambda}^{1/2} \hat{U}^\top  )\Big\}_{i \in [k-k_1]}$.
\end{enumerate}   
   \end{enumerate}
   \end{enumerate}
   \end{enumerate}
   
   \end{description}
 \end{algorithm}
\end{mdframed}

%% file: helper_lemmas.tex

We start by showing that when Condition \ref{cond:convergence-of-moment-tensors} holds, the uniform distribution on a $(1-\epsilon)$-fraction of the points is certifiably hypercontractive.

\begin{lemma}[Component Moments to Mixture Moments]\label{lem:component_to_mixture}
Let $\calM = \sum_{i \in [k]} w_i \calN(\mu_i, \Sigma_i)$ be a $k$-mixture with mean $\mu$ and covariance $\Sigma$ such that $w_i\geq \alpha$, 
for some $0 < \alpha <1$, \nnnew{and for all $i,j \in [k]$, $\Norm{\Sigma^{\dagger/2}\Paren{\Sigma_i - \Sigma_j}\Sigma^{\dagger/2}}_F \leq 1/\sqrt{\alpha}$}. 
Let $X$ be a multiset of $n$ samples satisfying Condition \ref{cond:convergence-of-moment-tensors} 
with respect to $\calM$  with parameters $(\gamma, t)$, for $0<\gamma<(dk/\alpha)^{-ct}$, 
for a sufficiently large constant $c$, and $t\in\N$. Let $\calD$ be the uniform distribution over $X$. 
Then, $\calD$ is $2t$-certifiably $(c/\alpha)$-hypercontractive \nnnew{and for $d \times d$-matrix-valued indeterminate $Q$, $\sststile{2}{Q}\Set{\E_{\calM} \Paren{x^{\top}Qx-\E_{\calM}x^{\top}Qx}^2 \leq \bigO{1/\alpha} \Norm{\Sigma^{1/2} Q\Sigma^{1/2}}_F^2}$}.
\end{lemma}


Next, we show how to robustly estimate the mean and covariance of an $\epsilon$-corrupted set of samples satisfying Condition~\ref{cond:convergence-of-moment-tensors} when the mixture is not partially clusterable, and make the inliers nearly isotropic.

\begin{lemma}[Robust Isotropic Transformation] \label{lem:sampling_from_nearl_isotropic}
Given $0<\epsilon <1$, and $k \in \mathbb{N}$, let $\alpha=\epsilon^{1/(10C^{k+1}(k+1)!)}$. 
Let $\calM=\sum_{i=1}^k w_i G_i$ with $G_i = \cN(\mu_i,\Sigma_i)$ be a $k$-mixture of Gaussians with $w_i\ge\alpha$ for all $i \in [k]$, 
and let $\mu$ and $\Sigma$ be the mean and covariance of $\calM$ such that $r= \rank(\Sigma)$ \nnnew{and for all $i,j \in [k]$, $\Norm{\Sigma^{\dagger/2}\Paren{\Sigma_i - \Sigma_j}\Sigma^{\dagger/2}}_F \leq 1/\sqrt{\alpha}$}. 
Let $X$ be a set of points satisfying Condition \ref{cond:convergence-of-moment-tensors} 
with respect to $\calM$ for some parameters $(\gamma, t)$. 
Given a set $Y$, an $\epsilon$-corrupted version of $X$, of size $n\ge n_0=d^{O(1)}$, 
there exists an algorithm that takes $Y$ as input and in time $n^{\mathcal{O}(1)}$ outputs 
estimators $\hat{\mu}$ and $\hat{\Sigma}$ such that 
$\hat{\Sigma} = \hat{U} \hat{\Lambda}\hat{U}^\top$ is the eigenvalue decomposition, 
where $\hat{U}\in\R^{n \times r}$ has orthonormal columns and $\Lambda\in \R^{r \times r}$ is a diagonal matrix. Further, we can obtain $n$ samples $Y'$ by applying the affine transformation 
$y_i \to \hat{U}^\top\hat{\Sigma}^{\dagger/2}\Paren{y_i - \hat{\mu}}$ to each sample,
such that a $(1-\epsilon)$-fraction have mean $\mu'$ and covariance $\Sigma'$ satisfying 
\begin{enumerate}
	\item $\Norm{\mu'}_2 \leq \bigO{\Paren{1+ \frac{\sqrt{\epsilon}k}{\alpha}}\sqrt{\epsilon/\alpha}}$,
	\item $\Paren{ \frac{1}{ 1+ (k \sqrt{\epsilon}/\alpha )} } I_r \preceq \Sigma' \preceq\Paren{ \frac{1}{ 1- (k \sqrt{\epsilon}/\alpha )} } I_r$,
	\item $\Norm{ \Sigma' - I_r }_F \leq \bigO{\sqrt{\epsilon}k/\alpha}$,
\end{enumerate}
where $I_r$ is the $r$-dimensional Identity matrix,
and the remaining points are arbitrary. 
Let $X'$ be the set obtained by $\hat{U}^\top \hat\Sigma^{\dagger/2}\Paren{x_i - \hat \mu}$. 
Then, $X'$ satisfies Condition \ref{cond:convergence-of-moment-tensors} with respect to 
$\sum_{i=1}^k w_i \calN\Paren{ \hat{U}^\top \hat{\Sigma}^{\dagger/2} (\mu_i-\hat{\mu}), 
\hat{U}^\top \hat{\Sigma}^{\dagger/2} \Sigma_i \hat{\Sigma}^{\dagger/2}\hat{U} }$ 
and parameters $(\gamma, t)$, and $Y'$ is an $\epsilon$-corruption of $X'$. 
\end{lemma}

\begin{proof}
For any $t' \in \mathbb{N}$, it follows from Corollary \ref{cor:hypercontractivity_of_mixtures} 
that $\calM$ has $2t'$-certifiably $(4/\alpha)$-hypercontractive degree-$2$ polynomials, 
 since $w_i\geq \alpha$ for all $i$.  
Next,  Lemma \ref{lem:component_to_mixture} implies that the uniform distribution over $X$ 
also has $2t'$-certifiably $(8/\alpha)$-hypercontractive degree-$2$ polynomials 
and \nnnew{for $d \times d$-matrix-valued indeterminate $Q$, 
\[\sststile{2}{Q}\Set{\E_{\calM} \Paren{x^{\top}Qx-\E_{\calM}x^{\top}Qx}^2 \leq \bigO{1/\alpha} \Norm{\Sigma^{1/2} Q\Sigma^{1/2}}_F^2}.
\]}
Then, it follows from Fact \ref{fact:param-estimation-main} that if $\frac{16}{\alpha} t'\epsilon^{1-4/t'} \ll 1$, there exists an algorithm that runs in time $n^{O(t')}$ and outputs estimates $\hat{\mu}$ and $\hat{\Sigma}$ satisfying:
\begin{enumerate}
\item \label{prop_mean}$\Norm{\Sigma^{\dagger/2}(\mu -\hat{\mu})}_2 \leq O(t'/\alpha)^{1/2} \epsilon^{1-1/t'}$, 
\item \label{prop_cov1}$\Paren{ 1- (k/\alpha) \epsilon^{1-2/t'} } \Sigma \preceq \hat{\Sigma} \preceq \Paren{ 1+ (k/\alpha) \epsilon^{1-2/t'} } \Sigma$  and, 
\item \label{prop_cov2}$\Norm{\Sigma^{\dagger/2} \Paren{ \hat{\Sigma} - \Sigma} \Sigma^{\dagger/2}}_F \leq (t'/\alpha) O(\epsilon^{1-1/t'})$.
\end{enumerate}

Setting $t' =2$,  compute $\hat{\Sigma} = \hat{U} \hat{\Lambda} \hat{U}^\top$, the eigendecomposition of $\hat{\Sigma}$, 
such that $\hat{U} \in \R^{n \times r}$ has orthonormal columns, where $r\leq d$ is the rank of $\hat{\Sigma}$ and 
$\hat\Lambda \in \R^{r \times r }$ is a diagonal matrix. Similarly, let $\Sigma = U \Lambda U^\top$ be the 
eigendecomposition of $\Sigma$. We apply the affine transformation $y_i \to \hat{U}^\top \hat{\Sigma}^{\dagger/2}\Paren{y_i - \hat{\mu}}$ 
to each sample and thus we can assume throughout the rest of our argument 
that we have access to $\epsilon$-corrupted samples from a $k$-mixture of Gaussians 
with mean $\mu' = \hat{U}^\top  \hat{\Sigma}^{\dagger/2}(\mu - \hat \mu)$ 
and covariance $\Sigma' = \hat{U}^\top \hat{\Sigma}^{\dagger/2} \Sigma \hat{\Sigma}^{\dagger/2} \hat{U}$. 
Then, we have that
\begin{equation*}
\begin{split}
\Norm{\mu'}_2 = \Norm{\hat{U}^\top  \hat{\Sigma}^{\dagger/2}(\mu - \hat \mu)}_2  
& \leq \Norm{ \hat{U}^\top}_{\textrm{op}}  \Norm{ \hat{\Sigma}^{\dagger/2}(\mu - \hat \mu)}_2   \\
& \leq \bigO{\Paren{1+ \frac{\sqrt{\epsilon}k}{\alpha}}\sqrt{\epsilon/\alpha}} \;,
\end{split}
\end{equation*}
where the last inequality follows from \eqref{prop_mean} and \eqref{prop_cov1}. 
It also follows from \eqref{prop_cov1} that 
\begin{equation}
\label{eqn:lowner_guarantee}
\Paren{ \frac{1}{ 1+ (k \sqrt{\epsilon}/\alpha )} } \hat{\Sigma} \preceq \Sigma \preceq\Paren{ \frac{1}{ 1- (k \sqrt{\epsilon}/\alpha )} } \hat{\Sigma} \;.
\end{equation}
Multiplying out \eqref{eqn:lowner_guarantee} with $\hat{U}^\top \hat{\Sigma}^{\dagger/2}$ on the left 
and $\hat{\Sigma}^{\dagger/2} \hat{U}$ on the right, we have
\begin{equation*}
\begin{split}
\Paren{ \frac{1}{ 1+ (k \sqrt{\epsilon}/\alpha )} } \hat{U}^\top \hat{\Sigma}^{\dagger/2}\hat{\Sigma} \hat{\Sigma}^{\dagger/2} \hat{U} \preceq \Sigma' \preceq\Paren{ \frac{1}{ 1- (k \sqrt{\epsilon}/\alpha )} } \hat{U}^\top \hat{\Sigma}^{\dagger/2} \hat{\Sigma} \hat{\Sigma}^{\dagger/2} \hat{U} \;.
\end{split}
\end{equation*}
Observe that \eqref{prop_cov1} implies that the rank of $\hat{\Sigma}$ and $\Sigma$ is the same, 
and thus $\hat{U}^\top \hat{\Sigma}^{\dagger/2}\hat{\Sigma} \hat{\Sigma}^{\dagger/2} \hat{U} = I_r$, 
where $I_r$ is the $r$-dimensional Identity matrix. Finally, we have that
\begin{equation*}
\begin{split}
\Norm{ \Sigma' - I_r }_F = \Norm{ \hat{U}^\top \hat{\Sigma}^{\dagger/2} \Sigma \hat{\Sigma}^{\dagger/2} \hat{U}  - \hat{U}^\top \hat{\Sigma}^{\dagger/2} \hat{\Sigma} \hat{\Sigma}^{\dagger/2} \hat{U}  }_F & \leq \Norm{ \hat{U} \hat{\Lambda}^{-1/2} \hat{U}^\top \Paren{  \Sigma  - \hat{\Sigma} }  \hat{U} \hat{\Lambda}^{-1/2} \hat{U}^\top    }_F \\
& = \Norm{ \hat{U} \hat{\Lambda}^{-1/2} \Lambda^{1/2}\Lambda^{-1/2}  \hat{U}^\top \Paren{  \Sigma  - \hat{\Sigma} }  \hat{U}\Lambda^{-1/2} \Lambda^{1/2}\hat{\Lambda}^{-1/2} \hat{U}^\top  }_F \\
& \leq \Norm{\hat{\Lambda}^{-1/2} \Lambda^{1/2}}^2_{\textrm{op}} \Norm{\Sigma^{\dagger/2} \Paren{ \hat{\Sigma} - \Sigma} \Sigma^{\dagger/2}}_F \\
&\leq \bigO{ \sqrt{\epsilon} k /\alpha} \;,
\end{split}
\end{equation*}
where we use that $ \hat{\Lambda}^{-1/2} =  \hat{\Lambda}^{-1/2} \Lambda^{1/2}\Lambda^{-1/2}$, 
the sub-multiplicative property of the Frobenius norm, the column span $U$ and $\hat{U}$ is identical 
(see \eqref{prop_cov1}), and the Frobenius recovery guarantee in \eqref{prop_cov2}.  




Finally, it follows from Lemma \ref{lem:covergence_affine_invariant} that Condition \ref{cond:convergence-of-moment-tensors} 
is affine invariant and is thus preserved under $x_i \to \hat{U}^\top\hat{\Sigma}^{-1/2}\Paren{x_i - \hat{\mu}}$, 
for $i\in[n]$, with parameters $(\gamma, t)$.
\end{proof}

The above robust isotropic transformation lemma allows us to obtain a covariance that is close to the identity matrix 
in a full-dimensional subspace (potentially smaller than the input dimension). 
Therefore, we will subsequently drop the subscript for the dimension, wherever it is clear from the context.

Next, we show that whenever the minimum mixing weight is sufficiently larger than the fraction of outliers, 
and a pair of components is covariance separated, we can partially cluster the samples. 


\begin{lemma}[Non-negligible Weight and Covariance Separation] \label{lem:covariance_sep}
Given $0<\epsilon<  1/k^{k^{O(k^2)}}$ and $k \in \mathbb{N}$, 
let $\alpha=\epsilon^{1/(10C^{k+1}(k+1)!)}$. 
\new{
Let $\calM=\sum_{i=1}^k w_i G_i$ with $G_i = \cN(\mu_i,\Sigma_i)$ be a $k$-mixture of Gaussians 
with mixture covariance $\Sigma$ such that $w_i \geq \alpha$ for all $i \in [k]$
 and there exist $i, j\in[k]$ such that  
$\Norm{ \Sigma^{\dagger/2}\Paren{\Sigma_i - \Sigma_j}\Sigma^{\dagger/2} }_F > 1/\sqrt{\alpha}$. Further, let $X$ be a set of points satisfying Condition \ref{cond:convergence-of-moment-tensors} with respect to $\calM$ 
for some parameters $\gamma\leq \eps d^{-8k}k^{-Ck}$, for a sufficiently large constant $C$, and $t\geq 8k$. 
Let $Y$ be an $\epsilon$-corrupted version of $X$ of size $n\ge n_0= \Paren{ dk }^{\Omega(1)}/\epsilon$,} 
Algorithm~\ref{algo:robust-partial-clustering} partitions $Y$ into $Y_1$, $Y_2$ in time $n^{O(1)}$ such that 
\new{with probability at least $\alpha^{k\log(k/\alpha)}$
there is a non-trivial partition of $[k]$ into $Q_1\cup Q_2$ so that letting $\calM_j$ 
be a distribution proportional to $\sum_{i\in Q_j}w_iG_i$ and $W_j=\sum_{i\in Q_j}w_i$, then
$Y_j$ is an $\bigO{\epsilon^{1/(10C^{k+1}(k+1)!)}}$-corrupted version of $\bigcup_{i\in Q_j} X_i$ 
satisfying Condition~\ref{cond:convergence-of-moment-tensors} with respect to $\calM$ 
with parameters $(\bigO{k\gamma/W_j},t)$.
}
\end{lemma}

\begin{proof}
We run Algorithm \ref{algo:robust-partial-clustering} with sample set $Y$, number of components $k$, 
the fraction of outliers $\epsilon$ and the accuracy parameter $\beta$. 
Since $X$ satisfies Condition~\ref{cond:convergence-of-moment-tensors}, 
we can set $t'\geq24$, $\beta=\alpha^{t'/4 - 4} k^{t'} (t')^{2t'}\le\alpha$ in Theorem \ref{thm:partial-clustering-poly}.
Then, by assumption, there exist $i,j$ such that 
$$
\Norm{\Sigma^{\dagger/2}\Paren{ \Sigma_i - \Sigma_j }\Sigma^{\dagger/2} }_F > \frac{1}{\sqrt{\alpha}}=\Omega\left(\frac{k^2(t')^4}{(\beta\alpha^4)^{2/t'}}\right) \;.
$$
We observe that we also satisfy the other preconditions for Theorem \ref{thm:partial-clustering-poly}, 
since $n \geq \Paren{ dk/ }^{\Omega(1)}/\epsilon$. 

Then, Theorem \ref{thm:partial-clustering-poly} implies that with probability at least $\alpha^{k\log(k/\alpha)}$, 
the set $Y$ is partitioned in two sets $Y_1$ and $Y_2$ such that 
\new{
there is a non-trivial partition of $[k]$ into $Q_1\cup Q_2$ so that letting $\calM_j$ be a distribution proportional 
to $\sum_{i\in Q_j}w_iG_i$ and $W_j =\sum_{i\in Q_j}w_i$, then
$Y_j$ is an $\bigO{\epsilon^{1/(10C^{k+1}(k+1)!)}}$-corrupted version of $\bigcup_{i\in Q_j} X_i$. 
By Lemma \ref{submixture condition lemma}, $\bigcup_{i\in Q_j} X_i$ 
satisfies Condition~\ref{cond:convergence-of-moment-tensors} with respect to $\calM$ with parameters $(\bigO{k\gamma/W_j},t)$.
}
\end{proof}

When the mixture is not covariance separated and nearly isotropic, we can obtain a small list of hypotheses such that one of them is 
close to the true parameters, via tensor decomposition.

\begin{lemma}[Mixture is List-decodable] \label{lem:list_decoding}
Given $0<\epsilon<  1/k^{k^{O(k^2)}}$ let $\alpha=\epsilon^{1/(10C^{k+1}(k+1)!)}$. 
\new{Let $\calM=\sum_{i=1}^k w_i G_i$ with $G_i = \cN(\mu_i,\Sigma_i)$ 
be a $k$-mixture of Gaussians with mixture mean $\mu$ and mixture covariance $\Sigma$,
such that $\Norm{\mu}_2 \leq \bigO{\sqrt{\epsilon/\alpha}}$,
$\norm{\Sigma-I}_F\le\bigO{\sqrt{\eps}/\alpha}$,
$w_i \geq \alpha$ for all $i \in [k]$, 
and $ \Norm{\Sigma_i - \Sigma_j }_F \leq 1/\sqrt{\alpha}$ for any pair of components,
and let $X$ be a set of points satisfying Condition \ref{cond:convergence-of-moment-tensors} 
with respect to $\calM$ for some parameters $\gamma= \eps d^{-8k}k^{-Ck}$, for a sufficiently large constant $C$, 
and $t= 8k$. Let $Y$ be an $\epsilon$-corrupted version of $X$ of size $n$,
Algorithm~\ref{algo:list-recovery-tensor-decomposition} 
outputs a list $L$ of hypotheses of size $\exp\left(1/\eps^2\right)$ in time $\poly(|L|,n)$ 
such that if we choose a hypothesis 
$\{\hat{\mu}_i, \hat{\Sigma}_i \}_{i \in [k]}$ uniformly at random,
$\Norm{\mu_i - \hat{\mu}_i}_2 \leq\bigO{\epsilon^{1/(20C^{k+1}(k+1)!)}}$ and 
$\Norm{\Sigma_i - \hat{\Sigma}_i }_F \leq \bigO{\epsilon^{1/(20C^{k+1}(k+1)!)}}$ for all $i$ with probability at least $\exp\left(-1/\eps^2\right)$.
}
\end{lemma}
\begin{proof}
Recall we run Algorithm \ref{algo:list-recovery-tensor-decomposition} on \new {the samples $Y$, the number of clusters $k$, the fraction of outliers $\eps$ and the minimum weight $\alpha=\epsilon^{1/(10C^{k+1}(k+1)!)}$}. Next, we show that the preconditions of Theorem \ref{thm:list-recovery-by-tensor-decomposition} are satisfied. First, \new{the upper bounds on $\norm{\mu}_2$ and $\norm{\Sigma-I}_F$ imply} $\sum_{i \in k} w_i \Paren{\Sigma_i  + \mu_i \mu_i^\top} = \Sigma+\mu\mu^\top \preceq (1+\bigO{\sqrt{\eps}/\alpha}) I$. Since the LHS is a conic combination of PSD matrices, it follows that  for all $i \in [k]$, $\mu_i \mu_i^\top \preceq \frac{1}{\alpha}\left(1+\bigO{\sqrt{\eps}/\alpha}\right) I$, 
and thus $\Norm{\mu_i \mu_i^\top}_F \leq \frac{2}{\alpha}$. Next, we can write:
\begin{align*}
\Norm{\Sigma_i  - I }_F 
&\le \norm{\Sigma_i-(\Sigma+\mu\mu^\top)}_F+\norm{\Sigma-I}_F+\norm{\mu\mu^\top}_F\\
&= \Norm{\Sigma_i  - \sum_{j \in [k]} w_j \Paren{\Sigma_j + \mu_j\mu_j^\top} }_F + \frac{\sqrt{\eps}k}{\alpha}+\frac{\eps}{\alpha}\\
&\leq \Norm{\sum_{j \in [k]} w_j \Paren{ \Sigma_i  -  \Sigma_j} }_F + \frac{2}{\alpha}+ \frac{\sqrt{\eps}k}{\alpha}+\frac{\eps}{\alpha} \\
&\leq  \frac{4}{\alpha} \;,
\end{align*}
where the first and the third inequalities follow from the triangle inequality and the upper bound on $\Norm{\mu_i \mu_i^\top}_F$, 
and the last inequality follows from the assumption that $\Norm{\Sigma_i - \Sigma_j }_F \leq 1/\sqrt{\alpha}$ 
for every pair of covariances $\Sigma_i,\Sigma_j$. So, we can set $\Delta=4/\alpha$ 
in Theorem~\ref{thm:list-recovery-by-tensor-decomposition}. 
Then, given the definition of $\alpha$, we have that
$$
\eta=2k^{4k}\bigO{1+\Delta/\alpha}^{4k}\sqrt{\eps} =\bigO{\epsilon^{2/5}}
$$
and $1/\eps^2\ge \log(1/\eta)(k+1/\alpha+\Delta)^{4k+5}/\eta^2$.
Therefore, Algorithm \ref{algo:list-recovery-tensor-decomposition} outputs a list $L$ of hypotheses 
such that $\abs{L} = \exp\left(1/\eps^2\right)$, and with probability at least $0.99$, 
$L$ contains a hypothesis that satisfies the following: for all $i \in [k]$, 
\begin{equation}
\label{eqn:true_hypothesis}
\begin{split}
\Norm{ \hat{\mu}_i - \mu_i }_2 & 
=\bigO{\frac{\Delta^{1/2}}{\alpha}}\eta^{G(k)}
=\bigO{\epsilon^{-1/(20C^{k+1}(k+1)!)}\cdot\eps^{1/(10C^{k+1}(k+1)!)}}
=\bigO{\epsilon^{1/(20C^{k+1}(k+1)!)}}  \textrm{ and }\\
 \Norm{ \hat{\Sigma}_i - \Sigma_i }_F & =\bigO{k^4}\frac{\Delta^{1/2}}{\alpha}\eta^{G(k)}
=\bigO{\epsilon^{1/(20C^{k+1}(k+1)!)}} \;.
\end{split}
\end{equation}
\new{
Then if we choose a hypothesis in $L$ uniformly at random, the probability that we choose the hypothesis satisfying (\ref{eqn:true_hypothesis}) is at least $1/|L|=\exp\left(-1/\eps^2\right)$.
} 
\end{proof}

{ Finally, if the mixture has a covariance matrix with small variance along any direction, we can further cluster the points by projecting the mixture along that direction.}

\begin{lemma}[Spectral Separation of Thin Components]\label{lem:hyperplane_clustering}
{
Given $0<\epsilon<  1/k^{k^{O(k^2)}}$, let $\alpha=\epsilon^{1/(10C^{k+1}(k+1)!)}$. Let $\calM=\sum_{i=1}^k w_i G_i$ with $G_i = \cN(\mu_i,\Sigma_i)$ be a $k$-mixture of Gaussians with mixture covariance $\Sigma$ 
such that $\Norm{\Sigma - I}_F \leq \bigO{\sqrt{\epsilon}k/\alpha}$, 
and let $X$ be a set of points satisfying Condition \ref{cond:convergence-of-moment-tensors} with respect to $\calM$ for some parameters $(\gamma,t)$. Given a set $Y$ being an $\epsilon$-corrupted version of $X$ of size $n$,
} 
and estimates $\{\hat{\mu}_i, \hat{\Sigma}_i \}_{i \in [k]}$, 
such that $\Norm{\mu_i - \hat{\mu}_i}_2 \leq\bigO{\epsilon^{1/(20C^{k+1}(k+1)!)}}$, 
$\Norm{\Sigma_i - \hat{\Sigma}_i }_F \leq \bigO{\epsilon^{1/(20C^{k+1}(k+1)!)}}$, 
suppose there exists a unit vector $v \in \R^d$ such that 
$v^\top \hat{\Sigma}_s v \leq \bigO{\epsilon^{1/(40C^{k+1}(k+1)!)}}$, 
for some $s \in [k]$. 
Then, there is an algorithm that efficiently partitions $Y$ into $Y_1$ and $Y_2$ 
such that 
\new{
there is a non-trivial partition of $[k]$ into $Q_1\cup Q_2$ so that letting $\calM_j$ be a distribution proportional to $\sum_{i\in Q_j}w_iG_i$ and $W_j=\sum_{i\in Q_j}w_i$, then
$Y_j$ is an $\left(\bigO{k^2\gamma}+\bigO{\epsilon^{1/(80kC^{k+1}(k+1)!)}/W_j}\right)$-corrupted version of $\bigcup_{i\in Q_j} X_i$ satisfying Condition~\ref{cond:convergence-of-moment-tensors} 
with respect to $\calM_j$ with parameter $(\bigO{k\gamma/W_j},t)$.
}
\end{lemma}
\begin{proof}
We run the algorithm from Lemma~\ref{cor:recursion} with the input being the samples $Y$, 
the current hypothesis $\{\hat{\mu}_i, \hat{\Sigma}_i \}_{i \in [k]}$, 
and the minimum eigenvalue $\eta=\bigO{\epsilon^{1/(40C^{k+1}(k+1)!)}}$. 
Observe that the mixture covariance satisfies $\Sigma \succeq \left(1-\bigO{\sqrt{\eps}k/\alpha}\right)I\succeq I/2$ 
and the upper bound on means and covariance is $\delta =\bigO{\epsilon^{1/(20kC^{k+1}(k+1)!)}n}$ by assumption. 
Therefore, we satisfy the preconditions of Lemma \ref{cor:recursion}. 
Thus, we obtain a partition $Y_1, Y_2$ such that 
\new{
there is a non-trivial partition of $[k]$ into $Q_1\cup Q_2$ so that letting $\calM_j$ be a distribution proportional to $\sum_{i\in Q_j}w_iG_i$ and $W_j=\sum_{i\in Q_j}w_i$, then it follows from Lemma \ref{submixture condition lemma} that
$Y_j$ is an $\left(\bigO{k^2\gamma}+\bigO{\epsilon^{1/(80kC^{k+1}(k+1)!)}/W_j}\right)$-corrupted version of $\bigcup_{i\in Q_j} X_i$ satisfying Condition~\ref{cond:convergence-of-moment-tensors} with respect to $\calM_j$ with parameter $(\bigO{k\gamma/W_j},t)$.
}
\end{proof}

%% file: partial-clustering-upgrade.tex

\section{More Efficient Robust Partial Cluster Recovery}
\label{sec:robust-partial-clustering-upg} 

In this section, we prove the following upgraded partial 
clustering theorem. In contrast to Theorem~\ref{thm:partial-clustering-poly}, here we obtain a probability of success that is inverse exponential in $k$ instead of $1/\alpha$.  

\begin{theorem}[Robust Partial Clustering in Relative Frobenius Distance] \label{thm:partial-clustering-poly-upg}
Let $0 \leq \epsilon < \alpha/k \leq 1$ and $t \in \N$. There is an algorithm with the following guarantees: 
Let $Y$ be an $\epsilon$-corruption of a sample $X$ of size $n\geq \paren{dk}^{Ct}/\epsilon$ for a large enough constant $C>0$, from 
$\calM = \sum_i w_i \cN(\mu_i,\Sigma_i)$ that satisfies Condition~\ref{cond:convergence-of-moment-tensors} 
with parameters $2t$ and $\gamma \leq \epsilon d^{-8t}k^{-Ck}$, for a large enough constant $C >0$. 
Suppose further that $w_i \geq \alpha > 2 \epsilon$ for each $i \in [k]$, 
and that for some $t \in \N$, $\beta > 0$ there exist $i,j \leq k$ such that 
$\Norm{\Sigma^{\dagger/2}(\Sigma_i -\Sigma_j)\Sigma^{\dagger/2}}_F^2 = \Omega\left((k^2 t^4)/\beta^{2/t} \alpha^4\right)$, 
where $\Sigma$ is the covariance of the mixture $\calM$. 
Then, for any $\eta \gg \sqrt{\epsilon/\alpha}$, the algorithm runs in time $n^{O(t)}$, and with probability at least $2^{-O\Paren{k}} (1- O(\eta/\alpha - \sqrt{\eta}) )$ over the random choices of the algorithm, 
outputs a partition $Y = Y_1 \cup Y_2$ satisfying:   
\begin{enumerate}
 \item \textbf{Partition respects clustering:} for each $i$, $\max \left\{ \frac{1}{w_i n}|Y_1 \cap X_i|, \frac{1}{w_i n}|Y_2 \cap X_i|\right\} \geq 1- O(\sqrt{\eta}) - O(\frac{\beta}{\eta \alpha^2})$, where $X_i \subset X$ corresponding to the points drawn from $\calN(\mu_i, \Sigma_i)$. 
 \item \textbf{Partition is non-trivial:} $\max_{i}\frac{1}{w_i n} |X_i \cap Y_1|, \max_{i}\frac{1}{w_i n} |X_i \cap Y_2| \geq 1- O(\sqrt{\eta}) - O(\frac{\beta}{\eta \alpha^2})$. 
\end{enumerate}
\end{theorem}

\subsection{Algorithm} 

Our algorithm will solve SoS relaxations of a polynomial inequality system. 
The indeterminates in this system are $X'$ (that is intended to be the guess for the original uncorrupted sample), a cluster of size $\alpha n$ within $X'$ (indicated by $z_i$s) with mean $\hat{\mu}$ and covariance matrix $\hat{\Sigma}$ and $\Pi$ (intended to be the square root of $\hat{\Sigma}$). The input corrupted sample $Y$ is a constant in this inequality system. Let $U\in \R^{d\times d}$ and $m, z \in \R^d$ also be indeterminates of the proof system.
The system can be thought of as encoding the task of finding clusters $\hat{C}$ within $Y$ that satisfies certifiable hypercontractivity of degree $2$ polynomials.

We present the constraints grouped together into meaningful categories below: 
The first set of constraints enforce that $\hat{\Sigma}$ is the square of $\Pi$. 
\begin{equation}
\text{Covariance Constraints: $\cA_1$} = 
  \left \{
    \begin{aligned}
      &
      &\Pi
      &=UU^{\top}\\
      &
      &\Pi^2
      &=\hat{\Sigma}\\
    \end{aligned}
  \right \}
\end{equation}

The intersection constraints force that $X'$ intersects $Y$ in all but an $\epsilon n$ points (and thus, $2\epsilon$-close to unknown sample $X$). 
\begin{equation}
\text{Intersection Constraints: $\cA_2$} = 
  \left \{
    \begin{aligned}
      &\forall i\in [n],
      & m_i^2
      & = m_i\\
      &&
      \textstyle\sum_{i\in[n]} m_i 
      &= (1-\epsilon) n\\
      &\forall i \in [n],
      &z_i (\tilde{y}_i-x'_i)
      &= 0
    \end{aligned}
  \right \}
\end{equation}

The subset constraints enforce that $z$ indicate a subset of size $\alpha n$ of $X'$. 
\begin{equation}
\text{Subset Constraints: $\cA_3$} = 
  \left \{
    \begin{aligned}
      &\forall i\in [n].
      & z_i^2
      & = z_i\\
      &&
      \textstyle\sum_{i\in[n]} z_i 
      &= \alpha n\\
    \end{aligned}
  \right \}
\end{equation}

Parameter constraints create indeterminates to stand for the covariance $\hat{\Sigma}$ and mean $\hat{\mu}$ of $\hat{C}$ (indicated by $z$).
\begin{equation}
\text{Parameter Constraints: $\cA_4$} = 
  \left \{
    \begin{aligned}
      &
      &\frac{1}{\alpha n}\sum_{i = 1}^n z_i \Paren{x'_i-\hat\mu}\Paren{x'_i-\hat\mu}^{\top}
      &= \hat{\Sigma}\\
      &
      &\frac{1}{\alpha n}\sum_{i = 1}^n z_i x'_i
      &= \hat\mu\\
    \end{aligned}
  \right \}
\end{equation}

\text{Certifiable Hypercontractivity : $\cA_4$}=\\
\begin{equation}
  \left \{
    \begin{aligned}
     &\forall t \leq 2s
     &\E_{z} (Q-\E_{z}Q)^{2t}
     &\leq (Ct/\alpha)^{t}2^{2t} \Paren{\E_{z} (Q-\E_{z}Q)^2}^t\\
     &
     &\E_{z}(Q-\E_{z}Q)^2
     &\leq 10\Paren{\frac{1}{\alpha}}^2 \Norm{Q}_F^2
    \end{aligned}
  \right \}
\end{equation}

where we write $\E_{z} Q$ as a shorthand for the polynomial $\frac{1}{\alpha^ n} \sum_{i} z_i Q(x_i)$ and $\E_{z} (Q-\E_{X_r}Q)^{2j}$ for the polynomial $\frac{1}{ \alpha n} \sum_{i} z_i \Paren{Q(x'_i)-\frac{1}{ \alpha  n} \sum_{i \leq n} z_i Q(x'_i)}^{2j}$ for any $j$. Note that $Q$ is a $d \times d$-matrix valued indeterminate. Observe that $Q$ itself can be eliminated from the system as is standard in several applications~\cite{KS17,KStein17,HopkinsL18,BK20} of SoS proofs in obtaining a succinct set of polynomial constraints (see Section 4.3 on ``Succinct Representation of Constraints'' in ~\cite{TCS-086} for an exposition).


\begin{mdframed}
  \begin{algorithm}[Polynomial Time Partial Clustering]
    \label{algo:robust-partial-clustering-upg}\mbox{}
    \begin{description}
    \item[Given:]
        A sample $Y$ of size $n$. An outlier parameter $\epsilon > 0$ and an accuracy parameter $\eta > 0$. 
    \item[Output:]
      A partition of $Y$ into partial clustering $Y_1 \cup Y_2$.
    \item[Operation:]\mbox{}
    \begin{enumerate}
      \item \textbf{Mean and Covariance Estimation}: Apply Robust Mean and Covariance Estimation (Fact~\ref{fact:ks-cov-est}) to estimate $\hat{\mu}$ and $\tilde{\Sigma}$ such that $\frac{1}{2} \Sigma \preceq \tilde{\Sigma} \preceq 1.5 \Sigma$ where $\Sigma$ is the covariance of the uncorrupted input mixture. 
      \item \textbf{Approximate Isotropic Transformation}: For each $y_i \in Y$, let $\tilde{y}_i = \tilde{\Sigma}^{\dagger/2} (y_i-\hat{\mu})$. Let $\tilde{Y} = \cup_{i \leq n} \tilde{y}_i$. 
      \item \textbf{SDP Solving:} Find a pseudo-distribution $\tzeta$ satisfying $\cup_{i =1}^4 \cA_i$ such that $\pE_{\tzeta}z_i \in \alpha\pm o_d(1)$ for every $i$. If no such pseudo-distribution exists, output fail.
        \item \textbf{Rounding:} Let $M = \pE_{z \sim \tzeta} [zz^{\top}]$. 
          \begin{enumerate} 
            \item \textbf{Generate candidate clusters}: For $\ell = O(1/\alpha \log \eta/\alpha)$ times, draw a uniformly random $i \in [n]$ and let $\hat{C}_i = \{j \mid M(i,j) \geq \alpha^2/2 \}$.  Let $\cL = \cup_{i \leq \ell} \hat{C}_i$. 
            \item \textbf{Candidate 2nd Moment Estimation:} For each $\hat{C}_i \in \cL$, let $S_i$ be the output of running robust 2nd moment estimation with Frobenius error (Lemma~\ref{lem:2nd-moment-est}) on $\hat{C}_i$ with outlier parameter $\eta_i'= O(\frac{\epsilon}{\alpha} + \frac{\beta}{\alpha^2 \eta})$ . 
             \item \textbf{Merge candidate clusters}: For each $i \leq \ell$, find $\cL_i$ to be the collection of all $j$ such that $\Norm{S_i- S_j}_F \leq 2C\tau$ for a large enough constant $C>0$. Set $\hat{C}_i \cup \cL_i = \hat{B}_i$. Repeat on $\cL \setminus \{\cL_i \cup i\}$. 
            \item \textbf{Output a union of a random subset of candidates}: For $\cL' = \cup_{i} \hat{B}_i$, choose a uniformly random subset $S$ of $\cL'$, set $Y_1 = \cup_{j \in S} \hat{B}_j$ and set $Y_2 = Y\setminus Y_1$. 
      \end{enumerate}
      \end{enumerate}
    \end{description}
  \end{algorithm}
\end{mdframed}

\paragraph{Analysis of Algorithm} 

\begin{lemma}[Success of Step 1]
Let $\tilde{\Sigma}$ be the output of the robust covariance estimation algorithm (Fact~\ref{fact:ks-cov-est}) applied to the input sample $Y$ with outlier parameter $\epsilon$. If $Y$ is an $\epsilon$-corruption of a sample $X$  from a GMM with minimum weight $\geq \alpha \geq  \Omega(\sqrt{\epsilon})$, mixture mean $\mu$ and covariance $\Sigma$ satisfying Condition~\ref{cond:convergence-of-moment-tensors}, then, 

\[
0.5 \Sigma \preceq \tilde{\Sigma} \preceq 1.5 \Sigma\mcom
\]
\[
\Norm{\tilde{\Sigma}^{-1/2}(\mu-\hat{\mu})}_2 \leq O(\sqrt{\epsilon}/\alpha)\mper
\]
\end{lemma}

\begin{proof}
The lemma immediately follows by noting that GMMs with minimum weight $\alpha$ are $4$-certifiably $1/\alpha$-subgaussian (Fact~\ref{fact:subgaussianity-of-mixtures}) and $\alpha \geq \Omega(\sqrt\epsilon)$.

\end{proof}

\begin{lemma}[Simultaneous Intersection Bounds for Frobenius Separated Case] 
\label{lem:simultaneous-intersection-bounds-poly-clustering-upg}
Let $X = X_1 \cup X_2 \cup \ldots X_k$ be a sample of size $n\geq \paren{dk}^{Ct}/\epsilon$ for a large enough constant $C>0$, from 
$\calM = \sum_i w_i \cN(\mu_i,\Sigma_i)$ that satisfies Condition~\ref{cond:convergence-of-moment-tensors} 
with parameters $2t$ and $\gamma \leq \epsilon d^{-8t}k^{-Ck}$. 
Suppose further that $\Norm{\mu_i}_2 \leq \frac{2}{\alpha}$ for every $i$, $\Norm{\Sigma_i}_2 \leq \frac{1}{\alpha}$ for every $i$ and the mixture mean $\mu$, covariance $\Sigma$ satisfy $\Norm{\mu}_2 \leq 1$ and $0.5 I \preceq \Sigma \preceq 1.5 I$. 
Let $\tau = 10^8 \frac{C^6 t^4}{\beta^{2/t}\alpha^2}$, for any $\beta>0$. 
Then, given any $\epsilon$-corruption $Y$ of $X$, for every $i,j$ such that $\Norm{\Sigma_i -\Sigma_j}_F^2 \geq  \Omega( \tau)$, 
\[
\Set{\cup_{i = 1}^4\cA_i } \sststile{2t}{z} \Set{z'(X_i) z'(X_j) \leq \beta}\mcom
\]
where $z'(X_i) = \frac{1}{w_i n} \sum_{j \in X_i} z_j \1(x_j=y_j)$ for every $i$. 
\end{lemma}

\begin{proof}[Proof of Theorem~\ref{thm:partial-clustering-poly-upg}]
First, since $Y$ is an $\epsilon$-corruption of a sample $X$ from a GMM such that $X$ satisfies Condition~\ref{cond:convergence-of-moment-tensors}, our robust mean and covariance estimation procedure (Step 1) applied to the mixture succeeds and recovers an estimate of the covariance that is multiplicative $1\pm 0.5$-factor approximation in Löwner order. Thus, for the rest of the analysis, we can assume that the smallest and largest eigenvalue of the mixture covariance are at least $0.5$ and at most $1.5$. Since each component has weight at least $\alpha$, this means that each of the constitute component covariance can now be assumed to have a spectral norm at most $1.5/\alpha$. 

Next, by an argument similar to the one presented in the proof of Theorem~\ref{thm:partial-clustering-non-poly}, the convex program we wrote is approximately solvable in polynomial time and is feasible whenever the uncorrupted sample $X$ satisfies Condition~\ref{cond:convergence-of-moment-tensors}. The only change here is in the certifiable hypercontractivity constraints where instead of the RHS of the bounded variance constraint is stated in terms of $\Norm{Q}_F^2$ instead of $\Norm{\Pi Q \Pi}_F^2$ with an additional slack of $O(1/\alpha^2)$. This modified constraint is satisfied by all true clusters by an application of Lemma~\ref{lem:frob-of-product} since each of their covariance has spectral norm at most $1.5/\alpha$.   


\paragraph{Rounding} Let $M = \pE_{\tzeta} z z^{\top}$. Then, by an argument similar to the proof of Theorem~\ref{thm:partial-clustering-non-poly}, we can conclude:
\begin{enumerate}
\item $o_d(1) + \alpha \geq M(i,j) \geq 0$.
\item $\sum_{j = 1}^n M(i,j) \geq (\alpha^2-o_d(1)) n$ for every $i$.
\item For every $i$, let $B_i$ be the set of ``large entries'': i.e. $j$ such that $M(i,j) \geq \alpha^2/2$. Then, $|B_i| \geq \alpha n/2$.
\end{enumerate}

In the following, let $M_i$ denote the $i$-th row of $M$ and $\Norm{M_i}_1$ for the sum of the non-negative entries of the vector $M_i$. 

\paragraph{Candidate Clusters}
For every $i$, let $F_i \subseteq [k]$ be the set of all $i' \in [k]$ such that $\Norm{\Sigma_i - \Sigma_{i'}}_F^2 \geq \tau$ (i.e., $F_i$ is the set of indices of true clusters whose covariances are far from that of the $i$-th cluster in Frobenius norm). 
For every row $j \in [n]$, let $C(j) \in [k]$ be such that $j \in X_{C(j)}$ 
Let's call $j$-th row of $M$ ``good'' if $x_j = y_j$ (i.e $j$-th sample is not an outlier) and the following condition holds:
\[
\sum_{r \in F_{C(j)}} \sum_{\ell \in X_r: x_{\ell} = y_{\ell}} M(j,\ell) \leq \Norm{M_j}_1 \Paren{\frac{\beta}{\eta}} \mper
\]

Thus, by Markov's inequality, the fraction of non-outlier entries in $B_j$ that come from $X_{r'}$ such that $r' \in F_r$ is at most $2\Paren{\frac{\beta}{\eta \alpha^2}}$.

Let us estimate the fraction of good rows now. 
From Lemma~\ref{lem:simultaneous-intersection-bounds-poly-clustering} and Fact~\ref{fact:sos-completeness}, we have that for every $r$ and $r' \in F_r$:

\[
 \pE[ z'(X_r)z'(X_{r'})] \leq \beta \mper
\]
Here, recall that $z'(X_r) = \frac{1}{w_r n} \sum_{i \leq n} z_i \1(y_i = x_i)$ for every $r$.
Summing up over $r' \in F_r$ yields:
\[
\frac{1}{w_r n} \sum_{r' \in F_r} \sum_{i \in X_r: x_i = y_i} \sum_{j \in X_{r'}: x_j = y_j} \pE[z_i z_j] \leq  n\beta \mper
\]

Thus, by Markov's inequality, with probability at least $1-\eta$ over the choice of $i \in X_r$ such that $x_i = y_i$, it must hold that:
\[
\sum_{r' \in F_r} \sum_{j \in X_{r'}: x_j = y_j} \pE[z_i z_j] \leq  n \Paren{\frac{\beta}{\eta}} \mper
\]

Using that $(1-\epsilon/\alpha)$-fraction of $i \in X_r$ satisfy $x_i = y_i$, for every $r$, we conclude that $1-\eta-\epsilon/\alpha$-fraction of the rows $X_r$ are good. 

Thus, with probability at least $(1-\eta-\epsilon/\alpha)^{\ell}\geq (1- O(\ell( \eta + \epsilon/\alpha) ))$, every candidate cluster picked in Step 1 of our rounding algorithm corresponds to the large entries from a good row of $M$. 

We next claim that we cover most of the points in the input in the union of the candidate clusters:
\begin{equation}
|\cup_{i \leq \ell} \hat{C}_i| \geq \Paren{ 1-2\sqrt{\eta} - \frac{\epsilon}{\sqrt{\eta} \alpha}}n \label{eq:coverage}
\end{equation} with probability at least $1-\sqrt{\eta}$. To see why, let's estimate the chance that an element $j \in [n]$ does not appear in any of the $\hat{C}_i$s. First, we can assume that $j$-th row of $M$ is good (this loses us $\eta+ \epsilon/\alpha$-fraction $j$s). For each such $j$, there are at least $\alpha n/2$ large entries. Since $M$ is symmetric, the $j$-th column of $M$ also has $\alpha n/2$ large entries. Further, $j$ appears in $\cup_{i \leq \ell} \hat{C}_i$ if at least one of the $\alpha n/2$ large entries are chosen in our rounding. The chance that this does not happen in any of the $\ell$ picks is at most $(1-\alpha/2)^{\ell}$. Since $\ell =\Theta\Paren{ \frac{1}{\alpha} \log (1/\eta)}$, this chance is at most $O(\eta)$. Thus, in expectation $|[n] \setminus \cup_{i \leq \ell} \hat{C}_i| \leq O\Paren{  \eta + \epsilon/\alpha}n$. By Markov's inequality, with probability at least $1-\sqrt{\eta}$,  $|[n] \setminus \cup_{i \leq \ell} \hat{C}_i| \leq  O\Paren{ \sqrt{\eta} + \frac{\epsilon}{\sqrt{\eta} \alpha}}n$. 

By a union bound, with probability at least $1- O\Paren{ \eta \ell - \epsilon\ell/\alpha } - \sqrt{\eta} \geq  1- O\Paren{ \eta\log(1/\eta)/\alpha  - \sqrt{\eta} } $, we must thus have both the following events hold simultaneously:
\begin{equation}
|\cup_{i \leq \ell} \hat{C}_i| \geq \Paren{ 1-2\sqrt{\eta} - \frac{\epsilon}{\sqrt{\eta} \alpha}}n \geq (1-3\sqrt{\eta})n \label{eq:coverage}
\end{equation}
and, for every $1\leq i \leq \ell$, 
\begin{equation}
|\hat{C}_i \cap (\cup_{r' \in F_{C(r)}} X_{r'})| \leq 2\Paren{\frac{\beta}{\eta \alpha} + \epsilon/\alpha} \cdot |\hat{C}_i| \mper \label{eq:no-bad-intersection}
\end{equation}

\paragraph{Merging Candidate Clusters}

Observe, following the proof of Theorem \ref{thm:partial-clustering-poly}, we know that there exists a partition of $Y$ into sets $Y_1$ and $Y_2$ such that for all $i$, 
\[
\max \left\{ \frac{1}{w_i n}|Y_1 \cap X_i|, \frac{1}{w_i n}|Y_2 \cap X_i|\right\} \geq 1- O(\sqrt{\eta}) - O(\frac{\beta}{\eta \alpha^2}) ,
\]
and
\[
\max_{i}\frac{1}{w_i n} |X_i \cap Y_1|, \max_{i}\frac{1}{w_i n} |X_i \cap Y_2| \geq 1- O(\sqrt{\eta}) - O(\frac{\beta}{\eta \alpha^2}).
\]

Next, we show that the merging step preserves this partition.
For each $\hat{C}_i$, let $\hat{C}_i' = \hat{C}_i \cap \cup_{j \not \in F_{C(i)}} X_j$. That is, $\hat{C}_i'$ is the subset of $\hat{C}_i$ obtained by removing points from ``far-off'' clusters and the outliers. Then, since we know that $|\hat{C}_i| \geq \alpha n/2$ and $|X \cap Y| \geq (1-\epsilon)n$, we must have $|\hat{C}_i|-|\hat{C}_i'| = \eta'_i |\hat{C}_i| \leq \Paren{\frac{3\epsilon}{\alpha} +  \frac{2\beta}{\eta \alpha^2}} |\hat{C}_i| $, where we note that $\eta'_i \leq \Paren{\frac{3\epsilon}{\alpha} +  \frac{2\beta}{\eta \alpha^2}}$.

Thus, $\hat{C}_i'$ is a collection of $\geq (1-\eta'_i) \alpha n/2$ points from the submixture $\cup_{j \not \in F_{C(i)}} X_j$. We know that each $\mu_i$ is of $\ell_2$ norm at most $1/\alpha$, each $\Sigma_i$ has spectral norm at most $1/\alpha$ and that for every $r,r' \not \in F_{C(i)}$, $\Norm{\Sigma_r -\Sigma_{r'}}_F^2 \leq \tau$. Further, $\Sigma_r$ is at most $\tau+1/\alpha = O(\tau)$-different in Frobenius norm from the covariance of the sub-mixture. By an argument similar to the proof of Lemma~\ref{fact:mean-variance-subgaussian-arbitrary-covariance}, we can establish that the submixture with components $r$ such that $r \not \in F_{C(i)}$ is $O(\tau)$-certifiably bounded variance. Since $\hat{C}_i'$ is a subset of this sub-mixture of size $\alpha n/2$, we immediately obtain that $\hat{C}_i'$ is $O(\tau/\alpha)$-certifiably bounded variance. 
Thus, applying Lemma~\ref{lem:2nd-moment-est} with outlier parameter $\eta_i'$ to input $\hat{C}_i$ yields an estimate $S_i$ of the 2nd moment of $\hat{C}_i'$ within a Frobenius error of at most $O(\tau/\alpha)$ 
From Lemma~\ref{lem:subsamples-close-bounded-variance}, this is an additional $O(1/\alpha)$ different in Frobenius norm from the 2nd moment of the sub-mixture which, as argued above, is itself at most $O(\tau)$ different in Frobenius norm from $\Sigma_i$. Chaining together yields that $\Norm{\Sigma_i - S_i}_F^2 \leq O(\tau/\alpha)$ for some constant $C$. 

Since for every $r \in S, r' \in T$ it holds that $\Norm{\Sigma_r -\Sigma_{r'}}_F^2 \gg \Omega(\tau/\alpha)$, conditioned on the good event above, our algorithm never merges $\hat{C}_i$ and $\hat{C}_j$ whenever $i,j$ are non-outliers and $i$ is in some cluster in $S$ and $j$ is in some cluster in $T$. On the other hand, if $i,j$ belong to the same cluster, then, the corresponding estimate $\Norm{S_i-S_j}_F^2 \leq 2 C \tau$. Thus, our merging process always merges together any such candidates. 

As a result, the output of the merging process can have at most one $i$ from any true cluster -- thus, the number of distinct members of $\cL'$ is at most $k$.  We note that the running time is dominated by computing a pseudo distribution satisfying the union of all the constraints (Step 3 in Algorithm \ref{algo:robust-partial-clustering-upg}) and requires $n^{O(t)}$ time. Step 4 computes a degree $O(1)$ sos relaxation for at most $O(\ell)$ components and the merging only requires a fixed polynomial in $d$ and $k$ time.  


\end{proof}




\subsection{Proof of Lemma~\ref{lem:simultaneous-intersection-bounds-poly-clustering-upg}} 
\label{sec:simultaneous-intersection-bounds-poly-clustering-upg}

In the following lemma, we show that the constraint system $\cA$, via a low-degree sum-of-squares proof, implies that a lower bound on the variance of any degree $2$ polynomial on $X'$ whenever the cluster $\hat{C}$ (indicated by $z$) appreciably intersects two well-separated true clusters. 

\begin{lemma}[Lower-Bound on Variance of Degree 2 Polynomials] \label{lem:lower-bound-variance-upg}
Let $Q \in \R^{d \times d}$ be any fixed matrix. 
Then, for any $i,j \leq k$,  and $z'(X_r) = \frac{1}{w_r n} \sum_{i \in X_r} z_i \cdot \1(y_i = x_i)$, we have for any $r \neq r' \in [k]$, 
\begin{align*}
\cA \sststile{4t}{z} \Biggl\{z'(X_r) z'(X_r') \leq  \frac{(32Ct/\alpha)^{2t}}{(\E_{X_r} Q - \E_{X_{r'}}Q)^{2t}}  & \Biggl(\frac{\alpha^4}{w_r^2 w_{r'}^2}  \Paren{\E_z (Q-\E_z Q)^{2}}^{t} \\
&+ \frac{\alpha^2}{w_r^2} \Paren{\E_{X_{r'}} (Q-\E_{X_{r'}}Q)^{2}}^t + \frac{\alpha^2}{w_{r'}^2} \Paren{\E_{X_r}(Q-\E_{X_r}Q)^{2}}^t\Biggr)\Biggr\}\mper
\end{align*}
\end{lemma}
\begin{proof}
Let $z_i' = z_i \1(y_i = x_i)$ for every $i$. 
For every $1 \leq r \leq k$, let $\E_{X_r} Q$ denote the expectation of the homogenous degree $2$ polynomial defined by $Q$: $\E_{X_r}Q = \frac{1}{w_r n} \sum_{i,j \in X_r} Q(x_i)$ for every $r$ where $Q(x_i) = x_i^{\top}Qx_i$. Similarly, let $\E_z Q$ be the quadratic polynomial in $z$ defined by $\E_z Q = \frac{1}{\alpha n} \sum_{i \leq n} z_i Q(x_i)$. 
Using the substitution rule and non-negativity of the $z_i'$s, we have for any $r,r' \in [k]$:
\begin{equation}
\label{eqn:lowerbound_expec_Q-upg}
\begin{split}
\cA \sststile{4t}{z} \Biggl\{ \E_z (Q-\E_z Q)^{2t}  & = \frac{1}{ \alpha n} \sum_{i  \in [n]} z_i \Paren{Q(x_i)- \E_z Q}^{2t} \\
& \geq \frac{1}{ \alpha n} \sum_{i \in X_r\cup X_{r'}: x_i = y_j} z'_i \Paren{Q(x_i)- \E_z Q}^{2t}\Biggl\}
\end{split}
\end{equation}

Then, using the SoS almost triangle inequality (Fact~\ref{fact:almost-triangle-sos}),  we have:

\begin{equation}
\label{eqn:lowerbound_expec_Q_almost_tri--upg}
\begin{split}
\cA \sststile{4t}{z} \Biggl\{  & \frac{1}{ \alpha n } \sum_{i \in X_{r} \cup X_{r'}} z'_i \Paren{Q(x_i)- \E_z Q}^{2t} \\
& \geq 2^{-2t} \Paren{ \frac{1}{ \alpha n} \sum_{i \in X_r: i} z'_i \Paren{\E_{X_r}Q -\E_z Q}^{2t} - \frac{1}{ \alpha n} \sum_{i \in X_r} z'_i \Paren{Q(x_i)- \E_{X_r}Q}^{2t} } \\
& \hspace{0.2in} + 2^{-2t} \Paren{ \frac{1}{ \alpha n} \sum_{i \in  X_r: i:x_i=y_i } z'_i \Paren{\E_{X_{r'}}Q- \E_z Q}^{2t}-  \frac{1}{ \alpha n} \sum_{i \in  X_{r'}, x_i = y_i } z'_i \Paren{Q(x_i)- \E_{X_{r'}}Q}^{2t}}\\
&= 2^{-2t} \Paren{ \frac{w_r}{\alpha } z'(X_r) \Paren{\E_{X_r}Q -\E_z Q}^{2t} - \frac{1}{\alpha n} \sum_{i \in  X_r} \Paren{Q(x_i)- \E_{X_r}Q}^{2t}}\\
&\hspace{0.2in}+ 2^{-2t}\Paren{ \frac{w_{r'}}{\alpha} z'(X_{r'}) \Paren{\E_{X_{r'}}Q- \E_z Q}^{2t}- \frac{1}{\alpha n} \sum_{i \in  X_{r'} } \Paren{Q(x_i)- \E_{X_{r'}}Q}^{2t} }\Biggl\}
\end{split}
\end{equation} 

Next, observe that by the SoS almost triangle inequality (Fact~\ref{fact:almost-triangle-sos}), we must have:
\[
\cA \sststile{4t}{} \Set{ \Paren{\E_{X_r}Q -\E_z Q}^{2t} + \Paren{\E_{X_{r'}}Q -\E_z Q}^{2t}  \geq 2^{-2t} \Paren{\E_{X_r}Q -\E_{X_{r'}} Q}^{2t} }\mper
\] 

Further, note that $ \cA \sststile{O(1)}{} \Set{\frac{w_r}{\alpha} z'(X_r) +  \frac{w_{r'}}{\alpha} z'(X_{r'}) \leq \frac{1}{\alpha n} \sum_{i} z_i \leq 1}$. Thus, using  Fact~\ref{fact:lower-bounding-sums-SoS} with $A = \frac{w_r}{\alpha} z'(X_r)$, $B= \frac{w_{r'}}{\alpha} z'(X_{r'})$, $C=\Paren{\E_{X_r}Q -\E_z Q}^{2t}$, and $D =\Paren{\E_{X_{r'}}Q -\E_z Q}^{2t}$ and $\tau = 2^{-2t} \Paren{\E_{X_r}Q -\E_{X_{r'}} Q}^{2t}$, we can derive:

\begin{equation}
\label{eqn:lowerbound_expec_Q}
\begin{split}
\cA \sststile{4}{z} \Biggl\{ (Ct/\alpha)^{2t} \Paren{\E_z (Q-\E_z Q)^2}^t &\geq \E_z (Q-\E_z Q)^{2t}  
 \geq  2^{-6t} \frac{w_r w_{r'}}{\alpha^2} z'(X_{r})z'(X_{r'}) \Paren{ \E_{X_r}Q-\E_{X_{r'}}Q}^{2t}  \\& -  2^{-6t}\frac{w_r}{\alpha}  \E_{X_r} \Paren{ Q-\E_{X_r}Q}^{2t} - 2^{-6t}
 \frac{w_{r'}}{\alpha}  \E_{X_{r'}} \Paren{Q-\E_{X_{r'}}Q}^{2t} \\
 &\geq 2^{-6t} \frac{w_r w_{r'}}{\alpha^2} z'(X_{r})z'(X_{r'}) (\E_{X_r}Q-\E_{X_{r'}}Q)^{2t} - \frac{w_r}{\alpha} (Ct/\alpha)^{2t}\Paren{\E_{X_r} \Paren{ Q-\E_{X_r}Q}^{2}}^t\\
&- \frac{w_{r'}}{\alpha}  (Ct/\alpha)^{2t}\Paren{\E_{X_{r'}} \Paren{ Q-\E_{X_{r'}}Q}^{2}}^t \Biggl\}
\end{split}
\end{equation}
where the first inequality uses the Certifiable Hypercontractivity constraints ($\cA_4$) and the last inequality follows from the Certifiable Hypercontractivity of $X_r$ and $X_{r'}$ (Condition~\ref{cond:convergence-of-moment-tensors}).
Rearranging completes the proof.

\end{proof}

We can use the lemma above to obtain a simultaneous intersection bound guarantee when there are relative Frobenius separated components in the mixture. 
\begin{lemma}[Lemma~\ref{lem:simultaneous-intersection-bounds-poly-clustering}, restated] \label{lem:pair-rel-frob-upg}
Suppose $\Norm{\Sigma^{-1/2}(\Sigma_r-\Sigma_{r'})\Sigma^{-1/2}}_F^2 \geq 10^8 \frac{C^6 t^4}{\beta^{2/t} \alpha^4}$. Then, for $z(X_r) = \frac{1}{w_r n} \sum_{i \in X_r} z_i \cdot \1(y_i = x_i)$,  
\[
\cA \sststile{2t}{} \Set{z(X_r) z(X_{r'}) \leq \beta}\mper
\]
\end{lemma}
\begin{proof}
WLOG, we will work with the transformed points $x_i \rightarrow \Sigma^{-1/2}x_i$ where $\Sigma$ is the covariance of the mixture. Note that our algorithm does not need to know $\Sigma$ -- this transformation is only for simplifying notation in the analysis that follows. 

Let $\tilde{\Sigma}_z = \Sigma^{-1/2}\Sigma_z \Sigma^{-1/2}$, $\tilde{\Sigma}_r = \Sigma^{-1/2} \Sigma_r \Sigma^{-1/2}$ and $\tilde{\Sigma}_{r'} =\Sigma^{-1/2} \Sigma_{r'}\Sigma^{-1/2}$ be the transformed covariances. Then, notice that $\Norm{\tilde{\Sigma}_{r}}_2 \leq \frac{1}{w_r} \Norm{\Sigma}_2 \leq \frac{1.5}{w_r}$ and $\Norm{\tilde{\Sigma}_{r'}}_2 \leq \frac{1}{w_{r'}} \Norm{\Sigma}_2 \leq \frac{1.5}{w_{r'}}$. 

We now apply Lemma~\ref{lem:lower-bound-variance} with $Q = \tilde{\Sigma}_{r}-\tilde{\Sigma}_{r'}$. Then, notice that $\E_{X_r}Q - \E_{X_{r'}} Q = \Norm{\tilde{\Sigma}_r - \tilde{\Sigma}_{r'}}_F^2 + \mu_{r}^{\top} (\tilde{\Sigma}_r-\tilde{\Sigma_{r'}})\mu_{r} - \mu_{r'}^{\top} (\tilde{\Sigma}_r-\tilde{\Sigma_{r'}})\mu_{r'} \geq \Norm{Q}_F^2 - \frac{4}{\alpha}$. Then, we obtain:

\begin{multline}\label{eq:basic-bound-upg}
\cA \sststile{2t}{z} \Biggl\{z(X_r) z(X_{r'}) \\\leq \Paren{\frac{32Ct/\alpha}{\E_{X_r} Q - \E_{X_{r'}}Q}}^{2t}  \Paren{\frac{\alpha^2}{ w_r w_{r'}} \Paren{\E_z (Q-\E_z Q)^2}^t + \frac{\alpha}{w_r} \Paren{\E_{X_{r'}} (Q-\E_{X_{r'}}Q)^2}^t + \frac{\alpha}{w_{r'}} \Paren{\E_{X_r}(Q-\E_{X_r}Q)^2}}^t\Biggr\}\mper
\end{multline}

\nnnew{Since $X_{r}$ and $X_{r'}$ have certifiably $C$-bounded variance polynomials for $C = 4$ (as a consequence of Condition~\ref{cond:convergence-of-moment-tensors} and Fact~\ref{fact:moments-to-analytic-properties} followed by an application of Lemma~\ref{lem:frob-of-product}), we have:}
\[
\E_{X_{r'}} (Q-\E_{X_{r'}}Q)^2 \leq \nnnew{6\Norm{{\tilde{\Sigma}_{r'}}^{1/2}Q{\tilde{\Sigma}_{r'}}^{1/2}}_F^2} \leq \frac{10}{w_{r'}^2} \Norm{Q}_F^2 \leq \frac{10}{\alpha^2} \Norm{Q}_F^2\;,
\]
and
\[
\E_{X_{r}} (Q-\E_{X_{r}}Q)^2 \leq \nnnew{6\Norm{{\tilde{\Sigma}_{r}}^{1/2}Q{\tilde{\Sigma}_{r}}^{1/2}}_F^2} \leq \frac{10}{w_{r}^2} \Norm{Q}_F^2 \leq \frac{10}{\alpha^2} \Norm{Q}_F^2\mper
\]
Finally, using the bounded-variance constraints in $\cA$, we have: 
\[
\cA \sststile{4}{Q,z} \E (Q-\E_z Q)^2 \leq \frac{10}{\alpha^2} \Norm{Q}_F^2\mper
\]

Plugging these estimates back in \eqref{eq:basic-bound-upg} yields:
\begin{multline}
\cA \sststile{4}{z} \Biggl\{z(X_r) z(X_{r'}) \leq \frac{(1000Ct/\alpha)^{2t}}{\Norm{Q}_F^{2t} \alpha^{2t}} \Paren{\frac{\alpha^2}{ w_r  } + \frac{\alpha}{ w_{r'}  }  + \frac{\alpha}{ w_r w_{r'}}} \leq \frac{3}{w_r w_{r'}} \frac{(1000Ct)^{2t}}{\alpha^{2t} \Norm{Q}_F^{2t}}\leq \frac{3(1000Ct)^{2t}}{\alpha^{2t} \Norm{Q}_F^{2t}}\Biggr\} \mper
\end{multline}

Plugging in the lower bound on $\Norm{Q}_F^{2t}$ and applying cancellation within SoS (Fact~\ref{fact:cancellation-within-sos-constant-rhs}) completes the proof. 
\end{proof}

\subsection{2nd Moment Estimation Subroutine}
The following lemma gives a 2nd moment estimation algorithm with error in Frobenius norm for distributions that have a certifiably bounded covariance. The proof is very similar to the SoS based mean and covariance estimation algorithms but we provide it in full for completeness here.

\begin{lemma}[2nd Moment Estimation in Frobenius Norm] \label{lem:2nd-moment-est}
Let $1/100 \geq \eta >0$. There is an $n^{O(1)}$ time algorithm that takes input an $\eta$-corruption $Y$ of an sample $X$ of size $n$ and outputs an estimate $M_2$ of the 2nd moment of $X$ with the following properties:
Let $X \subseteq \R^d$ be a collection of $n$ points satisfying $\sststile{2}{Q} \Set{ \frac{1}{|X|} \sum_{x \in X} \Paren{ Q(x)-\frac{1}{|X|} Q(x)}^2 \leq C \Norm{Q}_F^2}$ for a matrix-valued indeterminate $Q$. Let $M_2 = \frac{1}{n} \sum_{x \in X} xx^{\top}$. Then, the estimate $\hat{M_2}$ output by the algorithm satisfies:
\[
\Norm{\hat{M_2}-M_2}_F^2 \leq 80C \eta \mper
\]

\end{lemma}

\begin{proof}
Consider the constraint system with scalar-valued indeterminates $z_i$ for $1 \leq i \leq n$ and $d$-dimensional vector-valued indeterminates $x'_1, x'_2, \ldots, x'_n$ with the following set of constraints:
$\cA =$ \begin{equation}
  \left \{
    \begin{aligned}
     &\forall i \leq n
     &z_i^2
     &=z_i\\
     &
     &\sum_{i = 1}^n z_i
     &=(1-\eta)n
     &\\
     &
     &\tilde{M}_2
     &=\frac{1}{n} \sum_{i =1}^n x'_i {x_i'}^{\top}\\
     &\forall i \leq n
     &z_i x_i'
     & = z_iy_i
     \\
     &
     &\frac{1}{n}\sum_{i=1}^n \Paren{{x'}_i^{\top}Qx'_i-\frac{1}{n}\sum_{i=1}^n {x'}_i^{\top}Qx'_i}^2 
     &\leq C\Norm{Q}_F^2
    \end{aligned}
  \right \}
\end{equation}
Observe that $X'=X$ and $z_i$ set to the $0$-$1$ indicator of non-outliers satisfies the constraint system. Thus, the constraints are feasible. 

Our algorithm finds a pseudo-distribution $\tzeta$ of degree $10$ satisfying the above constraints and output $\pE[\tilde{M}_2]$.  Let us now analyze this algorithm. 
The key is the following statement that gives a sum-of-squares proof of closeness of $\tilde{M}_2$ and $M_2$ in Frobenius norm. 
We use the  notation $\E_X Q$ and $\E_{X'}Q$ to abbreviate $\frac{1}{n} \sum_{i = 1}^n x_i^{\top}Qx_i$ and $\frac{1}{n} \sum_{i = 1}^n {x'}_i^{\top} Q x'_i$ respectively. 

\begin{align*}
\cA \sststile{2}{Q} \Biggl\{ & \Paren{\frac{1}{n} \sum_{i = 1}^n x_i^{\top}Qx_i - \frac{1}{n} \sum_{i =1}^n {x'}_i^{\top}Q x'_i}^2 \\
&= \Paren{\frac{1}{n} \sum_{i = 1}^n (1-z_i \1(x_i = y_i)) x_i^{\top}Qx_i  - {x'}_i^{\top}Q x'_i}^2 \\
&\leq \Paren{\frac{1}{n} (1-z_i \1(x_i=y_i))^2 } \Paren{\frac{1}{n} \sum_{i = 1}^n \Paren{ x_i^{\top}Qx_i-{x'}_i^{\top}Q x'_i}^2}\\
&\leq 20\eta \cdot \Paren{\frac{1}{n} \sum_{i = 1}^n\Paren{x_i^{\top}Qx_i-\E_X Q}^2 + \frac{1}{n} \sum_{i = 1}^n \Paren{ {x'}_i^{\top}Q x'_i-\E_{X'}Q}^2 + \Paren{ \E_X Q-\E_{X'}Q}^2 }\\
&\leq 20 \eta (2C \Norm{Q}_F^2) + 20 \eta \Paren{ \E_X Q-\E_{X'}Q}^2  \Biggr\}
\end{align*}
where the first inequality follows by the SoS version of the Cauchy-Schwarz inequality and the 2nd by the SoS version of the Almost Triangle inequality. 

Rearranging and using that $1-20\eta>1/2$ now yields that:
\[
\cA \sststile{2}{Q} \Biggl\{ \Paren{\frac{1}{n} \sum_{i = 1}^n x_i^{\top}Qx_i - \frac{1}{n} \sum_{i =1}^n {x'}_i^{\top}Q x'_i}^2 
\leq 80C \eta \Norm{Q}_F^2  \Biggr\}
\]
Notice that the LHS above equals the linear polynomial $\langle \tilde{M}_2-M_2, Q\rangle$.
We now plug in $Q = \tilde{M}_2-M_2$ to obtain:
\[
\cA \sststile{2}{Q} \Biggl\{ \Norm{\tilde{M}_2-M_2}_F^4 \leq  80C \eta \Norm{\tilde{M}_2-M_2}_F^2  \Biggr\}
\]
Applying Fact~\ref{fact:sos-cancellation} yields:
\[
\cA \sststile{2}{Q} \Biggl\{ \Norm{\tilde{M}_2-M_2}_F^8 \leq   80^4 C^4 \eta^4 \Biggr\}
\]

Taking pseudo-expectations with respect to $\tzeta$ and using Hölder's inequality for pseudo-distributions yields that
\[
\Norm{\pE_{\tzeta} \tilde{M}_2 - M_2}_F^8 \leq \pE_{\tzeta} \Norm{\tilde{M}_2-M_2}_F^8 \leq 80^4C^4 \eta^4 \mper 
\] 

Taking the $4$-th root, we can conlcude our rounded value $\hat{M}_2 = \pE_{\tzeta} \tilde{M}_2$ satisfies:
\[
\Norm{\hat{M}_2 - M_2}_F^2 \leq 80 C  \eta \mper 
\] 
This completes the proof.

\end{proof}

We also note the following simple consequence of the certifiable bounded variance property that follows via an argument similar to the one employed in the proof of the previous lemma. 
\begin{lemma}[Subsamples of Bounded-Variance Distributions] \label{lem:subsamples-close-bounded-variance}
Let $X \subseteq \R^d$ be a collection of $n$ points satisfying $\sststile{2}{Q} \Set{ \frac{1}{|X|} \sum_{x \in X} (Q(x)-\frac{1}{|X|} Q(x))^2 \leq C \Norm{Q}_F^2}$ for a matrix-valued indeterminate $Q$. Let $M_2= \frac{1}{|X|} \sum_{x \in X} xx^{\top}$ be the 2nd moment of $X$. Let $S \subseteq X$ be a subset of size at least $\beta |X|$. Then, 
\[
\sststile{4}{} \Set{\Norm{\frac{1}{|S|} \sum_{x \in S} xx^{\top} - M_2}_F^2 \leq \frac{1}{\beta} }\mper
\]
\end{lemma}
\begin{proof}
We have by the Cauchy-Schwarz inequality:
\begin{align*}
\sststile{2}{Q} \Biggl\{ \Paren{\frac{1}{|X|} \sum_{x \in X} \1(x \in S) (x^{\top} Qx - M_2), Q \rangle)}^2 &\leq \Paren{\frac{1}{|X|} \1(x \in S)^2} \Paren{\sum_{x \in X}  (x^{\top} Qx - M_2), Q \rangle)^2}\\
&\leq \Paren{\frac{|S|}{|X|}} \Norm{Q}_F^2\mper
\end{align*}
We now substitute in $Q =\frac{1}{|X|} \sum_{x \in X} \1(x \in S) (x^{\top} Qx -  M_2$ to obtain:
\begin{align*}
\sststile{2}{Q} \Biggl\{ \Norm{\frac{1}{|X|} \sum_{x \in X} \1(x \in S) (x^{\top} Qx - M_2)}_F^4 &\leq \Paren{\frac{|S|}{|X|}} \Norm{\frac{1}{|X|} \sum_{x \in X} \1(x \in S) (x^{\top} Qx - M_2)}_F^2 \mper
\end{align*}
We now apply Fact~\ref{fact:sos-cancellation} to yield:
\begin{align*}
\sststile{2}{Q} \Biggl\{ \Norm{\frac{1}{|X|} \sum_{x \in X} \1(x \in S) (x^{\top} Qx - M_2)}_F^8 &\leq \Paren{\frac{|S|}{|X|}}^4 \mper
\end{align*}
We finally apply Fact~\ref{fact:cancellation-within-sos-constant-rhs} to conclude that:
\begin{align*}
\sststile{2}{Q} \Biggl\{ \Norm{\frac{1}{|X|} \sum_{x \in X} \1(x \in S) (x^{\top} Qx - M_2)}_F^2 &\leq \Paren{\frac{|S|}{|X|}} \mper
\end{align*}
Rescaling gives the claim.
\end{proof}

%% file: poly-list-size.tex
\def\proj{\textrm{proj}}




\section{Getting $\poly(\epsilon)$-close in TV Distance: Proof of Theorem \ref{thm:poly-proper-inf} }
\label{sec:poly-proper}



\begin{theorem}[Robustly Learning $k$-Mixtures with small error] \label{thm:robust-GMM-arbitrary-poly-eps}
Given $0< \epsilon < 1/k^{k^{O(k^2)}}$ and a multiset $Y = \{ y_1, y_2, \ldots , y_n\}$ of $n$ i.i.d. 
samples from a distribution $F$ such that $\dtv(F, \calM) \leq \eps$,
for an unknown $k$-mixture of Gaussians $\calM = \sum_{i \leq k} w_i\calN\Paren{\mu_i, \Sigma_i}$, 
where \new{$n \geq n_0=d^{O(k)}\poly_k(1/\eps)$,} there exists an algorithm that runs 
\new{in time $n^{O(1)}\poly_k\left(1/\eps\right)$} 
and with probability at least $0.99$ outputs a hypothesis $k$-mixture of Gaussians 
$\widehat{\calM} = \sum_{i \leq k} \hat{w}_i \calN\Paren{\hat{\mu_i}, \hat{\Sigma}_i}$ 
such that $d_{\textrm{TV}}\Paren{\calM, \widehat{\calM}} = \bigO{\eps^{c_k}}$, 
with $c_k=1/(100^k C^{(k+1)!}k!\mathrm{sf}(k+1))$, where $C>0$ is a universal constant 
and $\mathrm{sf}(k)=\Pi_{i \in[k]}(k-i)!$ is the super-factorial function.  
\end{theorem}

In order to obtain the above theorem, we require recovering a polynomial sized list of candidate parameters, in addition to the efficient partial clustering result we obtained in the previous section.
To this end, we show the following list-recovery theorem which is similar to Theorem \ref{thm:list-recovery-by-tensor-decomposition}, but the algorithm outputs a polynomial-size list instead.

\begin{theorem}[Recovering a small list of candidate parameters]\label{thm:polynomial-size-list-recovery-by-tensor-decomposition}
Fix any $\alpha> \epsilon > 0, \Delta > 0$.
Let $X$, a sample from a $k$-mixture of Gaussians
$\calM = \sum_i w_i \cN(\mu_i,\Sigma_i)$ satisfying Condition~\ref{cond:convergence-of-moment-tensors}
with parameters \new{$\gamma = \eps d^{-8k}k^{-Ck}$, for $C$ a sufficiently large universal constant,
and $t=8k$, and let $Y$ be an $\eps$-corruption of $X$}. Let $X'$ be a set of $n'= O\Paren{\epsilon \eta /\Paren{ k^5\Paren{\Delta^4 + 1/\alpha^4 }} }^{-4c k}$ fresh samples from $\cal M$ and $Z$ be an $\epsilon$-corruptionn of $X'$. 
If $w_i\ge\alpha$, $\Norm{\mu_i}_2 \leq \frac{2}{\sqrt{\alpha}}$ and  $\Norm{\Sigma_i - I}_F \leq \Delta$ for every $i \in [k]$,
then\new{, given $k,Y, Z$ and $\eps$,} the algorithm outputs a list $L$ of at most
$\ell' = O\Paren{ \Paren{ k^5\Paren{\Delta^4 + 1/\alpha^4 }}^{4k}/ \eta^{4k} }$
candidate hypotheses (component means and covariances),
such that \new{with probability at least $99/100$ } 
there exist $\{\hat \mu_i, \hat \Sigma_i\}_{i \in [k]} \subseteq L$ satisfying
$\Norm{\mu_i - \hat \mu_i}_2 \leq \bigO{\frac{\Delta^{1/2}}{\alpha}} \eta^{G(k)}$ and
$\Norm{\Sigma_i - \hat\Sigma_i}_F \leq \bigO{k^4} \frac{\Delta^{1/2}}{\alpha} \eta^{G(k)}$, for all $i \in [k]$.
Here, $\eta = (2 k)^{4k} \bigO{1/\alpha+ \Delta}^{4k} \eps^{1/(k^{O(k^2)})}$ and $G(k) = \frac{1}{C^{k+1} (k+1)!}$.
The running time of the algorithm is $\poly(|L|, \new{|Y|, d^k})\cdot\poly_k(1/\epsilon)$.
\end{theorem}

\subsection{Proof of Theorem \ref{thm:polynomial-size-list-recovery-by-tensor-decomposition} }

We use the following notation and background from Moitra-Valiant~\cite{MoitraValiant:10}:

\begin{definition}[Statistically Learnable]
Given $\eps>0$, we call a mixture of Gaussians $\calM=\sum_iw_i\calN(\mu_i, \Sigma_i)$ $\eps$-statistically learnable if $\min_iw_i\ge\eps$ and $\min_{i\neq j}d_{TV}(\calN(\mu_i, \Sigma_i),\calN(\mu_j, \Sigma_j))\ge\eps$.
\end{definition}

\begin{definition}[Correct Subdivision]
Given a Gaussian mixture of $k$ Gaussians, $\calM=\sum_iw_iN(\mu_i,\Sigma_i)$ and a mixture of $k'\le k$ Gaussians $\hat \calM=\sum_i\hat w_iN(\hat\mu_i,\hat\Sigma_i)$, we call $\hat \calM$ an $\eps$-correct subdivision of $\calM$ if there is a function $\pi:[k]\to[k']$ that is onto and 
\begin{enumerate}
\item $\forall j\in[k'], \left|\sum_{i:\pi(i)=j}w_i-\hat w_j\right|\le\eps$
\item $\forall i\in[k], \norm{\mu_i-\hat\mu_{\pi(i)}}+\norm{\Sigma_i-\hat\Sigma_{\pi(i)}}_F\le\eps$.
\end{enumerate}
\end{definition}

\begin{theorem}[Theorem 8 in \cite{MoitraValiant:10}]\label{thm:partition-pursuit}
Given an $\eps$-statistically learnable Gaussian mixture $\calM$ in isotropic position, for some $\eps>0$, there exists an algorithm that requires $n = \poly(d/\eps)$ samples and runs in time $O(\poly_k(n))$ and with probability at least $99/100$ recovers an $\eps$-correct sub-division $\hat \calM$. Let the corresponding algorithm be referred to as PARTITION PURSUIT.
\end{theorem}

The algorithm has two steps: first run the first three steps of Algorithm 3.2 to get the list $L'$ of $\hat S$ and $V'_{\hat S}$; then apply the following proposition to learn the mixture in the subspace $V'_{\hat S}$. This proposition is a generalization of Theorem~\ref{thm:partition-pursuit} without the assumption that the total variation distance between each pair of components is at least $\eps$. The sample and time complexities has a worse, but still polynomial dependence on $\eps$. Note that although the algorithm in the proposition is non-robust, we can take a sample without noise with constant probability because the algorithm only requires a polynomial number of samples in $\eps$.

\begin{mdframed}
  \begin{algorithm}[Efficient List-Recovery of Candidate Parameters]
    \label{algo:efficient-list-recovery-tensor-decomposition}\mbox{}
    \begin{description}
    \item[Input:] An $\epsilon$-corruption $Y$ of a sample $X$ from a $k$-mixture of
    Gaussians $\calM = \sum_i w_i \cN(\mu_i,\Sigma_i)$. Let Z be an additional $\epsilon$-corrupted sample of size $n'$ from $\calM$. 
    \item[Requirements: ] The guarantees of the algorithm hold if the mixture parameters
    and the sample $X$ satisfy:
    \begin{enumerate}
        \item $w_i \geq \alpha$ for all $i \in [k]$,
        \item $\Norm{\mu_i}_2 \leq 2/\sqrt{\alpha}$ for all $i\in[k]$,
        \item $\Norm{\Sigma_i - I}_F \leq \Delta$ for all $i\in[k]$.
        \item $X$ satisfies Condition~\ref{cond:convergence-of-moment-tensors} with parameters \new{$(\gamma, t)$,
        where $\gamma = \eps d^{-8k}k^{-Ck}$, for $C$ a sufficiently large universal constant, and $t=8k$}.
        \item The number of fresh samples $n'= O\Paren{\epsilon \eta /\Paren{ k^5\Paren{\Delta^4 + 1/\alpha^4 }} }^{-4c k}$, for a fixed constant $c$. 
    \end{enumerate}

    \item[Parameters: ] $\eta = (2 k)^{4k} \bigO{1/\alpha+ \Delta}^{4k} \eps^{1/(k^{O(k^2)})}$, $D = C(k^4/(\alpha\sqrt{\eta}))$,
    $\delta = 2 \eta^{1/(C^{k+1} (k+1)!)} $, $\ell'= 100\log k\Paren{ \eta /\Paren{ k^5\Paren{\Delta^4 + 1/\alpha^4 }} }^{-4k}$,
for some sufficiently large absolute constant $C>0$, $\lambda  = 4\eta$, $\phi = 10(1+\Delta^2)/(\sqrt{\eta}\alpha^5)$,  $\eps_1 =O\Paren{\sqrt{\Delta}\delta^{1/4}/\alpha}$.
    \item[Output:] A list $L$ of hypotheses such that there exists at least one, $\{\hat{\mu}_i, \hat{\Sigma}_i\}_{i \leq k}\in L$, satisfying: $\Norm{\mu_i - \hat \mu_i}_2 \leq \bigO{\frac{\Delta^{1/2}}{\alpha}} \eta^{G(k)}$ and
$\Norm{\Sigma_i - \hat\Sigma_i}_F \leq \bigO{k^4} \frac{\Delta^{1/2}}{\alpha} \eta^{G(k)}$, where $G(k) = \frac{1}{C^{k+1} (k+1)!}$.

    \item[Operation:]\mbox{}
    \begin{enumerate}
    \item \textbf{Robust Estimation of Hermite Tensors}: For $m \in [4k]$, compute $\hat T_m$ such that $\max_{m \in [4k]} \Norm{\hat T_m - \expecf{}{h_m(\calM)} }_F \leq \eta$ using the robust mean estimation algorithm in Fact~\ref{fact:robust-mean-estimation}.

    \item \textbf{Random Collapsing of Two Modes of $\hat{T}_4$}: Let $L'$ be an empty list. Repeat $\ell'$ times:
    For $ j \in [4k]$, choose independent standard Gaussians in $\mathbb{R}^d$, denoted by $x^{(j)}, y^{(j)} \sim \calN(0, I)$,
    and uniform draws $a_1, a_2, \ldots, a_t$ from $[-D,D]$. Let $\hat S$ be a $d \times d$ matrix such that
    for all $r,s\in[d]$, $\hat S(r,s) = \sum_{j \in [4k]} a_j \hat T_4(r,s,x^{(j)}, y^{(j)}) =
    \sum_{j \in [4k] } a_j \sum_{g,h \in [d]} \hat T_4 (r,s,g,h) x^{(j)}(g) y^{(j)}(h)$.
    Add $\hat S$  to the list $L'$.

    \item \textbf{Construct Low-Dimensional Subspace}: Let $V$ be the span of all singular vectors of the natural $d\times d^{m-1}$ flattening of $\hat T_m$ with singular values $\geq \lambda$ for $m \leq 4k$. For each $\hat S \in L'$, let $V'_{\hat S}$ be the span of $V$ plus all the singular vectors of $\hat S$ with singular value larger than $\delta^{1/4}$.

    \item \textbf{Moitra-Valiant for Low-Dimensional Subspace}: Initialize $L$ to be the empty list. For each $\hat S \in L'$, let $\hat P = U U^\top$ be the orthogonal projection matrix onto the span of $V'_{\hat S}$, where $U \in \R^{d \times d}$ has orthonormal columns. Let $m = \dim{V'_{\hat S}}$ and  let $\hat Z \subset Z$ be a randomly chosen subset of size $\poly(m/\eps_1)$. 
    Let $U_m$ denote the first $m$ columns of $U$ and for all $z \in \hat Z$, compute $U_m^\top z$. Run PARTITION PURSUIT on the resulting set of points and let $\{\hat\mu^{\hat P}_i , \hat\Sigma^{\hat P}_i \}_{i \in [k]}$ be the parameters corresponding to the $\epsilon$-correct subdivision output by PARTITION PURSUIT. Let $\hat \mu_i^\top = [(\hat\mu^{\hat P}_i)^\top, 0 ]$ be a $d$ dimensional vector padded with $0$s and $\hat \Sigma_i $ be a $d \times d$ matrix with $\hat\Sigma^{\hat P}_i$ in the top left $m \times m$ sub-matrix and $0$'s elsewhere. Add $\{ U \hat \mu_i, U \hat \Sigma_i U^\top + (\hat S + I) - \hat P\Paren{I + S} \hat P \}_{i \in [k]}$ to $L$. 


    \end{enumerate}

    \end{description}
  \end{algorithm}
\end{mdframed}

\begin{proposition}\label{prop:mv}
Given $\eps>0$ and a sample $X$ of size $\poly(d,1/\eps)$ from a $k$-mixture of Gaussians $\calM$ with mixture covariance $\Sigma$ such that $0.99 I \preceq \Sigma\preceq 1.01 I$ 
and satisfies $w_i\ge\eps$, the PARTITION PURSUIT algorithm runs in time $\poly(d,1/\eps)$ and with probability at least $9/10$ 
returns an $O(\eps)$-correct sub-division, denoted by $\hat\calM$.
\end{proposition}

Recall, the PARTITION PURSUIT algorithm satisfies Theorem~\ref{thm:partition-pursuit} and we will prove that with an appropriately chosen parameter $\eps$, PARTITION PURSUIT also satisfies Proposition~\ref{prop:mv}. The main idea is that if any two components are actually close enough in total variation distance, then any algorithm with access to only a polynomial number of samples could never distinguish these two components from a single Gaussian. So if all pairwise distances are either sufficiently large or sufficiently small, the algorithm will behave as if it were given sample access to a mixture that meets the requirements of Theorem~\ref{thm:partition-pursuit}.

\begin{lemma}\label{lemma:merge}
Given $0<\gamma, \delta<1$ and two distributions, $\calD_1$ and $\calD_2$ over $\mathbb{R}^d$ such that $d_{TV}( \calD_1,  \calD_2 )<\gamma$, let $X_1$ be set of $n$ i.i.d. samples from $\calD_1$ and  $X_2$ be $n$ i.i.d. samples from $\calD_2$. Let $\calA$ be any algorithm that takes as input $X_1$ and outputs a list of $m$ real numbers, $Y_1 = \{y_i\}_{i \in [m]}$, such that $y_i \in [-1,1]$ with probability at least $1-\delta$. Then, for any $\tau>0$, $\calA$ on input $X_2$ outputs a list of $m$ real numbers $Y_2 = \{ y'_i \}_{i \in [m]}$ such that with probability at least $1- \delta - (4 m n \gamma/ \tau)$, for all $i\in[m]$, $\abs{y_i - y'_i}\leq \tau$. 
\end{lemma}
\begin{proof}
Let $\calU_1$ be the uniform distribution over $X_1$ and $\calU_2$ be the uniform distribution over $X_2$. Then, 
\begin{equation}
\begin{split}
d_{TV}\Paren{\calU_1 , \calU_2 } \leq \sqrt{2}  H^2\Paren{\calU_1 , \calU_2 } &= \sqrt{2} n H^2\Paren{ \calD_1, \calD_2} \\
& \leq \sqrt{2} n d_{TV}\Paren{\calD_1, \calD_2} \\
&\leq  \sqrt{2}  \gamma n
\end{split} 
\end{equation} 
Consider the family of functions $\calF$ that take as input $n$ samples and output a single bit in $\{ 0,1\}$. We know that for any function $f \in \calF$, the probability that $f(X_1) \neq f(X_2)$ is at most $\sqrt{2}  \gamma n$. Recall, the algorithm outputs $m$ real numbers in the range $[-1,1]$, which we can discretize into a grid $\Delta$ of length $\tau$. There are at most $2/\tau$ distinct grid points and for any $y_i \in [-1,1]$, there exists a point $z_i \in \Delta$ such that $\abs{y_i - z_i}\leq \tau$. Further, observe we can represent each $y_i$ using $2/\tau$ functions $f \in \calF$. Then, union bounding over the events that each of the $2/\tau$ functions output different bits, for each of the $m$ parameters, we have that with probabiltiy at least $1- \Paren{2\sqrt{2}  \gamma n m/\tau}$, any algorithm outputs a list $\{y_i' \}_{i' \in [m]}$ such that $\abs{y_i -y'_i}\leq \tau$. Finally, union bounding over the event that algorithm $\calA$ fails with probability $\delta$ yields the claim. 
\end{proof}

We then prove there is a gap $[f(d,\eps_1),\eps_1)$ between pairwise distances of components so that if we merge components within distance $f(d,\eps_1)$, the resulting mixture is $\eps_1$-statistically learnable.

\begin{lemma}\label{lemma:tvd-gap}
Let $f(d)(\eps)=f(d,\eps)$. There exists $\ell\in[k^2]$ such that for every pair of components, either $d_{TV}(\calN(\mu_i, \Sigma_i) , \calN(\mu_j, \Sigma_j))<(f(d))^\ell(\eps)$ or $d_{TV}(\calN(\mu_i, \Sigma_i) , \calN(\mu_j, \Sigma_j))\ge (f(d))^{\ell-1}(\eps)$. Moreover, the set of Gaussians with total variation distance at most $(f(d))^\ell(\eps)$ is an equivalence class.
\end{lemma}

\begin{proof}
We can see that intervals $\left\{\left[(f(d))^\ell(\eps),(f(d))^{\ell-1}(\eps)\right)\right\}_{\ell\in[k^2]}$ are disjoint. There are at most $k^2-1$ distinct values of $d_{TV}(\calN(\mu_i, \Sigma_i) , \calN(\mu_j, \Sigma_j))$. So there exists an interval $\left[(f(d))^\ell(\eps),(f(d))^{\ell-1}(\eps)\right)$ such that for every pair of components $\calN(\mu_i, \Sigma_i) , \calN(\mu_j, \Sigma_j)$, either $d_{TV}(\calN(\mu_i, \Sigma_i) , \calN(\mu_j, \Sigma_j))<(f(d))^\ell(\eps)$ or $d_{TV}(\calN(\mu_i, \Sigma_i) , \calN(\mu_j, \Sigma_j))\ge (f(d))^{\ell-1}(\eps)$.

Next, we show for any $\ell$, Gaussians with pair wise TV distance $(f(d)^\ell)(\epsilon)$ form an equivalence class. Consider component Gaussians $G_1,G_2$ and $G_3$ such that $G_1$ and $G_2$ are at total variation distance at most $(f(d))^\ell(\eps)$ and $G_2$ and $G_3$ are also at total variation distance at most $(f(d))^\ell(\eps)$.
\begin{equation*}
\begin{split}
d_{TV}(G_1,G_3)& \le d_{TV}(G_1,G_2)+d_{TV}(G_2,G_3)\\
& \le2(f(d))^\ell(\eps)\\
& \ll (f(d))^{\ell-1}(\eps)
\end{split}
\end{equation*}
and since there is no pair of Gaussians with total variation distance inside the interval $\left[(f(d))^\ell(\eps),(f(d))^{\ell-1}(\eps)\right)$, this implies $d_{TV}(G_1,G_3)\le (f(d))^\ell(\eps)$.
\end{proof}

We can now complete the proof of Proposition~\ref{prop:mv} :

\begin{proof}[Proof of Proposition~\ref{prop:mv}]
By Lemma~\ref{lemma:tvd-gap}, there exists an interval $\left[(f(d))^\ell(\eps),(f(d))^{\ell-1}(\eps)\right)$ such that there is no pair of Gaussians with total variation distance inside the interval and $(f(d))^\ell(\eps),(f(d))^{\ell-1}(\eps)$ are polynomials in $d$ and $\eps$. Let $\eps_1=(f(d))^{\ell-1}(\eps)$ and $f(d,\eps_1)=(f(d))^\ell(\eps)$. Let $X$ be a set of $n=(d/\epsilon)^c$ samples from $\calM$, where $c$ is fixed universal constant. Let $\bar{\calM}$ be the mixture obtained by merging all components in an equivalence class with total variation distance at most $f(d,\eps_1)$ to a single Gaussian and observe $d_{TV}\Paren{\calM, \bar{\calM} }\leq k f(d,\eps_1)$. Next, observe that PARTITION PURSUIT outputs at most $k$ means and covariances, which can be represented as a list of at most $2 k d^2$ real numbers. Further, since $\Sigma \preceq 1.01 I$ and $w_i \geq \eps$, the means of each component $\|\mu_i\|_2^2 \leq 2/\eps$ and $\|\Sigma_i \|_F^2 \leq O(d^2 /\eps)$.  

Then, rescaling the instance by $O(\eps/d^2)$ and applying Lemma~\ref{lemma:merge} with $\calD_1 = \calM$, $\calD_2 =\bar{\calM}$, input samples $X$ and accuracy parameter $\tau = (\eps/d)^{c_2}$, for a large enough constant $c_2$, it follows that with probability at least
$1 - 0.99 - O( f(d,\eps_1) \cdot(\eps/d)^{c_3})$, for a fixed constant $c_3$, the resulting list of numbers is $\tau$-close to that obtained by running PARTITION PURSUIT on a set of $n$ samples from $\bar{\calM}$. Since $\bar{\calM}$ is $\eps_1$-statistically learnable, it follows from Theorem~\ref{thm:partition-pursuit} that with probability at least $9/10$, PARTITION PURSUIT will output an $O(\eps_1)$-correct sub-division $\hat\calM$.

\end{proof}

\begin{proof}[Proof of Theorem~\ref{thm:polynomial-size-list-recovery-by-tensor-decomposition}]


Recall, by part (1) of Proposition \ref{prop:construction-of-low-dim-subspace-for-enumeration}, the dimension of the subspace $V'_{\hat S}$ is $m = \dim{V'_{\hat S}}= O\Paren{ \frac{\Paren{ k (1 + \Delta + 1/\alpha) }^{4k +5 }}{\eta^2}}$ and let $\eps_1 = \sqrt{\Delta} \delta^{1/4} /\alpha$. Let $c_{mv}$ be a fixed constant such that $\Paren{ m/\eps_1}^{c_{mv}}$ samples suffice for applying Theorem \ref{thm:partition-pursuit}. Further, obseve in the fresh sample $Y$, the probability that any given sample is corrupted is $\epsilon$. Let $\zeta$ be the event that a random subset of $\Paren{ m/\eps_1}^{c_{mv}}$ samples from $Z$ does not contain any corrupted points. Then, the event $\zeta$ holds with probability at least $\Paren{ 1-\epsilon}^{\Paren{ m/\eps_1}^{c_{mv}}}$. Conditioning on $\zeta$ and running step 4 of Algorithm \ref{algo:efficient-list-recovery-tensor-decomposition}, it follows from Proposition~\ref{prop:mv} that we recover $O(\eps_1)$-accurate estimates to the parameters of $\calM$ in the subspace, i.e. $ \norm{ U^\top \mu_i - \hat \mu_i }_2 \leq O(\eps_1)$ and $\norm{ U^\top \Sigma_i U - \hat \Sigma_i  }_F \leq O(\eps_1)$. Since we repeat the above for $\ell'$ candidate subspaces in $L'$, the probability over all probability of success is $\Paren{ 1-\epsilon}^{\Paren{ m/\eps_1}^{c_{mv}  }\cdot \ell'} $. 

By part (2) in Proposition  \ref{prop:construction-of-low-dim-subspace-for-enumeration}, there is a vector $\mu_i'\in V'_{\hat S}$ such that $\norm{\mu_i-\mu_i'} \le \frac{20}{\alpha}\delta^{1/4}\Delta^{1/2}$ where $\delta=2\eta^{1/(C^{k+1}(k+1)!)}$ and $\eta=O(4k(1+1/\alpha+\Delta)^{4k}\sqrt{\eps_1})$. Let $\hat P = U U^\top$ be a projection matrix where the columns of $Q$ span $V'_{\hat S}$ and let $Q^\top \mu_i$ be the projection of the true means to the corresponding subspace. Then, 
\begin{equation*}
\begin{split}
\norm{ U \hat \mu_i  -\mu_i}_2 & \leq  \norm{ U \hat \mu_i  - \mu_i' }_2 +  \norm{  \mu_i' - \mu_i }_2\\
& \leq  \norm{ U \hat \mu_i  - P(\mu_i'-\mu_i + \mu_i) }_2 + O\Paren{\frac{\sqrt{\Delta}\delta^{1/4}}{\alpha} }\\
& \leq  \norm{ \hat \mu_i  -  U^\top\mu_i }_2 + O\Paren{\frac{\sqrt{\Delta}\delta^{1/4}}{\alpha} } \\
& \leq  O\Paren{\frac{\sqrt{\Delta}\delta^{1/4}}{\alpha} }.
\end{split} 
\end{equation*} 
where the third inequality follows from observing that $U^\top\mu_i$ is the true mean in the low dimensional subspace and applying Proposition \ref{prop:mv}.

By Proposition \ref{prop:tensor_decomposition_upto_lowrank}, there exists $\hat S\in L'$ such that $\hat S-(\Sigma_i -I)=P_i+Q_i$ where 
$\norm{P_i}_F=O(\sqrt{\eta/\alpha})$. 
Again by part (3) in Proposition  \ref{prop:construction-of-low-dim-subspace-for-enumeration}, there exists a symmetric matrix $Q_i'\in V_{\hat S}'\times V_{\hat S}'$ such that $\norm{Q_i-Q_i'}_F\le O(\frac{k^2}{\alpha}\delta^{1/4}\Delta^{1/2})$. We also know that in the subspace spanned by $V_{\hat S}'$, $\norm{ \hat{\Sigma}_i - U^\top \Sigma_i U  }_F^2 \leq \poly(\eps_2)$. Recall, Algorithm \ref{algo:efficient-list-recovery-tensor-decomposition} outputs the following estimate: $\hat M = U \hat \Sigma U^\top + \Paren{I +\hat S} - \hat P\Paren{I+\hat S} \hat P$.  Observe, for any matrix $M$ and projection matrix $P$, $M = PMP + (I-P)M(I-P) + PM(I-P)+ (I-P)MP$. Then,

\begin{equation}
\label{eqn:bounding_estimated_error_cov}
\begin{split}
\norm{\Sigma_i- \hat M  }_F 
& \leq   \underbrace{ \norm{ \hat P \Paren{ \Sigma_i- \hat M} \hat P }_F}_{(1)} + \underbrace{ \norm{ \Paren{I- \hat P} \Paren{ \Sigma_i-  \hat M } \Paren{I- \hat P}   }_F }_{(2)} \\
& \hspace{0.2in} + \underbrace{ \norm{\hat P\Paren{\Sigma_i- \hat M}\Paren{I- \hat P} }_F}_{(3)}  + \underbrace{ \norm{\Paren{I -\hat P} \Paren{\Sigma_i- \hat M}\hat P  }_F}_{(4)}\\
\end{split}
\end{equation}
We bound each of the terms above. Since $\hat P \Paren{ I+ S - P(I+S)P } \hat P = 0$, we can bound term (1) as follows
\begin{equation}
\label{eqn:term1_cov}
\norm{ \hat P \Paren{ \Sigma_i- \hat M} \hat P }_F =\norm{ \hat P\Sigma_i \hat P -    \hat P U \hat\Sigma_i U^\top  \hat P  }_F = \norm{ U^\top \Sigma_i U  - \hat\Sigma_i } \leq \ O\Paren{\frac{\sqrt{\Delta}\delta^{1/4}}{\alpha} }
\end{equation}
Similarly, since $\Paren{I- \hat P} \Paren{ U \hat \Sigma_i  U^\top} \Paren{I- \hat P}  = 0 $ and $\Sigma_i = I + \hat S - P_i - Q_i$, we can bound term (2) as follows:
\begin{equation}
\label{eqn:term2_cov}
\begin{split}
\norm{ \Paren{I- \hat P} \Paren{ \Sigma_i-  \hat M } \Paren{I- \hat P}   }_F & =\norm{ \Paren{I- \hat P} \Paren{ \Sigma_i-  \Paren{I + \hat S } } \Paren{I- \hat P}   }_F \\
& \leq \norm{ \Paren{I- \hat P} \Paren{ \Sigma_i-  \Paren{I + \hat S - Q_i } } \Paren{I- \hat P}   }_F + \norm{ \Paren{I- \hat P} Q_i \Paren{I- \hat P}   }_F \\
& \leq \norm{ P_i }_F^2 +  \norm{ \Paren{I- \hat P} \Paren{ Q_i - Q_i' } \Paren{I- \hat P}   }_F +  \norm{ \Paren{I- \hat P} Q_i' \Paren{I- \hat P}   }_F\\
& \leq O\Paren{ \sqrt{\frac{\eta}{\alpha} }  + \frac{k^2 \delta^{1/4} \Delta^{1/2}}{\alpha}} 
\end{split}
\end{equation}
Next, we bound term (3). Observe, $ \hat P \Paren{ U \hat \Sigma_i  U^\top} \Paren{I- \hat P}  = 0 $ and $\hat P (I  + S )\hat P (I - \hat P) =0$. Thus, 
\begin{equation}
\label{eqn:term3_cov}
\begin{split}
\norm{\hat P   \Paren{\Sigma_i- \hat M} \Paren{I -\hat P} }_F & = \norm{ \hat P\Paren{ \Sigma_i - \Paren{I + \hat S} } \Paren{I -\hat P} }_F \\
& = \norm{ \hat P\Paren{ P_i + Q_i } \Paren{I -\hat P} }_F\\
& \leq \norm{ \hat P P_i  \Paren{I -\hat P} }_F \\
& \leq O\Paren{ \sqrt{\frac{\eta}{\alpha}}}
\end{split}
\end{equation}
Obseve, term (4) follows from a similar argument. Combining equations \eqref{eqn:term1_cov}, \eqref{eqn:term2_cov},\eqref{eqn:term3_cov} and substituting back into \eqref{eqn:bounding_estimated_error_cov} we can conclude 
\begin{equation*}
\norm{\Sigma_i - \hat M }_F \leq O\Paren{ \sqrt{\frac{\eta}{\alpha} }  + \frac{k^2 \delta^{1/4} \Delta^{1/2}}{\alpha}} 
\end{equation*}




The size of $L'$ is $\ell' = O\Paren{\log k\Paren{ \eta /\Paren{ k^5\Paren{\Delta^4 + 1/\alpha^4 }} }^{-4k}}$ and since we add a single tuple of k means and covariances for each subspace in $L'$, the list $L$ has the same size. The running time is $\poly\Paren{|Y|, |L|, d^k, m, 1/\eps_1}$ concluding the proof. 
\end{proof}

\subsection{Proof of Theorem~\ref{thm:robust-GMM-arbitrary-poly-eps}}
Since we have all the main ingredients: the tensor decomposition algorithm recovering a polynomial size of list (Theorem~\ref{thm:polynomial-size-list-recovery-by-tensor-decomposition}), the upgraded partial clustering algorithm with high probability of success (Theorem~\ref{thm:partial-clustering-poly-upg}) and the spectral separation algorithm of thin components (Lemma~\ref{cor:recursion}), we can now complete the proof of Theorem~\ref{thm:robust-GMM-arbitrary-poly-eps}.

The algorithm establishing Theorem~\ref{thm:robust-GMM-arbitrary-poly-eps} is almost the same as Algorithm~\ref{algo:robust-GMM-arbitrary}. The only difference is we will replace Algorithm~\ref{algo:robust-partial-clustering} by Algorithm~\ref{algo:robust-partial-clustering-upg} and replace Algorithm~\ref{algo:list-recovery-tensor-decomposition}  by Algorithm~\ref{algo:efficient-list-recovery-tensor-decomposition}. The following two lemmas show that by modifying the parameters slightly and applying the upgraded partial clustering and tensor decomposition algorithms, we can have the same conclusions as in Lemma~\ref{lem:covariance_sep} and Lemma~\ref{lem:list_decoding} with a polynomial success probability. Then the proof of Theorem~\ref{thm:robust-GMM-arbitrary-poly-eps} is exactly the same as the proof of Theorem~\ref{thm:robust-GMM-arbitrary} in Section~\ref{sec:proof-main-thm} except for the use of Lemma~\ref{lem:covariance_sep-poly-eps} and Lemma~\ref{lem:list_decoding-poly-eps} instead of Lemma~\ref{lem:covariance_sep} and Lemma~\ref{lem:list_decoding}.

\begin{lemma}[Non-negligible Weight and Covariance Separation] \label{lem:covariance_sep-poly-eps}
Given $0<\epsilon<  1/k^{k^{O(k^2)}}$ and $k \in \mathbb{N}$, 
let $\alpha=\epsilon^{1/(45C^{k+1}(k+1)!)}$. 

Let $\calM=\sum_{i=1}^k w_i G_i$ with $G_i = \cN(\mu_i,\Sigma_i)$ be a $k$-mixture of Gaussians 
with mixture covariance $\Sigma$ such that $w_i \geq \alpha$ for all $i \in [k]$
 and there exist $i, j\in[k]$ such that  
$\Norm{ \Sigma^{\dagger/2}\Paren{\Sigma_i - \Sigma_j}\Sigma^{\dagger/2} }_F^2 > 1/\alpha^5$. Further, let $X$ be a set of points satisfying Condition \ref{cond:convergence-of-moment-tensors} with respect to $\calM$ 
for some parameters $\gamma\leq \eps d^{-8k}k^{-Ck}$, for a sufficiently large constant $C$, and $t\geq 8k$. 
Let $Y$ be an $\epsilon$-corrupted version of $X$ of size $n\ge n_0= \Paren{ dk }^{\Omega(1)}/\epsilon$, 
Algorithm~\ref{algo:robust-partial-clustering-upg} partitions $Y$ into $Y_1$, $Y_2$ in time $n^{O(1)}$ such that 
with probability at least $2^{-O(k)}(1-O(\alpha))$
there is a non-trivial partition of $[k]$ into $Q_1\cup Q_2$ so that letting $\calM_j$ 
be a distribution proportional to $\sum_{i\in Q_j}w_iG_i$ and $W_j=\sum_{i\in Q_j}w_i$, then
$Y_j$ is an $\bigO{\epsilon^{1/(45C^{k+1}(k+1)!)}}$-corrupted version of $\bigcup_{i\in Q_j} X_i$ 
satisfying Condition~\ref{cond:convergence-of-moment-tensors} with respect to $\calM$ 
with parameters $(\bigO{k\gamma/W_j},t)$.

\end{lemma}

\begin{proof}
We run Algorithm \ref{algo:robust-partial-clustering-upg} with sample set $Y$, number of components $k$, 
the fraction of outliers $\epsilon$ and the accuracy parameter $\eta$. 
Since $X$ satisfies Condition~\ref{cond:convergence-of-moment-tensors}, 
we can set $t=10$, $\beta=(k^2t^4\alpha)^{t/2}=O_k(\alpha^5)$ and $\eta=\alpha^2 \gg \sqrt{\eps/\alpha}$ in Theorem \ref{thm:partial-clustering-poly-upg}.
Then, by assumption, there exist $i,j$ such that 
$$
\Norm{\Sigma^{\dagger/2}\Paren{ \Sigma_i - \Sigma_j }\Sigma^{\dagger/2} }_F^2 > \frac{1}{\alpha^5}=\Omega\left(\frac{k^2t^4}{\beta^{2/t}\alpha^4}\right) \;.
$$
We observe that we also satisfy the other preconditions for Theorem \ref{thm:partial-clustering-poly-upg}, 
since $n \geq \Paren{ dk/ }^{\Omega(1)}/\epsilon$. 

Then, Theorem \ref{thm:partial-clustering-poly-upg} implies that with probability at least $2^{-O(k)}(1-O(\eta/\alpha-\sqrt{\eta}))=2^{-O(k)}(1-O(\alpha))$, 
the set $Y$ is partitioned in two sets $Y_1$ and $Y_2$ such that 
there is a non-trivial partition of $[k]$ into $Q_1\cup Q_2$ so that letting $\calM_j$ be a distribution proportional 
to $\sum_{i\in Q_j}w_iG_i$ and $W_j =\sum_{i\in Q_j}w_i$, then
$Y_j$ is an $\bigO{\epsilon^{1/(45C^{k+1}(k+1)!)}}$-corrupted version of $\bigcup_{i\in Q_j} X_i$. 
By Lemma \ref{submixture condition lemma}, $\bigcup_{i\in Q_j} X_i$ 
satisfies Condition~\ref{cond:convergence-of-moment-tensors} with respect to $\calM$ with parameters $(\bigO{k\gamma/W_j},t)$.

\end{proof}

When the mixture is not covariance separated and nearly isotropic, we can obtain a small list of hypotheses such that one of them is 
close to the true parameters, via tensor decomposition.

\begin{lemma}[Mixture is List-decodable] \label{lem:list_decoding-poly-eps}
Given $0<\epsilon<  1/k^{k^{O(k^2)}}$ let $\alpha=\epsilon^{1/(45C^{k+1}(k+1)!)}$. 
Let $\calM=\sum_{i=1}^k w_i G_i$ with $G_i = \cN(\mu_i,\Sigma_i)$ 
be a $k$-mixture of Gaussians with mixture mean $\mu$ and mixture covariance $\Sigma$,
such that $\Norm{\mu}_2 \leq \bigO{\sqrt{\epsilon/\alpha}}$,
$\norm{\Sigma-I}_F\le\bigO{\sqrt{\eps}/\alpha}$,
$w_i \geq \alpha$ for all $i \in [k]$, 
and $ \Norm{\Sigma_i - \Sigma_j }_F^2 \leq 1/\alpha^5$ for any pair of components,
and let $X$ be a set of points satisfying Condition \ref{cond:convergence-of-moment-tensors} 
with respect to $\calM$ for some parameters $\gamma= \eps d^{-8k}k^{-Ck}$, for a sufficiently large constant $C$, 
and $t= 8k$. Let $Y$ be an $\epsilon$-corrupted version of $X$ of size $n$,
Algorithm~\ref{algo:efficient-list-recovery-tensor-decomposition} 
outputs a list $L$ of hypotheses of size $O((1/\eps)^{4k^2})$ in time $\poly(|L|,n)$ 
such that if we choose a hypothesis 
$\{\hat{\mu}_i, \hat{\Sigma}_i \}_{i \in [k]}$ uniformly at random,
$\Norm{\mu_i - \hat{\mu}_i}_2 \leq\bigO{\epsilon^{1/(20C^{k+1}(k+1)!)}}$ and 
$\Norm{\Sigma_i - \hat{\Sigma}_i }_F \leq \bigO{\epsilon^{1/(20C^{k+1}(k+1)!)}}$ for all $i$ with probability at least $O(\eps^{4k^2})$.

\end{lemma}
\begin{proof}
Recall we run Algorithm \ref{algo:efficient-list-recovery-tensor-decomposition} on the samples $Y$, the number of clusters $k$, the fraction of outliers $\eps$ and the minimum weight $\alpha=\epsilon^{1/(20C^{k+1}(k+1)!)}$. Next, we show that the preconditions of Theorem \ref{thm:polynomial-size-list-recovery-by-tensor-decomposition} are satisfied. First, the upper bounds on $\norm{\mu}_2$ and $\norm{\Sigma-I}_F$ imply $\sum_{i \in k} w_i \Paren{\Sigma_i  + \mu_i \mu_i^\top} = \Sigma+\mu\mu^\top \preceq (1+\bigO{\sqrt{\eps}/\alpha}) I$. Since the LHS is a conic combination of PSD matrices, it follows that  for all $i \in [k]$, $\mu_i \mu_i^\top \preceq \frac{1}{\alpha}\left(1+\bigO{\sqrt{\eps}/\alpha}\right) I$, 
and thus $\Norm{\mu_i \mu_i^\top}_F \leq \frac{2}{\alpha}$. Next, we can write:
\begin{align*}
\Norm{\Sigma_i  - I }_F 
&\le \norm{\Sigma_i-(\Sigma+\mu\mu^\top)}_F+\norm{\Sigma-I}_F+\norm{\mu\mu^\top}_F\\
&= \Norm{\Sigma_i  - \sum_{j \in [k]} w_j \Paren{\Sigma_j + \mu_j\mu_j^\top} }_F + \frac{\sqrt{\eps}k}{\alpha}+\frac{\eps}{\alpha}\\
&\leq \Norm{\sum_{j \in [k]} w_j \Paren{ \Sigma_i  -  \Sigma_j} }_F + \frac{2}{\alpha}+ \frac{\sqrt{\eps}k}{\alpha}+\frac{\eps}{\alpha} \\
&\leq  \frac{2}{\alpha^{5/2}} \;,
\end{align*}
where the first and the third inequalities follow from the triangle inequality and the upper bound on $\Norm{\mu_i \mu_i^\top}_F$, 
and the last inequality follows from the assumption that $\Norm{\Sigma_i - \Sigma_j }_F^2 \leq 1/\alpha^5$ 
for every pair of covariances $\Sigma_i,\Sigma_j$. So, we can set $\Delta=2\alpha^{-5/2}$ 
in Theorem~\ref{thm:polynomial-size-list-recovery-by-tensor-decomposition}. 
Then, given the definition of $\alpha$, we have that
$$
\eta=2k^{4k}\bigO{1+\Delta/\alpha}^{4k}\sqrt{\eps} =\bigO{\epsilon^{2/5}}
$$
and $1/\eps^2\ge \log(1/\eta)(k+1/\alpha+\Delta)^{4k+5}/\eta^2$.
Therefore, Algorithm \ref{algo:efficient-list-recovery-tensor-decomposition} outputs a list $L$ of hypotheses 
such that $\abs{L} = \exp\left(1/\eps^2\right)$, and with probability at least $0.99$, 
$L$ contains a hypothesis that satisfies the following: for all $i \in [k]$, 
\begin{equation}
\label{eqn:true_hypothesis-poly-eps}
\begin{split}
\Norm{ \hat{\mu}_i - \mu_i }_2 & 
=\bigO{\frac{\Delta^{1/2}}{\alpha}}\eta^{G(k)}
=\bigO{\epsilon^{-1/(20C^{k+1}(k+1)!)}\cdot\eps^{1/(10C^{k+1}(k+1)!)}}
=\bigO{\epsilon^{1/(20C^{k+1}(k+1)!)}}  \textrm{ and }\\
 \Norm{ \hat{\Sigma}_i - \Sigma_i }_F & =\bigO{k^4}\frac{\Delta^{1/2}}{\alpha}\eta^{G(k)}
=\bigO{\epsilon^{1/(20C^{k+1}(k+1)!)}} \;.
\end{split}
\end{equation}

Then if we choose a hypothesis in $L$ uniformly at random, the probability that we choose the hypothesis satisfying (\ref{eqn:true_hypothesis}) is at least $1/|L|=\exp\left(-1/\eps^2\right)$.
\end{proof}


%% file: parameter_recovery.tex
\newcommand{\He}[1]{}

\section{Robust Parameter Recovery: Proof of Theorem \ref{thm:param-inf}}
\label{sec:param-recovery}

In order to show that our algorithm recovers the individual components and the parameters, we will prove the following identifiability theorem. Without any assumption on the mixtures, it is impossible to distinguish components within $\eps$ total variation distance with $\eps$-fraction of noise. So given two mixtures of Gaussians with $\eps$ total variation distance, the theorem shows that there exist two partitions of components of the two mixtures respectively such that any two components in the matched pair is are $\poly(\eps)$-close in total variation distance.

\begin{theorem}[Identifiability]\label{thm:parameter-recovery}
Let $\calM=\sum_{i=1}^{k_1}w_iG_i,\calM'=\sum_{i=1}^{k_2}w_i'G_i'$ be two mixtures of Gaussians such that $\dtv(\calM,
\calM')\le\epsilon$. Then there exists a partition of $[k_1]$ into sets $R_0,R_1,\dots,R_\ell$ and a partition of $[k_2]$ into sets $S_0,S_1,\dots,S_\ell$ such that 
\begin{enumerate}
\item 
Let $W_i=\sum_{j\in R_i}w_j$ for $i=0,1,\dots,k_1$,
$W_i'=\sum_{j\in S_i}w_j'$ for $i=0,1,\dots,k_2$. 
Then for all $i\in[\ell]$,
\begin{align*}
|W_i-W_{i}'|&\le\poly_k(\epsilon)\\
\dtv(G_j,G_{j'}')&\le\poly_k(\epsilon)\quad \forall j\in R_i,j'\in S_{i}
\end{align*}
\item
$W_0,W_0'\le\poly_k(\epsilon)$.
\end{enumerate}
\end{theorem}

\begin{corollary} \label{cor:parameter-estimation-main-technical}
There is an algorithm with the following behavior:
Given $\eps>0$ and a multiset of $n = d^{O(k)} \poly(\eps)$ samples from a distribution $F$ on $\R^d$ such that
$\dtv(F, \calM) \leq \eps$, for an unknown target $k$-GMM $\calM = \sum_{i=1}^k w_i \mathcal{N}(\mu_i, \Sigma_i)$, 
the algorithm runs in time $d^{O(k)} \poly_k(1/\eps)$ and outputs a $k'$-GMM hypothesis 
$\widehat{\calM}  = \sum_{i=1}^{k'} \widehat{w}_i \mathcal{N}(\widehat{\mu}_i, \widehat{\Sigma}_i)$ with $k'\le k$
such that with high probability
there exists a partition of $[k]$ into $k'+1$ sets $R_0,R_1,\dots,R_{k'}$ such that
\begin{enumerate}
\item 
Let $W_i=\sum_{j\in R_i}w_j$. 
Then for all $i\in[k']$,
\begin{align*}
|W_i-\hat w_{i}|&\le\poly_k(\epsilon)\\
\dtv(\mathcal{N}(\mu_j, \Sigma_j),\mathcal{N}(\widehat{\mu}_i, \widehat{\Sigma}_i))&\le\poly_k(\epsilon)\quad \forall j\in R_i
\end{align*}
\item
The sum of weights of exceptional components in $R_0$ is at most $\poly_k(\epsilon)$.
\end{enumerate}
\end{corollary}

Parameter estimation is implied by TV distance for individual Gaussians (in relative Frobenius norm). 
The corollary follows immediately from the identifiability theorem.

\paragraph{Outline of Proof.} The first step is to deal with the components in $\calM$ and $\calM'$ with small weights. We will construct $\tilde\calM,\tilde\calM'$ by removing components with small weights. If we prove the statement on $\tilde\calM,\tilde\calM'$, we can then deduce the theorem in the general case with worse, but still polynomial dependencies on $\eps$. The second step is a partial clustering, after which the components within each cluster have TV distance bounded by $1-\mbox{poly}(\eps)$. We prove this lemma in a separate section.
After that we modify the parameters slightly so that the resulting parameters for different components are either identical or have a minimum separation. After this, we can use a lemma from \cite{LM20} that provides a 1-1 mapping between the components of two such mixtures with small TV distance such that the mapped pairs have small TV distance.

\paragraph{Distance between Gaussians.} We use the following facts for Gaussian distributions.

\begin{lemma}[Frobenius Distance to TV Distance]\label{lemma:parameter-dtv}
Suppose $N(\mu_1,\Sigma_1),N(\mu_2,\Sigma_2)$ are Gaussians with $\norm{\mu_1-\mu_2}_2\le\delta$ and $\norm{\Sigma_1-\Sigma_2}_F\le\delta$. If the eigenvalues of $\Sigma_1$ and $\Sigma_2$ are at least $\lambda>0$, then $\dtv(N(\mu_1,\Sigma_1),N(\mu_2,\Sigma_2))=O(\delta/\lambda)$.
\end{lemma}

\begin{lemma}[Lemma 5.4 in {\cite{LM20}}]
\label{lemma:eigenvalue-lower-bound}
Let $\calM$ be a mixture of $k$ Gaussians that is connected if we draw edges between all components $i,j$ in $\calM$ such that $\dtv(G_i,G_j)\le 1-\delta$. Let $\Sigma$ be the covariance matrix of $\calM$. Then for any components $\Sigma_i$ of the mixture
\begin{enumerate}
\item $\Sigma_i\succeq\poly_k(\delta)\Sigma$
\item $\norm{\Sigma^{-1/2}(\Sigma-\Sigma_i)\Sigma^{-1/2}}_F\le\poly_k(\delta)^{-1}$.
\end{enumerate} 
\end{lemma}
The proof is identical to Lemma 5.4 in \cite{LM20}. The only difference is that in \cite{LM20} the authors assume that the minimal weight of $\calM$ is at least $\delta$ and TV distance between any pair of components is at least $\delta$ but here we do not need these two assumptions, which does not affect the proof.

\begin{fact}[Claim 3.9 in {\cite{LM20}}]
\label{fact:partial-same-1}
Let $\partial$ denote the differential operator with respect to $y$. If 
\[
f(y)=P(y,X)\exp\left(a(X)y+\frac{1}{2}b(X)y^2\right)
\]
where $P$ is a polynomial in $y$ of degree $k$ (whose coefficients are polynomials in $X$) and $a(X),b(X)$ are polynomials in $X$ then
\[
(\partial-(a(X)+yb(X)))f(y)=Q(y,X)\exp\left(a(X)y+\frac{1}{2}b(X)y^2\right)
\]
where $Q$ is a polynomial in $y$ with degree exactly $k-1$ whose leading coefficient is $k$ times the leading coefficient of $P$.
\end{fact}

\begin{fact}[Corollary 3.10 in {\cite{LM20}}]
\label{fact:partial-same-k+1}
Let $\partial$ denote the differential operator with respect to $y$. If 
\[
f(y)=P(y,X)\exp\left(a(X)y+\frac{1}{2}b(X)y^2\right)
\]
where $P$ is a polynomial in $y$ of degree $k$ then
\[
(\partial-(a(X)+yb(X)))^{k+1}f(y)=0.
\]
\end{fact}

\begin{fact}[Claim 3.11 in {\cite{LM20}}]
\label{fact:partial-different}
Let $\partial$ denote the differential operator with respect to $y$. If 
\[
f(y)=P(y,X)\exp\left(a(X)y+\frac{1}{2}b(X)y^2\right)
\]
where $P$ is a polynomial in $y$ of degree $k$. Let the leading coefficient of $P$ (viewed as a polynomial in $y$) be $L(X)$. Let $c(X)$ be a linear polynomial in $X$ and $d(X)$ be a quadratic polynomial in $X$ such that $\{a(X),b(X)\}\neq\{c(X),d(X)\}$. If $b(X)\neq d(X)$ then
\[
(\partial-(c(X)+yd(X)))^{k'}f(y)=Q(y,X)\exp\left(a(X)y+\frac{1}{2}b(X)y^2\right)
\]
where $Q$ is a polynomial of degree $k+k'$ in $y$ with leading coefficient 
\[
L(X)(b(X)-d(X))^{k'}
\]
and if $b(X)=d(X)$ then 
\[
(\partial-(c(X)+yd(X)))^{k'}f(y)=Q(y,X)\exp\left(a(X)y+\frac{1}{2}b(X)y^2\right)
\]
where $Q$ is a polynomial of degree $k$ in $y$ with leading coefficient 
\[
L(X)(a(X)-c(X))^{k'}.
\]
\end{fact}

\begin{lemma}\label{lemma:weight-gap}
Let $\calM=\sum_{i=1}^{k_1}w_iG_i,\calM'=\sum_{i=1}^{k_2}w_i'G_i'$ be two mixtures of Gaussians such that $\dtv(\calM,
\calM')\le\epsilon$.
For any constant $0<c_1<1$, there exists $i\in[k_1+k_2+1]$ such that $w_j,w_{j'}'\notin[\eps^{c_1^{i-1}},\eps^{c_1^{i}})$ for any $j\in[k_1],j'\in[k_2]$. 
Moreover, if 
\begin{align*}
\tilde \calM=\frac{\sum_{\{j:w_j\ge\eps^{c_1^{i}}\}}w_jG_j}{\sum_{\{j:w_j\ge\eps^{c_1^{i}}\}}w_j}\\
\tilde \calM'=\frac{\sum_{\{j:w_j'\ge\eps^{c_1^{i}}\}}w_j'G_j'}{\sum_{\{j:w_j'\ge\eps^{c_1^{i}}\}}w_j'}
\end{align*}
then $\dtv(\tilde \calM,\tilde\calM')\le O_k(\eps^{c_1^{i-1}})$.
\end{lemma}

\begin{proof}
We can see that $[\eps^{c_1^{i-1}},\eps^{c_1^{i}})$ with $i\in[k_1+k_2+1]$ are $k_1+k_2+1$ disjoint intervals and 
$w_j,w_{j'}'$ with $j\in[k_1],j'\in[k_2]$ have at most $k_1+k_2$ distinct values. So there is one interval containing no weights.

We then construct $\tilde \calM$ by removing the small components in $\calM$. The sum of weights removed is at most $k\eps^{c_1^{i-1}}$. So 
$\dtv(\calM,\tilde\calM)\le k\eps^{c_1^{i-1}}$. Similarly, we have $\dtv(\calM',\tilde\calM')\le k\eps^{c_1^{i-1}}$. By the triangle inequality,
$$
\dtv(\tilde \calM,\tilde\calM')\le \dtv(\calM,\calM')+\dtv(\calM,\tilde\calM)+\dtv(\calM',\tilde\calM')\le O_k(\eps^{c_1^{i-1}}).
$$
\end{proof}

Lemma~\ref{lemma:weight-gap} shows that we can remove components with tiny weights in the mixtures.
So in the following lemma, we will assume $M$ and $M'$ are Gaussian mixtures with minimal weights at least $\poly(\eps)$.
We will show that we can partition the union of components of two mixtures so that if we prove Theorem~\ref{thm:parameter-recovery} for each part of the partition, we can combine them to prove Theorem~\ref{thm:parameter-recovery} on the full mixtures.

\begin{lemma}\label{lemma:partition}
For any constant $0<c_3<1$, there exist $c_1,c_2>0$ that depend on $k$ and $c_3$, such that
if $\calM=\sum_{i=1}^{k_1}w_iG_i,\calM'=\sum_{i=1}^{k_2}w_i'G_i'$ with $k_1,k_2\le k$,
$\dtv(\calM,\calM')\le\eps$ and $w_i,w_i'\ge\eps^{c_1}$ for all $i$, then there exists a partition of $[k_1]$ into sets $R_1,\dots,R_\ell$ and a partition of $[k_2]$ into sets $S_1,\dots,S_\ell$ such that
\begin{enumerate}
\item 
For all $i\in[\ell]$, let $W_i=\sum_{j\in R_i}w_j,
W_i'=\sum_{j\in S_i}w_j'$ be the sum of weights in each piece.
Let $\calM_i=\frac{1}{W_i}\sum_{j\in R_i}w_jG_j,
\calM_i'=\frac{1}{W_i'}\sum_{j\in S_i}w_j'G_j'$ be the submixtures of Gaussians after partition.
Then for all $i\in[\ell]$,
\begin{align*}
|W_i-W_{i}'|&\le\poly_k(\epsilon)\\
\dtv(\calM_i,\calM_{i}')&\le O_k(\eps^{c_2})
\end{align*}
\item
Consider the graph with vertices corresponding to components in $\calM$ and $\calM'$ and two components are adjacent if the total variation distance between them is at most $1-\eps^{c_2c_3}$. Then the induced subgraph of vertices with indices $R_i\cup S_i$ is connected for all $i\in[\ell]$.
\end{enumerate}
\end{lemma}
The proof of Lemma~\ref{lemma:partition} is deferred to Section~\ref{sec:partial-clustering}.
In the following two lemmas, we then prove Theorem~\ref{thm:parameter-recovery} for each pair $\calM_i,\calM_{i}'$ defined in Lemma~\ref{lemma:partition}. In Lemma~\ref{lemma:merge-parameters}, we  construct two mixtures of which pairs of parameters are identical or separated. \Anote{cannot parse this sentence.} We also shows it suffices to work under this simplification.

\begin{lemma}\label{lemma:merge-parameters}
For any constant $0<c_4<1$, there exist $c_3,c_5$ that depend on $k$ and $c_4$, such that if $\calM=\sum_{i=1}^{k_1}w_iG_i,\calM'=\sum_{i=1}^{k_2}w_i'G_i'$ with $k_1,k_2\le k$ and
\begin{enumerate}
\item $\frac{1}{2}\calM+\frac{1}{2}\calM'$ is isotropic,
\item $\dtv(\calM,\calM')\le\eps$,
\item $w_i,w_i'\ge\eps^{c_3}$ for all $i$, 
\item Let $\mathcal{G}$ be a graph with components $G_i,G_i'$ in $\calM$ and $\calM'$ as vertex set and two components are adjacent if the total variation distance between them is at most $1-\eps^{c_3}$. Then $\mathcal{G}$ is connected
\end{enumerate}
then there exist two mixtures of Gaussians $\tilde\calM=\sum_{i=1}^{\tilde k_1}\tilde w_i\tilde G_i,\tilde\calM'=\sum_{i=1}^{\tilde k_2}\tilde w_i'\tilde G_i'$ such that 
\begin{enumerate}
\item \label{cond:identical-separated-mean}
Any pair in $\{\tilde\mu_i\}\cup\{\tilde\mu_i'\}$ is either identical or separated by at least $\eps^{c_4c_5}$ 
\item \label{cond:identical-separated-cov}
Any pair in $\{\tilde\Sigma_i\}\cup\{\tilde\Sigma_i'\}$ is either identical or separated by at least $\eps^{c_4c_5}$ in Frobenius norm. 
\He{It is not accurate to write the pair as $\tilde\Sigma_i,\tilde\Sigma_i'$ or $\tilde\Sigma_i,\tilde\Sigma_j'$.}
\item \label{cond:ht-difference}
$\norm{\E(h_m(\tilde\calM))-\E(h_m(\tilde\calM'))}_{F}\le O_k(\eps^{c_5})$ for any $m\le O(k)$
\item There exist $\pi_1:[k_1]\to[\tilde k_1]$ and $\pi_2:[k_2]\to[\tilde k_2]$ such that 
\begin{align*}
&\sum_{i:\pi_1(i)=j}w_i=\tilde w_{j}, \sum_{i:\pi_2(i)=j}w_i'=\tilde w_{j}', \\
&\dtv(G_i,\tilde G_{\pi_1(i)})\le\poly_k(\eps),\quad \text{for all }i\in[k_1]\\
&\dtv(G_i',\tilde G_{\pi_2(i)}')\le\poly_k(\eps),\quad \text{for all }i\in[k_2].
\end{align*}
\end{enumerate}
\end{lemma}

\begin{proof}
For any $0<c_4<1$, there is $\ell\in[k^2]$ such that the distance between any pair of parameters in $\{\mu_i\}\cup\{\mu_i'\}$ or the Frobenius distance between any pair in $\{\Sigma_i\}\cup\{\Sigma_i'\}$ is not in the interval $[\eps^{(c_4/2)^{\ell-1}},\eps^{(c_4/2)^\ell})$. 

Now consider a graph $\mathcal{G}$ on $k_1+k_2$ nodes where each node represents a vector in $\{\mu_i\}\cup\{\mu_i'\}$ and two vectors $a,b$ are adjacent if 
$$
\norm{a-b}\le\eps^{(c_4/2)^{\ell-1}}.
$$
We now construct new mixtures $\tilde \calM,\tilde\calM'$. For each connected component in $\mathcal{G}$ say $\{\mu_{i_1},\dots,\mu_{j_1}',\dots\}$, pick a representative say $\mu_{i_1}$ and set $\tilde\mu_{i_1}=\cdots=\tilde\mu_{j_1}'=\cdots=\mu_{i_1}$. Do this for all connected components and similar in the graph on covariance matrices with edges $(i,j)$ if 
$$\norm{\Sigma_i-\Sigma_j}_F\le\eps^{(c_4/2)^{\ell-1}}.$$
 After replacing close parameters with a representative, we may get some exactly same components in each new mixture. We then merge components with same means and covariances by adding their weights. 
Since all representatives of means and covariances are in different connected components of the graphs, they are separated by at least $\eps^{(c_4/2)^\ell}$. Setting $c_5=1/2(c_4/2)^{\ell-1}$ gives a separation of $\eps^{c_4c_5}$.

Next we prove \ref{cond:ht-difference}. There is a natural mapping $\pi_1:[k_1]\to[\tilde k_1]$ that maps any component in $\calM$ to the merged component in $\tilde\calM$ and a similar mapping $\pi_2:[k_2]\to[\tilde k_2]$ for $\calM',\tilde\calM'$. For all $i$, we have
\begin{equation}\label{eqn:parameter-distance}
\norm{\tilde\mu_{\pi_1(i)}-\mu_i},\norm{\tilde\mu_{\pi_2(i)}'-\mu_i'},
\norm{\tilde\Sigma_{\pi_1(i)}-\Sigma_i}_F,\norm{\tilde\Sigma_{\pi_2(i)}'-\Sigma_i'}_F\le O_k(1)\eps^{(c_4/2)^{\ell-1}}
\end{equation}
because for any pair of parameters above say $\tilde\mu_{\pi_1(i)}$ and $\mu_i$, there is a path of length at most $2k$ connecting $\mu_i$ to the representative of the connected component, and each edge connects a pair with TV distance at most $\eps$. Suppose $\norm{\mu_i},\norm{\Sigma_i-I}_F\le\Delta$. Then by Definition~\ref{def:degree_m_hermite_tensor}, we have for any integer $m$,
$$
\norm{\E(h_m(\calM))-\E(h_m(\tilde\calM))}_{F}\le O_k(m)\Delta^{m}\eps^{(c_4/2)^{\ell-1}}.
$$
Since $\frac{1}{2}\calM+\frac{1}{2}\calM'$ is isotropic and the minimum weight in $\frac{1}{2}\calM+\frac{1}{2}\calM'$ is at least $\frac{1}{2}\eps^{c_3}$, we have $\norm{\mu_i}\le \sqrt{2/\eps^{c_3}}$ for all $i$. Applying Lemma~\ref{lemma:eigenvalue-lower-bound} to $\frac{1}{2}\calM+\frac{1}{2}\calM'$, we have $\norm{I-\Sigma_i}_F\le\poly_k(\eps^{c_3})^{-1}$. 
So there is a constant $a$ such that $\Delta\le \eps^{-a c_3}$. 
If we take $c_3>0$ so that $a c_3O(k)\le1/2(c_4/2)^{\ell-1}$ and take $c_5=1/2(c_4/2)^{\ell-1}$,
then
$$\norm{\E(h_m(\calM))-\E(h_m(\tilde\calM))}_{F}\le O_k(m)\eps^{(c_4/2)^{\ell-1}-O(m)ac_3}=O_k(\eps^{c_5})$$
for $m\le O(k)$. By the same argument, we have the similar inequality for $\calM'$ and $\tilde\calM'$ 
$$
\norm{\E(h_m(\calM'))-\E(h_m(\tilde\calM'))}_{F}=O_k(\eps^{c_5}).
$$
Since we can use Proposition 3.3 to robustly estimate the Hermite tensors of a Gaussian mixture with $\eps$-fraction of noise and $\poly(\eps)$ error guarantee, we must have
$$
\norm{\E(h_m(\calM))-\E(h_m(\calM'))}_F\le\poly_k(\eps).
$$
Then by the triangle inequality, 
\begin{multline*}
\norm{\E(h_m(\tilde\calM))-\E(h_m(\tilde\calM'))}_F
\le\norm{\E(h_m(\calM))-\E(h_m(\tilde\calM))}_{F}+\\
\norm{\E(h_m(\calM))-\E(h_m(\calM'))}_F+\norm{\E(h_m(\calM'))-\E(h_m(\tilde\calM'))}_{F}
=O(\eps^{c_5}).
\end{multline*}
 
For the last conclusion, from the definition of $\pi_1$ and $\pi_2$, we know that 
$$
\sum_{i:\pi_1(i)=j}w_i=\tilde w_{j},\sum_{i:\pi_2(i)=j}w_i'=\tilde w_{j}'.
$$
Applying Lemma~\ref{lemma:eigenvalue-lower-bound} to $\frac{1}{2}\calM+\frac{1}{2}\calM'$, we have that eigenvalues of $\Sigma_i$ and $\Sigma_i'$ are at least $\poly(\eps^{c_3})$ for all $i$.
Then if $c_3$ is sufficiently small, by Lemma~\ref{lemma:parameter-dtv}, (\ref{eqn:parameter-distance}) implies $\dtv(G_i,\tilde G_{\pi(i)})\le\poly_k(\eps)$ and $\dtv(G_i',\tilde G_{\pi(i)}')\le\poly_k(\eps)$ for all $i$. \
\end{proof}

The following lemma shows the identifiability under the simplification of Lemma~\ref{lemma:merge-parameters}. It is proved in the proof of Lemma~8.2 in \cite{LM20}.

\begin{lemma}\label{lemma:main}
Suppose $\calM=\sum_{i=1}^{k_1}w_iG_i,\calM'=\sum_{i=1}^{k_2}w_i'G_i'$ satisfies \ref{cond:identical-separated-mean},\ref{cond:identical-separated-cov},\ref{cond:ht-difference} in the conclusion of Lemma~\ref{lemma:merge-parameters} with constants $c_4,c_5$ and the minimal weights are at least $\eps^{c_3}$. There exists a sufficiently small function $f(k)>0$ depending only on $k$ such that if $c_4\le f(k)$, then $k_1=k_2$ and there exists a permutation $\pi$ such that 
$|w_i-w_{\pi(i)}'|\le\poly_k(\eps)$ and $G_i=G_{\pi(i)}'$. 
\end{lemma}

\begin{proof}
Consider the component $G_{k_2}'=N(\mu_{k_2}',\Sigma_{k_2}')$ in $\calM'$. We claim that there must be some $i\in[k_1]$ such that 
\[
(\mu_i,\Sigma_i)=(\mu_{k_2}',\Sigma_{k_2}').
\]
Assume for the sake of contradiction that this is not the case. Let $S_1=\{i\in[k_1]:\Sigma_i=\Sigma_{k_2}'\}$ and $S_2=\{i\in[k_2-1]:\Sigma_i'=\Sigma_{k_2}'\}$.
Suppose $F,F'$ are the generating functions of $\calM$ and $\calM'$
\begin{align*}
F=\sum_{i=1}^{k_1}w_i\exp\left(\mu_i^T X+\frac{1}{2}X^T\Sigma_iXy^2\right)=\sum_{m=0}^{\infty}\frac{1}{m!}h_m(\calM)y^n\\
F'=\sum_{i=1}^{k_2}w_i'\exp\left({\mu_i'}^T X+\frac{1}{2}X^T\Sigma_i'Xy^2\right)=\sum_{m=0}^{\infty}\frac{1}{m!}h_m(\calM')y^n.
\end{align*}
Then define the differential operators
\begin{align*}
\calD_i=\partial-\mu_i^T X-X^T\Sigma_iXy\\
\calD_i'=\partial-{\mu_i'}^T X-X^T\Sigma_i'Xy
\end{align*}
where partial derivatives are taken with respect to $y$. Now consider the differential operator
\[
\calD=(\calD_{k_2-1}')^{2^{k_1+k_2-2}}\cdots(\calD_1')^{2^{k_1}}\calD_{k_1}^{2^{k_1-1}}\cdots \calD_1
\]
By Fact~\ref{fact:partial-same-k+1}, $\calD(F)=0$. By Fact~\ref{fact:partial-same-k+1} and Fact~\ref{fact:partial-different}, we have
\[
\calD(F')=P(y,X)\exp\left({\mu_{k_2}'}^T X+\frac{1}{2}X^T\Sigma_{k_2}'Xy^2\right)
\]
where $P$ is a polynomial of degree 
\[
\deg(P)=2^{k_1+k_2-1}-1-\sum_{i\in S_1}2^{i-1}-\sum_{i\in S_2}2^{k_1+i-2}
\]
with leading coefficient
\begin{multline*}
C_0=w_{k_2}'\prod_{i\in[k_1]\setminus S_1}(X^T(\Sigma_{k_2}'-\Sigma_i)X)^{2^{i-1}}\prod_{i\in S_1}((\mu_{k_2}'-\mu_i)^TX)^{2^{i-1}}\\
\prod_{i\in[k_2-1]\setminus S_2}(X^T(\Sigma_{k_2}'-\Sigma_i')X)^{2^{k_1+i-2}}\prod_{i\in S_2}((\mu_{k_2}'-\mu_i')^TX)^{2^{k_1+i-2}}.
\end{multline*} 
We now compare the following differentials evaluated at $y=0$
\begin{align*}
(\calD_{k_2}')^{\deg(P)}\calD(F)\\
(\calD_{k_2}')^{\deg(P)}\calD(F')
\end{align*}
The first quantity is 0 because $\calD(F)$ is identically 0 as a formal power series. The second one is $\Omega_k(1)C_0$.
Since for any $i$ $(\mu_i,\Sigma_i)\neq(\mu_{k_2}',\Sigma_{k_2}')$, our assumptions imply that the separation between $\mu_i,\mu_{k_2}'$ or $\Sigma_i,\Sigma_{k_2}'$ is at least $\eps^{c_4c_5}$.
Then we have $C_0\ge\eps^{c_4c_5 O_k(1)}$ for some $X$. On the other hand, the coefficients of the formal power series $F,F'$ are the Hermite polynomials $h_m(\calM)$ and $h_m(\calM')$. This is a contradiction with our assumption that 
\[
\norm{\E(h_m(\calM)-\E(h_m(\calM)}_F\le O_k(\eps^{c_5})
\]
as long as $c_4$ is smaller than some sufficiently small function $f(k)$ depending only on $k$. Thus there must be some component of $\calM$ that matches $G_{k_2}'=N(\mu_{k_2}',\Sigma_{k_2}')$. We can repeat the argument for each component in $\calM'$ and in $\calM$ to conclude that $\calM$ and $\calM'$ have the same components.

Next we will show that the weights of the same components in $\calM$ and $\calM'$ are close. We can assume that $\calM=\sum_{i=1}^{k}w_iG_i,\calM'=\sum_{i=1}^{k}w_i'G_i$ are two mixtures on the same set of components. Without loss of generality, 
\[
w_1-w_1'\le\cdots w_\ell-w_\ell'\le 0 \le w_{\ell+1}-w_{\ell+1}'\le\cdots \le w_k-w_k'.
\]
Then we can consider the following two mixtures
\begin{align*}
&(w_1-w_1')G_1+\cdots+(w_\ell-w_\ell')G_\ell\\
&(w_{\ell+1}-w_{\ell+1}')G_{\ell+1}+\cdots+(w_k-w_k')G_k.
\end{align*}
If 
\[
\sum_{i=1}^k|w_i-w_i'|>\eps^\zeta
\]
for some sufficiently small $\zeta$ depending only on $k$, we can then normalize each of the above into a distribution and repeat the same argument, using the fact that pairs of components cannot be too close, to obtain a contradiction. Thus, the mixing weights of $\calM$ and $\calM'$ are $\poly_k(\eps)$-close and this completes the proof.
\end{proof}

\begin{proof}[Proof of Theorem~\ref{thm:parameter-recovery}]
We first set $c_4=f(k)$ as in Lemma~\ref{lemma:main}, and then $c_3,c_5$ according to $c_4$ as in Lemma~\ref{lemma:merge-parameters}, and $c_1',c_2$ according to $c_3$ as in Lemma~\ref{lemma:partition}. Let $c_1=\min\{c_1',c_2c_3\}$. 

By Lemma~\ref{lemma:weight-gap}, we can find $i$ such that there is no $w_j,w_j'$ in $[\eps^{c_1^{i-1}},\eps^{c_1^{i}})$. 
Let $\tilde \calM=\sum_{\{j:w_j\ge\eps^{c_1^{i}}\}}w_jG_j$ and $\tilde \calM'=\sum_{\{j:w_j'\ge\eps^{c_1^{i}}\}}w_j'G_j'$. Then $\dtv(\tilde \calM,\tilde\calM')\le O(\eps^{c_1^{i-1}})$. Let $\eps_1=\eps^{c_1^{i-1}}$. We have $\dtv(\tilde \calM,\tilde\calM')\le O(\eps_1)$ and the minimum weights of $\tilde\calM,\tilde\calM'$ are at least $\eps_1^{c_1}$.

Now we can apply Lemma~\ref{lemma:partition} on $\tilde\calM,\tilde\calM'$ and get partitions of components of $\tilde\calM,\tilde\calM'$. For $i\in[\ell]$, let $\calM_i$ and $\calM'_i$ be the mixtures defined in Lemma~\ref{lemma:partition}. We can apply a linear transformation to make $\frac{1}{2}\calM_i+\frac{1}{2}\calM'_i$ in isotropic position. Since the total variation distance is invariant under linear transformations, so we still have both conclusions in Lemma~\ref{lemma:partition}. Let $\eps_2=\eps_1^{c_2}$. Then $\dtv(\calM_i,\calM_{\pi(i)})\le O(\eps_2)$ and $\frac{1}{2}\calM+\frac{1}{2}\calM'$ satisfies Lemma~\ref{lemma:eigenvalue-lower-bound} with $\delta=\eps_2^{c_3}$. 
Weights of both mixtures increase when we do the partition. So minimum weights are at least $\eps_1^{c_1}\ge\eps_1^{c_2c_3}=\eps_2^{c_3}$.

We now prove the statement on these smaller mixtures. First we can use Lemma~\ref{lemma:merge-parameters} to merge close parameters of $\calM_i,\calM_i'$ so that all pairs of parameters are either equal or separated by $\eps_2^{c_4c_5}$. Under this simplification, Lemma~\ref{lemma:main} shows that there is a perfect matching between the same components in two mixtures and their weights are almost the same. By the last statement in Lemma~\ref{lemma:merge-parameters}, it is also a matching between components of $\calM_i$ and $\calM_i'$ by combining $\pi$ and $\pi_1,\pi_2$. Moreover, if $\tilde G_j=\tilde G_{\pi(j)}'$, then $\dtv(G_\ell,G'_\ell)\le\poly(\eps_2)$ for all $\ell,\ell'$ such that $\pi_1(\ell)=j,\pi_2(\ell')=\pi(j)$. Repeating the argument for all pieces in $\tilde\calM,\tilde\calM'$ completes the proof.
\end{proof}

\subsection{Proof of Lemma~\ref{lemma:partition} }\label{sec:partial-clustering}
In this section, we will prove Lemma~\ref{lemma:partition}. 
The following fact in [Liu-Moitra] shows that a good set of clusters of one mixture exists.
\begin{fact}[Claim 7.6 in {[LM'20]}]\label{fact:partition-single-gaussian}
Let $\calM=\sum_{i=1}^k w_iG_i$ be a mixture of Gaussians. For any constants $0<\delta<1$ and $\eps>0$, there exists $t\in[k^2]$ such that there exists a partition (possibly trivial) of $[k]$ into sets $R_1,\dots,R_\ell$ such that 
\begin{enumerate}
\item If we draw edges between all pairs $i,j$ such that $\dtv(G_i,G_j)\le 1-\eps^{\delta^t}$, then each piece of the partition is connected
\item For any $i,j$ in different pieces of the partition, $\dtv(G_i,G_j)\ge 1-\eps^{\delta^{t-1}}$.
\end{enumerate}
\end{fact}

\begin{remark*}
Fact~\ref{fact:partition-single-gaussian} can be applied to a set of Gaussians instead of a mixture of Gaussians by randomly assigning positive weights for all Gaussians.
\end{remark*}

\begin{lemma}\label{lemma:components-dtv}
For any constant $0<c<1$, suppose $\calM=\sum_{i=1}^{k_1} w_iA_i,\calM'=\sum_{i=1}^{k_2} w_i'B_i$ are two mixtures of arbitrary distributions with $\dtv(\calM,\calM')\le\eps$ and $w_i,w_i'\ge \eps^c$. If for any $i\neq j$, $\dtv(A_i,B_j)\ge 1-\eps$, then $k_1=k_2$ and $\dtv(A_i,B_i)\le \poly_{k_1}(\eps)$ for all $i\in[k_1]$.
\end{lemma}

\begin{proof}
Suppose $\pi$ is any coupling of $\calM$ and $\calM'$ and $X,Y$ are random variables with distributions $\calM$ and $\calM'$. Then $\dtv(\calM,\calM')=\min_{\pi}\{\Pr_{\pi}(X\neq Y)\}$. We define $\pi$ to be the optimal coupling such that $\dtv(\calM,\calM')=\Pr_{\pi}(X\neq Y)$. 
Then we can define $\hat\pi$ on variables $i,j,X,Y$ such that $\sum_{i\in[k_1],j\in[k_2]}\hat\pi(i,j,X,Y)=\pi(X,Y)$ and the marginal distribution $\hat\pi_X$ with fixed $i$ of X is $w_iA_i$ for all $i\in[k_1]$ and the marginal distribution $\hat\pi_Y$ with fixed $j$ is $w_j'B_j$ for all $j\in[k_2]$.  Let $P_{ij}=\int_{X,Y}\hat\pi(i,j,X,Y)dXdY$ and $A_{ij}=\frac{1}{P_{ij}}\int_{Y}\hat\pi(i,j,X,Y)dY$ be distributions on $X$, $B_{ij}=\frac{1}{P_{ij}}\int_{X}\hat\pi(i,j,X,Y)dX$ be distributions on $Y$.
Then we have
\begin{align}
\begin{split}\label{eqn:dtv(M,M')}
\dtv(\calM,\calM')&=\Pr_{\pi}(X\neq Y)=\Pr_{\hat\pi}(X\neq Y)\\
&=\sum_{i,j}P_{ij}\Pr_{\hat\pi}(X\neq Y\mid i,j)\\
&\ge\sum_{i,j}P_{ij}\cdot\dtv(A_{ij},B_{ij}).
\end{split}
\end{align}
By the definition of $P_{ij},A_{ij},B_{ij}$,
\begin{align*}
w_iA_i=P_{ij}A_{ij}+\sum_{j'\neq j}P_{ij'}A_{ij'}\\
w_i'B_i=P_{ij}B_{ij}+\sum_{i'\neq i}P_{i'j}B_{i'j}
\end{align*}
Dividing both sides by $\max\{w_i,w_j'\}$, we get
\begin{align*}
A_i=\frac{P_{ij}}{\max\{w_i,w_j'\}}A_{ij}+\left(1-\frac{w_i}{\max\{w_i,w_j'\}}\right)A_i+
\sum_{j'\neq j}\frac{P_{ij'}}{\max\{w_i,w_j'\}}A_{ij'}\\
B_i=\frac{P_{ij}}{\max\{w_i,w_j'\}}B_{ij}+\left(1-\frac{w_j'}{\max\{w_i,w_j'\}}\right)B_i+
\sum_{i'\neq i}\frac{P_{i'j}}{\max\{w_i,w_j'\}}B_{i'j}
\end{align*}
From the above two equations, we can write $A_i,B_i$ as linear combinations of two distributions.
\begin{align*}
A_i=\frac{P_{ij}}{\max\{w_i,w_j'\}}A_{ij}+\left(1-\frac{P_{ij}}{\max\{w_i,w_j'\}})\right)A_i'\\
B_i=\frac{P_{ij}}{\max\{w_i,w_j'\}}B_{ij}+\left(1-\frac{P_{ij}}{\max\{w_i,w_j'\}}\right)B_i'
\end{align*}
Then by the triangle inequality,
$$
\dtv(A_i,B_j)\le\frac{P_{ij}}{\max\{w_i,w_j'\}}\dtv(A_{ij},B_{ij})+
\left(1-\frac{P_{ij}}{\max\{w_i,w_j'\}}\right)
$$
\begin{equation}\label{eqn:dtv(Aij,Bij)}
P_{ij}\cdot\dtv(A_{ij},B_{ij})\ge P_{ij}-(1-\dtv(A_i,B_j))\max\{w_i,w_j'\}.
\end{equation}
Combining (\ref{eqn:dtv(M,M')}) and (\ref{eqn:dtv(Aij,Bij)}), we have the following inequality on the TV distance between mixtures and the TV distance between components
\begin{equation}\label{eqn:dtv(M,M')-dtv(Ai,Bj)}
\dtv(\calM,\calM')
\ge\sum_{i,j}\Paren{ P_{ij}-(1-\dtv(A_i,B_j))\max\{w_i,w_j'\} }.
\end{equation}
By the lower bounds on $\dtv(A_i,B_j)$, we have
\begin{align}
\begin{split}\label{eqn:eps-max(wi,wi')}
\epsilon\ge\dtv(\calM,\calM')
&\ge\sum_{i,j}P_{ij}-\sum_{i\neq j}(1-\dtv(A_i,B_j))\max\{w_i,w_j'\}-\sum_{i}(1-\dtv(A_i,B_i))\max\{w_i,w_i'\}\\
&\ge1-\sum_{i\neq j}\eps-\sum_i\max\{w_i,w_i'\}+\sum_i\max\{w_i,w_i'\}\dtv(A_i,B_i)\\
&\ge1-\sum_{i\neq j}\eps-\sum_i\max\{w_i,w_i'\}+w_{\min}\dtv(A_1,B_1)
\end{split}
\end{align}
where $A_1,B_1$ can be replaced by any $A_i,B_i$ pair. Let $k=\max\{k_1,k_2\}$.
When $i\neq j$ and $\dtv(A_i,B_j)\ge 1-\eps$, we plug it into Equation~(\ref{eqn:dtv(M,M')-dtv(Ai,Bj)}) and get
\begin{align*}
\epsilon\ge\dtv(\calM,\calM')\ge\sum_{i\neq j}P_{ij}-(1-\dtv(A_i,B_j))\max\{w_i,w_j'\}
\ge\sum_{i\neq j}(P_{ij}-\epsilon).
\end{align*}
This implies
\[
\sum_{i\neq j}P_{ij}\le k^2\eps.
\]
Then we can bound $\sum_i\max\{w_i,w_i'\}-1$ in Equation (\ref{eqn:eps-max(wi,wi')})
\begin{align*}
\sum_i\max\{w_i,w_i'\}-1=\sum_i\max\{w_i,w_i'\}-w_i\le\sum_{i\neq j}P_{ij}\le k^2\eps.
\end{align*}
Plugging this bound into Equation (\ref{eqn:eps-max(wi,wi')}), for any $i$, we have
\begin{align*}
\dtv(A_i,B_i)\le \frac{1}{w_{\min}}\left(k^2\eps+\sum_i\max\{w_i,w_i'\}-1\right)\le\frac{2k^2\eps}{\eps^c}.
\end{align*}
\end{proof}

\begin{proof}[Proof of Lemma~\ref{lemma:partition}]
We apply Fact~\ref{fact:partition-single-gaussian} on the union set of components of $\calM$ and $\calM'$ with parameter $\delta$ to find a partition $R_1,\dots,R_\ell$. Let 
\begin{align*}
\calM_i&=\frac{\sum_{G_j\in R_i} w_jG_j}{\sum_{G_j\in R_i} w_j}\\
\calM_i'&=\frac{\sum_{G_j'\in R_i} w_j'G_j'}{\sum_{G_j'\in R_i} w_j'}.
\end{align*}
Then for any $i\neq j$, we know $\dtv(G_a,G_b')\ge 1-\eps^{\delta^{t-1}}$ for $G_a\in R_i,G_b'\in R_j$.
By (\ref{eqn:dtv(M,M')-dtv(Ai,Bj)}) in the proof of Lemma~\ref{lemma:components-dtv}, we have
$$
\dtv(M_i,M_j')\ge 1- 2k\eps^{\delta^{t-1}}.
$$
Then by Lemma~\ref{lemma:components-dtv}, for any $i$, there exists $a$ such that $\dtv(M_i,M_i')\le \eps^{a\delta^{t-1}}$. Let $c_2=a\delta^{t-1}$. If we set $\delta=c_2c_3/\delta^{t-1}=ac_3$, the partition satisfies the second conclusion.
\end{proof}

%% file: appendix.tex

\section{Omitted Proofs}
\label{sec:omitted_proofs}

In this subsection, we provide the proofs that were omitted from Section \ref{sec:prelims} and Section \ref{sec:full-algo-analysis}.

\subsection{Omitted Proofs from Section~\ref{ssec:prelims-gaussian}}

\begin{lemma}[Concentration of low-degree polynomials, Lemma \ref{lem:concentration_of_low_degree_gaussians} restated]
Let $T$ be a $d$-dimensional, degree-$4$ tensor such that $\Norm{T}_F \leq \Delta$
for some $\Delta>0$ and let $x, y \sim \calN(0, I)$. Then, with probability at least $1-1/\poly(d)$, the following holds:
\begin{equation*}
\Norm{ T\Paren{\cdot, \cdot, x, y}}^2_F \leq \bigO{\log(d) \Delta^2} \;.
\end{equation*}
\end{lemma}
\begin{proof}
We note that
\begin{align*}
\E\left[\norm{T(\cdot,\cdot,x,y)}_F^2\right]
&=\E\left[\sum_{i_1,i_2}\left(\sum_{i_3,i_4}T\Paren{i_1,i_2,i_3,i_4}x\Paren{i_3}y\Paren{i_4}\right)^2\right]\\
&=\E\left[\sum_{i_1,i_2}\left(\sum_{i_3,i_4}T\Paren{i_1,i_2,i_3,i_4}^2x\Paren{i_3}^2y\Paren{i_4}^2\right)\right]\\
&=\sum_{i_1,i_2,i_3,i_4}T\Paren{i_1,i_2,i_3,i_4}^2\le\Delta^2 \;.
\end{align*}
The second equality follows from the fact that $x\Paren{i_3},y\Paren{i_4}$ are independent and have zero means.
So the only non-zero terms are the squares. The third equality follows from the fact that
$x\Paren{i_3},y\Paren{i_4}$ are independent with unit variances. Observe that
$\Norm{ T\Paren{\cdot, \cdot, x, y}}^2_F$ is a degree-$2$ polynomial in Gaussian random variables.
Using standard concentration bounds for low-degree Gaussian polynomials
(see, e.g.,~Theorem 2.3 in \cite{DRST14}), we obtain
\begin{equation*}
\Pr\left[ \Norm{T(\cdot,\cdot,x,y)}_F^2 \geq t^2 \E\left[\norm{T(\cdot,\cdot,x,y)}_F^2\right] \right] \leq \exp\left( - c t \right) \;.
\end{equation*}
Setting $t = \Omega(\log(d))$ completes the proof.
\end{proof}

\subsection{Omitted Proofs from Section~\ref{sec:sos_proofs}}

\begin{lemma}[Spectral SoS Proofs, Lemma \ref{lem:spectral-sos-proofs} restated] 
Let $A$ be a $d \times d$ matrix. Then for $d$-dimensional vector-valued indeterminate $v$, we have:
\[
\sststile{2}{v} \Set{ v^{\top}Av \leq \Norm{A}_2 \Norm{v}_2^2}\mper
\]
\end{lemma}
\begin{proof}
Note that $v$ is the only variable in the proof here ($A$ is a matrix of constants). 
We note that $A \leq \Norm{A}_2 I$ or $\Norm{A}_2 I - A$ is PSD and thus $\Norm{A}_2 I - A = QQ^{\top}$ for some $d \times d$ matrix $Q$. Thus, $\Norm{Qv}_2^2 = v^{\top} (\Norm{A}_2 I - A) v = \Norm{A}_2 \Norm{v}_2^2 - v^{\top}Av$. Thus, $\Norm{A}_2 \Norm{v}_2^2 - v^{\top}Av$ is a sum of squares polynomial (namely $\Norm{Qv}_2^2$) in variable $v$. This completes the proof.

\end{proof}

\begin{lemma}[Frobenius Norms of Products of Matrices, Lemma \ref{lem:frob-of-product} restated]
Let $B$ be a $d \times d$ matrix valued indeterminate for some $d \in \N$. Then, for any $0 \preceq A \preceq I$, 
\[
\sststile{2}{B} \Set{\Norm{AB}_F^2 \leq \Norm{B}_F^2}\mcom
\]
and, 
\[
\sststile{2}{B} \Set{\Norm{BA}_F^2 \leq \Norm{B}_F^2}\mcom
\]
\end{lemma}
\begin{proof}
The proof of the second claim is similar so we prove only the first. We have:
\[
\sststile{2}{B} \Set{\Norm{B}_F^2= \Norm{(A+I-A)B}_F^2 = \Norm{AB}_F^2 + \Norm{(I-A)B}_F^2 + 2\tr((I-A)BB^{\top}A)} 
\]
Now, $A-A^2 \succeq 0$, thus, $A-A^2 = RR^{\top}$ for some $d \times d$ matrix $R$. Thus, $\tr((A-A^2)BB^{\top}) = \tr(RR^{\top}BB^{\top}) = \Norm{BR}_F^2$ - a sum of squares polynomial of degree $2$ in indeterminate $B$. Thus, $\sststile{2}{B} \Set{\tr((A-A^2)BB^{\top}) \geq 0}$. 
\end{proof}

\subsection{Omitted Proofs from Section~\ref{ssec:prelim-analytic}}

\begin{lemma}[Shifts Cannot Decrease Variance, Lemma \ref{lem:shifts-only-increase-variance} restated]
Let $\cD$ be a distribution on $\R^d$, $Q$ be a $d \times d$ matrix-valued indeterminate, and $C$ be a scalar-valued indeterminate.
Then, we have that
\[
\sststile{2}{Q,C} \Set{\E_{x\sim \cD} \left[ \Paren{Q(x)-\E_{x\sim\cD}[Q(x)] }^2 \right] \leq \E_{x\sim\cD} \left[\Paren{Q(x)-C}^2\right] }\mper
\]
\end{lemma}
\begin{proof}
\begin{align*}
\sststile{2}{Q,C}  \Biggl\{ \expecf{x\sim\cD}{\Paren{Q(x)-C}^2} &= \expecf{x\sim\cD}{\Paren{Q(x)-\expecf{x\sim\cD}{Q(x)}+ \expecf{x\sim\cD}{Q(x)} -C}^2}\\
&= \expecf{x\sim\cD}{\Paren{Q(x)-\expecf{x\sim\cD}{Q(x)}}^2} + \expecf{x\sim\cD}{(Q(x)-C)^2}\\
&\hspace{0.2in}+ 2 \expecf{x\sim\cD}{\Paren{Q(x)-\expecf{x\sim\cD}{Q(x)}} \Paren{\expecf{x\sim\cD}{Q(x)}-C}}\\
 &=  \expecf{x\sim\cD}{\Paren{Q(x)-\expecf{x\sim\cD}{Q(x)}}^2} + \expecf{x\sim\cD}{(Q(x)-C)^2} \\
 &\geq  \expecf{x\sim\cD}{\Paren{Q(x)-\expecf{x\sim\cD}{Q(x)}}^2} \hspace{0.2in}
\Biggr\}\mper
\end{align*}
\end{proof}

\begin{lemma}[Shifts of Certifiably Hypercontractive Distributions, Lemma \ref{lem:shifts-of-certifiably-hypercontractive-distributions} restated]
Let $x$ be a mean-$0$ random variable with distribution $\cD$ on $\R^d$ with $t$-certifiably $C$-hypercontractive degree-$2$ polynomials. Then, for any fixed constant vector $c \in \R^d$, the random variable $x+c$ also has $t$-certifiable $4C$-hypercontractive degree-$2$ polynomials.
\end{lemma}
\begin{proof}
Observe that using that $\expecf{x\sim\cD}{x} = 0$, we have that
\begin{equation*}
\sststile{2}{Q} \Set{ \expecf{x\sim\cD}{(x+c)^{\top} Q (x+c)} = \expecf{x\sim\cD}{x^{\top} Q x + c^{\top} Q c}} \;.
\end{equation*}

Next, by two applications of the SoS Triangle Inequality (Fact~\ref{fact:almost-triangle-sos}),
an application of Lemma~\ref{lem:shifts-only-increase-variance}  followed by certifiable hypercontractivity of $\cD$, we have:
\begin{align*}
\sststile{t'}{Q} \Biggl\{  & \expecf{x\sim\cD}{\Paren{(x+c)^{\top} Q (x+c)-\expecf{x\sim\cD}{(x+c)^{\top} Q (x+c)} }^{t'}} \\
& = \expecf{x\sim\cD}{\Paren{\Paren{x^{\top} Q x-\expecf{x\sim\cD}{x^{\top} Q x} } + x^{\top} Qc + c^{\top} Qx }^{t'}} \\
& \leq 4^{t'} \Paren{ \expecf{x \sim \cD}{\Paren{x^{\top}Qx-\E_{\cD}x^{\top} Q x}^{t'} }+ \expecf{x\sim\cD}{(x^{\top}Qc)^{t'}} + \expecf{x\sim\cD}{(c^{\top}Qx)^{t'}} }\\
&\leq 4^{t'} (Ct')^{t'} \Paren{ \expecf{x \sim \cD}{\Paren{x^{\top}Qx-\E_{\cD}x^{\top} Q x}^{2} }^{t'/2} + \expecf{x\sim\cD}{(x^{\top}Qc)^{2}}^{t'/2} + \expecf{x\sim\cD}{(c^{\top}Qx)^{2}}^{t'/2} } \Biggr\} \;.
\end{align*}
On the other hand, notice that
\begin{align*}
\sststile{2}{Q} \Biggl\{ & \expecf{x\sim\cD}{\Paren{(x+c)^{\top} Q (x+c)-\expecf{x\sim\cD}{(x+c)^{\top} Q (x+c)} }^{2}} \\
&=\Paren{ \expecf{x \sim \cD}{\Paren{x^{\top}Qx-\E_{\cD}x^{\top} Q x}^{2} }+ \expecf{x\sim\cD}{(x^{\top}Qc)^{2}} + \expecf{x\sim\cD}{(c^{\top}Qx)^{2}} } \hspace{0.2in}\Biggr\} \;.
\end{align*}
Thus,
\begin{align*}
\sststile{t'}{Q} \Biggl\{  & \expecf{x\sim\cD}{\Paren{x^{\top}Qx- \expecf{x\sim\cD}{x^{\top} Q x} }^2}^{t'/2}  + \Paren{ \expecf{x\sim\cD}{(x^{\top}Qc)^2} }^{t'/2} + \Paren{\expecf{x\sim\cD}{(c^{\top}Qx)^2} }^{t'/2} \\
&\leq 4^{t'} (Ct')^{t'} \Paren{\expecf{x\sim\cD}{\Paren{(x+c)^{\top} Q (x+c) - \expecf{x\sim\cD}{(x+c)^{\top}Q(x+c)} }^2 } }^{t'/2} \Biggr\} \;.
\end{align*}
As a result, we obtain:
\begin{align*}
\sststile{t'}{Q} \Biggl\{ & \expecf{x\sim\cD}{\Paren{(x+c)^{\top} Q (x+c)- \expecf{x\sim\cD}{(x+c)^{\top}Q(x+c)} }^{t'} } \\
 & \leq (4Ct')^{t'}\Paren{\expecf{x\sim\cD}{\Paren{(x+c)^{\top} Q (x+c) - \expecf{x\sim\cD}{(x+c)^{\top}Q(x+c)} }^2} }^{t'/2} \Biggr\} \;,
\end{align*}
which completes the proof.
\end{proof}

\begin{lemma}[Mixtures of Certifiably Hypercontractive Distributions,
Lemma \ref{lem:mixtures-of-certifiably-hypercontractive-distributions} restated]
Let $\cD_1, \cD_2, \ldots, \cD_k$ have $t$-certifiable $C$-hypercontractive degree-$2$ polynomials on $\R^d$,
for some fixed constant $C$. Then, any mixture $\cD=\sum_i w_i \cD_i$ also has $t$-certifiably $(C/\alpha)$-hypercontractive
degree-$2$ polynomials for $\alpha = \min_{i \leq k, w_i > 0} w_i$.
\end{lemma}

\begin{proof}
Applying Lemma~\ref{fact:almost-triangle-sos} followed by SoS H{\"o}lder's inequality on the second term
and followed by a final application of SoS H{\"o}lder's inequality (Fact \ref{fact:sos-holder}), we obtain:
\begin{align*}
\sststile{t'}{Q} \Biggl\{  \expecf{x\sim\cD}{ \Paren{ x^\top Qx-\expecf{x\sim\cD}{x^\top Qx } }^{t'} } & = \expecf{x\sim\cD}{\Paren{ x^\top Qx - \sum_i w_i \expecf{x\sim \cD_i}{ x^\top Qx} }^{t'} }\\
 &= \sum_i w_i \expecf{x \sim \cD_i}{ \Paren{x^\top Qx- \sum_i w_i \expecf{x\sim \cD_i}{x^\top Qx} }^{t'}}\\
 &\leq 2^{t'} \Biggl( \sum_i w_i \expecf{x\sim\cD_i}{ \Paren{x^\top Qx- \expecf{x\sim \cD_i}{x^\top Q x} }^{t'}} \\
 & \hspace{0.2in} + \Paren{\sum_i w_i \expecf{x\sim \cD_i}{x^\top Qx -\expecf{x\sim\cD_i}{x^\top Qx} }^{t'}} \Biggr) \\
&\leq 2^{t'} \Biggl( (Ct')^{t'} \Paren{\sum_i w_i \expecf{x\sim \cD_i}{ \Paren{x^\top Qx- \expecf{x\sim \cD_i}{x^\top Qx} }^{2} } }^{t'/2} \\
& \hspace{0.2in} + \sum_i w_i \Paren{  \expecf{x\sim \cD_i}{ \Paren{ x^\top Qx -\expecf{x\sim \cD_i}{x^\top Qx} }^{2}}  }^{t'} \Biggr)\\
&\leq \Paren{\frac{4 Ct'}{\alpha}}^{t'}    \Paren{\sum_i w_i \expecf{x\sim \cD_i}{ \Paren{ x^\top Qx -\expecf{x\sim \cD_i}{x^\top Qx} }^2 }  }^{t'/2}
 \Biggr\}\mper
\end{align*}
On the other hand, note that by Lemma~\ref{lem:shifts-only-increase-variance}, we know that
\begin{align*}
\sststile{2}{Q}\Biggl\{ \expecf{x\sim\cD}{ \Paren{x^\top Qx- \expecf{x\sim \cD}x^\top Qx  }^2 } &= \sum_i w_i \expecf{x\sim\cD_i}{ \Paren{x^\top Qx- \expecf{x\sim \cD}{x^\top Qx}  }^2} \\
&  \geq \sum_i w_i \expecf{x \sim \cD_i}{ \Paren{x^\top Qx -\expecf{x \sim \cD_i}{x^\top Qx} }^2 }\Biggr\}\mper
\end{align*}
Combining the two equations above completes the proof.
\end{proof}

\begin{corollary}[Certifiable Hypercontractivity of $k$-Mixtures of Gaussians, Corollary \ref{cor:hypercontractivity_of_mixtures} restated]
Let $\cD$ be a $k$-mixture of Gaussians $\sum_i w_i \cN(\mu_i,\Sigma_i)$ with weights $w_i \geq \alpha$ for every $i \in [k]$.
Then, $\cD$ has $t$-certifiably $4/\alpha$-hypercontractive degree-$2$ polynomials.
\end{corollary}
\begin{proof}
From \cite{DBLP:conf/soda/KauersOTZ14}, we know that the standard Gaussian random variable has $t$-certifiably $1$-hypercontractive degree-$2$ polynomials. From Fact~\ref{fact:affine-invariance-certifiable-hypercontractivity}, we immediately obtain that for any
PSD matrix $\Sigma$, the Gaussian $\cN(0,\Sigma)$ also has $t$-certifiable $1$-hypercontractive degree-$2$ polynomials.
From Lemma~\ref{lem:shifts-of-certifiably-hypercontractive-distributions}, we obtain that for any $\mu$, the Gaussian
$\cN(\mu,\Sigma)$ has $t$-certifiable $4$-hypercontractive degree-$2$ polynomials. Finally, applying
Lemma~\ref{lem:mixtures-of-certifiably-hypercontractive-distributions} to $\cD_i = \cN(\mu_i,\Sigma_i)$ and mixture weights
$w_1,w_2, \ldots w_k$, yields that $\cD = \sum_i w_i \cN(\mu_i,\Sigma_i)$ has $t$-certifiably $4/\alpha$-hypercontractive
degree-$2$ polynomials. This completes the proof.
\end{proof}

\begin{lemma}[Linear Transformations of Certifiably Bounded-Variance Distributions, Lemma \ref{lem:bounded-variance-linear-transform} restated]
For $d \in \N$, let $x$ be a random variable with distribution $\cD$ on $\R^d$ such that for $d \times d$ matrix-valued indeterminate $Q$, $\sststile{2}{Q} \Set{\E_{x \sim \cD}(x^{\top}Qx-\E_{\cD}x^{\top}Qx)^2 \leq \Norm{\Sigma^{1/2}Q\Sigma^{1/2}}^2_F }$. Let $A$ be an arbitrary $d \times d$ matrix and let $x' = Ax$ be the random variable with covariance $\Sigma' = AA^{\top}$. Then, we have that
\[
\sststile{2}{Q} \Set{\E_{x' \sim \cD'}({x'}^{\top}Qx'-\E_{\cD'}{x'}^{\top}Qx')^2 \leq \Norm{{\Sigma'}^{1/2}Q{\Sigma'}^{1/2}}^2_F}\mper
\]
\end{lemma}
\begin{proof}
The covariance of $x'$ is $AA^{\top} =\Sigma'$, say. Let ${\Sigma'}^{1/2}$ be the PSD square root of $\Sigma'$. The proof follows by noting that ${x'}^{\top}Qx' = (Ax)^{\top}Q(Ax) = x^{\top} (A^{\top}QA) x^{\top}$ and that $\Norm{A^{\top}QA}_F^2 = \tr(A^{\top}QAA^{\top}QA) = \tr(AA^{\top}QAA^{\top}Q) = \tr(\Sigma' Q\Sigma' Q) = \tr({\Sigma'}^{1/2}Q{\Sigma'}^{1/2} {\Sigma'}^{1/2}Q{\Sigma'}^{1/2}) = \Norm{{\Sigma'}^{1/2}Q{\Sigma'}^{1/2}}_F^2$. 
\end{proof}

\begin{lemma}[Variance of Degree-$2$ Polynomials of Standard Gaussians, Lemma \ref{lem:var-zero-mean-gaussians} restated] 
We have that
\[
\sststile{2}{Q} \Set{\E_{\cN(0,I)} \Paren{x^{\top}Qx - \E_{\cN(0,I)} x^{\top}Qx}^2 \leq 3 \Norm{Q}_F^2}\mper
\]
\end{lemma}
\begin{proof}
We will view $xx^{\top}$ and $I \in \R^{d \times d}$ as $d^2$-dimensional vectors. Consider the matrix 
$\E_{x \sim \cN(0,I)} (xx^{\top}-I)(xx^{\top}-I)^{\top}$. The diagonal of this matrix is $2I_{d^2}$. 
The off-diagonal part has exactly one non-zero entry in any row (which corresponds to entry indexed by $(i,j)$ and $(j,i)$ for $i \neq j$), 
and thus has spectral norm at most $1$ by the Gershgorin circle theorem. 
Thus, $\E_{x \sim \cN(0,I)} (xx^{\top}-I)(xx^{\top}-I)^{\top} \preceq 3I_{d^2}$. 

We thus have: 
\begin{multline}
\sststile{2}{Q} \Biggl\{ \E_{\cN(0,I)} \Paren{x^{\top}Qx - \E_{\cN(0,I)} x^{\top}Qx}^2 = \E_{\cN(0,I)} \Iprod{xx^{\top}-I,Q}^2 = \E_{\cN(0,I)} \Iprod{xx^{\top}-I,Q}^2 \\ \leq \Norm{\E_{x \sim \cN(0,I)}(xx^{\top}-I)(xx^{\top}-I')^{\top}}_2 \Norm{Q}_F^2 \leq 3\Norm{I_{d^2}}_2 \Norm{Q}_F^2 = 3 \Norm{Q}_F^2 \Biggr\}\mper
\end{multline}

\end{proof}

\begin{lemma}[Variance of Degree-$2$ Polynomials of Mixtures, Lemma \ref{fact:mean-variance-subgaussian} restated]
Let $\calM = \sum_i w_i \cD_i$ be a $k$-mixture of distributions $\cD_1, \cD_2, \ldots, \cD_k$ with means $\mu_i$ and covariances $\Sigma_i$. Let $\mu = \sum_i w_i \mu_i$ be the mean of $\calM$. Suppose that each of $\cD_1,\cD_2, \ldots, \cD_k$ have certifiably $C$-bounded-variance i.e. for $Q$: a symmetric $d \times d$ matrix-valued indeterminate. 
\[
\sststile{2}{Q} \Set{\E_{x' \sim \cD_i}({x'}^{\top}Qx'-\E_{\cD_i}{x'}^{\top}Qx')^2 \leq C\Norm{{\Sigma'}^{1/2}Q{\Sigma'}^{1/2}}^2_F}\mper
\]
Further, suppose that for some $H > 1$, $\norm{\mu_i-\mu}_2^2, \Norm{\Sigma_i-I}_F \leq H$ for every $1 \leq i \leq k$. 
Then, we have that
\[
\sststile{2}{Q} \Set{ \expecf{x\sim \calM}{ \Paren{ x^\top Qx-\expecf{x \sim \calM}{x^\top Qx } }^2 } \leq 100C H^2\Norm{Q}_F^2 }\mper
\]
\end{lemma}

\begin{proof}
We have the following sequence of (in-)equalities:
\begin{align}
\sststile{2}{Q} &\Biggl\{ \E_{x\sim \calM} \Paren{ x^\top Qx-\E_{x \sim \calM}x^\top Qx}^2 = \sum_{i \leq k} w_i \E_{x \sim \cD_i} \Paren{ x^\top Qx-\E_{x \sim \calM}x^\top Qx }^2\\
&=\sum_{i \leq k} w_i \E_{x \sim\cD_i}\Paren{ x^\top Qx-\E_{x\sim \cD_i} x^\top Qx + \E_{x \sim \cD_i}x^\top Qx -\E_{x \sim \calM}x^\top Qx }^2\\
& \leq 2 \sum_{i \leq k} w_i \E_{x \sim\cD_i} \Paren{x^\top Qx-\E_{x\sim \cD_i} x^\top Qx}^2 + 2 \sum_i w_i \Paren{\E_{x \sim \cD_i}x^\top Qx -\E_{x \sim \calM}x^\top Qx }^2
 \Biggr\} \;,\label{eq:split-componentwise}
\end{align}
where the third line follows from Fact~\ref{fact:almost-triangle-sos} (SoS Almost Triangle Inequality).

Let us first bound the 2nd term in the RHS above. 
Towards that, let $\Sigma = \sum_i w_i ((\mu_i-\mu)(\mu_i-\mu)^{\top} + \Sigma_i)$ 
be the covariance of the mixture $\calM$. Then, notice that 
$\Sigma = \sum_i w_i ((\mu_i-\mu)(\mu_i-\mu)^{\top} + \Sigma_i) = \sum_i w_i \mu_i \mu_i^{\top} + \sum_i w_i \Sigma_i -\mu\mu^{\top}$. 
Thus, we can write
\begin{align*}
\mu_i \mu_i^{\top} + \Sigma_i - \Sigma-\mu\mu^{\top} &= \sum_{j \neq i} w_j(\mu_i \mu_i^{\top} - \mu_j \mu_j^{\top}) + \sum_{j \neq i} w_j (\Sigma_i-\Sigma_j) \\
&= \sum_{j \neq i} w_j (\mu_j-\mu)(\mu_j-\mu)^{\top} -  \sum_{j \neq i}(\mu-\mu_j)(\mu-\mu_j)^{\top} + \sum_{j \neq i} w_j (\Sigma_i-\Sigma_j)\\
&=\sum_{j \neq i} w_j (\mu_j-\mu)(\mu_j-\mu)^{\top} -  \sum_{j \neq i}(\mu-\mu_j)(\mu-\mu_j)^{\top} + \sum_{j \neq i} w_j (\Sigma_i-I) - \sum_{j \neq i} w_j (\Sigma_j-I)\mper
\end{align*}
Here, in the second to last step, we added and subtracted $\sum_{j \neq i} w_j \mu \mu^{\top}$ 
and used that $\sum_i w_i \mu_i = \mu$, and in the last step we added and subtracted $\sum_{j \neq i}w_j I$.

By application of the triangle inequality for Frobenius norm to the RHS of the above, 
we have that: 
\begin{align*}
\Norm{\mu_i \mu_i^{\top} + \Sigma_i -\mu \mu^{\top} -\Sigma}_F 
&\leq \sum_{j \neq i} w_j \Norm{(\mu_i-\mu) (\mu_i-\mu)^{\top}}_F + \sum_{j \neq i} w_j \Norm{(\mu_j-\mu)(\mu_j-\mu)^{\top}}_F \\ 
&+ \sum_{j \neq i} w_j  \Norm{(\Sigma_i-I)}_F  + \sum_{j \neq i} w_j \Norm{ (I-\Sigma_j)}_F \leq H+H + H+H = 4H \;.
\end{align*}
Using the SoS version of the Cauchy-Schwarz inequality (Fact~\ref{fact:sos-holder}) 
on indeterminate $Q$ and constant $\mu \mu^{\top} - \mu_i \mu_i^{\top} + \Sigma_i - \Sigma$ and the above bound, 
we have: 
\begin{align*}
\sststile{2}{Q} 
&\Biggl\{ \sum_i w_i \Paren{\E_{x \sim \cD_i}x^\top Qx -\E_{x \sim \calM}x^\top Qx }^2 = \sum_i w_i \Paren{\Iprod{\mu_i \mu_i^{\top}+\Sigma_i,Q} -\Iprod{\mu\mu^{\top}+\Sigma,Q}}^2 \\
&\leq \sum_i w_i \Norm{\mu \mu^{\top} - \mu_i \mu_i^{\top} + \Sigma_i - \Sigma}_F^2 \Norm{Q}_F^2 \leq 16H^2 \sum_i w_i \Norm{Q}_F^2 = 16H^2  \Norm{Q}_F^2\Biggr\}\mper
\end{align*}
Let us now bound the first term in the RHS of \eqref{eq:split-componentwise} above. 
First, observe that $x^{\top}Qx - \E_{\cN(\mu_i,\Sigma_i)} x^{\top}Qx = (x-\mu_i)^{\top}Q(x-\mu_i) - \E_{\cN(\mu_i,\Sigma_i)} (x-\mu_i)^{\top}Q(x-\mu_i) + 2(x-\mu_i)^{\top}Q\mu_i$. Thus, using Fact~\ref{fact:almost-triangle-sos} and Lemma~\ref{lem:var-zero-mean-arbitrary-cov-gaussians}, we have:
\begin{align}
\sststile{2}{Q} &\Biggl\{\sum_{i \leq k} w_i \E_{\cD_i} \Paren{x^\top Qx-\E_{x\sim \cD_i} x^\top Qx}^2 \\
&\leq 2\sum_{i \leq k} w_i \E_{\cD_i} \Paren{(x-\mu_i)^\top Q(x-\mu_i)-\E_{x\sim \cD_i} (x-\mu_i)^\top Q(x-\mu_i)}^2 + 8 \sum_{i \leq k} w_i \E_{\cD_i} \Paren{(x-\mu_i)^{\top}Q\mu_i}^2\\
&\leq 6 \sum_i w_i C\Norm{\Sigma_i^{1/2}Q\Sigma_i^{1/2}}_F^2 + 8 \sum_{i \leq k} w_i \E_{\cD_i} \Paren{(x-\mu_i)^{\top}Q\mu_i}^2
 \Biggr\}\mper \label{eq:two-term-split}
\end{align}
For the first term, note that $\Norm{\Sigma_i}_2 \leq 1+\Norm{\Sigma_i-I}_F \leq 1+H$. 
Thus, $\Norm{\Sigma_i^{1/2}}_2 \leq \sqrt{1+H}$. Thus, we have that $\Sigma_i^{1/2} \preceq I + (\Sigma_i^{1/2}-I) \preceq \sqrt{1+H}I$. 
Using Lemma~\ref{lem:frob-of-product} with $A = (1+H)^{-1/2} \Sigma_i^{1/2}$ and $B = Q\Sigma_i^{1/2}$, 
we have: $\sststile{2}{Q} \Set{\Norm{\Sigma_i^{1/2} Q \Sigma_i^{1/2}}_F^2 \leq (1+H) \Norm{Q\Sigma_i^{1/2}}_F^2}$. 
By another application of Lemma~\ref{lem:frob-of-product}, we have: 
$\sststile{2}{Q} \Set{\Norm{Q\Sigma_i^{1/2}}_F^2 \leq (1+H) \Norm{Q}_F^2 }$. 
Thus, altogether, we have: $\sststile{2}{Q} \Set{\Norm{\Sigma_i^{1/2} Q \Sigma_i^{1/2}}_F^2 \leq (1+H)^2 \Norm{Q}_F^2}$. 
Using our assumption that $1 < H$, we thus have:
\begin{align*}
\sststile{2}{Q} &\Biggl\{\sum_i w_i C\Norm{\Sigma_i^{1/2}Q\Sigma_i^{1/2}}_F^2 \leq C(1+H)^2 \Norm{Q}_F^2 
\leq 4CH^2 \Norm{Q}_F^2\Biggr\}\mper
\end{align*}
For the second term, first observe that the following equality of quadratic polynomials in indeterminate $Q$: 
$\Paren{(x-\mu_i)^{\top}Q\mu_i}^2 = \Paren{(\Sigma_i^{\dagger/2}(x-\mu_i))^{\top} \Sigma_i^{1/2}Q\mu_i}^2$. 
Thus, $\E_{x\sim \cD_i} \Paren{(x-\mu_i)^{\top}Q\mu_i}^2 = \Norm{\Sigma_i^{1/2}Q\mu_i}_2^2$. 
Next, by the SoS Cauchy-Schwarz inequality (Fact~\ref{fact:sos-holder}), we have that 
$$\sststile{2}{Q} \Set{\Norm{\Sigma_i^{1/2}Q\mu_i}_2^2= \tr(\mu_i \mu_i^{\top} Q \Sigma_i Q) 
\leq H \tr(Q\Sigma_iQ)= H \Norm{\Sigma_i^{1/2}Q}_F^2} \;.$$ 
Applying Lemma~\ref{lem:frob-of-product} with the observation above that $\Sigma_i^{1/2} \leq (1+H)^{1/2} I$ yields: 
$\sststile{2}{Q}\Set{\Norm{\Sigma_i^{1/2}Q}_F^2 \leq (1+H)\Norm{Q}_F^2}$. 
Thus, altogether, we obtain: 
$\sststile{2}{Q} \Set{\E_{x\sim \cD_i} (x-\mu_i)^{\top}Q\mu_i}^2 
\leq H(1+H)^2 \Norm{Q}_F^2 \leq 4H^3 \Norm{Q}_F^2$. We thus have:
\begin{align*}
\sststile{2}{Q} &\Biggl\{\sum_{i \leq k} w_i \E_{\cD_i} \Paren{(x-\mu_i)^{\top}Q\mu_i}^2  \leq 4H^3 \Norm{Q}_F^2 \Biggr\}\mper
\end{align*}
Plugging in these bounds into \eqref{eq:two-term-split} completes the proof.
\end{proof}


As an immediate corollary of Lemma~\ref{lem:bounded-variance-linear-transform} and Lemma~\ref{fact:mean-variance-subgaussian}, 
we obtain:

\begin{lemma}[Variance of Degree-$2$ Polynomials of Mixtures of Gaussians, Lemma \ref{fact:mean-variance-subgaussian-arbitrary-covariance} restated]
Let $\calM = \sum_i w_i \cN(\mu_i,\Sigma_i)$ be a $k$-mixture of Gaussians with $w_i \geq \alpha$, 
mean $\mu = \sum_i w_i \mu_i$ and covariance $\Sigma = \sum_i w_i ((\mu_i-\mu)(\mu_i-\mu)^{\top} + \Sigma_i)$. 
Suppose that for some $H > 1$, $\Norm{\Sigma^{\dagger/2}(\Sigma_i-I)\Sigma^{\dagger/2}}_F \leq H$ for every $1 \leq i \leq k$. 
Let $Q$ be a symmetric $d \times d$ matrix-valued indeterminate. Then for $H' = \max \{H, 1/\alpha\}$, 
\[
\sststile{2}{Q} \Set{ \expecf{x\sim \calM}{ \Paren{ x^\top Qx-\expecf{x \sim \calM}{x^\top Qx } }^2 } \leq 100 {H'}^2\Norm{\Sigma^{1/2}Q\Sigma^{1/2}}_F^2 }\mper
\]
\end{lemma}
\begin{proof}
Let $\Sigma = U\Lambda U^{\top}$ be the covariance of the mixture $\calM$ along with its eigendecomposition. 
We want to apply Lemma~\ref{fact:mean-variance-subgaussian} and Lemma~\ref{lem:bounded-variance-linear-transform} 
with the linear transformation $x \rightarrow Ax$ for $A =\Lambda^{\dagger/2}U^{\top}$. 
For this, we need to check that the conditions of the Lemma~\ref{fact:mean-variance-subgaussian} are met after this linear transformation. 
The new component covariance is $\Sigma_i'= A\Sigma_iA^{\top}$ and the hypothesis implies that they are within $H$ in Frobenius distance of the new mixture covariance $I' = A\Sigma A^{\top}$ ($I$ in the range space of $\Sigma$). 
The new means of the components after the linear transformation are $\mu_i' = A \mu_i$ 
and the new mixture mean is $\mu' = A\mu$. Thus, noting that $I' = \sum_i w_i (\mu_i'-\mu')(\mu'-\mu')^{\top} + \sum_i w_i \Sigma_i'$, 
and since each of the terms in the RHS of the preceding equality are PSD, we must have that 
$I' \succeq  w_i (\mu_i'-\mu')(\mu_i'-\mu')^{\top}$ for every $i$. Thus, 
$1 = \Norm{I'}_2 \geq w_i \Norm{(\mu_i'-\mu')(\mu_i'-\mu')^{\top}}_2 =  \Norm{\mu_i'-\mu'}_2^2$. 
Rearranging yields that $\Norm{\mu_i'-\mu'}_2^2 \leq 1/w_i \leq 1/\alpha$. Thus, we can now apply 
Lemma~\ref{fact:mean-variance-subgaussian} to the linearly transformed mixture and the conclusion follows. 
\end{proof}

\subsection{Omitted Proofs from Section~\ref{ssec:det-conds}}

\begin{lemma}[Lemma~\ref{moment closeness lemma} restated]
If $X$ satisfies Condition \ref{cond:convergence-of-moment-tensors} with respect to $\calM=\sum_i w_i \cN(\mu_i,\Sigma_i)$
with parameters $(\gamma,t)$, then if $w_i \geq \gamma$ for all $i \in [k]$,
and if for some $B\geq 0$ we have that $\Norm{\mu_i}_2^2,\Norm{\Sigma_i}_{\mathrm{op}} \leq B$ for all $i \in [k]$,
then for all $m\leq t$, we have that:
$$
\Norm{\E_{x\in_u X}[x^{\otimes m}] - \E_{x\sim M}[x^{\otimes m}]}_F^2 \leq \gamma^2 m^{O(m)}B^m d^m \;.
$$
\end{lemma}
\begin{proof}
We begin by noting that for any symmetric $m$-tensor $A$ we have that
$\Norm{A}_F^2 \newblue{\leq} m^{O(m)}(\E_{v\sim \cN(0,I)}[\langle A,v^{\otimes m}\rangle^2]).$ 
\newblue{This is because, in the notation of~\cite{DiakonikolasKS18-mixtures}, 
the squared expectation of $\langle A,v^{\otimes m}\rangle$ is 
$\E[\mathrm{Hom}_A(v)^2]\geq m! \E[h_A(v)^2] = m! \Norm{A}_F^2$, where the first inequality holds 
because $\sqrt{m!}h_A(v)$ is the degree-$m$ harmonic part of $\mathrm{Hom}_A(v)$, 
and the equality is by Claim 3.22.} Therefore, to prove the lemma, it suffices to bound
\begin{align*}
& \E_{v\sim \cN(0,I)}\left[\left(\E_{x\in_u X}[(v\cdot x)^m]-\E_{x\sim\calM}[(v\cdot x)^m]\right)^2 \right]\\
= & \E_{v\sim \cN(0,I)}\left[\left(\sum_{i=1}^k\left(\frac{1}{n}\sum_{x\in X_i} (v\cdot x)^m - w_i \E_{x\sim \cN(\mu_i,\Sigma_i)}(v\cdot x)^m \right)  \right)^2 \right]\\
= & \E_{v\sim \cN(0,I)}\left[\left(\sum_{i=1}^k \sum_{j=0}^m \binom{m}{j} \mu_i^{m-j} \left(\frac{1}{n}\sum_{x\in X_i} (v\cdot (x-\mu_i))^j - w_i \E_{x\sim \cN(\mu_i,\Sigma_i)}(v\cdot (x-\mu_i))^j \right)  \right)^2 \right]\\
\leq & \E_{v\sim \cN(0,I)}\left[\left(\sum_{i=1}^k \sum_{j=0}^m \binom{m}{j} |v\cdot\mu_i|^{m-j} \left(w_i \gamma m! (v^T \Sigma v)^{j/2} \right) \right)^2 \right]\\
\leq & \gamma^2 m^{O(m)}\E_{v\sim \cN(0,I)}\left[\sum_{i=1}^k  \left(w_i  (|v\cdot\mu_i|+(v^T \Sigma v)^{1/2})^{2m} \right) \right]\\
\leq & \gamma^2 m^{O(m)}\E_{v\sim \cN(0,I)}\left[2B^m\Norm{v}_2^{2m} \right]\\
\leq & \gamma^2 m^{O(m)}B^m d^m.
\end{align*}
This completes the proof.
\end{proof}

\begin{lemma}[Lemma~\ref{submixture condition lemma} restated]
Let $\calM=\sum_i w_i \cN(\mu_i,\Sigma_i)$. Let $S\subset [k]$ with $\sum_{i\in S} w_i = w$, and let
$\calM'=\sum_{i\in S} (w_i/w) \cN(\mu_i,\Sigma_i)$. Then if $X$ satisfies Condition \ref{cond:convergence-of-moment-tensors} with respect to $\calM$ with parameters $(\gamma,t)$ for some $\gamma<1/(2k)$
 with the corresponding partition being $X=X_1\cup X_2 \cup\ldots \cup X_k$,
 then $X' = \bigcup_{i\in S} X_i$ satisfies Condition \ref{cond:convergence-of-moment-tensors}
 with respect to $\calM'$ with parameters $(O(k\gamma/w),t)$.
\end{lemma}
\begin{proof}
After noting that $|X'| = w|X|(1+O(k\gamma/w))$, the rest follows straightforwardly from the definitions using the partition $X'=\bigcup_{i\in S} X_i$.
\end{proof}

\begin{lemma}[Lemma~\ref{lem:det-suffices} restated]
Let $\calM=\sum_{i=1}^k w_i \cN(\mu_i,\Sigma_i)$ and let $n$ be an integer at least $k t^{Ct} d^{t}/\gamma^3$,
for a sufficiently large universal constant $C>0$, some $\gamma>0$, and some $t \in \N$.
If $X$ consists of $n$ i.i.d. samples from $\calM$, then $X$ satisfies Condition \ref{cond:convergence-of-moment-tensors}
with respect to $\calM$ with parameters $(\gamma, t)$ with high probability.
\end{lemma}
\begin{proof}
We will show that Condition \ref{cond:convergence-of-moment-tensors} holds with high probability
using that partition where $X_i$ is the set of samples drawn from the $i$-th component of $\calM$.
Note that the second part of Condition \ref{cond:convergence-of-moment-tensors} holds with high probability,
so long as $n$ is a sufficiently large multiple of $d/\gamma^2$ by the VC-Theorem~\cite{DL:01}.
In particular, if we think of samples as being drawn from $\R^d\times [k]$,
where the second coordinate denotes the component that the sample was drawn from,
the second part of Condition \ref{cond:convergence-of-moment-tensors} says that the empirical probability
of any event $H\times \{i\}$ is correct to within additive error $\gamma$. It is easy to see and well-known
that the class of such events has VC-dimension $O(d)$, from which the desired bound follows.

For the first part of Condition \ref{cond:convergence-of-moment-tensors}, we claim that it holds with high probability
so long as $n\geq k t^{Ct} d^{t}/\gamma^3$. To prove this, we show it separately for each $i$ with $w_i \geq \gamma$
(as otherwise there is nothing to prove) and take a union bound.
As Condition \ref{cond:convergence-of-moment-tensors} is invariant under affine transformations,
we may perform an invertible affine transformation so that $\mu_i=0$ and $\Sigma_i$
is the projection onto the first $d'$ coordinates, for some $d'$. It is clear that only the first $d'$ coordinates
of any element of $X_i$ will be non-zero. We claim that the first part of our condition will follow for a given $m$,
so long as $\left| |X_i|/n-w_i \right|\leq \gamma w_i$ (which holds with high probability if $n\gg \log(k)/\gamma^3$),
and
\begin{equation}\label{norm bound equation}
\Norm{\E_{x\in_u X_i}[x^{\otimes m}] - \E_{x\sim \cN(0,I_{d'})}[x^{\otimes m}]}_F^2 \leq \gamma^2 \;,
\end{equation}
as $\frac{1}{n}\sum_{x\in X_i} \langle v, x-\mu_i \rangle^m =
w_i(1\pm \gamma)\langle\E_{x\in_u X_i}[x^{\otimes m}],v^{\otimes m}\rangle.$
It is easy to see that each entry of the tensor on the left hand side of Equation \eqref{norm bound equation}
has mean $0$ and variance $m^{O(m)}/|X_i|$, and thus the expected size of the left hand side
is $m^{O(m)}d^m/|X_i|.$ Then, when $n \geq k^{Ck} d^{4k}/\gamma^3$ for a sufficiently large constant $C$,
all parts of our condition hold with high probability. This completes the proof.
\end{proof}

\subsection{Omitted Proofs from Section~\ref{sec:full-algo-analysis}}

\begin{lemma}[Frobenius Distance to TV Distance, Lemma \ref{lem:frobenius_to_tv}, restated]
Suppose $\cN(\mu_1,\Sigma_1),\cN(\mu_2,\Sigma_2)$ are Gaussians with $\norm{\mu_1-\mu_2}_2 \le\delta$ and
$\norm{\Sigma_1-\Sigma_2}_F\le\delta$. If the eigenvalues of $\Sigma_1$ and $\Sigma_2$ are at least $\lambda>0$, then
$$
\dtv(\cN(\mu_1,\Sigma_1),\cN(\mu_2,\Sigma_2))=O(\delta/\lambda) \;.
$$
\end{lemma}

\begin{proof}
By Fact~\ref{GaussianTVFact}, we have
$$
\dtv\Paren{ \calN(\mu_1,\Sigma_1), \calN(\mu_2,\Sigma_2) } = \bigO{ \Paren{ (\mu_1-\mu_2)^\top \Sigma_1^{-1} (\mu_1-\mu_2)}^{1/2}+\|\Sigma_1^{-1/2}\Sigma_2 \Sigma_1^{-1/2}-I \|_F )}.
$$
Then the first term is
$\Iprod{\mu_1-\mu_2,\Sigma_1^{-1} (\mu_1-\mu_2)}^{1/2}\le (\norm{\Sigma_1^{-1}}_{\mathrm{op}} \norm{\mu_1-\mu_2}_2^2)^{1/2}\le\delta/\sqrt{\lambda}$.
The second term is
\begin{align*}
\|\Sigma_1^{-1/2}\Sigma_2 \Sigma_1^{-1/2}-I \|^2_F
&=\norm{\Sigma_1^{-1/2}(\Sigma_1-\Sigma_2)\Sigma_1^{-1/2}}^2_F\\
&=\tr\left(\left(\Sigma_1^{-1/2}(\Sigma_1-\Sigma_2)\Sigma_1^{-1/2}\right)^2\right)\\
&\le\tr\left((\Sigma_1-\Sigma_2)^2\right)(1/\lambda)^2\\
&\le(\delta/\lambda)^2.
\end{align*}
Thus,
$$
\dtv(\cN(\mu_1,\Sigma_1),\cN(\mu_2,\Sigma_2))=O(\delta/\sqrt{\lambda}+\delta/\lambda)=O(\delta/\lambda) \;.
$$
\end{proof}


\begin{lemma}[Component Moments to Mixture Moments, Lemma \ref{lem:component_to_mixture} restated]
Let $\calM = \sum_{i \in [k]} w_i \calN(\mu_i, \Sigma_i)$ be a $k$-mixture such that $w_i\geq \alpha$, 
for some $0 < \alpha <1$, and $\calM$ has mean $\mu$ and covariance $\Sigma$ \nnnew{and for all $i \neq j \in [k]$, 
$\Norm{\Sigma^{\dagger/2}\Paren{\Sigma_i - \Sigma_j}\Sigma^{\dagger/2}}_F \leq 1/\sqrt{\alpha}$}. 
Let $X$ be a multiset of $n$ samples satisfying Condition \ref{cond:convergence-of-moment-tensors} 
with respect to $\calM$  with parameters $(\gamma, t)$, for $0<\gamma<(dk/\alpha)^{-ct}$, 
for a sufficiently large constant $c$, and $t\in\N$. Let $\calD$ be the uniform distribution over $X$. 
Then, $\calD$ is $2t$-certifiably $(c/\alpha)$-hypercontractive \nnnew{and for $d \times d$-matrix-valued indeterminate $Q$, $\sststile{2}{Q}\Set{\E_{\calM} \Paren{x^{\top}Qx-\E_{\calM}x^{\top}Qx}^2 \leq \bigO{1/\alpha} \Norm{\Sigma^{1/2} Q\Sigma^{1/2}}_F^2}$}.
\end{lemma}
\begin{proof}
First, since $\calM$ is a $k$-mixture of Gaussians with minimum mixing weight $w_{\min} \geq \alpha$, it follows from Corollary \ref{cor:hypercontractivity_of_mixtures} that $\calM$ is $t$-certifiably $(4/\alpha)$ hypercontractive. Further, since $X$ satisfies Condition \ref{cond:convergence-of-moment-tensors} with parameters $(\gamma, t)$, it follows from Lemma \ref{lem:covergence_affine_invariant} that the set $X' = \{ \Sigma^{\dagger/2}(x_i - \mu) \}_{x_i \in X}$ also satisfies  Condition \ref{cond:convergence-of-moment-tensors} with parameters $(\gamma, t)$ w.r.t. $\calM' = \sum_{i \in [k]} w_i \calN\Paren{ \Sigma^{\dagger/2}(\mu_i - \mu), \Sigma^{\dagger/2} \Sigma_i \Sigma^{\dagger/2} }$. Since $\Norm{\Sigma^{\dagger/2} \Sigma_i \Sigma^{\dagger/2}}_{\textrm{op}} \leq \bigO{1/\alpha}$, it follows from Lemma \ref{moment closeness lemma} that for all $m \leq t$, $\Norm{\E_{x\in_u X'}[x^{\otimes m}] - \E_{x\sim M'}[x^{\otimes m}]}_F^2 \leq \gamma^2 m^{O(m)} d^m (1/\alpha)^m $. Since $\gamma < (dk/\alpha)^{-O(t)}$, it follows from Fact~\ref{fact:moments-to-analytic-properties} that $X$ is $2t$-certifiably $(c/\alpha)$-hypercontractive. 

\nnnew{ 
By assumption, 
for all $i\neq j \in [k]$, we have that 
$\Norm{\Sigma^{\dagger/2}\Paren{\Sigma_i - \Sigma_j} \Sigma^{\dagger/2}}_F \leq 1/\sqrt{\alpha}$.
We can now apply Lemma \ref{fact:mean-variance-subgaussian-arbitrary-covariance}
to obtain 
\[
\sststile{2}{Q} \Set{ \expecf{x\sim \calM}{ \Paren{ x^\top Qx-\expecf{x \sim \calM}{x^\top Qx } }^2 } \leq \bigO{1/\alpha} \Norm{\Sigma^{1/2}Q\Sigma^{1/2}}_F^2 }\mper
\]
Therefore, it follows from Fact \ref{fact:moments-to-analytic-properties} that since $X$ satisfies Condition \ref{cond:convergence-of-moment-tensors} with parameters $(\gamma, t)$, the uniform distribution $\cD_X$ on $X$, $\sststile{2}{Q}\Set{ \E_{x \sim \cD_X}(x^{\top}Qx-\E_{x \sim \cD_X}x^{\top}Qx)^2 \leq \bigO{1/\alpha} \Norm{\Sigma^{1/2}Q\Sigma^{1/2}}_F^2}$. }
\end{proof}

%% file: bitcomplexity.tex

\section{Bit Complexity Analysis} \label{app:bit-comp}

Here we address numerical issues related to our computation. We begin wth the assumption that the eigenvalues of our covariance matrices are bounded below.
\begin{lemma}\label{basicBitComplexityLemma}
If $\mathcal{M} = \sum_{i=1}^k w_i G_i$ is a mixture of Gaussians $G_i$ where each $G_i$ has mean and covariance of norm at most $2^b$ for some positive integer $b$ and each $G_i$ has covariance matrix whose eigenvalues are bounded below by some $\lambda>0$. Let $\mathcal{M'}$ be an $\eps$-corruption of $\mathcal{M}$ whose outputs are bounded by $2^{O(b)}$. Let $N$ be a sufficiently large polynomial in $d^k/\eps$ and let $\eta$ be $\lambda$ divided by a sufficiently large polynomial in $2^b d/\eps$ (where sufficiently large is degree $O(1)$). Then if our algorithm is given $N$ i.i.d. samples from $\mathcal{M'}$ with each of their coordinates rounded to a nearby multiple of $\eta$ (by which we mean one of the two closest), then our algorithm runs in time $\poly(N,b,\log(1/\eta))$ and with high probability returns a list of mixtures of Gaussians $X_i$ with at least one of the $X_i$ $\poly_k(\eps)$-close to $\mathcal{M}$ in parameter distance.
\end{lemma}
\begin{proof}
This follows from noting firstly that with high probability the any subset of the rounded samples will have moments $\lambda/\poly(d/\eps)$-close to their moments before rounding. This means that with high probability these rounded samples will satisfy Condition~\ref{cond:convergence-of-moment-tensors}. This means that our algorithm satisfies the necessary correctness guarantees. Furthermore, given that our samples now all have bounded bit complexity, it is easy to see that the runtime of our algorithm is polynomial in $N$ and the bit complexity.
\end{proof}

More generally, as long as the parameters of the components of our mixture can be expressed with bounded bit complexity, we can prove a similar result, without needing any lower bound on the covariances.
\begin{theorem}
Let $\mathcal{M} = \sum_{i=1}^k w_i G_i$ be a mixture of Gaussians where the $G_i$ are Gaussians whose means and covariance matrices can all be written with coefficients given by rational numbers with bit complexity at most $b$ for some integer $b$. Let $\mathcal{M'}$ be an $\eps$-corruption of $\mathcal{M}$ so that with probability $1$ the returned points have size $2^{O(b)}$. Let $N$ be a sufficiently large polynomial in $d^k/\eps$. Then there exists an algorithm that given $b$ bit-oracle access to these samples runs in time $\poly(N,b)$ and with high probability returns a mixture of Gaussians $X$ so that $\dtv(X,\mathcal{M}) < \poly_k(\eps)$.
\end{theorem}
\begin{proof}
We begin by showing that we can find a list of hypotheses at least one of which is close. It is then straightforward to show that we can run a tournament over these hypotheses to find a specific one that works. We also assume for simplicity that each $w_i$ is at least $3\eps$.

We begin by setting $\lambda$ to be $2^{-b\cdot d^{kC}}$ for a sufficiently large constant $C$. By adding each sample to a random sample from $N(0,\lambda I)$, we can produce samples from $\tilde{\mathcal{M}}'$, and $\eps$-corruption of $\tilde{\mathcal{M}} =\sum_{i=1}^k w_i \tilde{G_i}$ where $\tilde{ G_i}$ is the convolution of $G_i$ with $N(0,\lambda I)$. Note that $\tilde G_i$ is a Gaussian whose covariance has eigenvalues at least $\lambda$. Furthermore, if the covariance matrix of $G_i$ is non-singular, the smallest eigenvalue of the covariance matrix must be at least $2^{O(b\cdot d)}$, and therefore $\dtv(G_i,\tilde G_i)<\eps$.

Since the eigenvalues of the components of $\tilde{\mathcal{M}}$ are bounded below, we can apply Lemma \ref{basicBitComplexityLemma} to our samples from $\tilde{\mathcal{M}}'$ rounded to an appropriate accuracy $\eta$, and in $\poly(N,b)$-time obtain a list of hypothesis mixtures at least one of which is (with high probability) close to $\tilde{\mathcal{M}}$ in total variation distance.

If the covariances of all of the $G_i$ with weights more than some sufficiently large $\poly_k(\eps)$ are all non-singular, then one of these hypotheses will be close to $\mathcal{M}$. Otherwise, there must be some $i$ for which $w_i$ is relatively large and for which $G_i$ has singular covariance matrix. In particular, there must be an integer vector $v$ with bit complexity $O(bd)$ in the kernel of the covariance matrix of $G_i$. The hypothesis mixture $X$ that is close to $\tilde{\mathcal{M}}$ in parameter distance must contain some component close to $\tilde G_i$. Since $\tilde G_i$ has covariance matrix $\tilde \Sigma_i = \lambda I + \Sigma_i$ where $\Sigma_i$ is the covariance matrix of $G_i$. We note that $\tilde \Sigma_i$ will have an eigenvalue of $\lambda$ and that therefore, our close hypothesis will have an eigenvalue at most $2\lambda$.

If any of our returned hypotheses have any component with a covariance matrix $\Sigma$ which has any eigenvalue less than $2\lambda$, we do the following. We consider the quadratic form on integer vectors $v$ defined by
$$
Q(v) = v^T \Sigma v + \sqrt{\lambda}|v|_2^2.
$$
We note that if this $\Sigma$ is close in parameter distance to a singular $\tilde \Sigma_i$ where $\Sigma_i$ had a null-vector $v$ of norm $2^{O(bd)}$, then for that same value of $v$ we will have that $Q(v) < \lambda^{1/4}.$ Using the Lov\'asz local lemma in \cite{erdHos1973problems}, we can find a $v$ so that $Q(v)$ is within a $2^{O(d)}$-factor of the minimum possible value over all non-zero, integer vectors $v$. If for this $v$, $Q(v) > 2^{\Omega(d)}\lambda^{1/4}$, we know that the hypothesis in question is not close to $\tilde{\mathcal{M}}$ in parameter distance and can be ignored. On the other hand, any $v$ with $Q(v)$ this small must have $|v| < 2^{O(d)} \lambda^{-1/2}$ and $v^T \Sigma v < 2^{O(d)} \lambda^{1/4}$. Note that the projection of $v$ onto the $\ker(\Sigma_i)^\perp$ is either zero or has magnitude at least $2^{O(bd)}$. In the latter case, it would need to be the case that $Q(v)$ is substantially larger. Thus, if such a hypothesis is close to $\tilde{\mathcal{M}}$ in parameter distance, then $v$ is in the kernel of some $\Sigma_i$.

If our algorithm finds some $v$ for some hypothesis, we then compute $v\cdot x$ to error $\lambda$ for each of our samples $x$. If $\mathcal{M}$ really has a component with $v$ in the kernel of its covariance matrix, all of the $x$'s taken from this component will have $v\cdot x$ the same. This means that at least a $(3/2)\eps$ fraction of our samples $x$ will have $v\cdot x$ within $\lambda$ of each other. Note that if $v$ is not in the kernel of any covariance matrix of any $G_i$ than $\var{v\cdot G_i}$ will be at least $2^{O(db)}$ for each $i$, and with high probability we will not find this many close samples.

To summarize, if our algorithm applies this procedure to every component of every hypothesis and does not find such a $v$, then it cannot be the case that $\mathcal{M}$ contains any components of weight more than $\poly_k(\eps)$ that are singular, and thus one of our original hypotheses must be close in total variational distance. We can then run a tournament to find a single one that is close. Otherwise, if we find such a $v$ for which many points do have $v\cdot x$ close by, then $v$ must be a null vector of the covariance matrix of some $G_i$. Furthermore, all of the samples within $\lambda$ of this common value of $v\cdot x$, with high probability are either errors or come from components contained in some lower dimensional subspace. We can determine what this subspace is by noting that it is defined by $v\cdot x = q$ for some rational number $q$ with bit-complexity at most $O(bd)$ and using continued fractions on a good numerical approximation of $q$ in order to determine its true value. Our algorithm can then recurse on the points in this subspace (a mixture of Gaussians in a lower dimensional space) and on the remaining points (which are from a mixture of fewer Gaussians), and return an appropriate mixture of the results.
\end{proof}